\pdfoutput=1
\documentclass[11pt,a4paper]{scrbook}
  \KOMAoption{headings}{optiontoheadandtoc}
  \KOMAoption{listof}{toc,leveldown}
  \KOMAoption{bibliography}{toc}
  \addtokomafont{subparagraph}{\footnotesize}

  \usepackage[pdf,greek,thm-chap]{preamble}
  \usepackage{extras}
  \usepackage{meta}

  \makeatletter
    \hearley@captionsetup{format=plain}
  \makeatother

\endofdump

  \def\fullName{Hannah Amelie Earley}
  \def\shortName{Hannah Earley}
  \def\myEmail{h.earley@damtp.cam.ac.uk}
\begin{document}
\hyphenation{Brown-ian}

\frontmatter
\begingroup

\makeatletter
\def\startcenterpage{%
  \ifcsname orig@twoside@local\endcsname%
    \PackageError{thesis}{Don't nest centerpages! KOMA doesn't like it...}{}%
  \fi%
  \def\orig@twoside{false}%
  \if@twoside\def\orig@twoside{true}\fi%
  \KOMAoptions{twoside=false}%
  \begingroup%
  \let\orig@twoside@local\orig@twoside}
\def\stopcenterpage{%
  \global\let\orig@twoside\orig@twoside@local%
  \endgroup%
  \KOMAoptions{twoside=\orig@twoside}}
\makeatother

\newenvironment{fmpage}[1]{%
\begin{titlepage}%
  \phantomsection%
  \null\vspace{3\baselineskip}%
  \begin{center}\LARGE\textsf{\textbf{%
   #1%
  }}\end{center}%
  \addcontentsline{toc}{section}{#1}%
  \vspace{2\baselineskip}%
}{\end{titlepage}}

\def\fmlistpage{%
  \cleardoubleoddpage%
  \thispagestyle{\titlepagestyle}}

\long\def\thesisDeclaration{
  \begin{fmpage}{Declaration}
    This thesis is the result of my own work and includes nothing which is the outcome of work done in collaboration except as declared in the Preface and specified in the text. It is not substantially the same as any that I have submitted, or, is being concurrently submitted for a degree or diploma or other qualification at the University of Cambridge or any other University or similar institution except as declared in the Preface and specified in the text. I further state that no substantial part of my thesis has already been submitted, or, is being concurrently submitted for any such degree, diploma or other qualification at the University of Cambridge or any other University or similar institution except as declared in the Preface and specified in the text. It does not exceed the prescribed word limit for the relevant Degree Committee.

    Variants of the core chapters of this thesis, \Cref{chap:revi,chap:revii,chap:reviii,chap:aleph}, have been posted as preprints~\cite{earley-parsimony-i,earley-parsimony-ii,earley-parsimony-iii,earley-aleph}. This thesis is my own work, though I will exercise nosism throughout.
    This work was supported by the Engineering and Physical Sciences Research Council, project reference 1781682~\cite{wearley-ukri}.

    \vfill
    \begin{flushright}
      \shortName\rlap\footnotemark \\
      December 2020
    \end{flushright}

    \footnotetext{\email<\myEmail>, \orcidFull{0000-0002-6628-2130}}
  \end{fmpage}
}

\makeatletter
  \if@hearley@deposit@
    \addtocounter{page}{-2}
    \includepdf[pages={1},offset={3mm 0}]{declaration.pdf}
    \shipout\null
    \addtocounter{page}{1}
  \fi
\makeatother

\startcenterpage
\begin{titlepage}
  \def\camfullstop{\rlap{\kern.5pt\raisebox{1.5pt}{.}}}
  \newcommand\camShield[2][]{\begingroup%
    \newdimen\shieldWidth%
    \shieldWidth=#2\relax%
    \begin{tabular}{c}
      \includegraphics[width=\shieldWidth]{_cam-shield#1.pdf}\\[.5em]
      \includegraphics[width=\shieldWidth]{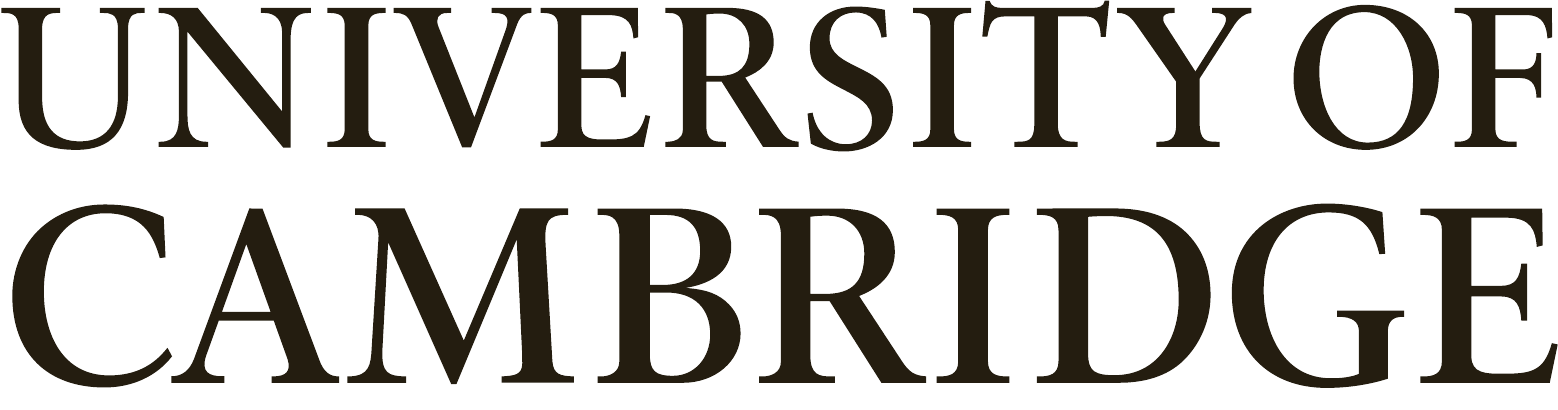}
    \end{tabular}\endgroup}
  \def\camShieldBW#1{\camShield[-bw]{#1}}
  {\centering
  \parindent=0pt

  \textsf{\textbf{\linespread{1.35}
    \Huge\thesisTitle\\[1.5em]
    \Large`\thesisSubtitle'\\[3em]}}
  {\large\fullName}
  \\[1.5em]
  {\large March 2021}
  \vfill
  \textit{\scriptsize This thesis is submitted for the degree of}\\
  {Doctor of Philosophy,}
  \\[1em]
  \textit{\scriptsize and was supervised by}\\
  {Gos Micklem}
  \\[1em]
  \textit{\scriptsize in the}\\
  {Department of Applied Mathematics and Theoretical Physics,}
  \\[1em]
  \textit{\scriptsize and}\\
  {Trinity Hall,}
  \\[1em]
  \textit{\scriptsize of the}\\[.5em]%
  \includegraphics[width=2.5cm]{_cam-text.pdf}\camfullstop
  \\[.5em]%
  \includegraphics[width=2.5cm]{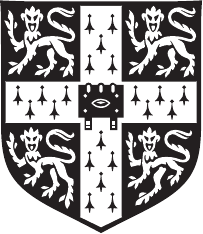}\\
  \vspace{-1cm}}

  {\color{white}.}\vspace{-\baselineskip}

\end{titlepage}
\stopcenterpage

\begin{titlepage}
  \centering
  \parindent=0pt
  \null\vfill\vfill
  \textit{\scriptsize This thesis was also advised by}\\
  {Stephen Eglen,}
  \\[1em]
  \textit{\scriptsize and was examined internally by}\\
  {Sebastian Ahnert}
  \\[1em]
  \textit{\scriptsize and externally by}\\
  {Michael P.\ Frank}
  \\[1em]
  \textit{\scriptsize on the}\\
  {30\textsuperscript{th} of April, 2021.}
  \vfill
\end{titlepage}

\makeatletter\if@hearley@deposit@
  \thesisDeclaration
  \begin{titlepage}%
    \phantomsection%
    \null\vspace{3\baselineskip}%
    \enlargethispage{4\baselineskip}%
    \begin{center}\parindent=0pt
      \textsf{\textbf{\linespread{1.35}
        \LARGE\thesisTitle\\[1em]
        \large`\thesisSubtitle'\\[1em]}}
      \shortName\\[2.5em]
    \end{center}%
    \addcontentsline{toc}{section}{Abstract}%
    \thesisAbstract
  \end{titlepage}
\fi\makeatother

\begin{titlepage}
  \centering
  \parindent=0pt
  \def\nl{\\[.7em]}
  \null\vfill\textit{\begin{tabular}{l}
    \hspace{10pt}to Ann, \nl
    \hspace{15pt}reunited with Dennis, \nl
    \hspace{0pt}and to Jim, \nl
    \hspace{8pt}who we and Patsy yet miss
  \end{tabular}}\vfill
\end{titlepage}

\begingroup
  \parindent=0pt
  \def\nl{\\[.7em]}
  \cleardoubleevenpage
  \thispagestyle{empty}
  \phantomsection
  \addcontentsline{toc}{section}{Acknowledgements}%
  \null\vspace{3.5\baselineskip}%
  \def\lineStyle{\textit}%
  \def\ackialign{{l}}%
  \def\lineNo#1{\lineStyle{#1}}%
  \def\linePunc#1#2{\lineStyle{#1#2}}%
  \def\btabu{\begin{tabular}}
  \def\etabu{\end{tabular}}
  \begin{center}
    \makeatletter\if@twoside%
      \def\ackialign{{r}}%
      \def\lineNo#1{\lineStyle{#1}}%
      \def\linePunc#1#2{{\lineStyle{#1}\rlap{\lineStyle{#2}}}\kern0.2ex}%
    \fi\makeatother%
    \expandafter\btabu\ackialign
      \LARGE\textsf{\textbf{Acknowledgements}}\\[2\baselineskip]
      \linePunc{Gos, without whom none of this would be possible};\nl
      \lineNo{your guidance, identical sleep schedule, and tangents}\nl
      \linePunc{have got me through this path ig-noble}.\nl
      \linePunc{P.S. I forgive you for trying to kill me with peanuts}.\nl
      \nl
      \linePunc{Mum and Dad, you gave me life},\nl
      \linePunc{and helped me through my toil and strife},\nl
      \linePunc{not to mention my sisters three},\nl
      \linePunc{whose love and kinship supported me}.\nl
      \nl
      \linePunc{My office-mate Pete(r)},\nl
      \lineNo{I miss our late-night pizza}\nl
      \linePunc{and basketball sessions};\nl
      \linePunc{may MIT sustain your passions}.\nl
      \nl
      \linePunc{Swaraj, though you moved to Oxford},\nl
      \linePunc{your dependability never wavered};\nl
      \linePunc{your friendship is immeasuarable},\nl
      \linePunc{and willingness to read my drafts inexplicable}.
    \etabu
  \end{center}\vfill%
  \clearpage
  \thispagestyle{empty}
  \null\vspace{3.5\baselineskip}%
  \begin{center}
    \expandafter\btabu\ackialign
      \LARGE\textsf{\textbf{Acknowledgements}}\\[2\baselineskip]
      \linePunc{Georgeos, in but a year of your acquaintance},\nl
      \linePunc{fast friends we have become};\nl
      \linePunc{though I have exploited it for articular comments},\nl
      \linePunc{fast friends I hope we go on}.\nl
      \nl
      \linePunc{Marianne, who accompanied me through much of my journey};\nl
      \linePunc{even though our relationship may not have survived the PhD},\nl
      \linePunc{I may not have survived the PhD if not for you},\nl
      \linePunc{for which you have my gratitude}.\nl
      \nl
      \linePunc{My office-mate subsequent Jiaee},\nl
      \linePunc{who cyclically accompanied me to Ely}:\nl
      \linePunc{about the distance, you and Eleanor lied to me},\nl
      \linePunc{nonetheless, around your upbeat attitude is fun to be}.\nl
      \nl
      \linePunc{Finally, these acknowledgements would be incomplete},\nl
      \lineNo{wihout mentioning the bodies whose funding replete}\nl
      \lineNo{upon me bestowed; DAMTP, Trinity Hall}\nl
      \linePunc{and EPSRC, I thank you all}.
    \etabu
  \end{center}\vfill%
\endgroup

\makeatletter\if@hearley@deposit@\else
  \thesisDeclaration
  \begin{fmpage}{Abstract}\thesisAbstract\end{fmpage}
\fi\makeatother

\cleardoubleoddpage%
  \phantomsection%
  \addcontentsline{toc}{section}{Table of Contents}%
  \tableofcontents
\fmlistpage\listoffigures
\fmlistpage\listoflistings

\endgroup

\mainmatter
\renewcommand{\thechapter}{\Roman{chapter}}
\begingroup
\chapter{Introduction}
\label{chap:intro}

  \def\langreview#1{\subpara{(#1)}}

In recent years, unconventional forms of computing ranging from molecular computers made out of DNA to quantum computers have started to be realised. Not only that, but they are becoming increasingly sophisticated and have a lot of potential to influence the future of computing. In this work, we will look at the intersection of molecular computing and another interesting class of unconventional computer: that of reversible computers. 

\begin{wrapfigure}{O}{.5\linewidth}
  \centering
  \includegraphics[width=\linewidth]{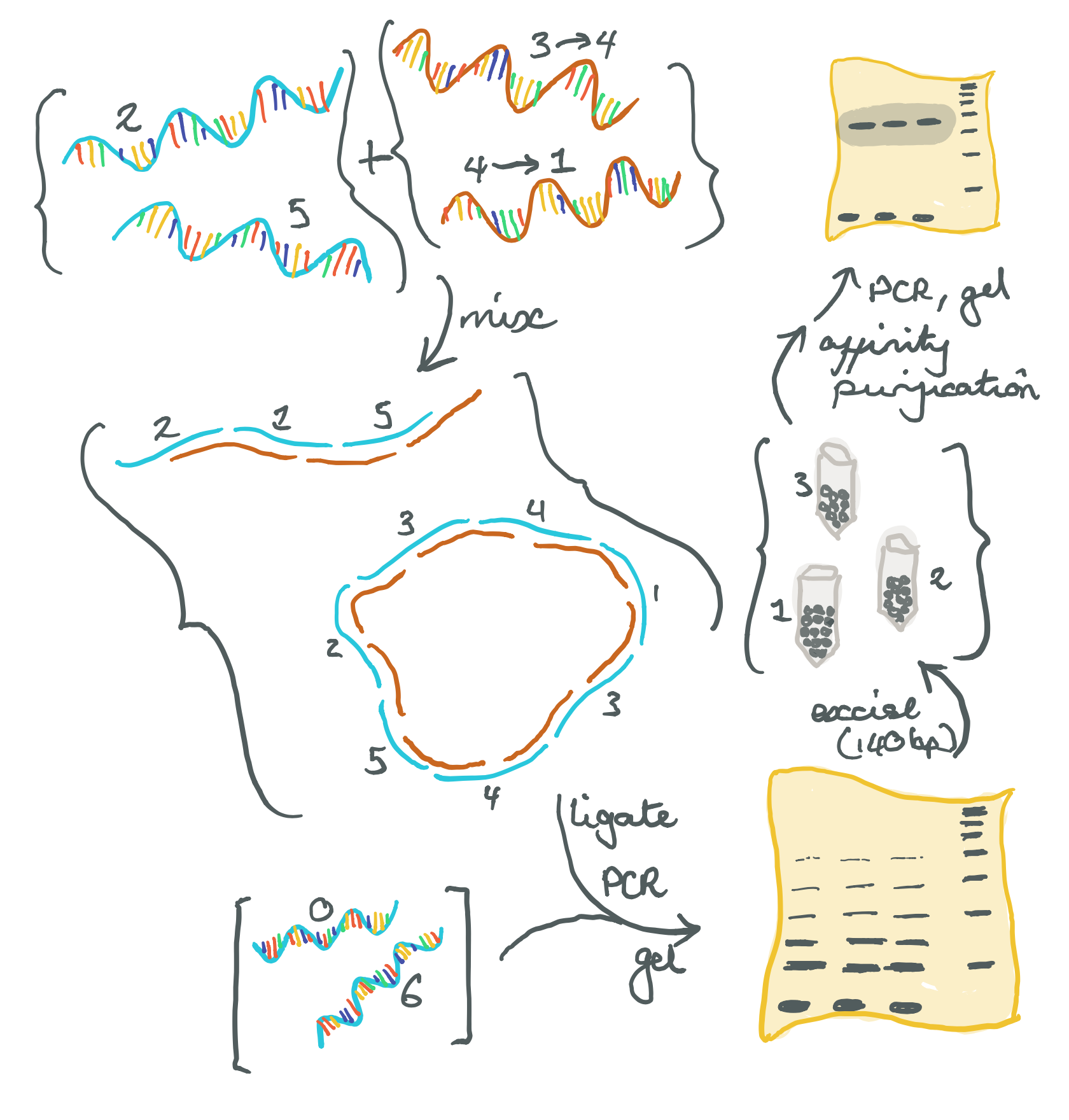}
  \caption{An illustration of Adleman's DNA scheme for solving the Hamiltonian path problem.}
  \label{fig:hampath}
\end{wrapfigure}
\para{Molecular Computation} The field of Molecular Computation began in 1994 with Len Adleman's proof-of-concept use of DNA~\cite{adleman} in the solution of the classic Hamiltonian path problem, as illustrated in \Cref{fig:hampath}. DNA and other nucleic acids proved an apt choice for building molecular/chemical computers, due to its information bearing capacity and the strong preference that single strands show for binding to complementary sequences, not to mention the low cost with which it can be synthesised and analysed as well as the large body of knowledge available from decades of biochemical research. Inspired by Adleman's paper, a community of researchers designing and building DNA computers arose in concert with the DNAX (and shortly thereafter, the complementary FNANOX) series of conferences\footnote{\url{https://isnsce.org/conferences/}}.

DNA, or deoxyribonucleic acid, is a family of polymeric molecules that is employed by all known organisms to store their genomic information, amongst other uses. Each molecule can be characterised by its sequence: a string of the monomers adenine (\dna{A}), cytosine (\dna{C}), guanine (\dna{G}) and thymine (\dna{T}). These monomers are each bound to a deoxyribose sugar molecule, which are joined together by phosphodiester bonds. Whilst this ability to store information through the arbitrary combination of these monomers is advantageous, the key to its success is the special hydrogen-bonding properties that these monomers possess. When two DNA strands are aligned adjacently, if adenine and thymine are opposite each other they will be perfectly positioned to form two favourable hydrogen bonds; similarly, if cytosine and guanine are opposite each other they will be perfectly positioned to form three favourable hydrogen bonds. Whilst other combinations may also form hydrogen bonds, the positions are suboptimal and thermodynamically unfavourable when compared to the pairs \dna{A-T} and \dna{C-G}. In nature, DNA is found almost exclusively in double-stranded form, whereby each strand has a sequence complementary to the other: if one strand has sequence \dna{GGCACAACGT} then the other will have sequence \dna{CCGTGTTGCA}. These complementary strands then \emph{hybridise} to one another. In fact, this is an oversimplification: in addition to their information content, DNA molecules have an intrinsic orientation, conventionally notated via the free 5' phosphoryl and 3' hydroxy groups on the terminal ribose sugars. DNA strands bind anti-parallel, and so the sequences should be written \dna{{5'}-GGCACAACGT-{3'}} and \dna{{5'}-ACGTTGTGCC-{3'}}. Schematically, this looks like
\begin{center}
\begin{tikzpicture}
  \node at  (0,1) (A0) {\dna{5'}};  \node at  (0,0)  (B0) {\dna{3'}};
  \node at  (1,1) (A1) {\dna{G}};   \node at  (1,0)  (B1) {\dna{C}}; \draw  (A1) --  (B1);
  \node at  (2,1) (A2) {\dna{G}};   \node at  (2,0)  (B2) {\dna{C}}; \draw  (A2) --  (B2);
  \node at  (3,1) (A3) {\dna{C}};   \node at  (3,0)  (B3) {\dna{G}}; \draw  (A3) --  (B3);
  \node at  (4,1) (A4) {\dna{A}};   \node at  (4,0)  (B4) {\dna{T}}; \draw  (A4) --  (B4);
  \node at  (5,1) (A5) {\dna{C}};   \node at  (5,0)  (B5) {\dna{G}}; \draw  (A5) --  (B5);
  \node at  (6,1) (A6) {\dna{A}};   \node at  (6,0)  (B6) {\dna{T}}; \draw  (A6) --  (B6);
  \node at  (7,1) (A7) {\dna{A}};   \node at  (7,0)  (B7) {\dna{T}}; \draw  (A7) --  (B7);
  \node at  (8,1) (A8) {\dna{C}};   \node at  (8,0)  (B8) {\dna{G}}; \draw  (A8) --  (B8);
  \node at  (9,1) (A9) {\dna{G}};   \node at  (9,0)  (B9) {\dna{C}}; \draw  (A9) --  (B9);
  \node at (10,1) (A10) {\dna{T}};  \node at (10,0) (B10) {\dna{A}}; \draw (A10) -- (B10);
  \node at (11,1) (A11) {\dna{3'}}; \node at (11,0) (B11) {\dna{5'}};
  \draw (A0) -- (A1) -- (A2) -- (A3) -- (A4) -- (A5) -- (A6) -- (A7) -- (A8) -- (A9) -- (A10) -- (A11);
  \draw (B0) -- (B1) -- (B2) -- (B3) -- (B4) -- (B5) -- (B6) -- (B7) -- (B8) -- (B9) -- (B10) -- (B11);
\end{tikzpicture}
\end{center}
but in reality the strands twist together into a double helix, the structure of which was found by Rosalind Franklin, James Watson and Francis Crick~\cite{wc-dna}. Eponymously, the characteristic complementary `base pairing' is known as Watson-Crick base pairing. The DNA computing motifs designed by Adleman and others make use of this Watson-Crick base pairing, but exploit partial hybridisation and the dynamics and interactions of single-stranded DNA.

\begin{figure}
  \begin{subfigure}{\linewidth}
    \centering
    \includegraphics[width=.9\linewidth]{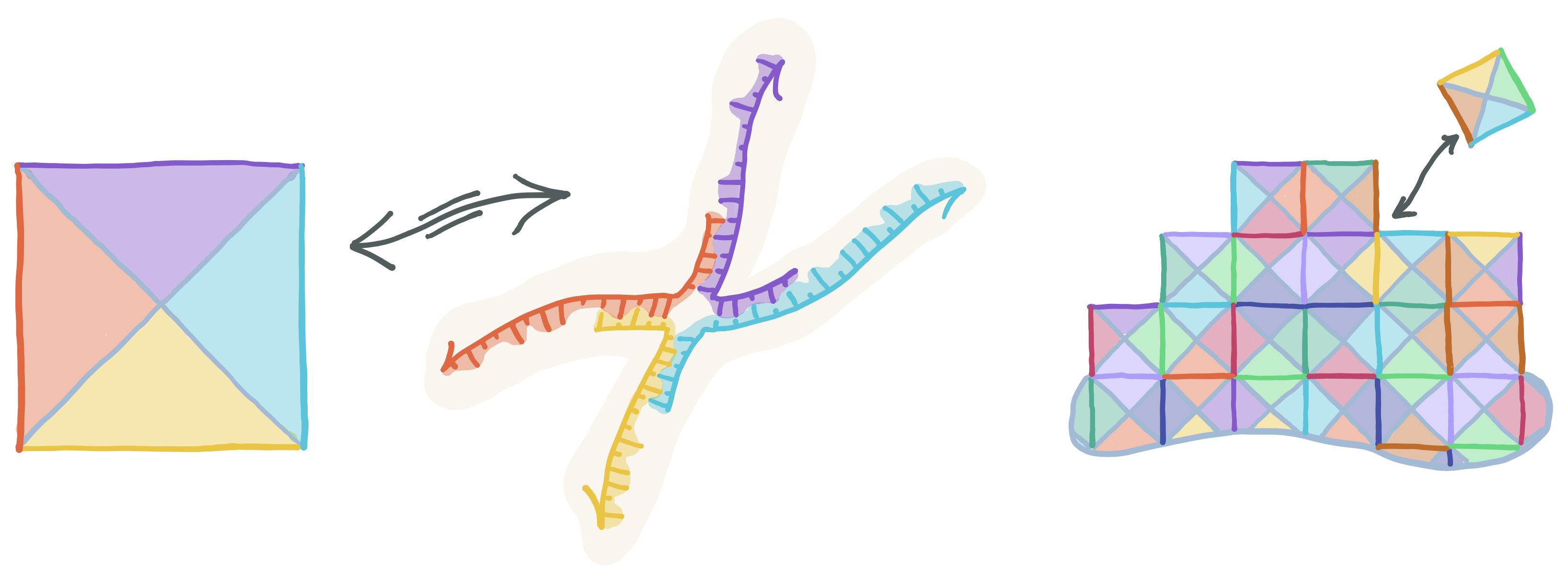}
    \caption{The canonical embodiment of the Tile Assembly Model: in the abstract, square `tiles' are designed with four sticky edges. The edges are assigned colours such that only edges of the same colour will stick to one another. By choosing the tile set, different patterns can be assembled, both deterministic and non-deterministic. In fact, it was shown by Hao Wang in the `60s~\cite{wang-tiles} that appropriately chosen tile sets can simulate a Turing Machine, and so arbitrary computation is possible. During his PhD, Erik Winfree showed how to realise these tilesets using DNA~\cite{winfree-tam}, per the illustration above. Although in practice the tiles adopt a non-square shape, the same general properties of Wang Tiles are observed. In fact, not only is the TAM Turing complete, but there even exists a universal tile set that can simulate any other~\cite{tam-univ}.}
    \label{fig:dna-principle-tam}
  \end{subfigure}
  \begin{subfigure}{\linewidth}
    \centering
    \includegraphics[width=.9\linewidth]{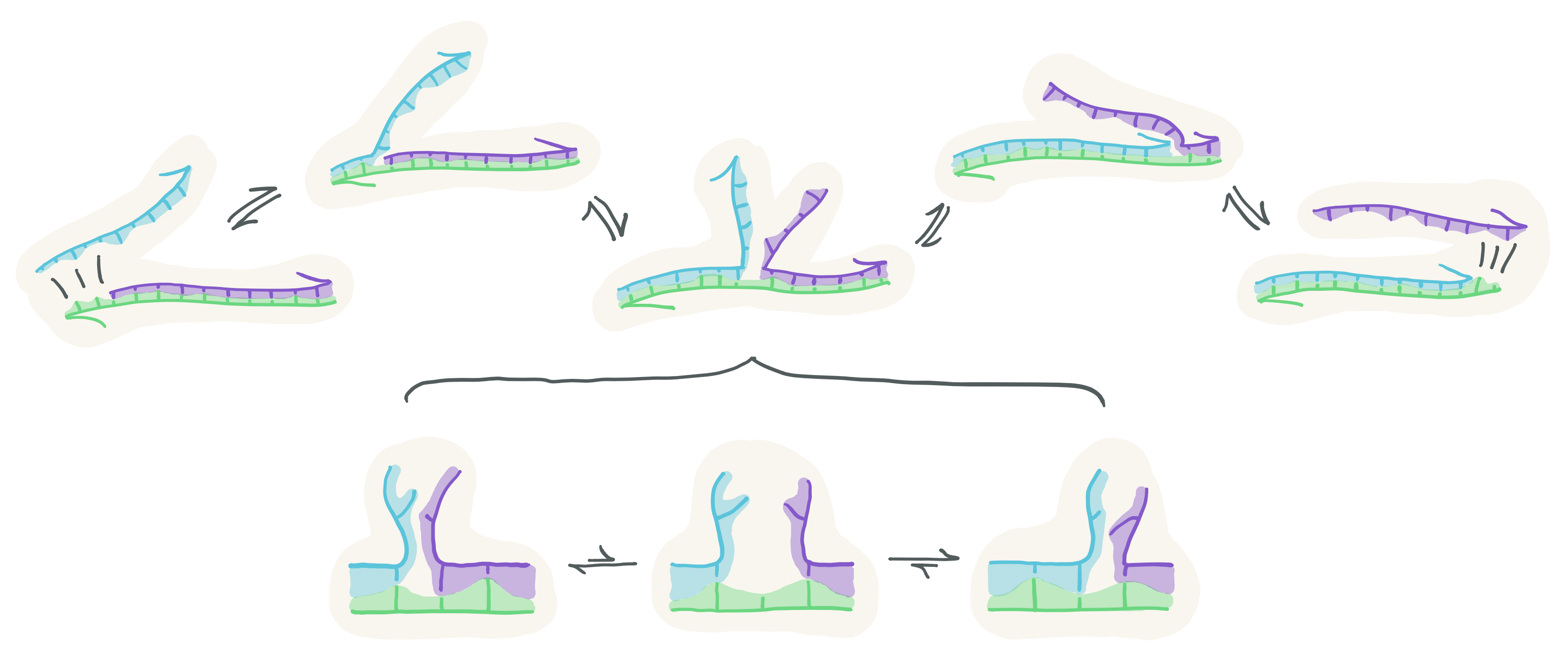}
    \caption{DNA-strand displacement schemes rely on the primitive illustrated here of branch migration, in which an incumbent strand $B$ (purple) is displaced by an incident strand $A$ (blue). We start with $B$ hybridised to a gate strand (green), except that the gate strand is slightly longer and has an exposed single-stranded `toehold' domain. This toehold is complementary to part of $A$, whereby $A$ is reversibly recruited to the gate strand. Moreover, the remainder of $A$ is complementary to much of the gate strand that $B$ is bound to, and so can compete with $A$ for hybridisation. In a thermal context, then, $B$ may transiently `fray' away from the gate strand at its ends, providing an opportunity for $A$ to partially hybridise in its place. A random walk ensues, in which the junction between $A$ and $B$ may migrate in either direction with roughly equal chance. If the branch fully migrates to the other end, then $B$ is left bound by only a short region: another toehold. As with $A$, $B$ can then reversibly attach or detach from the gate strand. If $A$ was slightly longer and bound the entirety of the gate strand, then the reaction could be (and routinely \emph{is}) made effectively irreversible. Sophisticated cascades can then be prepared whereby $B$ may interact with downstream gates, or where a gate may bind multiple strands. Other sophisticated behaviours can be realised, for example, by the use of common or orthogonal toeholds and other sequence domains. The culmination of this is that, up to a tunable probability of error, DSD supports universal computation~\cite{winfree-dna-crn-univ}.}
    \label{fig:dna-principle-dsd}
  \end{subfigure}
  \caption{Illustrations of the principles underlying two common schemes of DNA computation.}
  \label{fig:dna-principle}
\end{figure}

Whilst Adleman's proof of concept required manual actuation of each step of the computation, and was limited to NP-complete problems, schemes demonstrating universal computation and autonomous operation were soon developed. An example from the first of the DNAX series is the proposed DNA and restriction enzyme implementation of a Turing Machine of \textcite{rothemund-tm}. Whilst universal in principle, it was still not autonomous and has never been experimentally realised due to its complexity. Nevertheless, successful approaches and motifs employing both DNA and enzymes have been developed, such as the PEN toolbox of Yannick Rondelez's group~\cite{pen-toolbox} and the PER reaction of Jocelyn Kishi~\cite{per}. Arguably the most popular approaches in the field are enzyme-free approaches. Two famous examples are given by the Tile-Assembly Model (TAM), introduced by Erik Winfree~\cite{winfree-tam}, and DNA-strand displacement (DSD) schemes, first popularised by Georg Seelig et al~\cite{dsd}. These are illustrated in \Cref{fig:dna-principle}.

Some impressive feats have been achieved, from neural networks~\cite{qian-nn,cherry-nn} to cargo sorting robots~\cite{cargo-sorting} to reprogrammable iterative binary logic circuits~\cite{woods-21}. Nevertheless, DNA computing is not without its challenges. For example, DSD systems often exhibit spurious `leak' reactions, and optimisation and tricks are required to mitigate this~\cite{leakless-dsd}. Moreover, traditional DSD systems occupy the entire reaction volume, which can therefore only perform a single computation. Strictly speaking, one could perform multiple computations by using orthogonal DNA sequences such that a set of DNA species performing one computation doesn't interact with another performing a different computation; however, toeholds typically can't be much longer than 6 nucleotides in length to ensure reversibility of association, limiting the set of distinct toeholds to $4^6=4096$. Worse still, it transpires that DNA's sequence specificity is not absolute, and so only a fraction of this design space is useful: in fact, as a single mismatch or two does not prevent hybridisation but does affect the rate of hybridisation, it is often exploited for tuning kinetics~\cite{dsd-kin-i,dsd-kin-ii,dsd-kin-iii,dsd-kin-iv}. We have not even mentioned the influence of palindromic sequences and other motifs manifesting non-trivial secondary structures. This seriously limits the complexity that can be realised in these systems, because the CRNs that are compiled into DSD schemes (using, e.g., \tool{VisualDSD}~\cite{visual-dsd}) do not make use of complexation of CRN species, and therefore each sub-module in a program requires a fresh set of CRN species with their own unique DNA sequences---even if the sub-module is identical to another, just used in a different place in the program. Additionally, these DNA computers are analogue which significantly limits their computational complexity if one wants to be certain of the result~\cite{crn-probab}, although if one is comfortable with slightly less certainty then better complexity is possible~\cite{winfree-dna-crn-univ}. They are also very slow, taking on the order of several hours or even days to progress.
The \emph{tour-de-force} examples of neural networks~\cite{qian-nn} and (4-bit) square rooting~\cite{qian-sqrt}, for example, took on the order of \emph{10 hours} to converge to an answer.

If instead one makes use of localisation, as with the TAM, then much of these issues with DSD can be avoided~\cite{chatterjee-localised}: localising DSD species to a surface allows `circuits' to be constructed which are fast, modular, and can reuse sequence domains. The problem is that constructing these is substantially more difficult than free-solution DSD, and that gate strands are `consumed' during computation making these one-shot reactions. Free-solution DSD \emph{also} consumes its gate strands, but at least in this case it is possible (if non-trivial) to continually supply fresh gate and fuel species and remove waste complexes. The TAM does not get off completely unscathed either, and methods to suppress errors to mitigate its share of problems have also been developed~\cite{tam-err}.

This is not to say that computing with molecules such as DNA is intrinsically flawed, but the difficulties in overcoming these issues has left many in the field pessimistic about initial beliefs that DNA could be used to perform arbitrary computation in an effective manner. Instead, the general consensus is that DNA computing may be best used to perform small logical calculations as part of a synthetic biological system. My view is more optimistic, in that I believe that an as-yet undiscovered DNA computing motif exists that will solve all these problems simultaneously. Whilst perhaps the existence of such a panacea is wishful thinking, I believe that the abstract molecular schemes discussed in \Cref{app:seq-klona,sec:ex1} lend evidence to this view, and developing these further will be the subject of future work.

\para{Reversible Computation}
\label{app:rev-comp}

\begin{figure}[p]
  \centering
  \includegraphics[width=.7\textwidth]{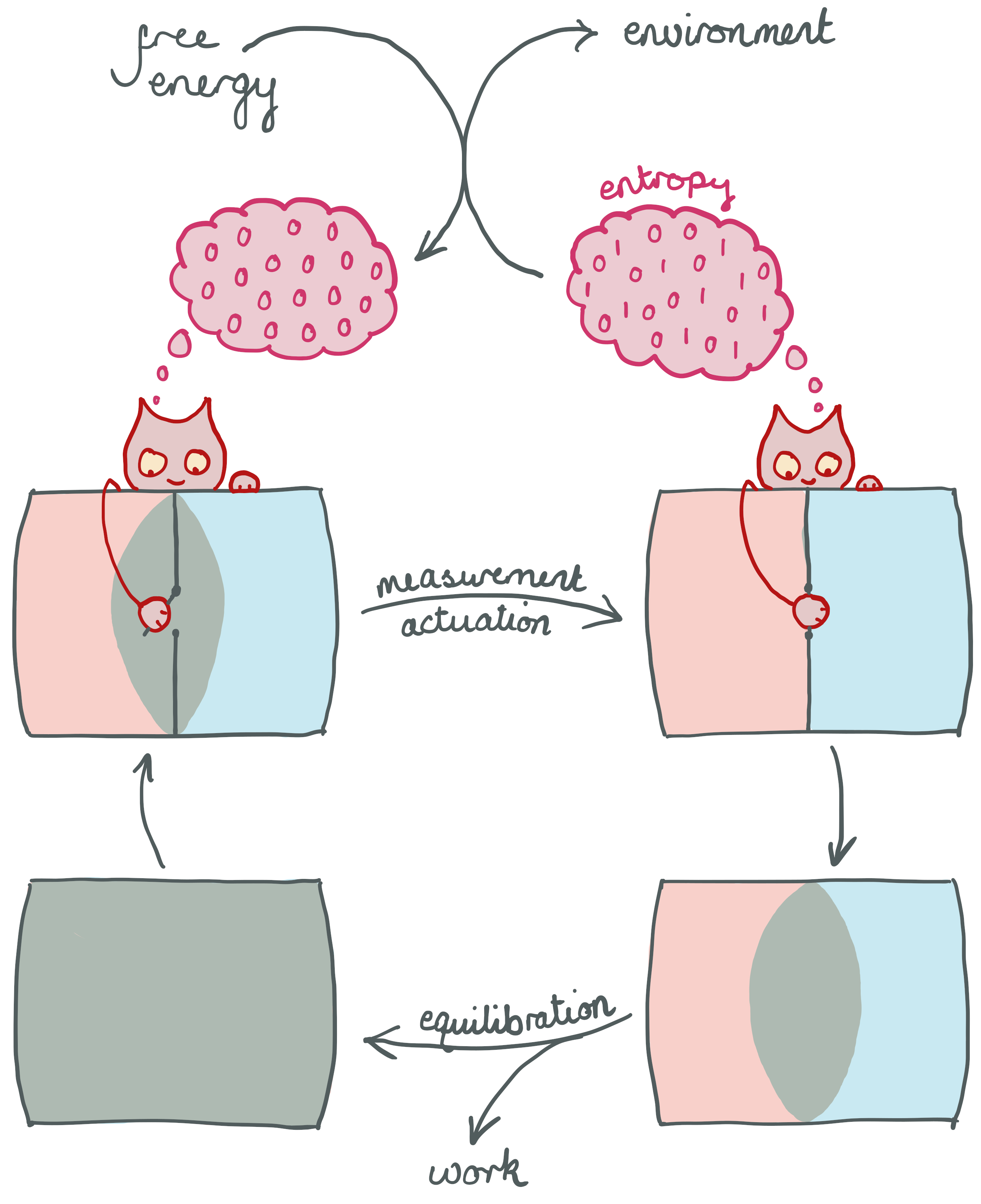}
  \caption[Maxwell's D{\ae}mon (1867).]{Maxwell's D{\ae}mon (1867) is a thought experiment which, if valid, would offer a method to violate the second law of thermodynamics. Consider a box divided in two; to the left of the divider is a red gas, and to the right a blue gas. If the divider is removed, then mechanical work can be extracted as the gases mix and equilibrate. Now suppose the divider is re-inserted, and endowed with a small window whose hinge is frictionless and can therefore be opened and closed without energetic cost. Imagine that a d{\ae}mon or other such entity, with particularly keen eyesight and nimble fingers, is positioned at the window. If she observes a red particle in the right compartment moving towards the window, or a blue particle in the left compartment moving towards the window, she briefly opens it to allow the particle's passage; otherwise she leaves the window closed. Over time, this process will bring the box back to its initial state, thereby apparently reducing its entropy.\captionnl The problem is that this process requires the d{\ae}mon to learn information about the state of the gas particles. In effect, the entropy has simply moved from the state of the particles to the state of the d{\ae}mon's memory. No entropy reduction has occurred. Moreover, eventually her memory will be filled up, and she will likely want to `erase' it to continue her task. Forgetting all this information will correspond to moving the entropy into the state of another system. If this entropy now manifests as heat, then Landauer showed that discarding a quantity of information $\Delta I$ requires the production of heat $\Delta Q\ge kT\Delta I$ where $k$ is Boltzmann's constant, $T$ the temperature, and with equality only in the limit of a thermodynamically reversible process (i.e.\ taking infinite time).}
  \label{fig:daemon}
\end{figure}
\afterpage{\clearpage}

As mentioned, the other class of unconventional computer we shall consider is that of reversible computers. Reversible computers impose the requirement that all state transitions be invertible, and therefore conserve information. The reason for this restriction is that the laws of physics, as far as is currently known\footnote{This includes quantum mechanics; the Schrödinger equation specifies that the time evolution operator is unitary and hence invertible. Furthermore, it seems that even the supposedly irreversible process of wavefunction collapse is in fact a reversible one of progressive decoherence, as espoused by the Many-Worlds Interpretation~\cite{mwi}.}, are fundamentally reversible. In contrast, conventional computers are irreversible in the sense that they do not conserve information and that it is not generally possible to deterministically wind back the state of computation. This has surprising consequences: from a thermodynamic perspective, it transpires that this loss of information manifests as a commensurate increase in entropy, whilst from a quantum perspective, the loss of information of some quantum state necessarily results in decoherence of any quantum state entangled with said state. This quantum mechanical behaviour is not of great concern in classical computers, but is untenable for quantum computers and hence they must necessarily be programmed reversibly.

\begin{wrapfigure}{o}{.5\linewidth}
  \centering
  \begin{subfigure}{\linewidth}%
    \centering%
    \includegraphics[width=3.2cm]{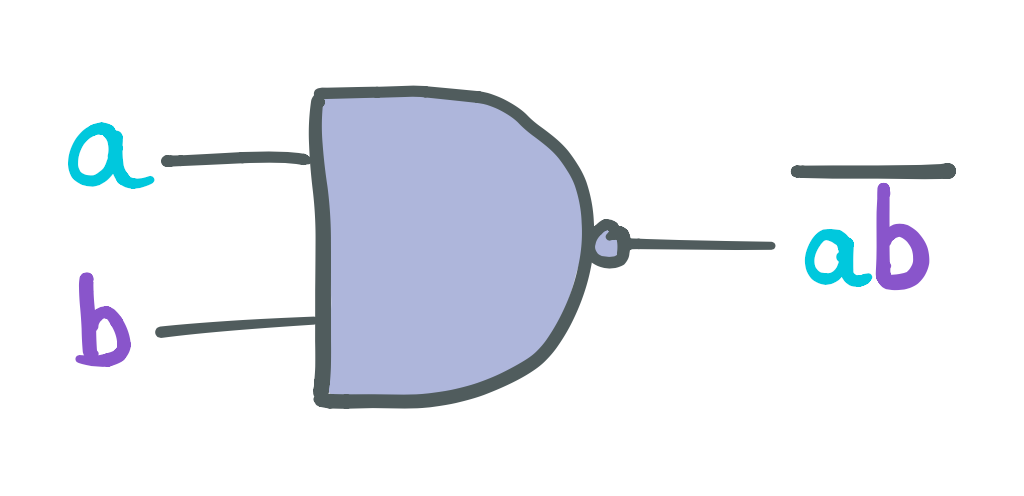}%
    \includegraphics[width=3.2cm]{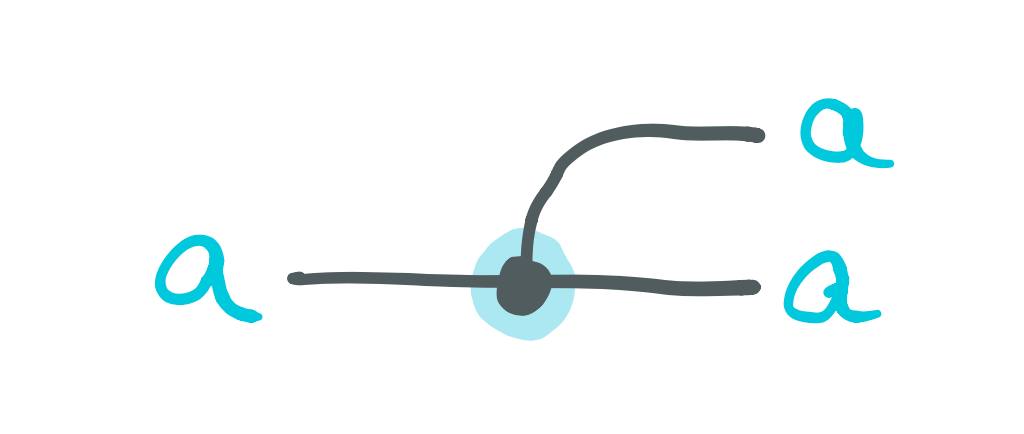}
    \caption{\gate{NAND} and \gate{FANOUT} gates}
  \end{subfigure}
  \begin{subfigure}{\linewidth}%
    \centering%
    \includegraphics[width=2cm]{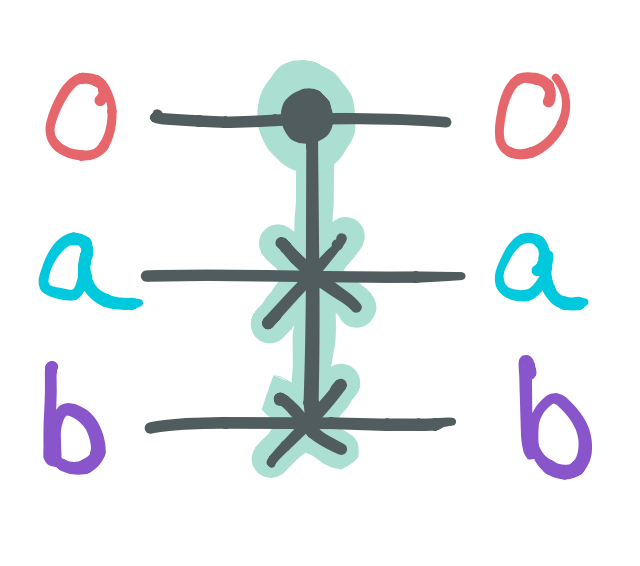}%
    \includegraphics[width=2cm]{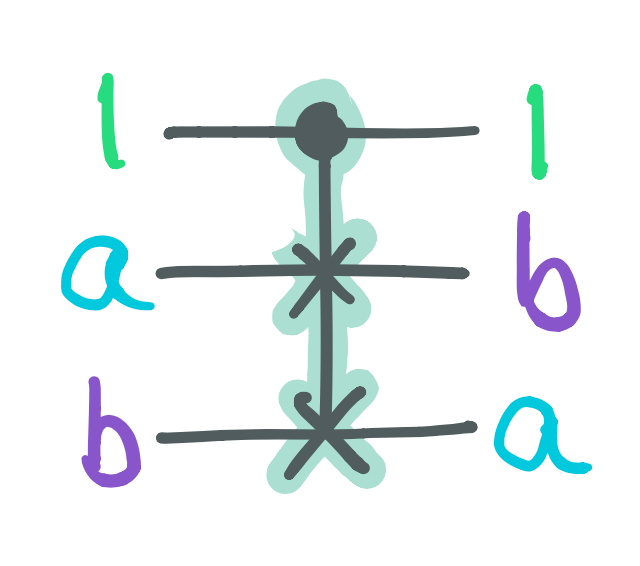}%
    \kern.2cm%
    \includegraphics[width=3.2cm]{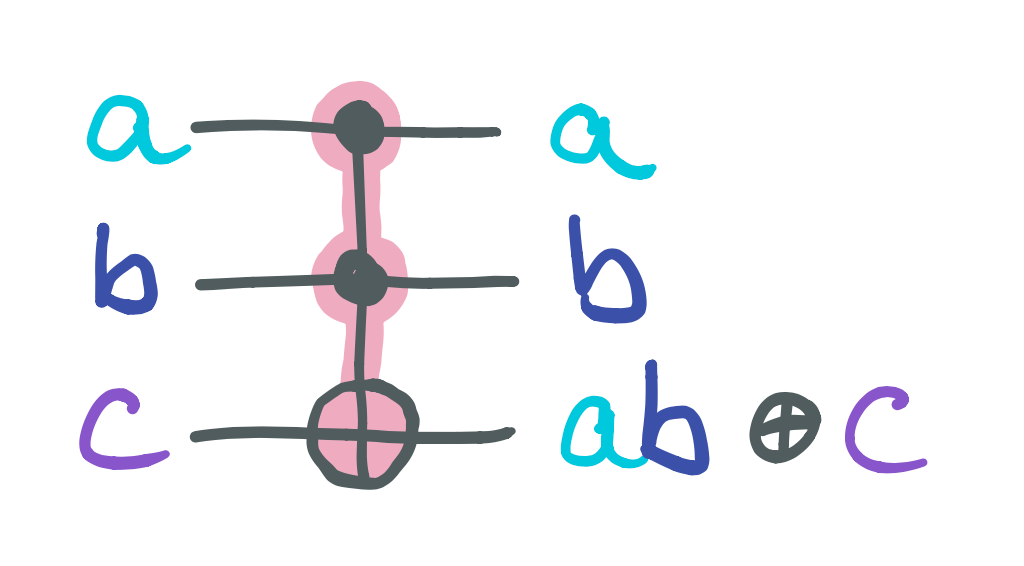}
    \caption{Fredkin and Toffoli gates}
  \end{subfigure}
  \begin{subfigure}{\linewidth}%
    \centering%
    \includegraphics[width=4cm]{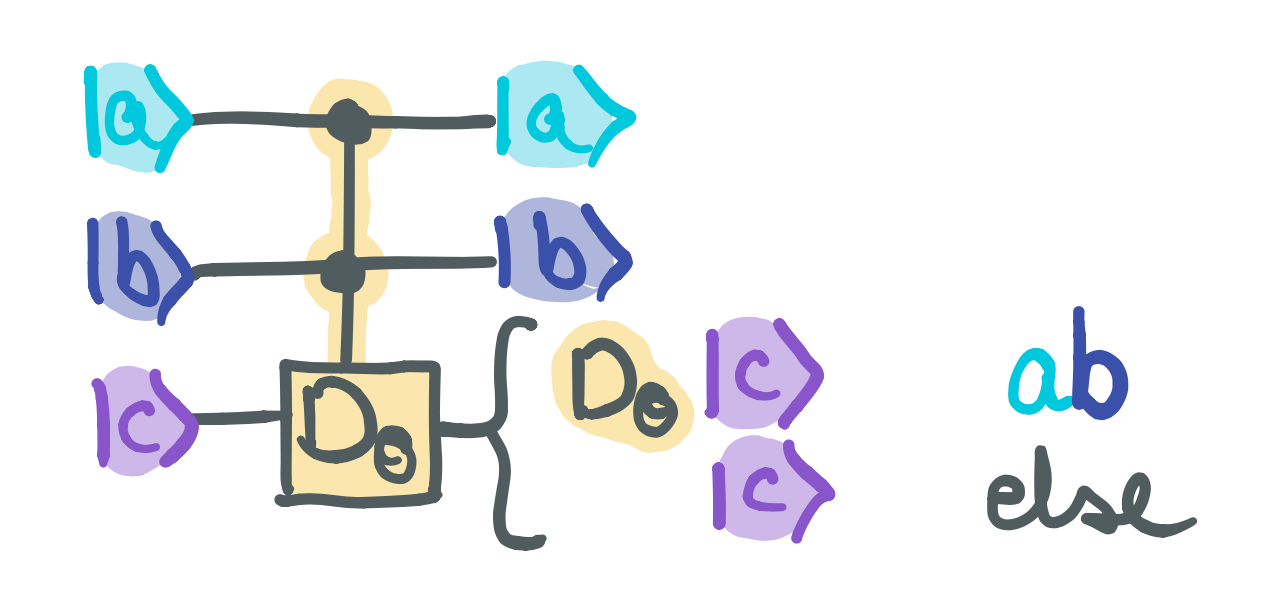}%
    \kern.2cm%
    \raisebox{1cm}{\begin{minipage}{3cm}\footnotesize\begin{align*}
      {\mathrm D}_\theta &= \begin{pmatrix}
        i\cos\theta & \sin\theta \\ \sin\theta & i\cos\theta
      \end{pmatrix}
    \end{align*}\end{minipage}}
    \caption{Deutsch Gate}
  \end{subfigure}
  \caption[Three sets of universal gates for Boolean logic, both classical and quantum.]{Three sets of universal gates. (a)~Two gates which can be used to synthesise any Boolean function $\mathbb B^m\to\mathbb B^n$. (b)~Two reversible logic gates, \emph{each} of which can simulate \gate{NAND} and \gate{FANOUT} gates given suitable ancillary bits (inputs fixed to either 0 or 1). (c)~An example gate that is universal for quantum computation: that is, any quantum circuit can be constructed from a series of Deutsch gates (parameterised by the angle $\theta$).}
  \label{fig:universal-gates}
\end{wrapfigure}

Expanding more on the entropic manifestation of information loss, simulating the ability to discard information in a reversible setting touches deeply on the connection between thermodynamics and information theory, and was formalised by \textcites{szilard-engine,landauer-limit} in the twentieth century in order to resolve Maxwell's eponymous thought experiment, the Maxwell D{\ae}mon (\Cref{fig:daemon}). For every bit of information discarded, a commensurate entropy increase of at least $k\log 2$ is produced elsewhere, where $k$ is Boltzmann's constant. Physically, the information is encoded into environmental degrees of freedom, and this information promptly decorrelates with the original computational system. This decorrelation is effectively irreversible, and hence the information becomes entropic. More generally, whenever a quantity of information $I\,\text{bits}$ is discarded, an entropy increase of $\Delta S\ge kI\log 2$ manifests, with equality only in the limit of thermodynamic reversibility in which the discarding process takes infinitely long. If this entropy were allowed to remain in the computational system, it would eventually accumulate to such an extent that it would interfere with the well ordered operation of the computational mechanism. As a visceral example, should the entropy take the form of heat, $\Delta Q=T\Delta S$ where $T$ is the system temperature, the system will eventually transmute into a form incompatible with its function---such as melting or turning into a plasma---as the temperature becomes too great. Even if the system is physically heat-resistant, the increased entropy will result in errors so frequent that almost all the computational capacity is directed to error correction, slowing productive computation to a crawl.

At room temperature, this heat generation will be on the order of $\sim\SI{e-21}{\joule}\sim\SI{e-2}{\electronvolt}$ per bit discarded. Conventional computers typically expend a factor of around \num{e8} times this much for reasons relating to reliability, speed, and the relatively large size of their transistors in comparison to the atomic scale. It is somewhat challenging to precisely quantify this overhead as the information processing capacity of a modern consumer processor is difficult to ascertain, due to the extensive use of pipelining, vectorisation and multiple-instruction dispatch. Nevertheless, we can obtain a reasonable estimate. arm's processors are well known for emphasising power efficiency, and are used extensively in mobile electronics as well as some laptops and desktops (notably, three of Apple's new device lineup). Considering arm's A-78 processor micro-architecture~\cite{arm-a78}, we find numbers of around \SI{1}{\watt\per{core}}, with each core running at a clock speed of \SI{3}{\giga\hertz} and executing up to \SI{6}{(macro)instructions} per clock cycle with a width of up to \SI{128}{\bit}. Taken together, these give an energy dissipation of roughly \SI{1.5e-13}{\joule\per\bit}. In contrast, at a temperature of around \SI{300}{\kelvin}, Landauer's bound gives just \SI{2.9e-21}{\joule\per\bit}, a \num{5e7} overhead for the arm processor. Intel's chips have thermal design powers typically exceeding \SI{100}{\watt}, but have compensatingly greater vectorisation widths of some instructions, such that a similar Landauer overhead can be achieved when fully exploiting the processor's capabilities.

The origin of irreversibility in conventional computers can be subtle, but it is in fact ubiquitous. This may be considered to ultimately originate in early mathematical models of computations such as the Turing Machine and Lambda Calculus in which overwriting and discarding, respectively, are implicit to the primitive operations of the models.  At the lowest levels of abstraction of modern computers, semiconductor logic gates and buses are operated in an irreversible manner: somewhat oversimplifying, traditional CMOS-style logic dissipates energy every time the transistor states switch, and during each clock cycle an equilibration process occurs in which the nodes of the circuit are
connected to bus lanes at various reference voltages and any transient voltage disparity is eliminated. Both of these processes are capable of discarding information in the system, although exactly where in these processes the information erasure should be attributed is a non-trivial matter. In fact it is possible to drive CMOS-style logic gates \emph{adiabatically}~\cite{crl,sl-crl,s2lal}, achieving asymptotically zero dissipation in the limit of infinitely slow switching. This is assuming that the instructions are logically invertible, perhaps subject to certain pre- and post-conditions in which case it is sometimes referred to as \emph{conditionally reversible}. Reintroducing irreversibility is possible, but dissipation of information should be performed separately to maintain adiabaticity in the operation of the logic gates. The idea of using reversible logic gates in irreversible computation may seem counter-productive, but it is a useful potential way to make substantial gains in energy efficiency and to get closer to the Landauer bound.

At higher levels of abstraction potential sources of information erasure include the obvious, such as overwriting/mutating a variable, to the less obvious such as variables going out of scope or ignoring a function's return value. Even more subtly, any iteration or recursion which does not maintain a record of its initiation condition as well as its termination condition is in fact irreversible, as shall later become clear.

With such ubiquitous application of information-destructive primitives, it is hard to see how algorithms can be rewritten reversibly. Indeed, whilst Landauer could envisage what logically reversible logic gates might look like, even describing~\cite{landauer-limit} what would eventually become known as the Toffoli or \gate{CCNOT} gate (see \Cref{fig:universal-gates}), he came to the conclusion that one would end up having to manually keep track of the lost information thus rapidly filling the computer's memory with useless `garbage' or `history' data. 12 years later, Charles Bennett---often considered the `Father of Reversible Computing'---developed a remarkably simple algorithm~\cite{bennett-tm} (\Cref{fig:bennett-algo1}) for embedding any irreversible function $x\mapsto f(x)$ as the reversible injection $x\leftrightarrow(x,f(x))$ without producing any net garbage. Later, he extended this algorithm~\cite{bennett-pebbling} to reduce the amount of interim garbage data produced via a recursive application of the original algorithm, the general principle of which is illustrated in \Cref{fig:bennett-pebble}.

\begin{figure}
  \begin{subfigure}[b]{.5\linewidth}
    \centering
    \includegraphics[width=\linewidth]{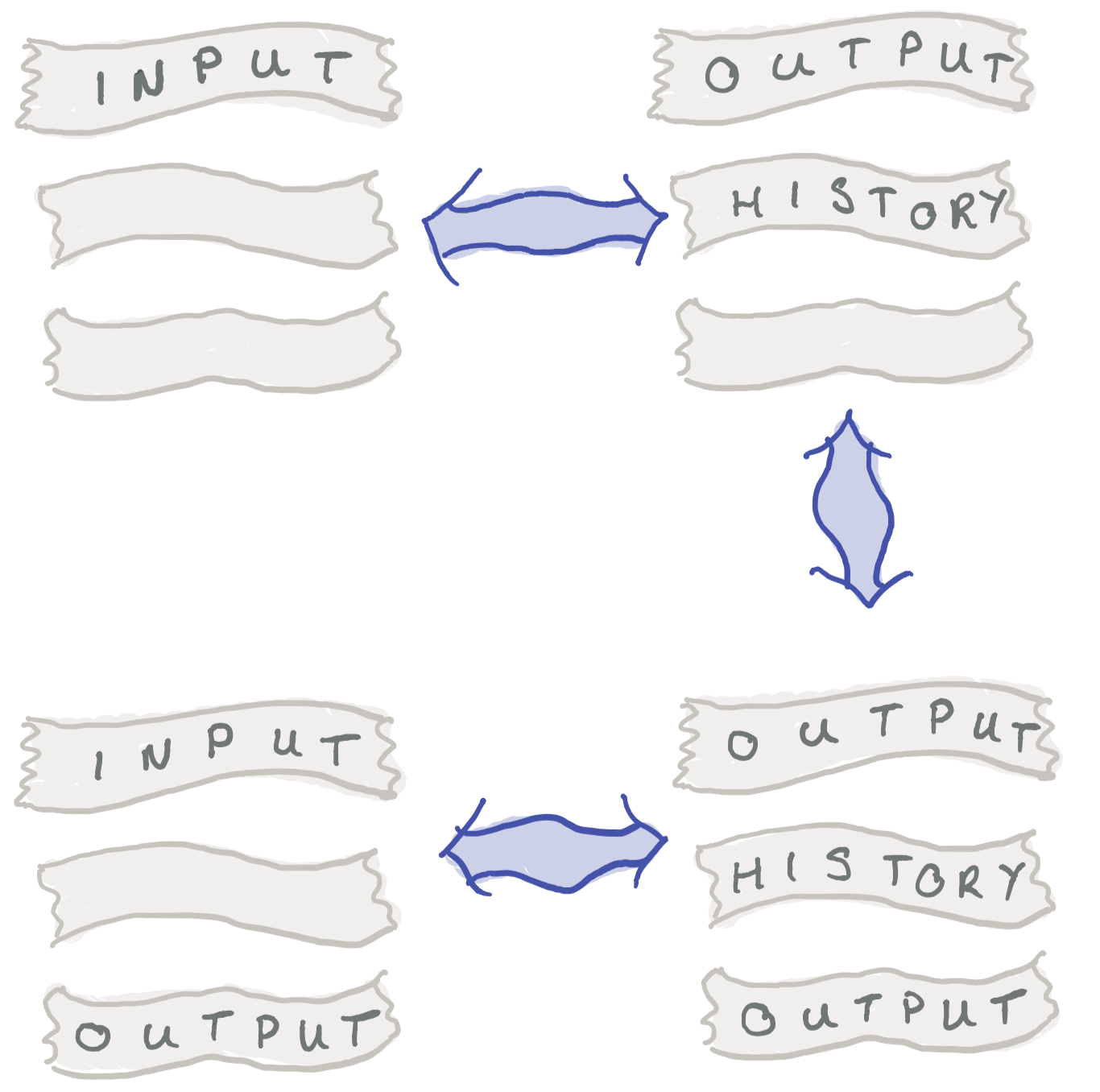}
    \caption{}
    \label{fig:bennett-algo1}
  \end{subfigure}%
  \begin{subfigure}[b]{.5\linewidth}
    \centering
    \includegraphics[width=\linewidth]{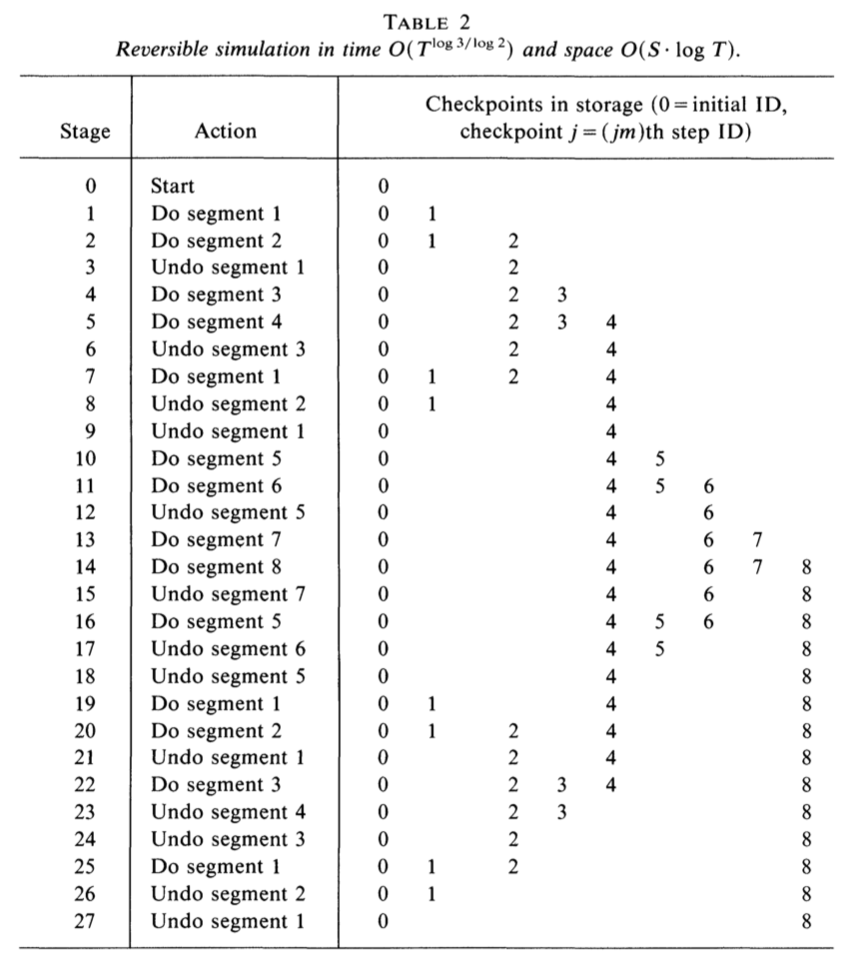}
    \caption{}
    \label{fig:bennett-pebble}
  \end{subfigure}
  \caption[Two of Bennett's algorithms for reversibly simulating irreversible programs.]{Two of Bennett's algorithms for reversibly simulating irreversible programs.
  \capnl(a) Bennett's initial algorithm~\cite{bennett-tm} for reversibly simulating an irreversible program with no net garbage data, except for a retained copy of the original input.
  In the same paper, Bennett described a reversibilised analogue of the Turing Machine (TM), the \emph{Reversible Turing Machine} (RTM), whose every operation is logically invertible and which could simulate any Turing Machine by placing the garbage data on a separate tape instead of erasing it.
  Bennett's algorithm makes use of the reversibility of this process: after first running the simulated TM forwards, we peek at the output and make a copy onto a third tape; we then forget about the third tape and reverse the simulation on the first two tapes, consuming the garbage and output and yielding the original input.
  Copying can be performed reversibly by, for example, using the output as a \emph{one-time pad} to encrypt the contents of the third tape (if the RTM alphabet is binary, this is just the \gate{XOR} operation).
  If the third tape is empty, then this will result in a copy; if it is not, then the result will generally be garbled.
  \capnl(b) A reproduced figure of Bennett's `pebbling' algorithm~\cite{bennett-pebbling}, illustrating the general principle. The pebbling algorithm runs the algorithm of \Cref{fig:bennett-algo1} partially forwards and backwards, making and un-making partial copies of the intermediate states of the computation. This yields a spatial overhead logarithmic in the total computation time, instead of linear as with the original algorithm. The overhead in the time complexity is more significant however, but Bennett presented even more sophisticated algorithms to render the exponent arbitrarily small.}\label{fig:bennett-algos}
\end{figure}

\begin{figure}[!t]
  \long\def\fredkinEx#1{\begin{subfigure}{.222\linewidth}
    \centering\includegraphics[width=.9\linewidth]%
      {gate-fredkin-bb-#1.png}
   \end{subfigure}}
  \centering
  \begin{minipage}[m]{.7\linewidth}\centering
    \includegraphics[width=.8\linewidth]{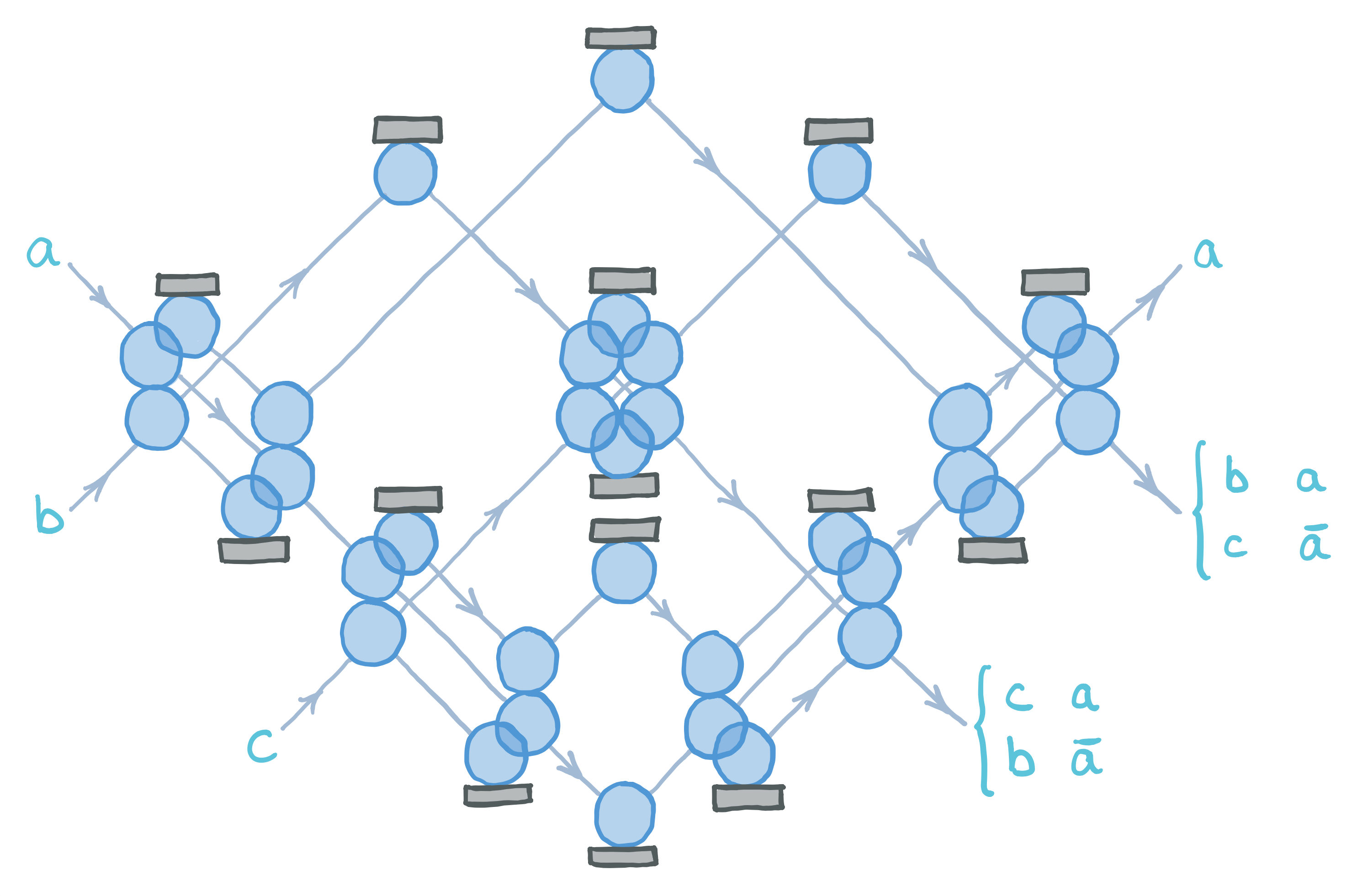}
  \end{minipage}%
  \hspace{.3cm}%
  \fbox{\begin{minipage}[m]{.2\linewidth}
      \includegraphics[width=\linewidth]{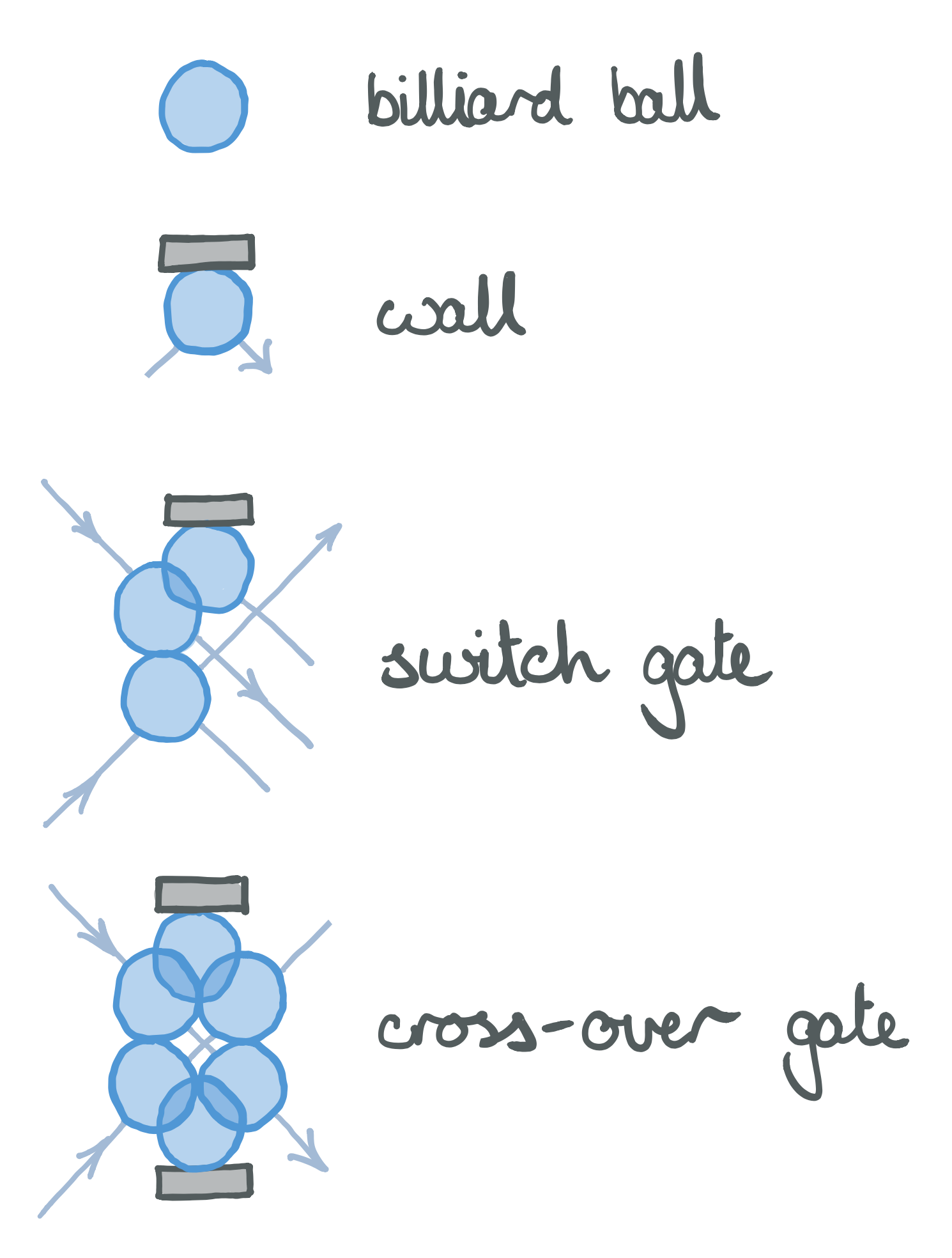}%
    \end{minipage}}
  \caption[A billiard-ball implementation of a Fredkin gate due to Ressler.]{A billiard-ball implementation of a Fredkin gate due to \textcite{ressler}. As primitives, it makes use of eight `switch' gates, one `cross-over' gate (middle), and five lone walls for simple redirection. The implementation differs slightly from the canonical Fredkin gate in that the second and third outputs are swapped. Note that not all the ball positions depicted necessarily occur for each possible combination of inputs; an input is given by the presence or absence of a ball for each of $a$, $b$ and $c$, and these balls are synchronised in that at any given time they are aligned vertically, as shown in the below gate instantiations:}
  \label{fig:gate-fredkin-bb}
  \vspace{1em}
  \begin{minipage}{\linewidth}
    \fredkinEx{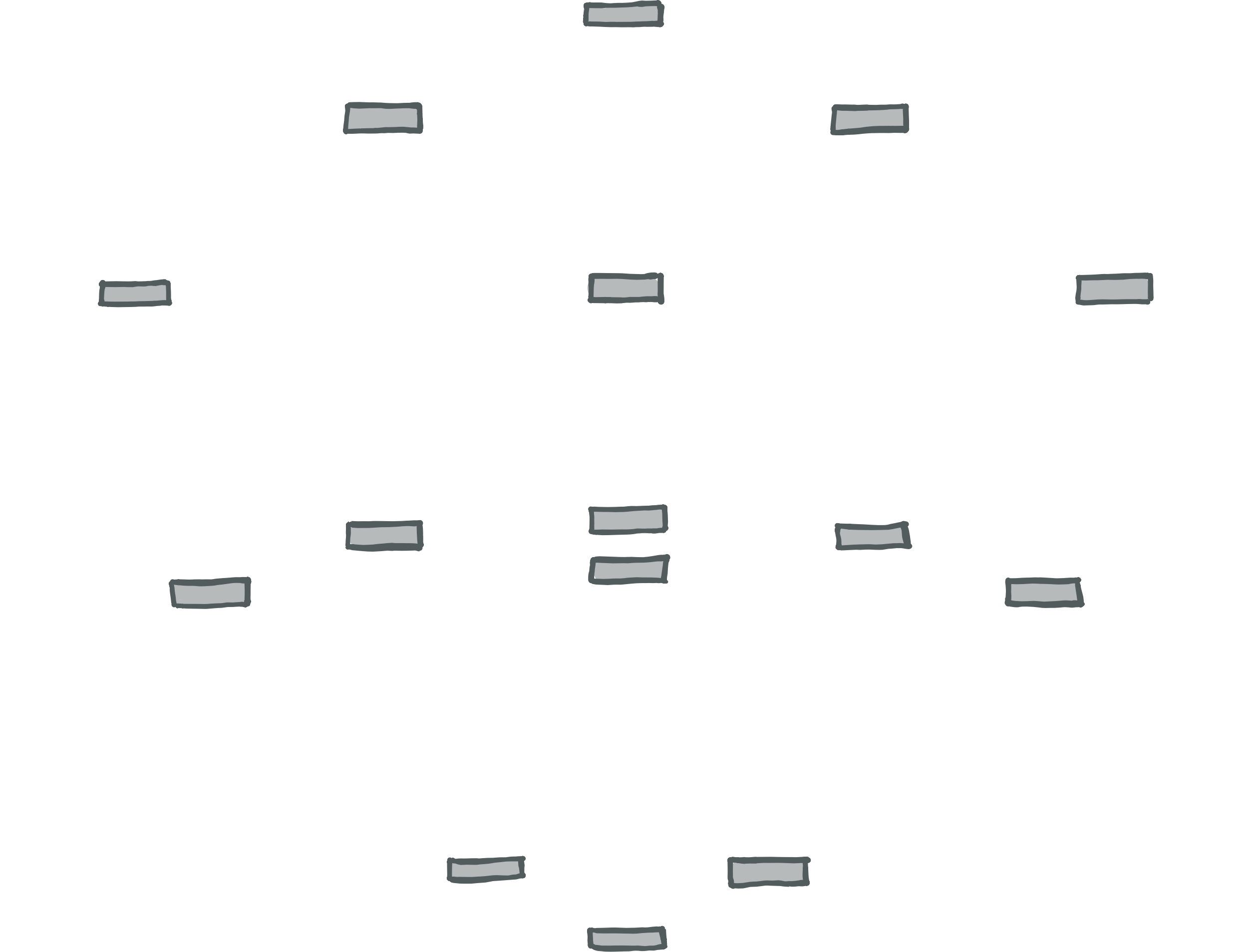}\fredkinEx{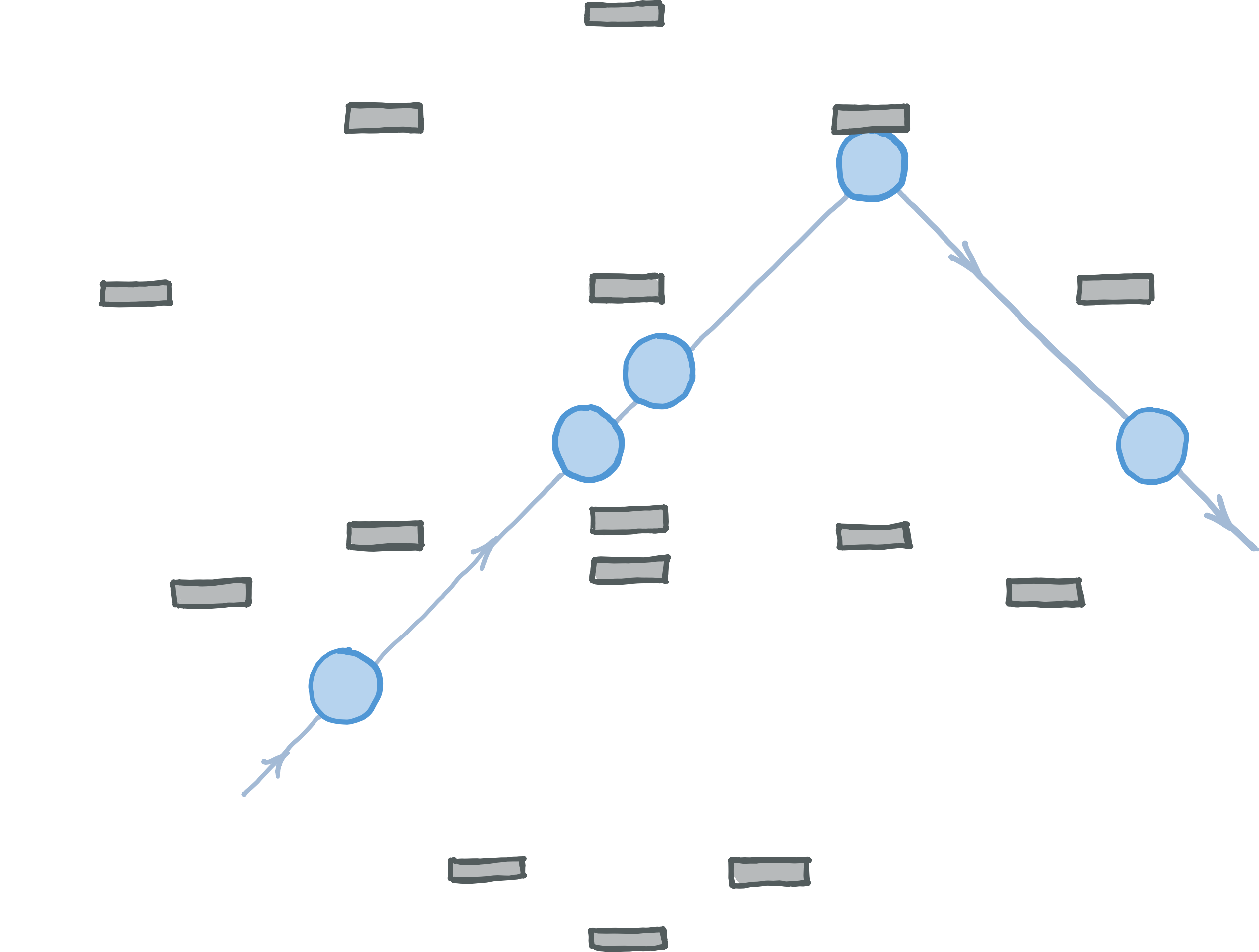}\fredkinEx{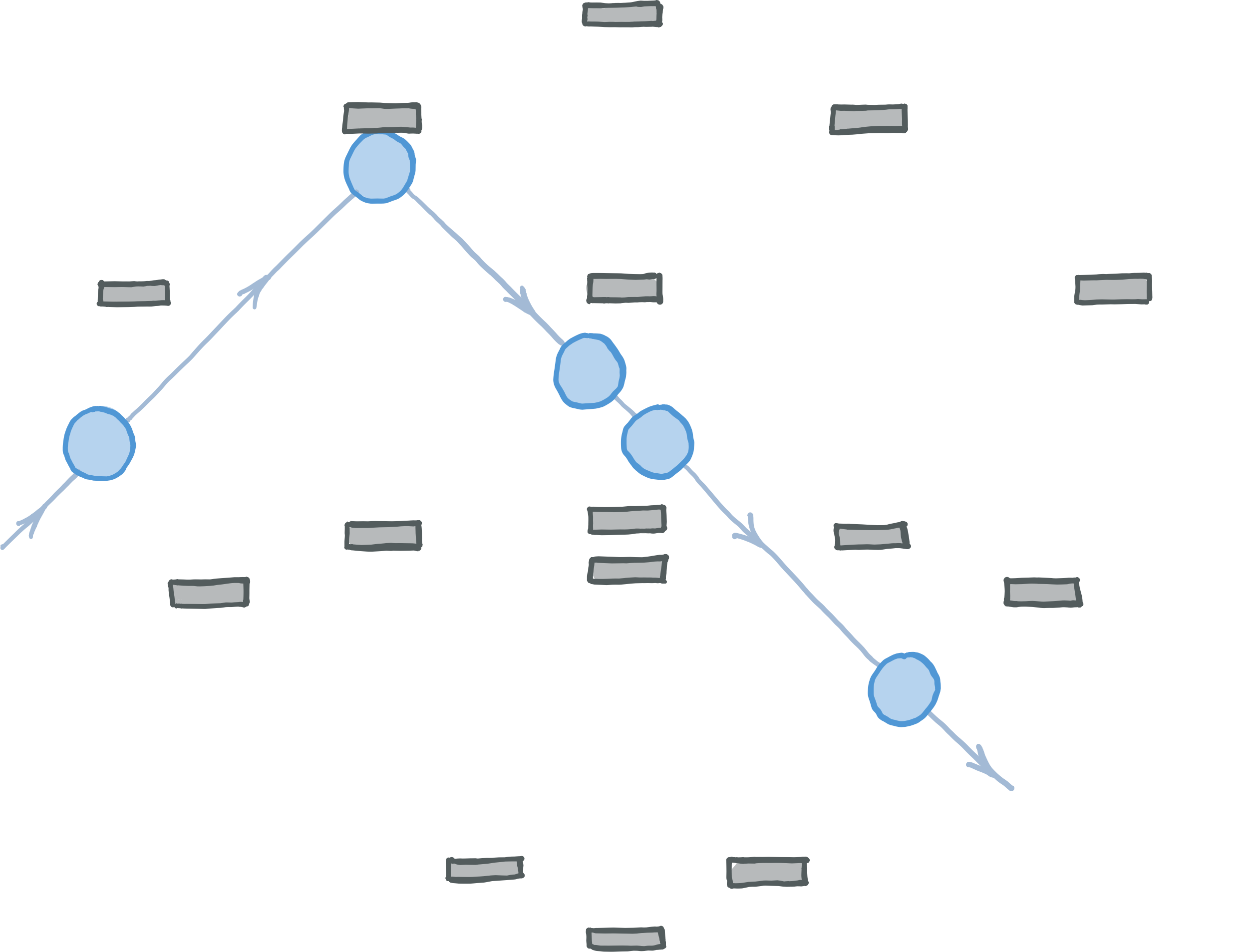}\fredkinEx{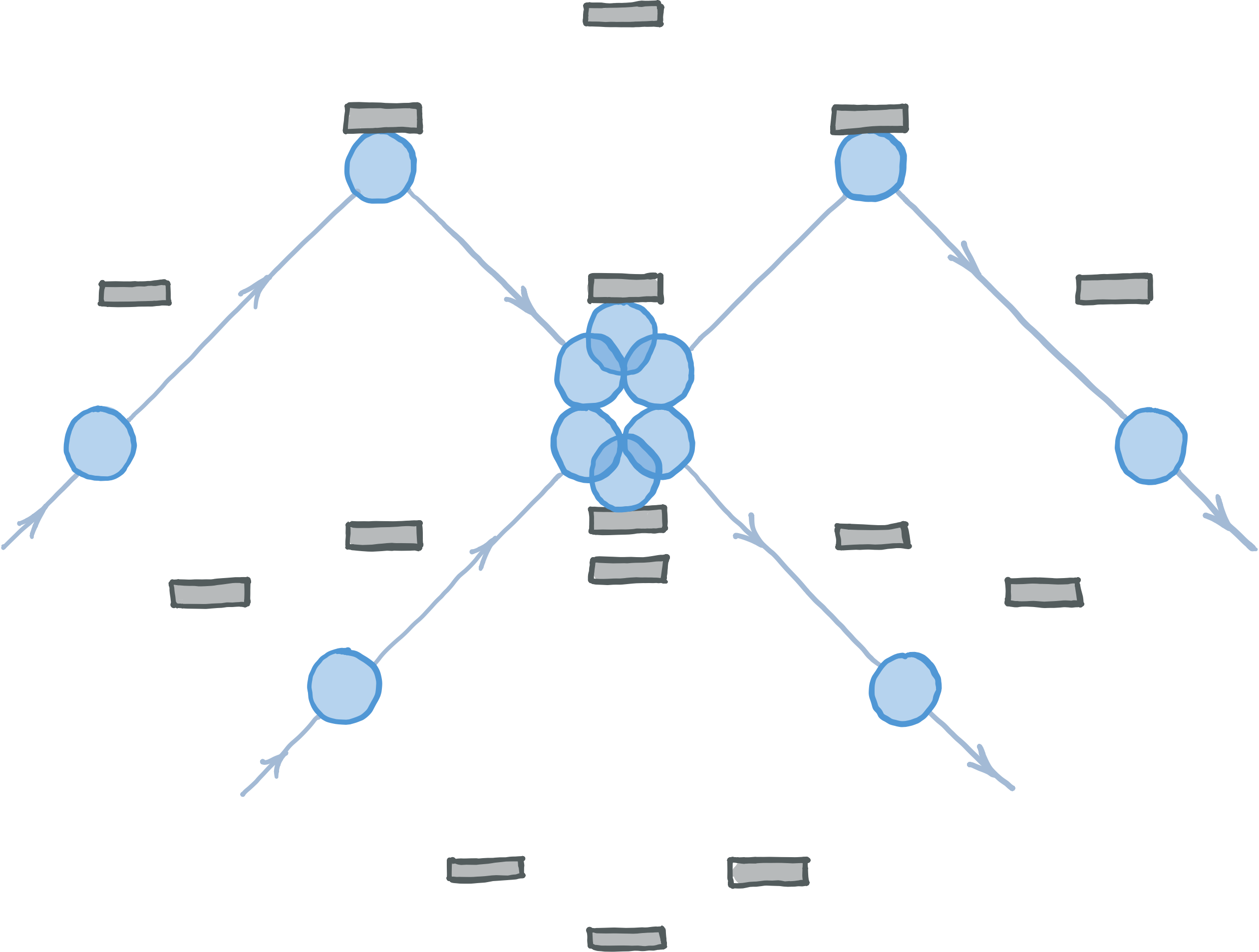}\hfil\\[.5em]
    \null\hspace{.111\linewidth}\fredkinEx{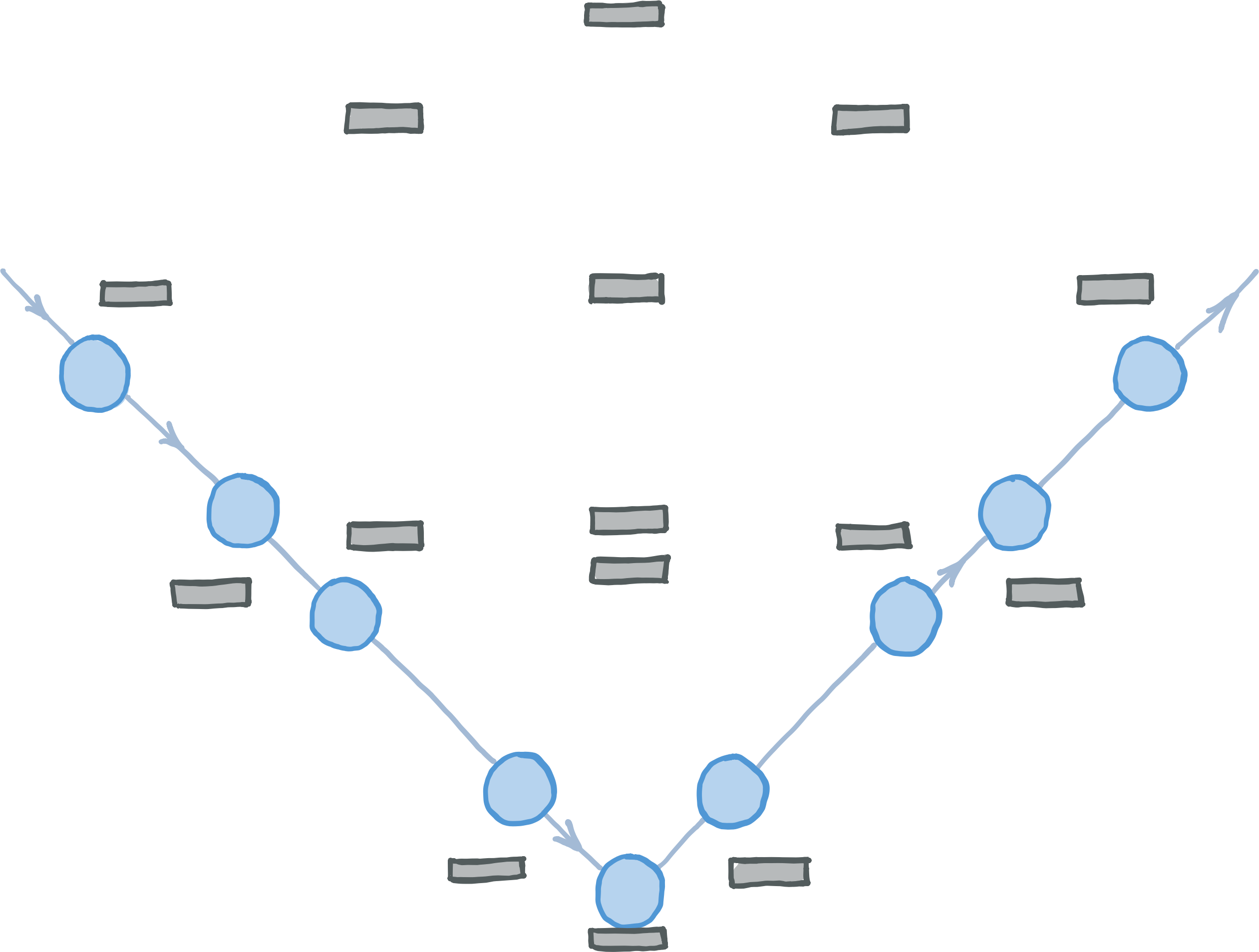}\fredkinEx{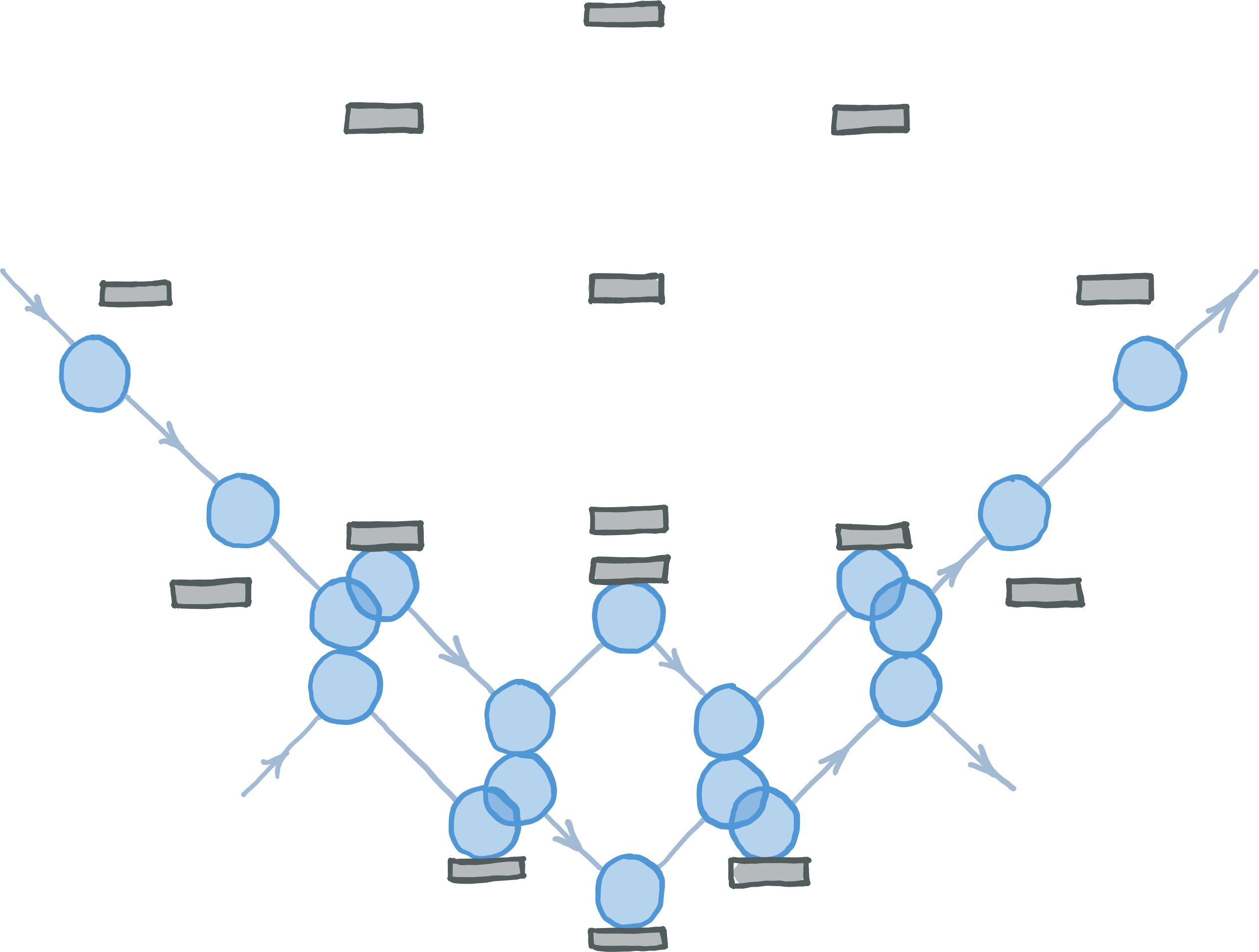}\fredkinEx{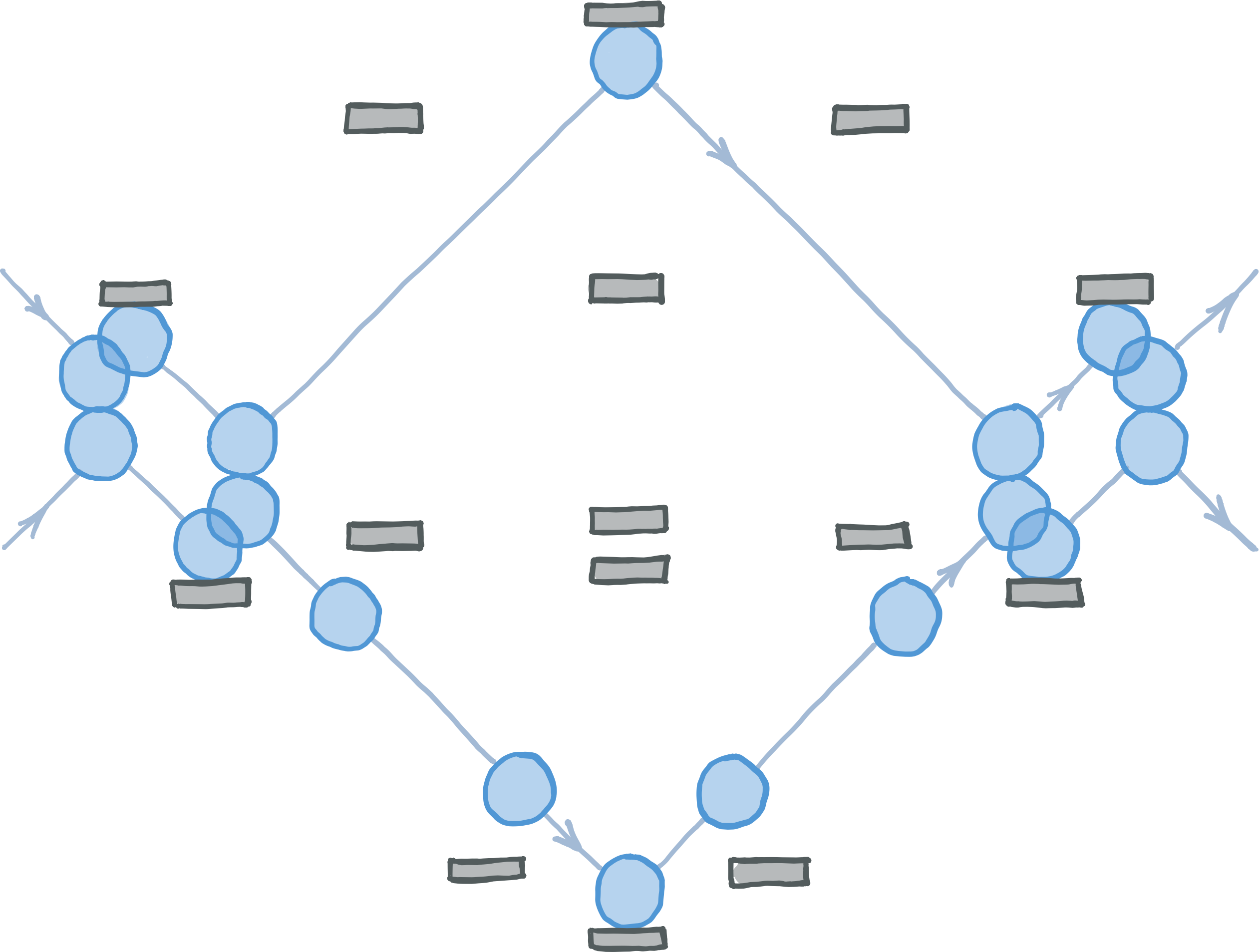}\fredkinEx{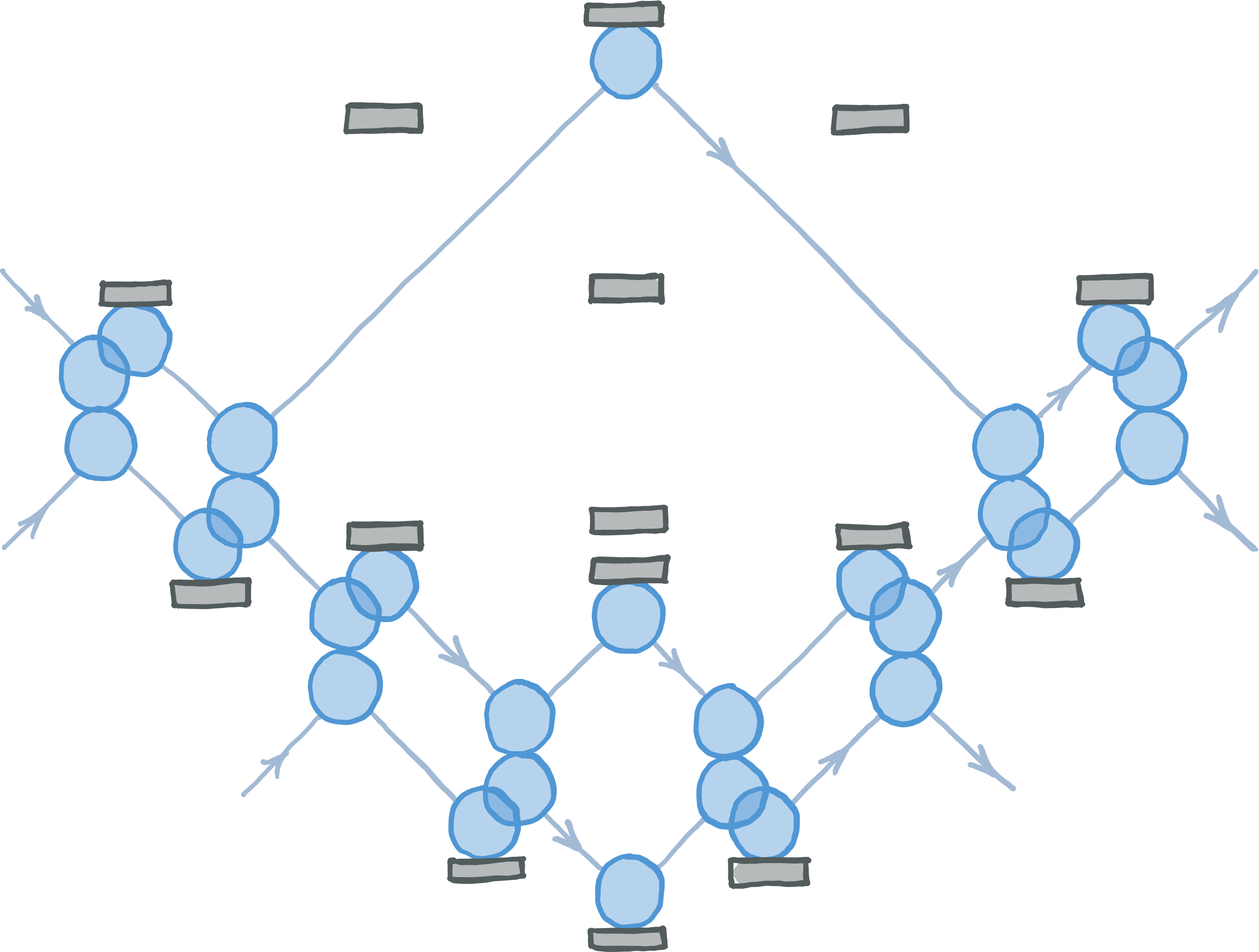}
  \end{minipage}
  \vspace{-1em}
\end{figure}

With the viability of reversible computing proven, the stage was then set for more faithful programming of reversible computers to be demonstrated beyond simple injective embedding of irreversible programs. Whilst Bennett did introduce the idea of a Reversible Turing Machine, arguably the first description of a reversible computer that might be physically realisable was given by \textcite{fredkin-conlog}. Fredkin and Toffoli introduced conservative logic, in which the number of 1 bits (and optionally also the number of 0 bits) is conserved; they proceeded to show that this logic could be implemented by a system of ideal elastic `billiard balls', projected across an ideal frictionless surface decorated with ideally rigid walls. The motivation for this construction was that it is precisely the sort of system considered by classical mechanics. If one could build such a billiard-ball computer, then aside from the initial kinetic energy supplied to the system, there would be no need to supply additional energy nor any dissipation; moreover, one could recover the kinetic energy at the end for amortised-free computation. To illustrate the principles of billiard-ball computation, Andrew Ressler's implementation of a Fredkin gate is shown in \Cref{fig:gate-fredkin-bb}. Going further, Ressler designed an entire reversible computer processor---complete with arithmetic logic unit---within the billiard-ball scheme for his Master's thesis~\cite{ressler}, demonstrating that not only were reversible computers theoretically as capable as irreversible ones from a Turing universality perspective, but that they were also practically as capable.

Unfortunately, Bennett's analysis of the billiard-ball model showed it to be impracticable as a fully dissipationless model of reversible computation:
\enlargethispage{\baselineskip}\begin{quote}
  ``Even if classical balls could be shot with perfect accuracy into a perfect apparatus, fluctuating tidal forces from turbulence in the atmospheres of nearby stars would be enough to randomise their motion within a few hundred collisions. Needless to say, the trajectory would be spoiled much sooner if stronger nearby noise sources (e.g., thermal radiation and conduction) were not eliminated.''
  \signed{\textrm{ --- \textcite{bennett-rev}}}
\end{quote}
We shall have more to say about this in \Cref{chap:revi}, but since then a number of adiabatic designs of classical reversible computers have been developed and experimentally verified, not to mention the recent progress in building quantum computers. We highlight the \emph{Pendulum Instruction-Set Architecture} (\pisa), first introduced by \textcite{pendulum} and later refined by \textcite{frank-thesis,pendulum2}. This architecture was successfully realised adiabatically with semiconductors by employing the `Split-Level Charge Recovery Logic' (SCRL) motif~\cite{sl-crl}, as shown in \Cref{fig:pendulum}. Adiabatic reversible computers achieve far lower dissipation than irreversible computers, but still dissipate energy at a rate inversely proportional to the time per operation. Questions of whether or not lower energy dissipation, or even the `holy grail' of fully dissipationless computation, are achievable have continued to be asked\footnote{Such as at the recent 2020 \emph{Physics and Engineering Issues in Adiabatic/Reversible Classical Computing Visioning Workshop}.}, and we believe we have answered these in the negative in \Cref{chap:revi}.

\begin{figure}[htbp]
  \begin{subfigure}[b]{.33\linewidth}
    \centering
    \includegraphics[width=.95\textwidth]{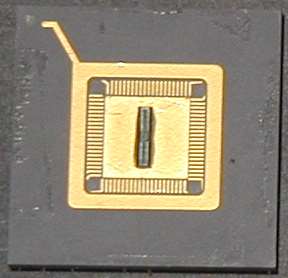}
    \caption{\archTick}
  \end{subfigure}%
  \begin{subfigure}[b]{.33\linewidth}
    \centering
    \includegraphics[width=.95\textwidth]{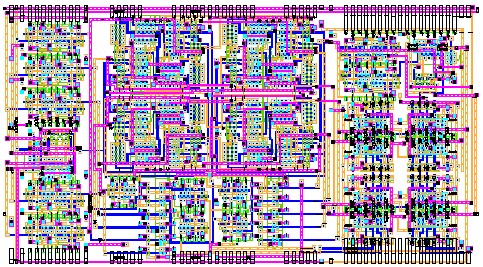}\\[.2em]
    \includegraphics[width=.95\textwidth]{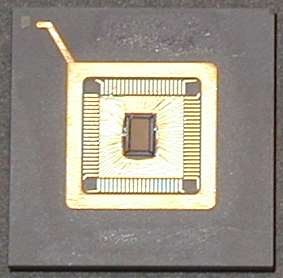}%
    \caption{\archFT\ and its unit cell}
  \end{subfigure}%
  \begin{subfigure}[b]{.33\linewidth}
    \centering
    \includegraphics[width=.95\textwidth]{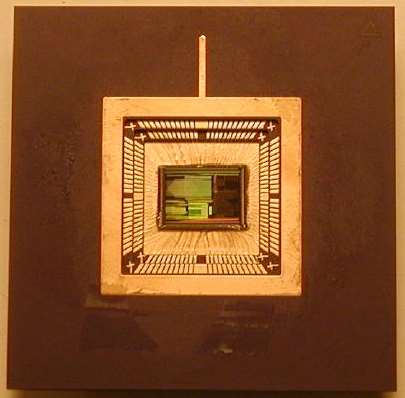}
    \caption{\pendulum}
  \end{subfigure}
  \caption[Three iterations of experimentally realised reversible processor.]{%
    Three iterations of experimentally realised reversible processor between 1996 and 1999 by Frank, Ammer, Love, Rixner and Vieri.
    \archTick\ (1996) was an 8-bit subset of \pisa, but was non-adiabatic.
    \archFT\ (1996) was an adiabatic programmable gate array implementing the \emph{Billiard Ball-Model Cellular Automaton} (BBMCA) of \textcite{bbmca}.
    \pendulum\ (1999~\cite{pendulum2}) was the culmination of Vieri's work for his graduate thesis, and a 12-bit implementation of \pisa.}
  \label{fig:pendulum}
\end{figure}

\begin{wrapfigure}{O}{.5\linewidth}
  \centering
  \includegraphics[width=\linewidth]{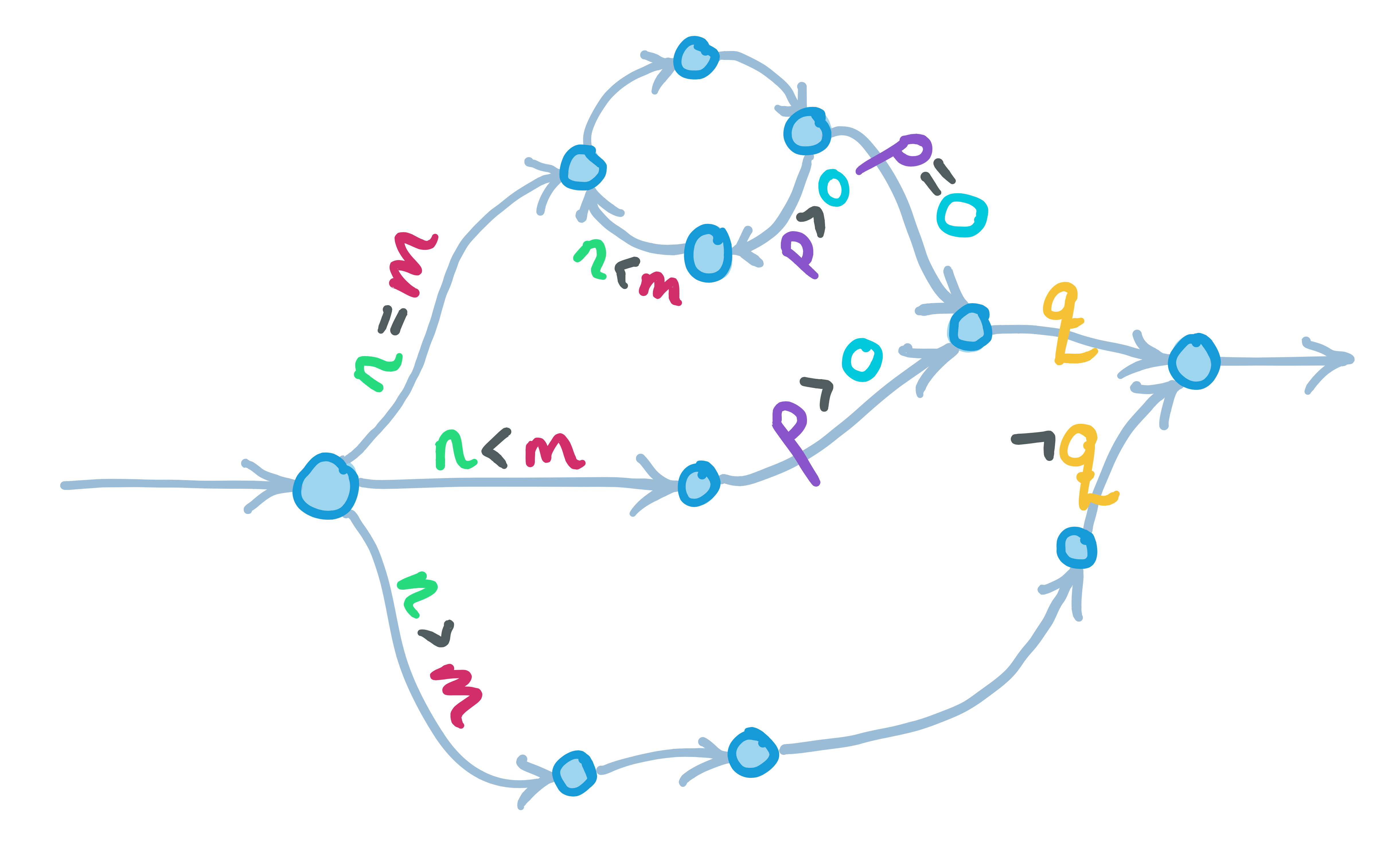}
  \caption[An illustration of the principles of reversible control flow.]{An illustration of the principles of reversible control flow in the form an abstract state transition diagram. The meanings of the variables are not important. Arrow direction specifies the `forwards' direction of the program, and can be reversed to obtain the reversed program.}
  \label{fig:rev-control-flow}
\end{wrapfigure}
\para{Programming Reversibly}

With such ubiquitous application of information-destructive primitives, it is hard to see how algorithms and programs can be rewritten reversibly. In \Cref{chap:aleph} we shall introduce a model of reversible computation appropriate for Brownian computers such as the molecular computers discussed earlier, in addition to many examples illustrating the principles of programming reversibly. Learning to do so requires a paradigm shift in one's thought process, in much the same way as an imperative/procedural programmer undergoes when learning a functional programming language for the first time. Reversibility as a concept is orthogonal to other language paradigms such as imperative/procedural, functional, and declarative, and indeed examples of reversible languages exist for each of these (we believe the \textAleph-calculus and \alethe, introduced in \Cref{chap:aleph}, constitute the first example of a reversible declarative language).

To program reversibly takes some care.
At a high level, one must ensure that each block (whether a function, a subroutine, a loop, a conditional branch, a statement, etc) is information preserving:
that is, there must be a bijective map between the set of \emph{legal} inputs and \emph{legal} outputs.
We say legal, because it is possible for some inputs to not map to outputs and vice-versa, instead producing a logical error during execution or entering an infinite loop.
For example, consider Bennett's embedding protocol applied to the factorial function.
This would produce a reversible program performing the map $(n)\leftrightarrow(n,n!)$, e.g.\ $(4)\leftrightarrow(4,24)$.
If, however, one attempted to feed in the output $(4,30)$ and run this program backwards then an error\footnote{%
  If we assume this is implemented on an RTM, then the error occurs when reversing the copy step:
  tapes 1 and 3 are initialised to 4 and 30 respectively, then the irreversible factorial TM is run forwards on tapes 1 and 2 to obtain 24 and some entropic garbage;
  the RTM then attempts to consume the copy of 24 on tape 3, but $30\neq24$ and so the machine will either stall or its state may become corrupt and meaningless.}\ 
would occur because there is no input mapping to this output.
Once one has ensured that a block is information preserving, the next step is to try to decompose it into smaller information-preserving blocks, just as with regular programming.
Expertise with this comes with practice, and my experience with reversible programming leads me to believe that it is possible to become just as proficient with reversible programming as with irreversible.
Whilst programming reversibly is to be preferred where practicable, a useful `trick' for implementing a difficult bijection $f$ can be achieved by first writing irreversible programs for $f:x\mapsto y$ and $f^{-1}:y\mapsto x$, embedding these using Bennett's algorithm as $f':x\leftrightarrow(x,y)$ and $(f^{-1})':y\leftrightarrow(x,y)$, and then constructing the composition $((f^{-1})')^{-1} \circ f' : x \leftrightarrow y$.

\begin{figure}[t]
  \begin{subfigure}{\linewidth}
    \centering
    \includegraphics[width=.8\linewidth]{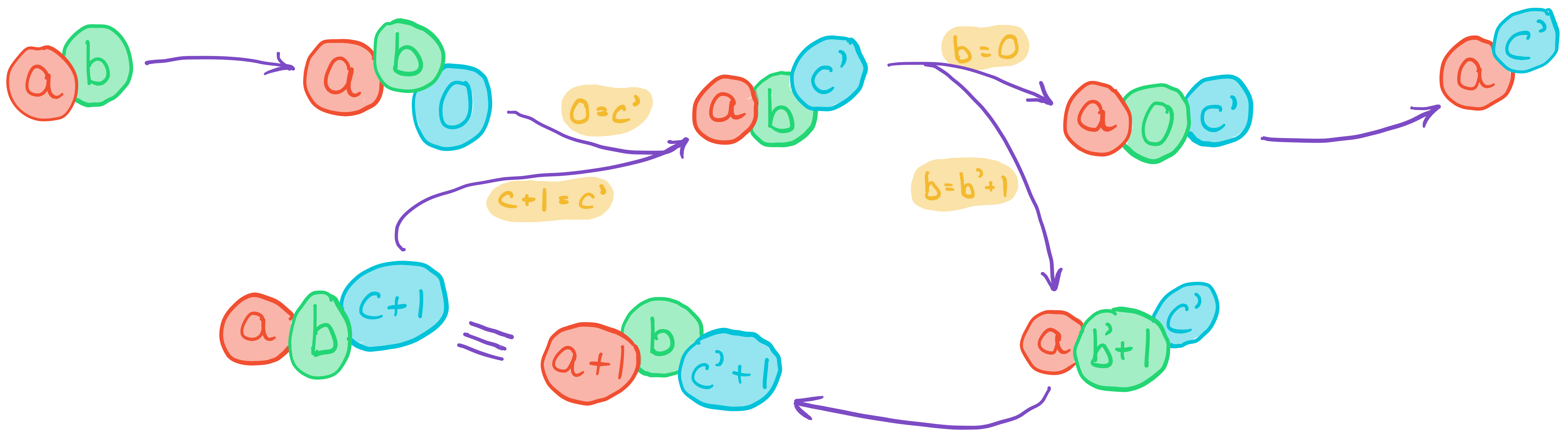}
    \caption{The logical state transition diagram for the reversible addition. This is a reversible embedding of the Peano-axiomatic definition of natural number addition, and therefore assumes as primitives that we can increment/decrement and check whether or not a number is zero.}
  \end{subfigure}
  \begin{subfigure}{\linewidth}
    \def\cond#1#2{%
      \overset{c=#1}{%
        \underset{\smash{\raisebox{0.1ex}{$\scriptstyle b=#2$}}}%
        \longleftrightarrow}}
    {\small\begin{align*}
      (3,4) \leftrightarrow (3,4,0) \cond{0=0}{4\neq0} (4,3,1) \cond{1\neq0}{3\neq0} (5,2,2) \cond{2\neq0}{2\neq0} (6,1,3) \cond{3\neq0}{1\neq0} (7,0,4) \cond{4\neq0}{0=0} (7,0,4) \leftrightarrow (7,4)
    \end{align*}}\vspace{-1em}
    \caption{An example execution path of the program shown in (a). Seen in the forward direction, we are adding $3$ and $4$ to get $7$; seen in the reverse, we are subtracting $4$ from $7$ to get $3$.}
  \end{subfigure}
  \caption[An example of the reversible control flow for a possible reversible implementation of addition of two natural numbers.]{%
    An example of the reversible control flow for a possible reversible implementation of addition of two natural numbers, $a$ and $b$.
    In order to be reversible, there must be a bijection between legal inputs and legal outputs and one such way is to retain one of the inputs---here $b$.
    This differs from the Bennett-style embedding, which would retain both inputs i.e.\ $(a,b)\leftrightarrow(a,b,a+b)$, and by being more parsimonious in our use of information we obtain the converse operation---subtraction---for free.}
  \label{fig:add-control-flow}
\end{figure}
Arguably the most challenging concepts to learn are those of reversible control flow: conditional branching and loops.
Provided the information used to discriminate the different cases in a branch is available before and after, then this branching is clearly reversible. 
The problem comes when merging these branches later: in irreversible languages this merging of control flow is implicit, but in a reversible language we must ensure that we are able to tell from which branch we came immediately after merging and so merging must be made explicit.
In fact, reversibility provides an elegantly symmetric solution to this without introducing a new primitive for merging.
The idea is that merging is simply branching in reverse, and so one reuses the branching primitive and supplies a set of orthogonal conditions that must be satisfied by each branch of the computation before control is merged.
If the information used to perform the original branch is preserved, then the original conditions can be reused; it is not necessary that the final conditions be the same as the initial, however, and this can often be used for good effect when data is transformed to a different---but provably isomorphic---representation.
A loop is similar to a two-branch conditional, except that instead of the execution direction on both sides being in the same direction, they are in opposite directions in order to produce a loop.
The consequence is that effectively the entry to a loop is a merge of control flow: one must know whether a loop iteration is being entered for the first time, or as a continuation, so that on reversal it is possible to `un-enter' the loop.
The termination condition remains the same as in the irreversible case: at the end of each iteration, it is checked whether we should continue looping or exit the loop.
These principles are illustrated abstractly in \Cref{fig:rev-control-flow}, more concretely in \Cref{fig:add-control-flow}, and more extensively in \Cref{chap:aleph}.

\enlargethispage{2\baselineskip}
We conclude with a brief review\footnote{%
  In some code snippets, notation of some arithmetic expressions and operators has been slightly modified where it is thought that the medium of mathematical typesetting provides greater clarity than the \tool{ASCII} source.}\ 
of some (classical) reversible programming languages. 

\langreview\janus To our knowledge, the first such programming language---\janus---was designed by \textcite{janus} in 1982 for a Caltech class project, and later formalised by \textcite{janus-analysis}. It is an imperative procedural language, and is named after the Roman god associated with duality and time. \Cref{lst:revex-janus} illustrates most of its syntax through three increasingly sophisticated examples, or see \Cref{dfn:janus-syntax} for a more detailed exposition.

\begin{listing}[htbp]
    \def\proc{\kw{procedure}}
  \def\If{\kw{if}}
  \def\Then{\kw{then}}
  \def\Else{\kw{else}}
  \def\Fi{\kw{fi}}
  \def\From{\kw{from}}
  \def\Do{\kw{do}}
  \def\Loop{\kw{loop}}
  \def\Until{\kw{until}}
  \def\call{\kw{call}}
  \def\swap{\Leftrightarrow}
\vspace{-\baselineskip}
\begin{sublisting}[t]{.27\linewidth}%
  \small\begin{align*}
    & n ~ x[2] \\[.5em]
    & \proc ~ \ident{fib} \\
    &\qquad \If ~ {n=0} \\
    &\qquad \Then ~ {x_0 \inpl+ 1} \\
    &\qquad \phantom\Then ~ {x_1 \inpl+ 1} \\
    &\qquad \Else ~ {n \inpl- 1} \\
    &\qquad \phantom\Else ~ \call ~ \ident{fib} \\
    &\qquad \phantom\Else ~ {x_0 \inpl+ x_1} \\
    &\qquad \phantom\Else ~ {x_0 \swap x_1} \\
    &\qquad \Fi ~ {x_0 = x_1}
  \end{align*}%
\end{sublisting}%
\begin{sublisting}[t]{.32\linewidth}%
  \small\begin{align*}
    & n ~ m ~ k \\[.5em]
    & \proc ~ \ident{fac} \\
    &\qquad m \inpl+ 1 \\
    &\qquad \From ~ m = 1 \\
    &\qquad \Loop ~ m \swap k \\
    &\qquad \phantom\Loop ~ \From ~ m = 0 \\
    &\qquad \phantom\Loop ~ \Loop ~ m \inpl+ n \\
    &\qquad \phantom{\Loop ~ \Loop} ~ k \inpl- 1 \\
    &\qquad \phantom\Loop ~ \Until ~ k = 0 \\
    &\qquad \phantom\Loop ~ n \inpl- 1 \\
    &\qquad \Until ~ n = 1 \\
    &\qquad n \inpl- 1 \\
    &\qquad m \swap n
  \end{align*}%
\end{sublisting}%
\begin{sublisting}[t]{.41\linewidth}%
  \small\begin{align*}
    & n ~ \ell[99] ~ \sigma[99] ~ i ~ j \\[.5em]
    & \proc ~ \ident{sort} \\
    &\qquad \From ~ {i=0} \\
    &\qquad \Loop ~ {j \inpl+ n-2} \\
    &\qquad \phantom{\Loop} ~ \From ~ {j=n-2} \\
    &\qquad \phantom{\Loop} ~ \Loop ~ \If ~ \ell_j > \ell_{j+1} \\
    &\qquad \phantom{\Loop ~ \Loop} ~ \Then ~ \ell_j \swap \ell_{j+1} \\
    &\qquad \phantom{\Loop ~ \Loop ~ \Then} ~ \sigma_j \swap \sigma_{j+1} \\
    &\qquad \phantom{\Loop ~ \Loop} ~ \Fi ~ \sigma_j > \sigma_{j+1} \\
    &\qquad \phantom{\Loop ~ \Loop} ~ j \inpl- 1 \\
    &\qquad \phantom{\Loop} ~ \Until ~ j = i-1 \\
    &\qquad \phantom{\Loop} ~ j \inpl- i-1 \\
    &\qquad \phantom{\Loop} ~ i \inpl+ 1 \\
    &\qquad \Until ~ i = n-1 \\
    &\qquad i \inpl- n-1
  \end{align*}%
\end{sublisting}
  \caption{Three select examples of \janus\ programs.}
  \label{lst:revex-janus}
\end{listing}

The first example (due to \textcite{janus-analysis}), $\ident{fib}$, is a recursive procedure which computes the $n^{\text{th}}$ and $n+1^{\text{th}}$ \emph{Fibonacci} numbers into $x_0$ and $x_1$ respectively, consuming $n$ in the process.
The second (due to us), $\ident{fac}$, is an iterative procedure that replaces $n$ with $n!$ (providing $n>0$), using $k$ as a `scratch' variables for implementing $m\leftrightarrow m\cdot n$ and $m$ to hold the burgeoning factorial as $n$ is consumed.
The third (due to \textcite{janus}), $\ident{sort}$, performs a bubble sort on the list $\ell$ of length $n<100$, and sets $\sigma$ to the permutation mapping between the original and sorted lists (as a precondition, $\sigma$ must be set to the identity permutation $(1\,2\,\cdots\,99)$).
The absence of dynamic-length arrays or memory allocation (and hence use of 99 as a fixed buffer size in the above) shows a limitation of \janus, though the augmentation of \janus\ to support such features would be a straightforward extension.

\begin{0dfn}[Syntax and semantics of \janus]\label{dfn:janus-syntax}
A program consists of a series of (global) variable declarations and procedure declarations.
Variables are either machine-precision integers, initialised to 0, or fixed-length arrays thereof as indicated by the square brackets.
Variables are modified by in-place updates, or can be swapped with $\Leftrightarrow$.
There are three possible in-place operators, $\inpl+$, $\inpl-$ and $\inpl\oplus$ (in-place bit-wise \gate{XOR}), and the right-hand side of a statement can be any compound expression making use of
  the unary operators $-$ (arithmetic negation), $\neg$ (bit-wise negation);
  the arithmetical binary operators $+$, $\times$, $-$, $/$ (quotient), $\%$ (remainder);
  the bit-wise binary operators $\oplus$, $\land$, $\vee$;
  or the Boolean binary operators $<$, $\le$, $>$, $\ge$, $=$, $\neq$
  where \textsc{False} is represented by 0 and \textsc{True} by $-1$.
These expressions need not be information-preserving, as the semantics make use of a Bennett embedding: computing in the forward direction the value, performing the modification, and then uncomputing the value.
As a result, it is illegal for the variable being modified to appear in the expression modifying it.

Procedures can be called in their forward direction by the aptly-named $\kw{call}$ statement, and in their reverse direction by $\kw{uncall}$.
\janus\ also supports conditionals and looping, per the aforementioned principles.
Conditionals are represented by $\kw{if}~p$ $[\kw{then}~\sigma]$ $[\kw{else}~\sigma']$ $\kw{fi}~q$.
They consist of a branch and unbranch condition, respectively $p$ and $q$, and optional statements to be executed in the \textsc{True} ($\kw{then}$) and \textsc{False} ($\kw{else}$) cases.
If $p$ evaluates to \textsc{True} before the branch, then $q$ \emph{must} evaluate to \textsc{True} afterwards, and similarly for the converse, or else a run-time error occurs.
Loops are similarly defined, represented by $\kw{from}~p$ $[\kw{do}~\sigma]$ $[\kw{loop}~\sigma']$ $\kw{until}~q$.
Upon reaching a loop, $p$ must evaluate to \textsc{True}. The statement(s) $\sigma$, if any, are then executed. Then $q$ is tested; if it evaluates to \textsc{False} then the statement(s) $\sigma'$ are executed, else the loop exits. Otherwise, $p$ is evaluated once again (and must this time evaluate to \textsc{False}), and execution continues from $\sigma$.
Not described here are the rudimentary I/O operations and \janus\ run-time system.
\end{0dfn}

\langreview\psilisp A little while later, in 1992, Henry Baker introduced two \lisp\ variants: \linlisp~\cite{linear-lisp} and \psilisp~\cite{psi-lisp}.
The former employs \emph{linear logic}~\cite{linlog}, in which every value must be consumed precisely once (usage optional), in order to do away with expensive garbage collection in functional languages.
This idea is seeing a resurgence, with support in recent updates to languages like \haskell~\cite{linear-hs}.
Whilst linear logic shares some properties with reversible computation, indeed quantum logic is often contextualised as a subset of linear logic, it is not restrictive enough to guarantee reversibility.
  Baker's etymologically-palindromic \psilisp,
however, does guarantee injectivity of its primitives and is therefore reversible.
It is itself a variant of \linlisp\ appropriately modified.
The most conspicuous modification is that of the $\kw{if}$-expression.
Baker introduces the merge condition in an asymmetrical way, as the return value is not necessarily named: the first two sub-expressions are, respectively, a Boolean \emph{expression} indicating which branch to take, and a Boolean \emph{predicate} to be applied to the result.
These predicates are handled specially, in comparison to the rest of the language, in that they are permitted to be non-invertible.
As with expressions in \janus, these predicates are run forwards to select a branch, and then backwards to erase the garbage.
Unfortunately the paper is somewhat lacking in precise detail of the language design, and so it is not clear how one defines new predicates (other than the mention of an implicitly instantiated history stack created when these predicates are evaluated), nor is it clear what happens if a predicate is called outside of such a scope.

\begin{listing}[htbp]
  \begin{minipage}{\linewidth}\begin{align*}
    & (\kw{defun} ~ \ident{fact} ~ (n) \\
    &\quad (\kw{assert} ~ (\Ident{and} ~ (\ident{integerp} ~ n) ~ ({>} ~ n ~ 0))) \\
    &\quad (\kw{if} ~ (\ident{onep} ~ n) ~ \#'\!\ident{onep} \\
    &\quad \phantom{(\kw{if}} ~ n \\
    &\quad \phantom{(\kw{if}} ~ (\times ~ n ~ (\ident{fact} ~ ({1-} ~\, n)))))
  \end{align*}\end{minipage}
  \caption{The (erroneous) reference implementation of the factorial function in \psilisp}
  \label{lst:revex-psilisp-orig}
\end{listing}

To showcase the syntax and semantics of \psilisp, \textcite{psi-lisp} provides a reference implementation of the factorial function (reproduced in \Cref{lst:revex-psilisp-orig}).
Unfortunately, this definition is \emph{not} reversible as it is not possible, given output 24, to mechanically invert the expression $(\times ~ n ~ (\ident{fact} ~ ({1-} ~~ n)))$ to obtain $n=4$.
In fact this is not a failure of the reversibility of \psilisp, but rather a failure to adhere to the principles of linearity: $n$ is used twice in this expression!
Moreover, the presumed definition of $\times$ is not invertible or possible.
This demonstrates the ease with which one can introduce a bug when programming reversibly, although were an interpreter/compiler for \psilisp\ available it would surely not accept the proposed program.

\begin{listing}[hbtp]
  \begin{minipage}{\linewidth}\begin{align*}
    & (\kw{defun} ~ \ident{first-onep} ~ (\ident{car}~\cdot~\ident{cdr}) ~ (\kw{discard}~\ident{cdr}) ~ (\ident{onep} ~ \ident{car})) \\
    & (\kw{defun} ~ \ident{fac} ~ (n ~ d) \\
    &\quad (\kw{assert} ~ (\Ident{and} ~ (\ident{integerp} ~ n) ~ ({>} ~ n ~ 0))) \\
    &\quad (\kw{if} ~ (\ident{onep} ~ n) ~ \#'\!\ident{first-onep} \\
    &\quad \phantom{(\kw{if}} ~ (\Ident{list}~n~d) \\
    &\quad \phantom{(\kw{if}} ~ (\kw{let}~((n'~d')~(\ident{fac}~({1-} ~~ n)~({1+} ~\, d)) \\
    &\quad \phantom{(\kw{if} ~ (\kw{let}~(}(n''~d'') ~ (\times~n'~d')) \\
    &\quad \phantom{(\kw{if} ~ (\kw{let}} ~ (\Ident{list}~{n''}~({1-}~\,d''))))) \\
    & (\kw{defun}~\ident{fact}~(n)~(\kw{let}~((m~1)~(\ident{fac}~n~1))~m))
  \end{align*}\end{minipage}
  \caption{A corrected implementation of the factorial function in \psilisp}
  \label{lst:revex-psilisp-fixed}
\end{listing}

In \Cref{lst:revex-psilisp-fixed}, we attempt to provide a corrected \psilisp\ implementation of the factorial.
This implementation is an iterative loop disguised in recursive form, as it does not appear that tail-recursion is possible in a reversible context due to the loss of information pertaining to recursion depth.
The original paper is somewhat light on details, and so this definition is somewhat speculative, particularly with regard to user-defined predicates: here we define an auxiliary function $\ident{first-onep}$ which returns whether its first argument is 1 or not, and assume the existence of a $\kw{discard}$ primitive to access the implicit history stack.
We also (reasonably) assume an invertible implementation of $\times$ in which $(\times~x~y)$ evaluates to $(\Ident{list}~xy~y)$, although the linearity restrictions in \psilisp\ mean we must refer to the output $y$ by a different name than the input $y$.
The parameter $d$ tracks recursion depth, increasing from an initial value of $1$ to a maximum of $n_0$, whilst $n$ is simultaneously decremented from $n_0$ to 1.
As we retrace our steps, we multiply $n$ by the depth counter $d$, and then decrement $d$, such that when we finally return from $\ident{fac}$ the depth counter has been restored to its initial value (e.g.\ 1).
The wrapper function $\ident{fact}$ handles supplying and consuming this depth counter, and assumes that \psilisp\ permits pattern matching on constants or some other way to consume a known constant.

\begingroup
  \def\dest{\psi'\!}
  \def\src{\psi}
  \let\alphas\cursivealphaS
  \def\fracprod{\mathbin{\times\kern-0.5ex/}}
  \def\access#1#2{({#1}~{\underline{~~}}~{#2})}
  \long\def\rfor#1=#2to #3 #4{(\kw{for}~{#1}={#2}~\kw{to}~{#3}~#4)}
  \long\def\defsub#1(#2)#3{(\kw{defsub}~\ident{#1}~(#2)~#3)}

\langreview\rlang Another language with \lisp-like syntax, \rlang\ (``yar''), was introduced in 1997 by \textcite{frank-r}. However, unlike \psilisp, the language is procedural and \langc-like---using the \lisp-like syntax for convenience of implementation (which is written in \comlisp). To the best of our knowledge, this is the first, and perhaps only, (classical) language for which a compiler targeting a reversible processor exists; namely, the aforementioned \emph{Pendulum Instruction-Set Architecture} (\pisa). Indeed, the example program in \Cref{lst:revex-r} was compiled to \pisa\ assembly and executed on an emulator (though, to our knowledge, not on a physical \pisa\ implementation such as \pendulum).

\begin{listing}[htbp]
  \begin{minipage}{\linewidth}\begin{align*}
    &\defsub pfunc (\dest~\src~i~\alphas~\varepsilon) {\\
    &\quad (\access{\dest}{i} \inpl+ (\access\alphas i \fracprod \access\psi i))\\
    &\quad (\access{\dest}{i} \inpl- (\varepsilon \fracprod \access\psi{((i+1)\land127)}))\\
    &\quad (\access{\dest}{i} \inpl- (\varepsilon \fracprod \access\psi{((i-1)\land127)}))} \\[.5em]
    &\defsub schstep (\psi_R~\psi_I~\alphas~\varepsilon) {\\
    &\quad \rfor i = 0 to 127 {
      (\kw{call}~\ident{pfunc}~{\psi_R}~{\psi_I}~{i}~{\alphas}~{\varepsilon})}\\
    &\quad \rfor i = 0 to 127 {
      (\kw{rcall}~\ident{pfunc}~{\psi_I}~{\psi_R}~{i}~{\alphas}~{\varepsilon})}}
  \end{align*}\end{minipage}
  \caption{A \rlang\ implementation of the discretised one-dimensional Schrödinger wave equation}
  \label{lst:revex-r}
\end{listing}

Semantically, \rlang\ is not dissimilar to \janus, but it does support local variable scoping which makes it more powerful as a structured programming language.
On the other hand, it does not (as of 1999) possess the ability to have different conditions on entry and exit for conditionals and loops, nor does it appear to support unbounded $\kw{for}$ loops.
Nonetheless, such features can be simulated and so this is not a fundamental limitation of the language.

The example in \Cref{lst:revex-r} begs some further discussion.
The Schrödinger equation was first discretised in a reversible way by \textcite{schro-rev}.
It has the remarkable property that, although its update rule is in some sense inexact, it nevertheless conserves amplitude (observed by Fredkin and proved by Feynman).
To understand the example, we first briefly sketch how the Schrödinger equation,
  $i\hbar\pdv{t}\psi = \hat H\psi$
which describes the time evolution of the wavefunction $\psi$ under the action of the Hamiltonian $\hat H$, may be discretised.
This equation has the formal solution $\psi(t)=\exp(\smash{-i\hat H t/\hbar})\psi(0)$; this is formal in the sense that $\hat H$ is an operator, which in common formulations takes the form of a Hermitian matrix or a partial differential operator, and so some care must be taken when manipulating it.
Supposing we wish to use a timestep $\delta t$, we can write $\psi(t+\delta t)=\exp(\smash{-i\hat H\delta t/\hbar})\psi(t)$ for any value of $t$.
It may be more convenient to supply the program with $\hat U\equiv \exp(\smash{-i\hat H\delta t/\hbar})$, having also pre-discretised $\psi$ such that $\hat U$ is a matrix and $\psi$ a vector, in which case the simulation can be performed exactly and reversibly\footnote{%
  It may be observed that matrix multiplication for non-singular $\hat U$ (which $\hat U$ is, being unitary) may be performed reversibly using Gaussian-elimination.
  Specifically, applying the row reduction algorithm to $[\hat U^{-1}|\mathbb 1|\vec\psi]$ yields $[\mathbb 1|\hat U|\hat U\vec\psi]$ where $\mathbb 1$ is the identity matrix.
  As $\hat U$ is unitary, computing its inverse is trivial (the inverse of a unitary matrix is its conjugate transpose).}.
Here, however, the program receives $\hat H$.
Specifically, we are considering a one-dimensional cyclic potential well in which $\hat H=-\frac{\hbar^2}{2m}\pdv[2]{x}+V$, and the value of the potential $V$ is supplied at each discretised point.
Expanding $\hat U$ for small $\delta t$ and discretising space, we find
$$\psi(x,t+\delta t)-\psi(x,t)=i\frac{\hbar\delta t}{2m\delta x^2}[\psi(x+\delta x,t)-2\psi(x,t)+\psi(x-\delta x,t)]-i\frac{V(x)\delta t}{\hbar}\psi(x).$$
Here the approaches of Fredkin and Frank differ; Frank identifies that the real part of the wavefunction at the next timestep is only informed by the imaginary part of the wavefunction at the previous and vice-versa, in order to use this directly.
In contrast, Fredkin seeks a more symmetric approach, transforming this to a form $\psi(x,t+\delta t)-\psi(x,t-\delta t)$: in this form the reversible implementation is more obvious without needing the observation of real-imaginary orthogonality, but one must retain copies of the wavefunction at two timepoints instead of the one in Frank's version.
With the discretised Schrödinger equation thus derived, Frank's program is a simple transcription of this (after performing the routine algebra to discriminate the real and imaginary parts) with the array $\alphas$ corresponding to $V\delta t/\hbar$ and $\varepsilon$ to the coefficient $\hbar\delta t/m\delta x^2$ (twice the phase rotation per time-step due to the particle's kinetic energy).
It also assumes space has been discretised 128-fold.
Finally we comment on the operator $\underline{~~}$ which indexes into an array, and the operator $\fracprod$, peculiar to \rlang, which returns the upper 32 bits of the 64-bit product of two 32-bit signed integers (the purpose being to implement the rounded product of an integer with a fixed point number in $[-1,1)$).

\endgroup
\begingroup
  \def\pname{\ident}
  \def\call#1(#2)#3 {\pname{#1}(#2)\pname{#3}}
  \long\def\defnn#1#2#3#4#5{\pname{#1}(#2)\,\{#5\}\,(#3)\pname{#4}}
  \long\def\defn#1(#2)#3 #4{\defnn{#1}{#2}{#2}{#3}{#4}}
  \def\v{}
  \def\nl{\\[.5\baselineskip]}

\langreview\kayak Perhaps marking the debut of reversible programming into public perception, Ben Rudiak-Gould introduced \kayak\ in 2002 as an entry~\cite{kayak} for the second instance of the Esoteric Awards (`\emph{Essies}'). The \emph{Essies} are an annual competition in design and use of \emph{esoteric} programming languages. Esoteric programming languages are not targeted for real-world use but instead are experiments in designing languages that may, for example, be intentionally difficult to program in or are very minimal whilst remaining Turing complete. The most well-known esoteric language is \Brainfuck~(\brainfuck), created in 1993 by \textcite{brainfuck}. As well as being known for being hard to program in (though certainly not the most difficult, the honour of which arguably goes to \malbolge~\cite{malbolge}), its syntactic minimalism and semantic simplicity make \brainfuck\ attractive for proving Turing-completeness of new languages and indeed the original distribution of \kayak\ included a \brainfuck-to-\kayak\ compiler (as does our language, \alethe, introduced in \Cref{chap:aleph}). 

\begin{0dfn}[Syntax and semantics of \kayak]\label{dfn:kayak-syntax}
  \kayak\ has very few primitives.
  A program is a series of procedures, each consisting of a bipartite name (e.g.\ $\pname{swap}()\pname{paws}$), an input and output set of $|$-delimited variables (these sets should be of equal length), and a body.
  Procedure bodies can contain one of four `commands': variable invocation, temporary register negation, conditional execution and procedure invocation.
  Variables are infinite stacks of bits, and invoking a variable will pop the top-most bit and place it in the `temporary' register if this register is empty, otherwise it will empty the temporary register onto the top of the named stack.
  The temporary register is local to procedure bodies and conditional blocks.
  Variables not named in the inputs/outputs are also local to procedure bodies (but not conditionals), and are initialised to an infinite number of 0s (and should be restored as such upon procedure exit).
  If the temporary register is occupied, then the $|$ command will negate its bit.
  The $[{}\cdots{}]$ command is a conditional block, which will run the nested commands if and only if the temporary register is occupied and set to 1.
  Finally, a procedure can be invoked by calling its bipartite name with some variables; for example, $\call swap(\v x|\v y)paws $ will swap the contents of the stacks $\v x$ and $\v y$.
  Procedures can be run in reverse by reversing their names; for example, a factorial calculated using the program in \Cref{lst:revex-kayak} can be un-calculated by calling $\pname{lairo}(n!)\pname{tcaf}$.
  Of course, if a name is palindromic then its reverse is inaccessible, and so it is recommend that only self-inverse procedures have palindromic names.
  A fascinating property of this syntax is that inverting a program is as simple as reversing the characters of its file and mirroring the characters $\{([<>])\}$.
\end{0dfn}

  \begin{listing}[hbtp]
    \begingroup
    \begin{minipage}[t]{.6\linewidth}\begin{align*}
        & \defn swap(\v a|\v b)paws {} \nl
        & \defn add(\v a|\v b)bus {\v a[~\call add({\v a|\v b})bus ~\v z|\v b~]\v a} \nl
        &\defn less-than(\v a|\v b|\v p)or-equal {\v a\\
          &\quad \phantom|[b[~\call less-than(\v a|\v b|\v p)or-equal ~]b] 
          ~|~ [~\v p|\v p~]\\
        &|\v a} \nl
        &\defn remquo(\v n|\v d|\v q)ddalum {\\
          &\quad \call less-than(\v d|\v n|\v p)or-equal \\
          &\quad p[~\call sub(\v n|\v d)dda ~\call remquo(\v n|\v d|\v q)ddalum ~]q \\
        &}
    \end{align*}\end{minipage}%
    \begin{minipage}[t]{.4\linewidth}\begin{align*}
        & \defn fact(\v n|\v d)orial' {\v n[\v n[ \\
          &\quad \v z|\v n ~ \v z|\v d \\
          &\quad \call fact(\v n|\v d)orial' \\
          &\quad \call muladd(\v n|\v d|\v z)ouqmer \\
          &\quad \call swap(\v n|\v z)paws \\
          &\quad \v d|\v z ~ \v n|\v z ~ \v n|\v z \\
        &]\v n]\v n} \nl
        &\defn fact(\v n)orial {\v z|\v d \\
          &\quad \call fact(\v n|\v d)orial' \\
        & \v d|\v z}
    \end{align*}\end{minipage}
    \endgroup
    \caption{Our \kayak\ implementation of the factorial function}
    \label{lst:revex-kayak}
  \end{listing}
  \begin{listing}[hbtp]
    \begingroup
      \begin{minipage}{\linewidth}\begin{align*}
        & \defn s(\v a|\v b). {} \\
        & \defn t(\v a|\v b). {\v a[\call t({\v a|\v b}). ~\v z|\v b]\v a} \\
        &\defn u(\v a|\v b|\v p). {\v a[b[\call u(\v a|\v b|\v p). ]b] 
          ~|~ [\v p|\v p]|\v a} \\
        &\defn v(\v n|\v d|\v q). {\call u(\v d|\v n|\v p). ~p[\call .(\v n|\v d)t ~\call v(\v n|\v d|\v q). ]q} \\
        & \defn w(\v n|\v d). {\v n[\v n[ \v z|\v n ~ \v z|\v d ~
          \call w(\v n|\v d). ~ \call .(\v n|\v d|\v z)v ~
          \call s(\v n|\v z). ~ \v d|\v z ~ \v n|\v z ~ \v n|\v z ]\v n]\v n} \\
        &\defn x(\v n). {\v z|\v d~\call w(\v n|\v d). ~\v d|\v z}
      \end{align*}\end{minipage}
    \endgroup
    \caption{Our \kayak\ implementation of the factorial function, `golfed' and obfuscated}
    \label{lst:revex-kayak-golfed}
  \end{listing}

\kayak\ seems inspired by \brainfuck\ in terms of its syntactic obscurity, as may be gleaned from \Cref{dfn:kayak-syntax}.
Nevertheless, it is possible to write useful programs in it as we demonstrate in \Cref{lst:revex-kayak} with an implementation of the factorial function.
In fact, the algorithm employed is much the same as that we used in \Cref{lst:revex-psilisp-fixed}.
Moreover, we also provide\footnote{\url{https://github.com/hannah-earley/kayak/blob/master/arith.kayak}} convenience routines for the reading and writing of base-10 \tool{ASCII} representations of numbers, thus showing that the presence of modular abstractions---even in a language as cryptic as \kayak---permits arbitrarily complex programs to be readily written.
In keeping with its obscure intention, we also provide a `golfed' variant in \Cref{lst:revex-kayak-golfed}.

\endgroup
\begingroup
  \def\treq{\overset{\scriptscriptstyle\triangle}=}
  \long\def\defn#1 #2={\ident{#1}~{#2} \treq}
  \long\def\case#1of{\kw{case}~{#1}~\kw{of}}
  \long\def\llet#1= #2 #3in{\kw{let}~{#1}=\ident{#2}~#3~\kw{in}~}
  \def\tuple<#1>{{\langle #1\rangle}}
  \def\Z{\ident Z}
  \def\S#1{\ident S(#1)}
  \def\dup<#1>{\lfloor\tuple<#1>\rfloor}

\langreview\langYAG Whilst \psilisp\ went some way towards the construction of a functional reversible programming language, its incomplete definition leaves it unclear whether this was achieved---not to mention that \lisp-esque languages are not deemed to be `pure' functional languages.
In 2012, \textcite{yag-lang} introduced a language---hereafter referred to as \langYAG---which seems to meet the criterion of being a functional language.
The example of computing Fibonacci numbers from the paper is reproduced in \Cref{lst:revex-yag}, and the semantics are summarised in \Cref{dfn:yag-syntax}.

\begin{listing}[htbp]
  \begin{minipage}{\linewidth}\begin{align*}
    & \defn plus \tuple<x,y> = \case y of \\
    &\phantom{\defn plus \tuple<x,y> ={}}\quad \rlap\Z\phantom{\S u} \rightarrow \dup<x> \\
    &\phantom{\defn plus \tuple<x,y> ={}}\quad \S u \rightarrow \llet \tuple<x',u'> = plus \tuple<x,u> in \tuple<x',\S{u'}> \\[.5em]
    & \defn fib n = \case n of \\
    &\phantom{\defn fib n ={}}\quad \rlap\Z\phantom{\S m} \rightarrow \langle\S\Z, \S\Z\rangle \\
    &\phantom{\defn fib n ={}}\quad \S m \rightarrow \llet \tuple<x,y> = fib n in\\
    &\phantom{\defn fib n ={}\quad\S m \rightarrow{}}\llet z = plus \tuple<y,x> in z
  \end{align*}\end{minipage}
  \caption{A \langYAG\ implementation of the Fibonacci numbers}\label{lst:revex-yag}
\end{listing}

\begin{0dfn}[Syntax and semantics of \langYAG]\label{dfn:yag-syntax}
  Like most functional languages, \langYAG\ is expression oriented.
  A primary difference, however, is that function application may only occur in $\kw{let}$ and $\kw{rlet}$ expressions; this increased indirection helps enforce the reversibility of the semantics.
  These expressions take the form $\kw{(r)let}~p=f~q~\kw{in}~e$ where $p$ and $q$ are what we shall call\footnote{In the original paper, these are called left-expressions by analogy to the common understanding of left-expressions in language syntaxes, however the occurence of \emph{left}-expressions on the \emph{right}-hand side of the equals sign is somewhat confusing.} \emph{patterns}, $f$ is a function name, and $e$ an expression in which to substitute the bindings of $p$.
  The $\kw{let}$ form runs $f$ forwards with input $q$, and the $\kw{rlet}$ runs $f$ backwards with output $q$.
  A pattern is either a variable, e.g.\ $x$, a constructor applied to some number of sub-patterns, e.g.\ the Peano numbers would be encoded inductively by the forms $\Z\equiv\Z()$ and $\S n$, or the special `duplication' pattern $\dup<x>$ (where $\tuple<x>$, $\tuple<x,y>$, etc are syntactic sugar for the tuple constructors $\ident{Tuple}_1(x)$, $\ident{Tuple}_2(x,y)$, etc). The `duplication' pattern is unique\footnote{Duplication and equality testing as concepts are of course not unique to \langYAG, but the particular place this concept holds in the syntax and semantics \emph{is} a seemingly novel feature.} to \langYAG; for a full explanation, the paper~\cite{yag-lang} is the best source, but essentially $\dup<x>$ reduces to $\tuple<x,x>$ whilst $\dup<x,y>$ reduces to $\tuple<x>$ if $x=y$, else to $\tuple<x,y>$, and thus by matching on these two forms via a $\kw{case}$ expression, equality can be tested (and exploited to un-copy in a safe way).
  The remaining expression form is given by $\case p of$ $\{q_i \rightarrow e_i\}_{i=1}^m$, where the body defines some number of branches $m$.
  Here also \langYAG\ takes a slightly different approach from other languages in this space.
  Whilst most languages require that the potential matching patterns $q_i$ be orthogonal, in \langYAG\ this is not the case; instead, the first pattern $q_j$ matching the input $p$ selects the branch taken.
  Similarly, once the branch returns an output from $e_j$, this must be unique \emph{only amongst the expressions $i=1\ldots j$} in contrast to most other languages where it must be unique amongst \emph{all} the expressions.
  Finally, a program is a series of function definitions of the form $f~p\treq e$.
\end{0dfn}

\endgroup
\begingroup
  \def\type{\mathrel{\ensuremath{:\kern0.11ex:}}}
  \def\nat{\mathbb{N}}
  \long\def\body#1{\left|~\begingroup%
    \arraycolsep=0pt%
    \def\arraystretch{1.2}%
    \begin{array}{l>{{}\leftrightarrow{}}l}%
      #1%
    \end{array}%
  \endgroup\right.}
  \long\def\defn#1::#2<->#3.#4;{\ident{#1}\!\type{#2}\leftrightarrow{#3}\\[-0.3em]&\body{#4}}
  \long\def\Defn#1::#2<->#3::#4.#5;{\defn{#1}::{#2}<->{{#3}\type\!\ident{#4}}.{#5};}
  \long\def\Where#1{\\[-0.3em]&\kern0.7em\begin{aligned}\kw{where}~#1\end{aligned}}
  \long\def\where#1::#2;{\Where{&\labl{#1}\type{#2}}}
  \def\labl#1{\ident{#1}\!} 
  \def\Labl#1.{\labl{#1}\mathbin{\$}}
  \def\iter{\Labl iter .}
  \def\Z{\Ident{Z}}
  \def\S#1{\Ident{S}~#1}
  \def\Left{\Ident{InL}}
  \def\Right{\Ident{InR}}
  \def\bind#1~#2:#3 {\ident{#1}~~\tilde{~}{#2}{:}\ident{#3}~}
  \def\trace{\bind{trace}~f:}

\langreview\theseus

\begin{listing}[htb!]
  \begin{minipage}[b]{.47\linewidth}\begin{align*}
    &\kw{type}~a+b = \Left~a ~|~ \Right~b \\
    &\kw{type}~a\times b = (a,b) \\
    &\kw{type}~\nat = \Z ~|~ \S\nat \\[.5em]
    &\defn trace :: {f{:}(a+b\leftrightarrow a+c)} \rightarrow b <-> c.
      \comment{\cmtSLopen\ definition omitted} ;\\[.5em]
    &\defn addSub :: \nat+\nat <-> \nat+\nat.
      \Left~(\S n) & \Left~n \\
      \Left~\Z & \Right~\Z \\
      \Right~n & \Right~(\S n); \\[.5em]
    &\Defn add$_1$ :: \nat <-> \nat :: sub$_1$.
      n & \trace addSub n; \\[.5em]
    &\Defn add :: \nat\times\nat <-> \nat\times\nat :: sub.
      a,b & \iter a,\Z,b \\
      \iter \S a, a', b & \iter a, \S a', \ident{add$_1$}\,b \\
      \iter \Z, a, b & a, b;
    \where iter :: \nat\times\nat\times\nat;
  \end{align*}\end{minipage}%
  \begin{minipage}[b]{.53\linewidth}\begin{align*}
    &\Defn cantorUnpair :: \nat <-> \nat\times\nat :: cantorPair.
      n & \iter n, (\Z, \Z) \\
      \iter \S n, (\Z, b) & \iter n, (\S b, \Z) \\
      \iter \S n, (\S a, b) & \iter n, (a, \S b) \\
      \iter \Z, (a, b) & (a, b);
    \where iter :: \nat\times(\nat\times\nat); \\[.616em]
    &\defn fib$'$ :: (\nat\times\nat) + \nat <-> (\nat\times\nat) + (\nat\times\nat).
      \Right~n & \iter (n, (\Z,\Z)) \\
      \iter (\S n, (a, b)) & \Labl iter' . (n, \ident{add}\,(b, a)) \\
      \iter (\Z, (a, b)) & \Right~(\ident{add$_1$}~a, \ident{add$_1$}~b) \\
      \Labl iter' . (n, (a, b)) & \iter (n, (a, \S b)) \\
      \Left~(n, (\S a)) & \iter (n, (\S a, \Z)) \\
      \Left~(n, \Z) & \Left~(\ident{cantorUnpair}~n);
    \Where{%
      &\labl{iter}\,\type\nat\times(\nat\times\nat)\\
      &\labl{iter'}\,\type\nat\times(\nat\times\nat)} \\[.616em]
    &\defn fib :: \nat <-> \nat\times\nat.
      n & \trace fib' n;
  \end{align*}\end{minipage}
  \caption{A \theseus\ implementation of the Fibonacci numbers}\label{lst:revex-theseus}
\end{listing}

The last language we will review is \theseus, whose name is inspired by the legendary \emph{Ship of Theseus} in that its computation `[replaces] values by apparently equal values'.
Based on the family of point-free\footnote{%
  Point-free languages do not support explicit variable binding, instead modelling computation as `flow through pipelines'; a point-free approach can often be used to elegant effect, and in languages supporting a mixed approach such as \haskell\ this is encouraged.
  Honourable mentions in this space are the languages \lang{Inv} and \lang{Fun} of \textcite{inv-lang}, whose motivation was in developing a language for bi-directional transformations.
  Bi-directional transformations are another thread in the field of reversible computation, useful for such applications as editors in which there is an isomorphism between the (editable) document description and the resultant document view.%
} combinator calculi $\Pi$ introduced by Roshan James for his graduate thesis~\cite{pi-calc}, \textcite{theseus} introduced an equivalent higher level language.
The motivation for $\Pi$ is to provide a Category-Theoretic underpinning of reversible programming: specifically, $\Pi^0$ corresponds to dagger symmetric traced bimonoidal categories.
Category Theory has long been a source of fruitful cross-fertilisation between the realms of mathematics and computer science, with monadic effects in \haskell\ a prime example.
The typical direction is to attempt to find a category that models an existing programming language, such as \ctgry{Hask}\footnote{%
  The familiar reader will object that \ctgry{Hask} is not, strictly speaking, a category due to the interaction of the primitive $\ident{seq}$ with bottom values.%
}
for \haskell; \theseus\ goes the other direction, deriving a language from a category.
The objective of this is as an aid to reasoning about categorical constructions rather than to produce a useful language as such, nevertheless a reference interpreter exists\footnote{Mirrored at \url{https://github.com/hannah-earley/theseus}}.

\begin{0dfn}[Syntax and semantics of \theseus]\label{dfn:theseus-syntax}
  \theseus\ is based on type isomorphisms, and has the strongest and most developed type system of those reviewed here.
  A \theseus\ program is a series of type definitions and isomorphism (`map') definitions, and its syntax largely resembles a subset of \haskell's.
  A map is defined by a name, a type isomorphism, an optional name for its inverse, and a series of pattern clauses.
  Unlike in \haskell, \emph{either} side of a pattern clause may invoke an isomorphism; if an invocation occurs on a `target' side then it will be called and its return value substituted, as expected, if it occurs on the `matching' side, however, then its inverse will be called and the result then matched.
  There are two further extensions to maps.
  Firstly, maps may be parameterised in order to simulate higher-order isomorphisms (this is just a simulation, as maps are not first-class values in \theseus).
  Secondly, maps may have typed `labels', syntactically represented by $\labl{label}\,\mathbin\$ v$ where $v$ is its appropriately typed argument.
  The function of labels is as a sort of `go-to' operator.
  Implicitly a map with type $a\leftrightarrow b$ and label types $\ell_1$ and $\ell_2$ is transformed to a map $(\ell_1+\ell_2)+a\leftrightarrow (\ell_1+\ell_2)+b$.
  Then, the map is called repeatedly (`traced') until it returns a value of type $b$ and can thus be used to implement looping.
  For example, if such a labelled map $f$ yields a loop of 3 iterations, then we would call the modified map $f'$ four times on $\Right~x$ and then match on the output, i.e.\ $\Right~y=f^4(\Right~x)$.
  A core difference between this language and the others reviewed here is that the clauses of a map must not only be orthogonal, but also \emph{exhaustive}.
  This complicates definitions of even something as simple as adding 1 to a natural number, because the output of $\ident{add$_1$}$ cannot be 0; to circumvent this, explicit divergence of these cases must be employed, i.e.\ introducing an infinite loop.
\end{0dfn}

The semantics of \theseus\ are summarised in \Cref{dfn:theseus-syntax} and illustrated in the example implementation of Fibonacci numbers in \Cref{lst:revex-theseus}.
In \Cref{lst:revex-theseus}, the definitions of $\ident{trace}$, $\ident{addSub}$, $\ident{add$_1$}$ and $\ident{add}$ are due to \textcite{theseus} whilst the remaining definitions, $\ident{cantorUnpair}$, $\ident{fib'}$ and $\ident{fib}$, are due to us\footnote{%
  The original paper also supplies a definition of $\ident{fib}$, but as a Bennett-esque embedding.}.
The aforementioned exhaustivity requirement makes this substantially more complicated than, e.g., that in \Cref{lst:revex-yag} because we must manufacture an exhaustive handling of the error condition. 
In $\ident{fib}$ this is very prominent because the subset of $\nat\times\nat$ corresponding to valid Fibonacci pairs has vanishing measure.
Additionally, we attempt to reduce the extraneous cases by internally operating on $(F_n-1,F_{n+1}-1)$.
James and Sabry do suggest at the end of their paper that it may be possible to extend \theseus\ with more conventional error handling, which we believe would be the more optimal choice for a programming language.
Explicit divergence has a number of downsides in our opinion: from a programmer perspective, coding up divergence conditions is quite cumbersome\footnote{We would estimate it occupied $95\%$ of our time in constructing the $\ident{fib}$ example!}, and there is no `correct' choice of how to do so as these cases are meaningless, thus making it more of an art than a science.
Furthermore, from a user's perspective this divergence behaviour is sub-optimal as one cannot distinguish a long-running/intentionally non-halting computation from an error condition or bug, not to mention the excessive resultant CPU utilisation.
Of course, these objections are besides the point of \theseus, but would serve as a small modification with disproportionate benefit for adoption of \theseus\ as a viable typed functional language for reversible computation.

\endgroup

We will revisit the matter of reversible programming languages in \Cref{chap:aleph}, but for the next three chapters we shall return to the physics of reversible computation.

\endgroup

\begingroup
\pdfsuppresswarningpagegroup=1

\begin{chapter-summary}

  In this chapter we will consider the maximum sustained computational performance---in terms of operations per second---that can be extracted from a given region of space, under fairly general assumptions. Namely, we assume the known laws of physics, and that one will need to supply the system with energy over time in order to keep it going. We show, similarly to early work by \textcite{frank-thesis}, that for any realisable computer reversible computers are strictly better than irreversible computers at any size. We also derive universal scaling laws describing just how much better a reversible computer could be compared to an irreversible computer, proving that the adiabatic \emph{Time-Proportionally Reversible Architectures} (TPRAs) of Frank are the best possible, and suggest a path to achieving this bound with molecular computers.
  
  To illustrate these results, summarised in \Cref{fig:scaling}, suppose we wish to build the most powerful computer we can within some spherical region of space, of radius $R$. If the computer is irreversible, such as conventional silicon-based processors, then it is found that essentially only the surface of this sphere can be used for computation and the interior must be empty or inert. This clearly limits the rate of computation proportional to $4\pi R^2$. The reason for this restriction is thermodynamic, arising from constraints on both supply of power and rejection of heat. If $R$ is very big, such that the computer threatens to collapse into a black hole, then it is found that the computational rate can now only scale proportional to $R$.
  
  For a reversible computer the situation is substantially improved. In principle, a reversible computer can operate without producing an increase in entropy and without requiring any input of energy, and so the rate of computation could scale with the volume. Unfortunately in practice this is not possible, and a more careful consideration of the system's entropy is required. The rate of computation is found to scale proportional to  $\frac{4}{\sqrt{3}}\pi R^{5/2}$, geometrically between the area and volume. For large $R$, general relativistic effects become significant and this falls to the more restrictive bound $R^{3/2}$. For even larger $R$ this eventually falls to proportional to $R$, coinciding with the irreversible computer. The reason for this additional threshold between scaling laws is that the sub-relativistic reversible computers---where gravitational effects on spacetime are negligible---are not taking full computational advantage of their volume due to thermodynamic constraints; as the system gets larger beyond the collapse threshold (when the geometry transitions to a thick shell rather than a sphere) the thermodynamic constraints gradually relax, allowing a scaling that `temporarily' exceeds the limit value of $R$.
  
  Therefore we see that at almost all scales, reversible computers substantially outperform irreversible ones, whilst at extremely small---typically sub-nanometer---and large---order of the visible universe in size---scales they coincide.

\end{chapter-summary}

\chapter{Limits on Computational Rates in Physical Systems}
\label{chap:revi}

\begin{figure}[h!]
  \centering
  \includegraphics[width=.9\textwidth]{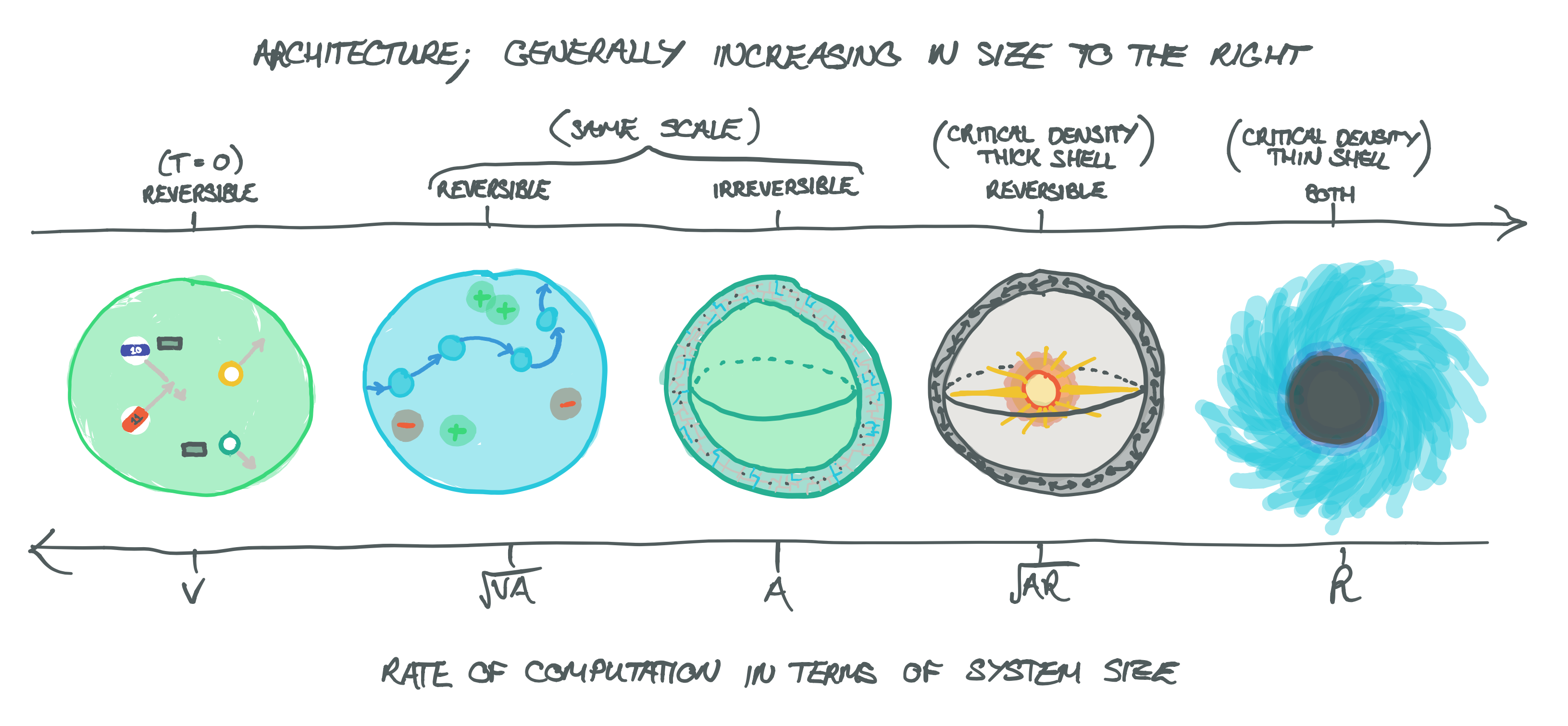}
  \captionsetup{singlelinecheck=off,type=figure}
  
    \newlength\figprefixlen\settowidth{\figprefixlen}{Figure~99:~}
    
  \caption[An illustration of the primary results of this chapter.]{%
    An illustration of the primary results of this chapter. Consider a spherical region of computational matter with radius~$R$, convex surface area~$A$ and enclosed volume~$V$. The expressions indicate how the computational rate of each architecture scales, proportionally. From left to right, these regimes are:%
    \\[1em]%
    \begin{tabularx}{.99\textwidth-\figprefixlen}{rX}%
        $V:$ & reversible computers at absolute zero; \\%
        $\sqrt{AV}:$ & reversible computers at finite temperature; \\%
        $A:$ & irreversible/canonical computers; \\%
        $\sqrt{RA}:$ & critical density (thick shell on the cusp of gravitational collapse); \\%
        $R:$ & critical density (thin shell on the cusp of gravitational collapse).%
    \end{tabularx}%
  }
  \label{fig:scaling}
\end{figure}

\section{Introduction}
\label{sec:intro}

\begin{figure}[tb!]
  \centering
  \begin{subfigure}[b]{.45\linewidth}\centering
    \includegraphics[width=.9\linewidth]{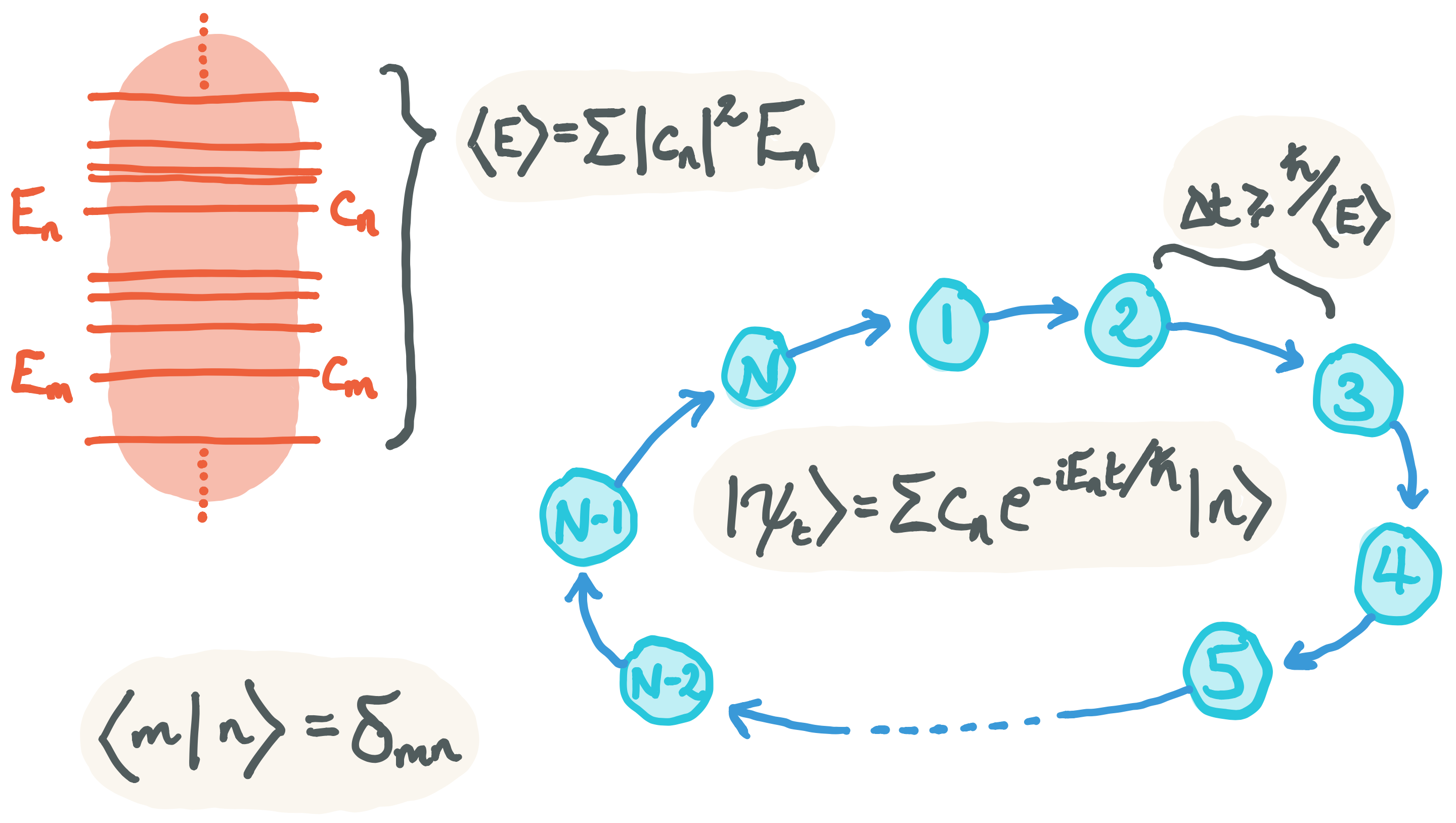}
    \vspace{1cm}
    \caption{}\label{fig:constraint-qm}
  \end{subfigure}%
  \begin{subfigure}[b]{.45\linewidth}\centering
    \includegraphics[width=.9\linewidth]{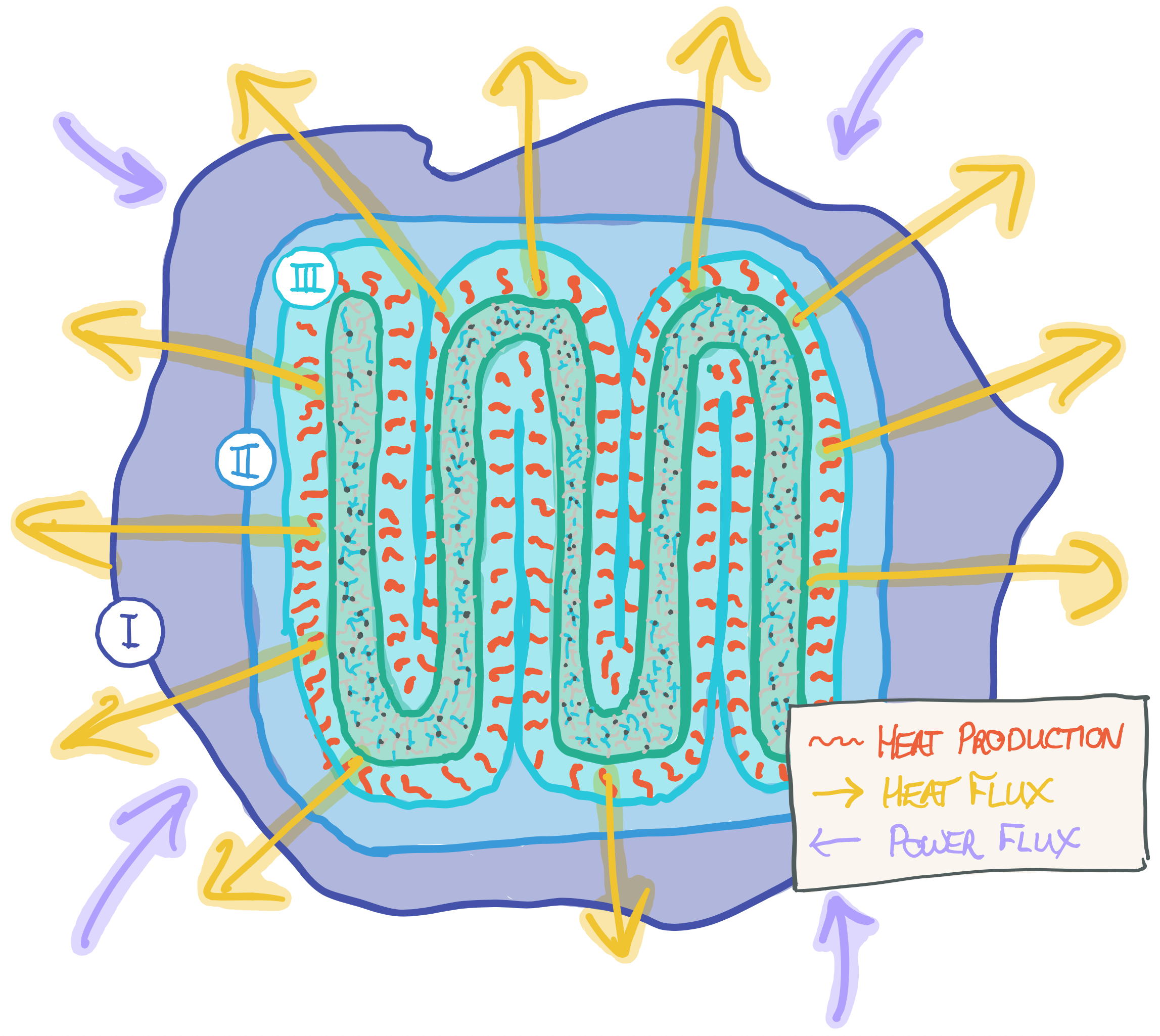}
    \caption{}\label{fig:constraint-td}
  \end{subfigure}
  \caption[An illustration of the three dominant physical constraints that apply to computational systems.]{An illustration of the three dominant physical constraints that apply to computational systems. Constraints (a) and (b) are described in \Cref{sec:intro} and (c) in \Cref{sec:gr}. (a) The quantum mechanical constraint bounding the minimum time for a state transition for a system of given average mass-energy $\evqty{E}$ where $E_0=0$. This illustration includes example energy levels of the eigenstates and a cyclic sequence of orthogonal states $\{1,\ldots,N\}$ which are each a superposition of the eigenstates, with the same average energy $\evqty{E}$. (b) A combination geometric-thermodynamic constraint. Thermodynamically, each computational operation generates entropy (generally heat) and this is bounded from below in the case of irreversible computation, but still non-zero in the reversible case. Geometrically, we can only dissipate this entropy at a rate scaling with the convex bounding surface, (II). The surface (III), whilst larger than (II), is not useful as the entropy flux must still pass through surface (II). Surface (I) is also larger, but the flux must first pass through surface (II). The green region is the computational system proper. At steady state there is an amortised balance between the entropy dissipation flux and input power flux. (continued on next page)}\label{fig:constraint}
\end{figure}
\begin{figure}[tb!]\ContinuedFloat
  \centering
  \begin{subfigure}[b]{.9\linewidth}\centering
    \vspace{.02\linewidth}
    \includegraphics[width=.9\linewidth]{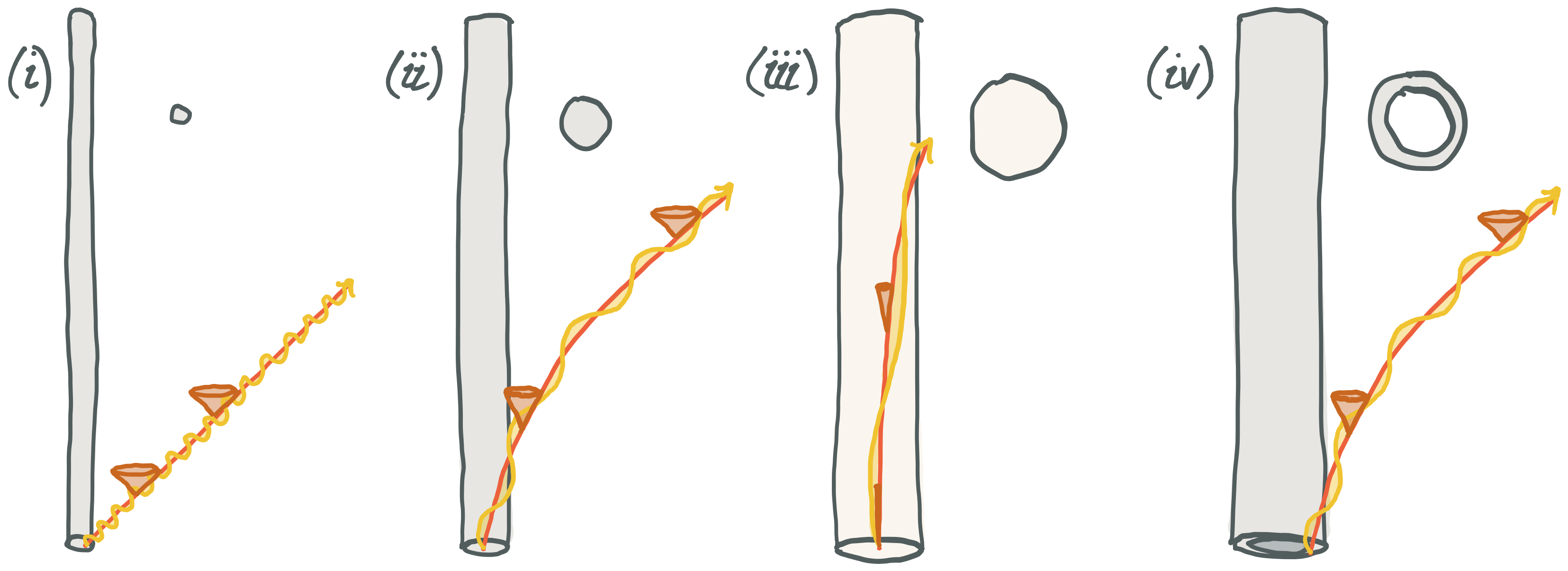}
    \vspace{.02\linewidth}
    \caption{}\label{fig:constraint-gr}
  \end{subfigure}
  \caption[]{(continued) (c) The (general) relativistic constraint which informs the optimal mass distribution of the system. Each example computational system is shown as a space time diagram, with the horizontal axis corresponding to space and the vertical to time. The above-right circles are cross-sections showing the mass distribution. Cones are light-cones showing the local causal structure and indicating spacetime curvature due to the stress-energy tensor. The emanating light waves correspond to the output of the computer. For small computers such as (i), general relativistic effects are negligible. For larger computers such as (ii), space is noticeably curved within the vicinity of the system and so the frequency of emitted photons is red-shifted, meaning that the computer appears slower to a distant observer. For computers on the cusp of gravitational collapse such as (iii)---even if the density is lowered to keep the Schwarzschild radius small, as indicated by the lighter shade of the system---light takes a very long time to escape, diverging to infinity as the system radius approaches $\frac98$ of the Schwarzschild radius.
  The factor $\frac98$ comes from the Chandrasekhar limit~\cite{schwarz-lims} and is explained in \Cref{sec:gr}.
  In fact, this is an ideal case; in practice the computational rate will vanish and the system will spontaneously collapse at even larger radii depending on the heat capacity ratio.
  To avoid this while keeping the same total mass and radius, the system should be reconfigured into a spherical shell of the same density as (i,ii), as depicted in (iv).
  The geometry can then be optimised to minimise the effect of time dilation to at most a threefold slowdown, thus retaining our qualitative power-law scaling.}
\end{figure}

The ubiquity of computer technology---and the increasing demands set upon it by intensive algorithms from fields such as machine learning, physics simulations, and cloud computing, among others---render the question of computer performance of considerable interest. Perhaps the most well known observation on computer performance was published by \textcite{moore-law}, who noticed the trend of microchip component density doubling every 18 months, later dubbed `Moore's law' in his honour. This law seemed to apply to many quantities in computing, including clock speed, \textit{FLOPS}\footnote{\textit{FLOPS}, or FLoating point Operations Per Second, is a common measure of supercomputer performance.} of the top 500 supercomputers combined, and inverse storage cost. A popular debate arose over when, if ever, Moore's law would stagnate; for clock speed, this point has come and gone, with consumer processor clock speeds frozen around \SIrange{3}{4}{\giga\hertz}. 

Whilst it is not surprising that contemporary technology may pose limits on computational performance, it is pertinent to ask whether the laws of physics impose hard bounds on performance. Indeed, our knowledge of quantum physics, thermodynamics and relativity reveals the answer to be affirmative, as hinted at in the introduction chapter. For a brief overview of some of these constraints see, for example, Lloyd's analysis~\cite{lloyd-ultimate}.

In this chapter, we investigate how these bounds vary as a given computational system is scaled up. The results we find apply, with different constants of proportionality, for any given computational architecture. 
A computer constructed from a mix of architectures can be treated as a linear combination, and therefore these scaling results apply to any computer.
This builds upon the work of \textcite{frank-thesis}, proving that his notion of adiabatic \emph{Time-Proportionally Reversible Architectures} (TPRAs) are the best possible, confirming the relevant scaling results in the mesoscopic regime (as well as analysing scaling in very small and very large regimes in more detail), and excluding the possibility of architectures of lower dissipation. Moreover, the detailed calculations herein yield specific constants of proportionality for broad classes of reversible computational architectures.

There are many metrics by which one may wish to measure computational performance. In this chapter, we shall investigate `speed' as defined by the rate of state transitions the computer is able to execute. As we are more concerned with the general form of scaling law rather than the specific constants for a specific architecture, the exact definition of `state' and `transition' are not too important. In general the most `obvious' definitions can be assumed, for example a conventional silicon-based computer would have as its states each possible value of all its registers, memory and any attached storage device, and as its transitions a single machine instruction\footnote{To be specific, we refer to the combination of a single op-code and its parameters.}. More rigorous definitions can be found in quantum mechanics, wherein a state would be an eigenstate of the computational basis Hamiltonian, and a transition would correspond to the complete evolution of the state between orthogonal eigenstates.

Other performance metrics of interest might concern synchronisation between distinct computational elements within a parallel system, and interaction with some arbitrary non-equilibrium system such as a supply of additional memory resources. These shall be covered in \Cref{chap:revii,chap:reviii}, respectively.

\para{Quantum Constraint on Computer Performance}

As a first approach to bounding computational performance, we turn to quantum mechanics. \textcite{bremermann-limit} gave a back of the envelope calculation in 1962, subsequently refined and elaborated upon by \textcite{margolus-levitin} and summarised in \Cref{fig:constraint-qm}, to show that a system with energy $E$ can change state at a maximum rate of $\nu\le E/h_P$ where $h_P$ is Planck's constant. Assuming this energy is due primarily to rest mass, this gives $\nu\le\SI{1.36e50}{\hertz\per\kilogram}$. Restricting our attention to a specific architecture as defined earlier, we assume a fixed density---or at least that the density is bounded from above---such that $\nu\le V\rho c^2/h_P$. This immediately implies that the maximum rate of computation scales with volume. 

\para{Thermodynamic and Geometric Constraint on Computer Performance}

This quantum limit assumes a closed computational system, requiring no external power source.
As we saw earlier this is not possible for conventional computing architectures, being irreversible and hence subject to the Landauer bound~\cite{landauer-limit,szilard-engine}.
To recall, the Landauer bound quantifies the increase in entropy that occurs whenever information is irreversibly discarded into the environment (as must occur when trying to overwrite data or otherwise reset some arbitrary state to a known state).
Explicitly, whenever a quantity of information $I\,\text{bits}$ is discarded, an entropy increase of $\Delta S\ge kI\log 2$ manifests where $k$ is Boltzmann's constant, with equality only in the limit of thermodynamic reversibility in which the discarding process takes an infinite period of time.

In order to sustain computation then, it is necessary to remove the additional entropy at the same amortised rate as its production. In the case of heat, this is achieved by cooling the system; this process can be generalised to other forms of entropy, such as disorder in a spin bath~\cite{landauer-spin}, but we shall restrict out attention to heat for conceptual convenience. We must thus enact a flow of work into the system and heat out in order to sustain computation, but here we encounter a geometric constraint. This energy flow applies at every bounding surface, assuming that the heat is ultimately rejected to infinity. Consider a convex bounding surface of area $A$, and assume our technology of choice is capable of transporting energy up to a certain flux $\phi$, then the maximal power that can be exchanged with the system is given by $P\le\phi A$. As has been established, however, the rate of heat generation scales with the rate of computational transitions for an irreversible computer, and therefore the rate of computation is ultimately bounded by the system's convex surface area, rather than its volume as contended by the quantum mechanical limit. See \Cref{fig:constraint-td} for an illustration.

This is far more restrictive than the volumetric upper bound derived earlier, and suggests that for large computers only a subvolume equivalent to the outermost shell of the volume can perform useful computation with the remaining bulk relegated to dormant inactivity\footnote{The remaining bulk need not be entirely devoid of purpose. It may, for example, be used for storage and memory. It must however have negligible entropy generation, such as that resulting from structural deterioration due to cosmic rays, spontaneous tautomerisation, etc.}. That is, we take the more restrictive of the two bounds, which for small systems is the volumetric bound and for larger (irreversible) systems is this areametric bound.

\para{Ballistic Computation}

Does this areametric thermodynamic bound render the volumetric quantum bound inaccessible? A trivial exception is found at small scales, where the surface area to volume ratio becomes so large that the volumetric bound is in fact more restrictive than the areametric bound, and thus takes priority. To more robustly improve on the rate of computation, however, we must reduce the entropic cost associated with a computational transition.
As was elaborated upon in the introduction chapter, this necessitates conserving information during computation by using the concept of \emph{reversible} computation.
Whilst the utility of reversible computation was initially controversial, with Landauer initially~\cite{landauer-limit} believing it to not be useful, it was soon shown by \textcite{bennett-tm} that reversible computation was viable.
Moreover, he showed how any irreversible program could be efficiently simulated on such a reversible computer without excessive memory overhead, thus proving that reversible computing was equipotent with conventional computing. For a brief exposition of reversible programming, refer back to \Cref{chap:intro}.

An idealised reversible computer would produce no entropy in its transitions.
\textcite{fredkin-conlog} described such an idealised system, a physical model of reversible computing that could operate without being actively powered.
This model consisted of a frictionless table upon which hard elastic billiard balls were projected.
The balls would bounce off of each other and some strategically placed walls, with the precise configuration describing the resultant computation.
As a testament to its capability, Fredkin's student, \textcite{ressler}, proposed the design of a fully programmable billiard ball computer, complete with arithmetic logic unit.
See \Cref{fig:gate-fredkin-bb} for an illustration of the general principles of the billiard ball model.
Given that such a computer would be powered solely by its initial kinetic energy, the volumetric quantum bound would be attainable. Unfortunately, \textcite{bennett-rev} calculated that such a design would ultimately be infeasible as even the most distant and subtle influences would be sufficient to rapidly and completely thermalise the motions of the system.
Such influences could in principle be suppressed by error correction mechanisms, but error correction corresponds to the discarding of excess entropy, and would seem to reconstitute the very issues we were trying to avoid. 

Our last refuge against unwarranted entropic influences is in quantum ground state condensates. Phenomena such as superfluidity and superconductivity arise in a sufficiently cooled system, wherein a macroscopic fraction of particles are found to inhabit their ground state, manifesting macroscopic quantum behaviours and a subsystem with vanishing entropy and temperature. The utility of such systems is manifold, but unfortunately it is doubtful that they can be exploited for dissipationless computation; the reason is that, in order to enact transitions to orthogonal quantum states, the Schrödinger equation requires us to prepare a superposition of (distinct) energy states which then rotates under the action of the Hamiltonian. The presence of non-ground eigenstates will necessarily result in a renewed vulnerability to thermodynamic perturbations. Nevertheless, such cold temperatures are not altogether useless, as unwanted fluctuations would still be significantly suppressed.

\para{Brownian Dynamics}

The fluctuations and dissipation that result from these thermodynamic perturbations lead to diffusion of the distribution in phase space, and hence loss of absolute control over it. Consequently there will be some level of unpredictability to the state of the system and its (generalised) velocities. In the limit of negligible control, the dynamics are (mostly) dominated by the noise source; this regime is referred to variously as Brownian or Langevin dynamics. In \Cref{sec:crn} it shall be shown how to leverage the minimal extant control over the system to make net progress, and to even robustly surpass irreversible computers in performance.

\para{Summary}

Whilst fully dissipationless computation is all but impossible, it is still possible in principle to tune the system so as to bring the entropic transition cost as close to zero as desired. In exchange, the rate of computation may itself be diminished. In the chapter that follows, this compromise is evaluated across the range of known physics---from quantum to classical, non-relativistic to relativistic---covering the full spectrum of feasible computational architectures. In so doing, a universal scaling limit is discovered, exceeding that of irreversible computers yet still falling short of the volumetric quantum bound.

\section{Quantum Ballistic Architectures and the Quantum Zeno Effect}
\label{sec:qze}

We first consider a quantum architecture. A viable quantum architecture must be ballistic in the sense of following an exact prescribed trajectory in phase space, as to do otherwise would lead to decoherence undermining the computational state. Of course we have established that a ballistic quantum system cannot truly be isolated from entropic effects; in particular, the third law of thermodynamics precludes lowering the temperature of any system to absolute zero, and so we must incorporate a heat bath into our analysis.

Suppose the ballistic Hamiltonian is given by $H_0$ and the perturbative effect of the heat bath by $V(t)$, such that the true Hamiltonian is $H=H_0+V$. Further, let the initial density be $\rho$ and introduce the projector $P$ corresponding to the subspace of valid states, such that $P\rho=\rho P=\rho$. To ensure processive computation, we must therefore periodically correct errors introduced by the perturbation. The cost of such error correction will be bounded from below by the entropy increase due to $V$, which can be expressed in terms of the probability of an error $\delta p$ corresponding to the erroneous subspace $P^\perp=1-P$. 

If corrections are to be made at intervals $\delta t$, then this error will be found to be $\delta p(t)=\tr[P^\perp \rho(t+\delta t)]$.
For a time-dependent Hamiltonian, the quantum Liouville equation tells us that $\dot\rho=i[\rho,H/\hbar_P]$ where $\hbar_P$ is Planck's constant.
We can expand $\rho$ as
\begin{align*}
  \rho(t+\delta t) &= \rho + \delta t \pdv{\rho}{t} + \tfrac12\delta t^2 \pdv[2]{\rho}{t} + \bigO{\delta t^3} \\
  &= \rho + i[\rho, \varepsilon] +
    \tfrac12( i\delta t[\rho,\dot\varepsilon] - (\varepsilon^2\rho - 2\varepsilon\rho\varepsilon + \rho\varepsilon^2)) + \bigO{\delta t^3}
\end{align*}
where we have omitted explicit time dependence for brevity and written $\varepsilon\equiv H\delta t/\hbar_P$. We then find
\begin{align*}
  \delta p &= \tr[P^\perp (\rho + i[\rho,\varepsilon + \tfrac12\delta t\dot\varepsilon] - \tfrac12(\varepsilon^2\rho - 2\varepsilon\rho\varepsilon + \rho\varepsilon^2))] + \bigO{\varepsilon^3} \\
    &= \tr[P^\perp \varepsilon\rho\varepsilon] + \bigO{\varepsilon^3}
       \equiv \tr[\rho\varepsilon P^\perp \varepsilon] + \bigO{\varepsilon^3}
\end{align*}
where the last line uses the fact that $P^\perp\rho=\rho P^\perp=0$. Note that if the $\rho\varepsilon P^\perp\varepsilon$ term vanishes, then so do all higher terms, and therefore this term is always the leading order approximant.

In order to evaluate this trace, we first write $H_0$ and $V$ as block matrices in the basis $(P,P^\perp)$, yielding
\begin{align*}
  H_0 &= \begin{pmatrix}
    h_{00} & 0 \\ 0 & h_{11}
  \end{pmatrix}\,, &
  V &= \begin{pmatrix}
    v_{00} & v_{01} \\ v_{10} & v_{11}
  \end{pmatrix}\,.
\end{align*}
Using the idempotence of projectors, i.e.\ $P^\perp\equiv P^\perp P^\perp$, and the fact $\rho\equiv P\rho$, we can see that $\rho HP^\perp H\equiv \rho VP^\perp V=\rho(V^2-VPV)$. Therefore, we have
\begin{align*}
  \delta p &= \qty\Big(\frac{\delta t}{\hbar_P})^2 (\evQty{V^2} - \evQty{VPV})\,.
\end{align*}
The first thing to notice about this expression is that it varies with the square of $\delta t$. This is characteristic of the Quantum Zeno Effect (QZE~\cite{qze})---that is, in the limit of constant measurement (i.e.\ $\delta t\to 0$), we can freeze dynamical evolution. Here, we are performing a partial measurement by projecting onto the subspace of either valid or erroneous computational states, thus allowing computational evolution to continue. Of course, the Zeno rate itself is subject to the Margolus-Levitin limit~\cite{margolus-levitin} and so such constant measurement is not possible.

To proceed in further determining $\delta p$, and thereby $\delta h$, it is important to elaborate on certain architectural details. Margolus and Levitin's analysis shows that one gets the same $E/h_P$ total rate regardless of whether one considers transitions of the system as a whole (`serial') or the combined transitions of a system partitioned into subsystems (`parallel'). For a physically realistic system, measuring large subsystems as a whole is impractical and so a more fine-grained architecture is likely preferable. Nevertheless, we shall proceed generally, finding the result is independent of this detail.

Let the subsystems be indexed by $i\in\mathcal I$. For a fully serial system, $|\mathcal I|=1$. We allow each subsystem to evolve independently, assuming that interaction events are negligibly frequent relative to our Zeno corrections. Specialising our above result for $\delta p$, we have
\begin{align*}
  \delta p_i &= \qty\Big(\frac{\delta t_i}{\hbar_P})^2 (\evQty{V_i^2}-\evQty{V_iP_iV_i})
\end{align*}
The strength of the perturbations $V_i$ will be controlled by an effective temperature for that subsystem, $T_i$. Suppose the subsystem has $n_i$ degrees of freedom\footnote{Excluding any frozen out by low temperatures.}, hereafter referred to as \textit{computational primitives}. By the equipartition theorem, the energies of each primitive will be individually perturbed. For an uncorrelated perturbation, we expect that $\overline{\evQty{V}}=0$ where $\overline x$ is the time-average value of $x$. Therefore, $\overline{\evQty{V^2}}\equiv\var_{\evQty{\bar\cdot}} V$. The expected aggregate perturbation over all primitives, then, is $n_i (k T_i)^2$ assuming $V$ is Gaussian and where $k\equiv k_B$ is Boltzmann's constant.

As for the second term, $VPV$, we expect that on average $V$ will mix (near-)degenerate states indiscriminately. $VPV$ represents the probability that the perturbation stays within the intended computational subspace. The state space can be divided into disjoint computational subspaces, within which the dynamics perform a reversible cycle over computational states. Let the total state space have cardinality $\Omega_i$, and the number of microstates associated with each subspace be $\omega_{ij}$ where $j$ indexes the different subspaces such that $\Omega_i=\sum_j\omega_{ij}$. Assuming similar energy distributions between subspaces, the chance of remaining within the original subspace is $\omega_{ij}/\Omega_i$. We further approximate this by $\omega_i/\Omega_i$ where $\omega_i=\evQty{\omega_{ij}}_j$ is the average subspace cardinality. Equivalently, we can define $\Omega_i/\omega_i \equiv W_i$, the number of distinct programs that the subsystem is capable of executing. Therefore,
\begin{align*}
  \overline{\delta p_i} &= n_i \qty\Big(\delta t_i \frac{k T_i}{\hbar_P})^2 (1-\sfrac1{W_i})
\end{align*}
In fact, we may wish to correct the second term errors in which we jump within the same subspace, as this may interfere with synchronisation between different subsystems (or it may be difficult to ascertain whether we have jumped within the same subspace). In this case, the expression reduces to
\begin{align*}
  \overline{\delta p_i} &= n_i \qty\Big(\delta t_i \frac{k T_i}{\hbar_P})^2 (1-\sfrac1{\Omega_i})
\end{align*}
In any case, $1-\sfrac1{W_i}$ is always finite and $\in[\tfrac12,1)$ for any useful subsystem; that is, a useful subsystem should have $W_i\ge2$. For composite/non-primitive subsystems, we would like the number of useful programs to scale exponentially with the number of primitives, i.e.\ $W_i=g_i^{n_i}$ for some $g_i$ (typically on the order of unity but greater than 1). For a composite subsystem of even moderate size, this exponential scaling will render the $1-\sfrac1{W_i}$ term effectively unity.

We are now able to determine the entropy increase between Zeno cycles. The indiscriminate mixing of the perturbation means that, in the event of an error, the entropy attains its maximum value. This yields the following expression for the information entropy\footnote{We use $H$ to refer to the information theoretical entropy, connected to the thermodynamical entropy as $S=k_BH$. Furthermore, we use the lowercase quantity $h$ to refer to the entropy of a subsystem or particle.},
\begin{align*}
  \delta h_i &= [-(1-\delta p_i)\log(1-\delta p_i) - \delta p_i\log\delta p_i
      + (1-\delta p_i)\log\omega_i + \delta p_i\log\Omega_i] - [\log\omega_i] \\
      &= \delta p_i (1 + \log\tfrac{\Omega_i}{\omega_i} - \log\delta p_i) - \bigO{\delta p_i^2} \\
      &\equiv \delta p_i (\log W_i + 1 - \log\delta p_i) - \bigO{\delta p_i^2} \\
      &\ge n_i\delta p_i \log g_i - \bigO{\delta p_i^2}\,.
\end{align*}
This expression assumes ignorance of states within the current computational subspace; again, we may wish to correct these errors too, in which case we may simply substitute $\omega_i=1$. This would effectively lead to taking $W_i=\Omega_i=g_i'^{n_i}$ where $g_i'\ge g_i$, and so the inequality remains valid.

Putting these together, we get
\begin{align*}
  \dot h_i &\ge \frac{n_i^2}{r_{Z,i}} \underbrace{\zeta_i\qty\Big(\frac{kT_i}{\hbar_P})^2 \log g_i}_{\equiv\gamma_i}
\end{align*}
where $r_{Z,i} \equiv 1/\delta t_i$ is the Zeno measurement rate for the subsystem and $\zeta_i=1-\sfrac1{W_i}\in[\tfrac12,1)$. The aggregate rate of entropy generation is given by $\dot H = \sum_i \dot h_i$ and is subject to the constraint $k\hat T\dot H\le P$ where $\hat T$ is the system temperature that governs the Landauer bound and $P=\phi A$ is the heat dissipation power, proportional to the system's surface area. It may seem surprising that $\dot h_i\propto n_i^2$. In fact, this only applies for sufficiently small $\delta p_i$; when $\delta p_i$ becomes significantly large, $\dot h_i$ approaches its maximum of $n_i r_{Z,i} \log g_i$.

We wish to maximise the computational rate, subject to this constrained entropy production. Per Margolus and Levitin, this rate depends on the combined energy of the computational primitives. If the average energy per primitive in subsystem $i$ is $\varepsilon_i$, then the energy available for computation is $E_C=\sum_i \varepsilon_i n_i$ and the computational rate is subject to $R_C\le E_C/h_P$.

Introducing the Lagrangian multiplier $\alpha$, we wish to maximise $\Lambda$ with respect to $n_i$;
\begin{align*}
  \Lambda = \sum_i\varepsilon_i n_i - \frac{1}{2\alpha} \sum_i\frac{n_i^2}{r_{Z,i}} \gamma_i 
  \qquad\implies\qquad
  0 = \varepsilon_i - \frac1\alpha \frac{n_i}{r_{Z,i}} \gamma_i\,.
\end{align*}
This gives $\dot h_i \ge \alpha \varepsilon_i n_i$ and so
\begin{align*}
  \dot H \ge \alpha E_C\,;
\end{align*}
solving for $\alpha$ and summing over $i$,
\begin{align*}
  \alpha r_{Z,i} &= \frac{n_i \gamma_i}{\varepsilon_i} \equiv \frac{n_i\varepsilon_i \gamma_i}{E_C\varepsilon_i^2} E_C \\
  \alpha R_Z &= E_C \evqty\Big{\frac{\gamma_i}{\varepsilon_i^2}}_{E_{C,i}}\,.
\end{align*}
where the average is taken over the computational energy distribution and $R_Z$ gives the net Zeno rate of the system as a whole. Substituting into the $\dot H$ constraint and making use of the fact that $R_Z\le E_Z/h_P$,
\begin{align*}
  \frac{P}{k\hat T} \ge \dot H &\ge \frac{E_C^2}{R_Z} \evqty\Big{\frac{\gamma}{\varepsilon^2}} \\
  \frac{P}{k\hat T}\frac{E_Z}{h_P} \evqty\Big{\frac{\gamma}{\varepsilon^2}}^{-1} &\ge E_C^2\,
\end{align*}
where $\hat T$ can be called the `Szilard' temperature: the temperature of the system in which entropy is generated and must be subsequently erased.
Finally, we obtain an expression for the net computational rate $R_C$,
\begin{align*}
  R_C &\le \sqrt{\frac{P}{k\hat T} \frac{E_Z}{h_P} \evqty\bigg{\zeta_i\qty\bigg(\frac{2\pi kT}{\varepsilon})^2 \log g}^{-1}}
  = \frac1{2\pi} \evqty\bigg{\frac{\zeta_i\log g}{\beta^2\varepsilon^2}}^{-1/2} \sqrt{\frac{P}{k\hat T}\frac{E_Z}{h_P}}\,.
\end{align*}
Now, as the upper bound of $P$ is proportional to the bounding surface area, and the upper bound of $E_Z$ is proportional to the system's volume, we get the scaling law
\begin{align*}
  R_C &\lesssim \sqrt{AV} \sim V^{5/6}\,,
\end{align*}
consistent with the scaling law for adiabatic architectures found by \textcite{frank-thesis} and showing it to be an upper bound.

We note that \textcite{levitin-toffoli-full} derive a related expression for the rate of energy dissipation/heat production within a quantum harmonic oscillator (QHO),
\begin{align*}
  P &= \frac{\varepsilon h_P R_C^2}{(1-\varepsilon)N}\,,
\end{align*}
where $\varepsilon$ is the probability of error per state transition and $N$ is the number of QHO energy levels. Given that their analysis assumes maximal $R_C$, however, this simplifies to
\begin{align*}
  P &= \frac{\varepsilon\Delta E}{1-\varepsilon} R_C
\end{align*}
where $\Delta E$ is the separation of energy levels. Given that $\varepsilon$ is assumed constant in the analysis, this then yields an areametric bound on $R_C$. Our analysis permits $\varepsilon$ to vary, and its optimum value as a function of system parameters is obtained. We also remain general in the systems studied, and thus confirm the conjecture that for a fixed error probability the energy-dissipation rate increases quadratically with the rate of computation.

We also note the recent results\footnote{The author gratefully acknowledges Mike Frank for pointing them towards this body of work.} of \textcite{pidaparthi-lent} on the Landau-Zener effect (LZE)~\cite{lze1,lze2}. In a closed quantum system, with no thermal coupling to the external environment, the Landau-Zener effect predicts that the energy dissipation decays \emph{exponentially} with the switching time of the system (the inverse of the computation rate). That this is non-zero even in the absence of an external environment shows the difficulty in achieving ballistic computation even in `idealised' conditions, but it certainly points towards a route to achieving near-ballistic computation in regimes in which thermal coupling is negligible. As expected, once thermal coupling is re-introduced to the system Pidaparthi and Lent show how the adiabatic behaviour is recovered. Moreover, their numerical results indicate that the LZE does not appear to be a practicable approach to achieving super-adiabaticity, except perhaps in a very narrow region of system conditions and size. Moreover, the results of Pidaparthi and Lent show the occurrence of a local minimum of dissipation with respect to switching time that would lead to engineering difficulties in increasing computer size for this parameter range. In any case, our asymptotic results stand, although it would be instructive to characterise the parameter range in which the LZE is exploitable and in which the aforementioned engineering difficulties arise.

\section{Classical Architectures I: A General Lagrangian Approach}

If use of the QZE is not possible, perhaps because of high temperatures or an excessive Zeno rate, then Fermi's golden rule applies, giving $\delta p\propto\delta t$. Consequently $\dot h$---and thus $\dot H$---is subject to a finite lower bound. We find $\dot h\ge n\dot p\log g$, yielding $\dot H\ge E_C\evqty{\dot p\log g/\varepsilon}$. Combining with our constraint $\dot H\le\hat\beta P$, the inequality in $R_C$ becomes
\begin{align*}
  R_C &\le \evqty{\frac{\dot p\log g}{\varepsilon/h_P}}^{-1}\frac{P}{k\hat T}
\end{align*}
where $h_P$ here is Planck's constant, and so the scaling law falls to $R_C\lesssim A\sim V^{2/3}$. That is, any ballistic system not making use of the QZE will be subject to the same areametric limit as irreversible computers.

The breakdown of the QZE corresponds to an increase in thermal coupling. By abandoning the desire to maintain well defined quantum states, we instead enter the incoherent regime of the classical realm wherein quantum effects are significantly suppressed. Proceeding generally, we adopt a Lagrangian formalism in order to remain noncommittal in our choice of coordinates $\vec q$.

We introduce a reasonably general Lagrangian, $\mathcal L=T-V$,
\begin{align*}
  \mathcal L &= [\tfrac12 c_{ij}\dot q_i\dot q_j] - [V + w_i\dot q_i + \mathcal O(\dot{\vec q}{\,}^3)] \equiv \tfrac12\dot q^\top C\dot q - V - W\dot q + \mathcal O(\dot q^3)
\end{align*}
where the coefficients $c$, $V$, $w$...\ may depend continuously on the coordinates $\vec q$, and Einstein notation is used for repeated indices. We have rewritten the Lagrangian in matrix form for concision, where $C$ is a matrix, $V$ a scalar, $W$ a row vector, and $q$ a column vector. We assume the system is (effectively) closed, self-contained, or otherwise independent of its environment, and thus $\mathcal L$ will not formally depend on the time parameter. In addition, this property implies invariance under global translation of any of its coordinates and thus conservation of all associated momenta. We obtain the Hamiltonian $\mathcal H$ and generalised momenta thus
\begin{align*}
  \mathcal H &= \pdv{\mathcal L}{\dot q_i}\dot q_i - \mathcal L = \tfrac12\dot q^\top C\dot q + V + \mathcal O(\dot q^3) \\
  p_i &= \pdv{\mathcal L}{\dot q_i} = C\dot q - W + \mathcal O(\dot q^3)
\end{align*}
Notice that the terms of our potential linearly dependent on the velocities drop out of the Hamiltonian. Those dependent on the cube or higher remain, but will turn out to be negligible at the velocities optimal for computation (and are typically not physically relevant in any case).

For brief collisions, the change in coordinates is negligible and so the spatiotemporal variance in coefficients can be neglected. In this case, we should conserve the Hamiltonian and momenta. We introduce superscripts to the coordinates, $q_i^{(n)}$, to represent different particles. We also introduce notation for the velocities before, $u$, after, $v$, and their change, $\Delta u=v-u$. For a collision between particles $m$ and $n$,
\begin{align*}
  \Delta\mathcal H &= [\tfrac12 v^{(m)\top}C^{(m)}v^{(m)} - \tfrac12 u^{(m)\top}C^{(m)}u^{(m)}] + [\tfrac12 v^{(n)\top}C^{(n)}v^{(n)} - \tfrac12 u^{(n)\top}C^{(n)}u^{(n)}] \\
    &= \tfrac12 (u^{(m)} + v^{(m)})^\top C^{(m)}\Delta u^{(m)} + \tfrac12 (u^{(n)} + v^{(n)})^\top C^{(n)}\Delta u^{(n)} = 0 \\
  \Delta p &= C^{(m)}\Delta u^{(m)} + C^{(n)}\Delta u^{(n)} = 0\,,
\end{align*}
where we use the fact that $C$ is symmetric. Substituting $\Delta p$ into $\Delta\mathcal H$, we find,
\begin{align*}
  2\Delta\mathcal H^{(m,n)} &= [(u^{(m)}+v^{(m)})-(u^{(n)}+v^{(n)})]^\top C^{(m)}\Delta u^{(m)} \\
  0 &= [(2u^{(m)}+\Delta u^{(m)})-(2u^{(n)}+\Delta u^{(n)})]^\top C^{(m)}\Delta u^{(m)} \\
    &= [2(u^{(m)}-u^{(n)}) + (1+C^{(n)-1}C^{(m)})\Delta u^{(m)}]^\top C^{(m)}\Delta u^{(m)} \\
    &= [2(u^{(m)}-u^{(n)}) + (C^{(m)-1}+C^{(n)-1})C^{(m)}\Delta u^{(m)}]^\top C^{(m)}\Delta u^{(m)} \\
    &\equiv [2(u^{(m)}-u^{(n)}) + \mu^{(m,n)-1}C^{(m)}\Delta u^{(m)}]^\top C^{(m)}\Delta u^{(m)} \\
    &= [2\underbrace{\mu^{(m,n)+\sfrac12}(u^{(m)}-u^{(n)})}_{-y} + \underbrace{\mu^{(m,n)-\sfrac12}C^{(m)}\Delta u^{(m)}}_{x}]^\top \mu^{(m,n)-\sfrac12}C^{(m)}\Delta u^{(m)}\,,
\end{align*}
where $\mu$ is the generalised reduced mass of the $(m,n)$ system. We can solve for $x$, and thereby $\Delta u^{(m)}$, by rewriting thus,
\begin{align*}
  |x-y|^2 &= |y|^2\,.
\end{align*}
This form reflects the fact that we lack sufficient constraints to unambiguously solve for the post-collision picture. In Cartesian coordinates for point particles, this manifests as an unknown direction of motion afterwards. This is usually resolved by introducing the constraint that momenta perpendicular to the normal vector between the colliding bodies are unchanged. In our generalised coordinate system, the appropriate constraint is unspecified. The general solution is given by
\begin{align*}
  x &= y+|y|\hat n = |y|(\hat y + \hat n)
\end{align*}
where $\hat n$ is a unit vector of unknown direction. We now find an expression for the change in kinetic energy of our particle of interest, $(m)$,
\begin{align*}
  \Delta\mathcal H^{(m)} &= (u^{(m)} + \tfrac12\Delta u^{(m)})^\top C^{(m)}\Delta u^{(m)} \\
    &= u^{(m)\top} \mu^{(m,n)\sfrac12}|y|(\hat y+\hat n) + \bigOO{\Delta u^{(m)2}} \\
    &= \eta|u^{(m)\top}\mu^{(m,n)}(u^{(m)}-u^{(n)})| + \bigOO{\Delta u^{(m)2}}
\end{align*}
where we have introduced a factor $\eta\in[-2,2]$ to take into account our uncertainty in the final direction, and where we have assumed $|\Delta u|\ll u$. Being more careful, we can determine $\Delta u^{(m)}$ thus,
\begin{align*}
  \Delta u^{(m)} &= C^{(m)-1} \mu^{(m,n)+\sfrac12} |\mu^{(m,n)+\sfrac12}(u^{(n)}-u^{(m)})| (\hat y+\hat n) \\
  |\Delta u^{(m)}| &= \eta' |C^{(m)-1} \mu^{(m,n)}(u^{(n)}-u^{(m)})| \\
    &= \eta' |(C^{(m)}+C^{(n)})^{-1} (C^{(n)} u^{(n)} + C^{(m)} u^{(m)}) - u^{(m)}| \\
    &= \eta' |\tilde u^{(m,n)} - u^{(m)}|
\end{align*}
where $\tilde u^{(m,n)}$ is the average of prior velocities, weighted by their generalised masses $C^{(\cdot)}$, and $\eta'\in[0,2]$ is another uncertainty factor related to $\eta$. When $u^{(m)}$ is large compared to $u^{(n)}$, $(u^{(m)}+\Delta u^{(m)})\approx u^{(m)}$. When it is small, $(u^{(m)}+\Delta u^{(m)})\approx u^{(n)}$. Thus a better approximation to $\Delta\mathcal H$ is given by
\begin{align*}
  \Delta\mathcal H^{(m)} &\approx \eta|\tilde u^{(m,n)\top}\mu^{(m,n)}(u^{(m)}-u^{(n)})|\,.
\end{align*}

$\Delta\mathcal H$ gives the energy lost from the computational particle $(m)$ to some environmental particle $(n)$, and thus is the amount of energy that must be supplied to the computational system and dissipated from the thermal system. In order to proceed we must determine the rate of this energy transfer; if we assume that the kinetic dynamics of the system are significantly faster than the computational transitions, as suggested by our inability to sufficiently control the quantum state, then we can assume negligible extant correlations between the positions and velocities of different particles. Thus we can employ a kinetic theory approach to determine this rate.

We first determine the mean free path $\ell$, the average distance a particle travels between collisions. Given the vanishing correlations, we are able to treat collisions between different species of particle separately. For spherically symmetric particles, we find $\pi (r^{(\alpha)} + r^{(\beta)})^2 \ell^{(\alpha;\beta)} = 1/n^{(\beta)}$ for the mean free path of $\alpha$ particles between collisions with $\beta$ particles, where $n^{(\beta)}=N^{(\beta)}/V^{(\beta)}$ is the number density of $\beta$ particles. More generally, we can replace $\pi(r^{(\alpha)}+r^{(\beta)})2$ with $A^{(\alpha,\beta)}$, the effective collision cross-section for arbitrarily shaped particles averaged over impact orientation. We also need to determine the average relative speed between $\alpha$ and $\beta$ particles,
\begin{align*}
  \bar v_{\text{rel}}^{(\alpha,\beta)2} &= \evQty{(v^{(\alpha)} - v^{(\beta)})^2} \\
    &= \evQty{v^{(\alpha)2}} + \evQty{v^{(\beta)2}} - 2\evQty{v^{(\alpha)} \cdot v^{(\beta)}} \,,\\
  \bar v_{\text{rel}}^{(\alpha,\beta)} &= \sqrt{\bar v^{(\alpha)2} + \bar v^{(\beta)2}} \,,\\
  \nu^{(\alpha;\beta)} &= \bar v_{\text{rel}}^{(\alpha,\beta)} / \ell^{(\alpha;\beta)} \\
    &= A^{(\alpha,\beta)}n^{(\beta)}\sqrt{\bar v^{(\alpha)2} + \bar v^{(\beta)2}}\,,
\end{align*}
where we have used the fact that the velocities are uncorrelated to cancel the cross term. Summing over these rates for each class, we can find the rate of loss of kinetic energy for $\alpha$ particles,
\begin{align*}
  \dot{\mathcal H}^{(\alpha)} &= \sum_\beta \eta^{(\alpha;\beta)} |\tilde u^{(\alpha,\beta)\top} \mu^{(\alpha,\beta)} (u^{(\beta)} - u^{(\alpha)})| A^{(\alpha,\beta)}n^{(\beta)} \sqrt{\bar u^{(\alpha)2} + \bar u^{(\beta)2}}\,.
\end{align*}
When $\alpha$ particles are heavy and fast, this reduces to
\begin{align*}
  \dot{\mathcal H}^{(\alpha)} &= \sum_\beta \eta^{(\alpha;\beta)} |u^{(\alpha)\top} C^{(\beta)} u^{(\alpha)}| A^{(\alpha)}n^{(\beta)} \bar u^{(\alpha)} \\
    &= \sum_\beta \eta^{(\alpha;\beta)} \rho^{(\beta)} A^{(\alpha)} \bar u^{(\alpha)3}
\end{align*}
where $\rho=n|C|$ is the generalised density. Notice that this takes the same form as the equation for hydrodynamic drag, taking the drag coefficient to be $2\eta$. Assuming there exists some correspondence between the generalised velocity $u^{(\alpha)}$ and the rate of computation, we find
\begin{align*}
  P &\ge \sum_\alpha N^{(\alpha)} A^{(\alpha)}\qty\bigg(\ell^{(\alpha)}\frac{R^{(\alpha)}}{N^{(\alpha)}})^3 \sum_\beta \eta^{(\alpha;\beta)}\rho^{(\beta)}
\end{align*}
where $\alpha$ ranges over computational particles and $\ell$ is the characteristic generalised displacement corresponding to a single computational transition. To maximise the net computational rate $R=\sum_\alpha R^{(\alpha)}$, we find that we need to let $\bar u^{(\alpha)}\to 0$, but this violates our assumption that $\alpha$ particles are fast. For finite velocity computational particles, we thus find that our net computational rate scales as $R\lesssim A\sim V^{2/3}$.

In the limit of slow $\alpha$ particles, $|u^{(\alpha)}|\ll|u^{(\beta)}|$, $\dot{\mathcal H}$ instead reduces to
\begin{align*}
  \dot{\mathcal H}^{(\alpha)} &= \sum_\beta \eta^{(\alpha;\beta)} |\tilde u^{(\alpha,\beta)\top}C^{(\beta)}u^{(\alpha)}|A^{(\alpha)}n^{(\beta)}\bar u^{(\beta)}
\end{align*}
where we have taken $\evQty{u^{(\beta)}-u^{(\alpha)}}=\bar u^{(\alpha)}$. If the generalised mass of the $\alpha$ particles is not significantly heavier than the $\beta$ particles, then $\tilde u^{(a,b)}$ will have a strong dependence on $u^{(\beta)}$ and this will lead to the same scaling limit of $R\lesssim A$. In order to improve on this, we must let the $\alpha$ particles be significantly heavier. In this limit, $\tilde u^{(\alpha,\beta)}\approx u^{(\alpha)}$ and we get
\begin{align*}
  \dot{\mathcal H}^{(\alpha)} &= \sum_\beta \eta^{(\alpha;\beta)} |u^{(\alpha)\top}C^{(\beta)}u^{(\alpha)}|A^{(\alpha)}n^{(\beta)}\bar u^{(\beta)} \,,\\
  P &\ge \sum_\alpha N^{(\alpha)}A^{(\alpha)}\qty\bigg(\ell^{(\alpha)}\frac{R^{(\alpha)}}{N^{(\alpha)}})^2 \sum_\beta \eta^{(\alpha;\beta)} \rho^{(\beta)}\bar u^{(\beta)} \\
    &\ge \sum_\alpha \frac{R^{(\alpha)2}}{N^{(\alpha)}} \underbrace{A^{(\alpha)}\ell^{(\alpha)2} \sum_\beta \eta^{(\alpha;\beta)} \rho^{(\beta)}\bar u^{(\beta)}}_{\gamma^{(\alpha)}}\,.
\end{align*}
To proceed, we maximise the rate subject to this power constraint,
\begin{align*}
  \Lambda = \sum_\alpha R^{(\alpha)} - \frac{1}{2\lambda}\sum_\alpha \frac{R^{(\alpha)2}\gamma^{(\alpha)}}{N^{(\alpha)}}
  \qquad\implies\qquad
  0 = 1 - \frac{R^{(\alpha)}\gamma^{(\alpha)}}{\lambda N^{(\alpha)}}
\end{align*}\begin{align*}
  \frac R\lambda &= \sum_\alpha \frac{n^{(\alpha)}}{N\gamma^{(\alpha)}}N = N\evqty\bigg{\frac1{\gamma^{(\alpha)}}} \\
  P &\ge \sum_\alpha \lambda R^{(\alpha)} = \lambda R = \frac{R^2}{N}\evqty\bigg{\frac1{\gamma^{(\alpha)}}}^{-1} \\
  R &\le \sqrt{PN\evqty\bigg{\frac1\gamma}}
\end{align*}
where $N=\sum_\alpha N^{(\alpha)}$ is the total number of computational particles. To maximise the computational rate $R$, then, we let $N$ scale with the volume of the system, leading once again to the scaling law
\begin{align*}
  R \lesssim \sqrt{AV} \sim V^{5/6}\,.
\end{align*}

\section{Classical Architectures II: Brownian Machines}
\label{sec:crn}

Whilst the previous two sections should suffice to cover all physical computational systems, the high mass-low speed limit of the classical system is not ideal from an engineering perspective. For improved practicality, we reformulate this limit in terms of an abstract chemical reaction network (CRN) near equilibrium, and show that we obtain the same result. This result is thus more robust in terms of its attainability.

\paragraph{Entropy generation rate}
We first seek a general expression for the rate of entropy generation for a chemical reaction network. Consider any such dynamical system consisting of a set of species $\{C_j:j\}$ and a set of reversible reactions $\Gamma = \{ \nu_{ij}X_j \longleftrightarrow_{\gamma_i} \nu_{ij}'X_j : i \}$, where we sum over $j$ and the $\nu_{ij}$ are stoichiometries. By the convergence theorem for reversible Markov systems, the system will converge to a unique steady state/equilibrium distribution for any initial conditions. As the reactions are reversible, the steady state will also satisfy detailed balance such that the forward and backward rates coincide separately for each reaction, i.e.\ $\smfwd R_i=\smbwd R_i$, where $R$ is a rate of the whole system rather than per unit volume.

Let us look for a quantity $H$ that measures progress towards equilibrium and satisfies the following properties:
\begin{enumerate}
  \item $H$ is a state function of the system\footnote{More accurately, on the joint state of an ensemble of iid systems.}, depending only on its instantaneous description;
  \item $H$ increases monotonically with time;
  \item $H$ is additive, i.e.\ $H(\cup_iV_i)=\sum_iH(V_i)$ for any set of disjoint regions $V_i$.
\end{enumerate}
By properties 1 and 2 and the convergence theorem, $H$ will approach a unique maximum at equilibrium. Microscopically, reactions occur discretely and so the rate of change of $H$ can be written
\begin{align*}
  \dot H_i &= \fwd R_i\fwd{\Delta h}_i+\bwd R_i\bwd{\Delta h}_i = (\fwd R_i-\bwd R_i)(h_{\text{reactants}}-h_{\text{products}})
\end{align*}
for any reaction $i$, where we have used reversibility to identify $\smfwd{\Delta h}_i=-\smbwd{\Delta h}_i=h_{\text{reactants}}-h_{\text{products}}$ and where $h(\sum_j\nu_{ij}X_j)$ is the entropy of the region associated with precisely $\nu_{ij}$ particles of each species $X_j$. In order to satisfy property 2, the sign of $\smfwd{\Delta h}_i$ must equal that of $\smfwd R_i-\smbwd R_i$.

Now, consider fixing $h_{\text{products}}$ and $\smbwd R_i$, and varying $\smfwd R_i$; this is possible unless the reaction is trivial with $\nu_{ij}=\nu_{ij}'$ for all species $X_j$, i.e.\ it does nothing. To maintain property 2, it must therefore be the case that $\smfwd{\Delta h}_i=h(\smfwd R_i)-h(\smbwd R_i)$. To satisfy property 3, we require $h(\alpha\smfwd R_i)-h(\alpha\smbwd R_i)=h(\smfwd R_i)-h(\smbwd R_i)$ where $\alpha$ is some arbitrary scaling factor of the system. We can therefore infer the functional form $h(x)=\log_b ax$ for some constants $a$ and $b$.

We have therefore derived an entropy-like quantity, unique up to choice of logarithmic unit via $b$ and entropy at $T=0$ via $a$. Without loss of generality we choose $a=1$ and $b=e$, and therefore find that
\begin{align*}
  \dot H &= \sum_i (\fwd R_i - \bwd R_i)\log\frac{\fwd R_i}{\bwd R_i} \equiv 2\sum_i R_i\beta_i\arctanh \beta_i
\end{align*}
where $\beta_i=(\smfwd R_i-\smbwd R_i)/(\smfwd R_i+\smbwd R_i)$ is the effective bias of reaction $i$ and $R_i=\smfwd R_i+\smbwd R_i$ is its gross reaction rate.

\paragraph{Entropy generation rate for elementary chemical reactions}

We now calculate our entropy quantity for elementary chemical reactions and compare against the expected value. The rates for elementary chemical reactions are given by collision theory as
\begin{align*}
  \fwd r_i &= \fwd k_i \prod_j[X_j]^{\fwd\nu_{ij}} &
  \bwd r_i &= \bwd k_i \prod_j[X_j]^{\bwd\nu_{ij}}
\end{align*}
where $\smfwd k_i$ and $\smbwd k_i$ are rate constants and $r$ are rates per unit volume. Assuming no inter-particle interactions, the canonical entropy is given by the Sackur-Tetrode equation which can be obtained simply by considering the spatial distribution of the particles along with any other degrees of freedom. For species $X_i$, consider an arbitrary volume $V_i$ available to it; if the thermal volume of the particles is $\mathcal V_i=\Lambda^3$ where $\Lambda$ is the thermal \emph{de Broglie} wavelength, then there will be $V_i/\mathcal V_i$ loci available to the particle, and so the associated entropy within the volume will be
\begin{align*}
  N_ih(X_i) &= N_i \log\frac{V_i}{\mathcal V_i} - \underbrace{\log N_i!}_{\mathclap{\text{Gibb's factor}}} + N_i(\varepsilon_i-1) \\
  h(X_i) &= \varepsilon_i - \log[X_i]\mathcal V_i + \bigO{\frac{\log N_i}{N_i}}
\end{align*}
where $N_i=[X_i]V_i$ is the number of particles in the volume, and $[X_i]$ the concentration. We have also taken into account indistinguishability of the particles with the Gibb's factor, and allowed for additional degrees of freedom with the $\varepsilon_i$ term. Strictly speaking the spatial distribution should be calculated combinatorially via the multinomial coefficient
\begin{align*}
  \log\frac{(V/\mathcal V)!}{((V/\mathcal V - \sum N_i))!\prod N_i!}
\end{align*}
where we have assumed all the particles inhabit the same region and have the same wavelength. Applying Stirling's approximation and assuming that $\sum N_i\ll V/\mathcal V$, this reduces to $\sum N_i(1-\log[X_i]\mathcal V)$ however, so our simplistic derivation is valid.

Therefore we find the canonical entropy change due to reaction $i$ is given by
\begin{align*}
  \Delta h_i &= \sum_j(\bwd\nu_{ij}-\fwd\nu_{ij})h(X_j)
     = \sum_j(\bwd\nu_{ij}-\fwd\nu_{ij})(\varepsilon_j-\log\mathcal V_j) - \sum_j(\bwd\nu_{ij}-\fwd\nu_{ij})\log[X_j]
\end{align*}
using a slight abuse of notation (the logarithmands are not unitless, but their combination is). Comparing with our quantity $\Delta h=\log\smfwd r/\smbwd r$, we find
\begin{align*}
  \Delta h_i &= \log\frac{\fwd k_i}{\bwd k_i} - \sum_j(\bwd\nu_{ij}-\fwd\nu_{ij})\log[X_j] \\
  \frac{\fwd k_i}{\bwd k_i} &= \qty\bigg(\prod_j\mathcal V_j^{\fwd\nu_{ij}-\bwd\nu_{ij}})\qty\bigg(\exp\sum_j(\fwd\nu_{ij}-\bwd\nu_{ij})\varepsilon_j)
\end{align*}
which should be compared to the Arrhenius equation. Indeed, the $\mathcal V_i$ terms confer the appropriate units to the pre-exponential factor and, by expanding the enthalpy term in terms of $k_BT$ as $\varepsilon_i=\varepsilon_i^{(0)}+\frac1{k_BT}\varepsilon_i^{(1)}+\bigOO{\frac1{k_BT}}^2$, we can group the roughly constant terms of the enthalpy into the pre-exponential factor, leaving the exponential term in the form $e^{-\Delta E/k_BT}$.

\paragraph{Maximising performance in Brownian machines}

We wish to maximise the net reaction rate of the system by selecting the biases $\beta_i$ of each individual reaction, subject to bounded total entropy production. If $R=\sum R_i$ is the total gross reaction rate, then the total net rate is $R_c=\sum R_i\beta_i=R\evqty{\beta}$ where the expectation value is with respect to the fractional contribution of each reaction, i.e.\ weighted by $R_i/R$. We denote this by $R_c$ assuming that each reaction contributes to computational progress; if this is not true then $R_c$ will be less than this value. The entropy rate is similarly given by $\dot H = 2R\evqty{\beta\arctanh\beta}$. We use a Lagrange multiplier approach as usual,
\begin{align*}
  \Lambda &= R\evqty{\beta} - \lambda R\evqty{\beta\arctanh\beta} &
  &\implies&
  R_i &= \lambda R_i\partial_{\beta_i}\beta_i\arctanh\beta_i;
\end{align*}
that is, the optimum is achieved by setting all the $\beta_i$ equal to a constant, $\beta$. This gives $\dot H = 2R\beta\arctanh\beta=\frac{2R_c^2}{R} + \bigOO{\frac{R_c^4}{R^3}}$ and hence
\begin{align*}
  R_c \le \sqrt{\frac{PR}{2kT}}\,,
\end{align*}
leading once again to the scaling limit $R_c\lesssim\sqrt{AV}\sim V^{5/6}$, as the power $P$ scales with area and the gross computation rate $R$ with volume. Here we have presented a constructive proof of this scaling limit, as any valid reversible CRN implementation will yield this scaling limit.

\section{Relativistic Effects at Scale}
\label{sec:gr}

At different scales, from the microscopic to the cosmic, different constraints predominate in the analysis of maximum computation rate. For macroscopic systems at worldly scales, the aforementioned thermodynamic constraints are most directly relevant, yielding an upper bound of $R_C\lesssim\sqrt{AV}$. At sufficiently small scales, however, the surface area to volume ratio will be great enough that this thermodynamic constraint will no longer be limiting, with the volumetric Margolus-Levitin bound instead dominating. This reflects the fact that there is simply insufficient energy enclosed in the system for the resultant rate of computation to saturate the heat dissipation capacity at the boundary. If one were to use a denser computational architecture, the computational capacity of the region could be increased to the point that the power-flux bound is saturated, and thus recovering the more restrictive $\sqrt{AV}$ bound.

It transpires that there are two further regimes. As our systems approach cosmic scales, the threat of gravitational collapse becomes pertinent. In order to avoid this fate, we must lower the average density such that the system's radius always exceeds its Schwarzschild radius\footnote{In reality, the threshold radius for gravitational collapse exceeds the Schwarzschild radius by a factor not less than $\tfrac98$. The reasons for this shall be discussed in due course.}. As the Schwarzschild radius is proportional to the mass of the system, this implies that the computational rate of such large systems varies \textit{linearly} in radius. In fact, there is an intermediate regime: as we have not been utilising the full computational potential of our mass due to the thermodynamic constraint, we can gradually increase said utilisation whilst simultaneously reducing our overall density until the Margolus-Levitin limit is attained, at which point we default to the linear regime. This intermediate regime is thus still described by the $R_C\lesssim\sqrt{PM}$ limit, which in this post-Schwarzschild realm equates to $R_C\lesssim V^{1/2}$. In summary, the scaling laws of these four regimes go as $V$, $V^{5/6}$, $V^{1/2}$ and $V^{1/3}$, in order of increasing scale.

A more detailed analysis reveals that these large systems are subject to a further consideration. So far we have assumed a Galilean invariance worldview. At small and medium scales, this approximation is very good. At larger scales, however, relativistic effects threaten to reduce our overall computation rate via time dilation, and at even larger scales they can constrain the amount of mass we can fit within a volume lest the system undergo gravitational collapse as mentioned earlier. Special relativistic time dilation is easily avoided by minimising relative motion within the computational parts of the system (this implies that message packets will have reduced computational capacity compared to the static computational background, but message packets can generally be assumed static anyway).

Gravitational time dilation is a more serious issue, which cannot be eluded at large scales. In the case of a hypothetical supermassive computational system, local computation proceeds unabated, but distant users interfacing with the system will observe slower than expected operation, as depicted in \Cref{fig:constraint-gr}. To calculate this slowdown factor, we shall proceed by modelling the system as spherically symmetric. This is not unreasonable at these scales, where a body's self-gravitation makes maintaining other geometries unstable and impractical. We will not consider rotational systems which could allow for an ellipsoidal shape, as the requisite angular velocities would almost certainly abrogate nearly all computational progress solely due to special relativistic effects.

The metric\footnote{We use a $(+---)$ signature.} of an isotropic and spherically symmetric body in hydrostatic equilibrium can be obtained via the \emph{Tolman-Oppenheimer-Volkoff} (TOV) equation~\cite{tov-eqn},
\begin{equation}\label{eqn:tov-canonical}
\begin{split}
  \dv{P}{r} &= -\frac{Gm}{r^2}\rho\qty\Big(1+\frac{P}{\rho c^2})\qty\Big(1+\frac{4\pi r^3 P}{mc^2})\qty\Big(1-\frac{2Gm}{rc^2})^{-1} \\
  \dv{\nu}{r} &= -\qty\Big(\frac{2}{P+\rho c^2})\dv{P}{r} \\
      &= \frac{2Gm}{r^2c^2}\qty\Big(1+\frac{4\pi r^3P}{mc^2})\qty\Big(1-\frac{2Gm}{rc^2})^{-1} \\
  \dd s^2 &= e^\nu c^2 \dd t^2 - \qty\Big(1-\frac{2Gm}{rc^2})^{-1}\dd r^2 - r^2 \dd\Omega^2\,,
\end{split}
\end{equation}
where $P$ is the pressure, $m(r)$ is the mass enclosed by the concentric spherical shell of radius $r$, and $\rho$ is the local mass-energy density. All masses are as observed by a distant observer.

The rate of dynamics of a point $P$ as observed from a point $Q$ is well known to be
\begin{align*}
  \frac{R(Q)}{R(P)} &= \sqrt{\frac{g_{00}(P)}{g_{00}(Q)}}\,,
\end{align*}
where $g_{\mu\nu}$ is the metric tensor. We assume our distant observer to be moving slowly (relative to the computer) in approximately flat space, and thus that $g_{00}(Q)=1$. This yields a slowdown factor of $f(P)=\sqrt{g_{00}(P)}$. The total slowdown for an extended body is hence found to be
\begin{align*}
  f &= \frac1M \int_V \dd V \rho \sqrt{g_{00}} \\
    &= \frac1M \int_0^{r_1} \dd r \dv{m}{r} e^{\nu/2}\,, & \text{(spherically symmetric)}
\end{align*}
where $M$ is the total mass and $r_1$ is the least upper bound of its radial extent. 

In order to maximise $f$, we must clearly maximise $\nu$ throughout the body. It can be seen that $\dv{\nu}{r}\ge 0$, and in fact within solid regions where $\rho>0$ this inequality is strict. Furthermore, in empty space the TOV equation breaks down; instead, we use our Schwarzschild matching conditions to find
\begin{align*}
  \dv{\nu}{r} &= \frac{2Gm}{r^2c^2}\qty\Big(1-\frac{2Gm}{rc^2})^{-1} \\ \Delta\nu &= \Delta\log\qty\Big(1-\frac{2Gm}{rc^2})\,,
\end{align*}
which shows that this inequality is always strict whenever $m>0$. $\nu$ at the surface is fixed by the Schwarzschild metric, in accordance with Birkhoff's theorem, and thus for a fixed extent our only control over $\nu$ is the distribution $m$ between the surface and core. Approaching the core from the surface, $\nu$ strictly decreases until $m$ vanishes. Furthermore $\dd\nu/\dd r$ is at least as great as in empty space, being strictly greater in massive regions. Therefore it is desirable to concentrate mass towards the surface.

We assume that the density of our architecture of choice is bounded from above, thus limiting the degree to which we can concentrate mass towards the surface. The optimum is then achieved by a thick shell of limiting density from the surface inwards. The two extreme cases of a thin shell and a solid sphere can be treated exactly, but an arbitrarily thick spherical shell must be treated numerically.

\para{Solid Sphere}

A solid sphere coincides with the Schwarzschild geometry. The usual Schwarzschild metric corresponds to the exterior of the object, instead we must use the interior solution which can be derived from the TOV equation~(\ref{eqn:tov-canonical}), yielding
\begin{align*}
  \sqrt{g_{00}} &= \frac32\sqrt{1-\frac{r_s}{r_1}} - \frac12\sqrt{1-\qty\Big(\frac{r}{r_1})^2\frac{r_s}{r_1}}
\end{align*}
where $r_s=2GM/c^2$ is the Schwarzschild radius. Before computing $f$, note that $g_{00}$ vanishes when $r_s=\tfrac89r_1$. This limit is in fact of great importance; whilst an object of size $r_1=\tfrac98 r_s$ would exceed the Schwarzschild radius and be presumed stable against gravitational collapse, this is not the case. At this point, such an object is unstable towards spontaneous oscillations, its core pressure diverges, and it readily collapses to a black hole~\cite{schwarz-lim-main}. Furthermore, depending on its heat capacity ratio $\gamma$, a gaseous hydrostatic sphere may be unstable at even larger radii as tabulated by \textcite{schwarz-lims}. This $r_1\ge\tfrac98r_s$ threshold is thus a stronger bound on the size of massive objects. Integrating $\sqrt{g_{00}}$ over the mass of this solid sphere, we find
\begin{align*}
  f_{\text{solid}} &= \frac{3}{16 v_s}( (1+6v_s)\sqrt{1-v_s} - v_s^{-1/2}\arcsin\sqrt{v_s} )
\end{align*}
where we have introduced the normalised coordinate $v=r/r_1$. At the $v_s=\tfrac89$ threshold, we get $f_{\text{solid}}\approx 0.1699$.

\para{Thin Shell}

For a thin shell, we first identify that the pressure on the outer boundary vanishes. Let the inner radius be $r_0$ and the thickness be $\delta r=r_1-r_0$ such that $M\approx 4\pi\rho r_1^2\delta r$ and $\delta v\ll 1$. We also introduce the unitless variable $u=P/\rho c^2$,
\begin{align*}
  \dv{u}{r} &= -\frac{4\pi\rho rG}{c^2}(1+u)\qty\Big(\frac{r-r_0}{r} + u)\qty\Big(1-\frac{2Gm}{rc^2})^{-1} \\
  \dv{u}{v} &= -\frac{4\pi\rho rr_1 G}{c^2}(1+u)\qty\Big(\frac{r-r_0}{r} + u)\qty\Big(1-\frac{2GM}{c^2}\frac{r-r_0}{r\delta r})^{-1} \\
    &\approx -\frac{v_s}{2\delta v}(u+1)(u+\beta\delta v)(1-\beta v_s)^{-1} \\
  \Delta u &\approx -\tfrac12v_s(u+1)(u+\beta\delta v)(1-\beta v_s)^{-1}\,,
\end{align*}
where $\beta = \frac{r-r_0}{r_1-r_0} \in[0,1]$. Integrating from the surface inwards, we have $u_1=0$ and $\beta=1$, yielding
\begin{align*}
  u_0 &= \frac{\delta v}{2} \frac{v_s}{1-v_s} + \bigO{\delta v^2}\,.
\end{align*}
Therefore in the limit $\delta v\to 0$, the pressure throughout the thin shell vanishes. As a result, we also have that $\nu$ is constant for all $v\in[0,1]$, taking the value $1-v_s$ and giving
\begin{align*}
  f_{\text{thin}} &= \sqrt{1-v_s}\,.
\end{align*}
At the threshold, we get $f_{\text{thin}}=\tfrac13$.

\para{Thick Shell}

For general $\delta v$, we integrate numerically. We rewrite the TOV equations in normalised and unitless form thus
\begin{align*}
  \dv{\log g}{v} &= \frac{vv_s}{2}\frac{\mu+3u}{\mu_1-v^2v_s\mu}\,, &
  \dv{u}{v} &= -(1+u)\dv{\log g}{v}\,, &
  \dv{f}{v} &= -\frac{3v^2g}{\mu_1}\,,
\end{align*}
where $\mu=1-(v_0/v)^3$, $\mu_1=1-v_0^3$ and $g\equiv \sqrt{g_{00}}$. We integrate the system from $v=1$ down to $v=v_0$, with initial conditions $u_1=0$, $g_1=\sqrt{1-v_s}$ and $f_1=0$. Our slowdown factor is given by $f_0$, and explicit results are given in the following section.

\newlength{\limexsize}
\setlength{\limexsize}{4.5cm}
\newcommand*{\limex}[1]{
  \begin{subfigure}[t]{\limexsize}
    \centering
    \input{fig-lims-sub-#1}
    \caption{}\label{fig:lims-#1}
  \end{subfigure}}
\newcommand*{\limkey}{\begin{subfigure}[t]{5.5cm}
  \small
\begingroup
  \makeatletter
  \providecommand\color[2][]{%
    \GenericError{(gnuplot) \space\space\space\@spaces}{%
      Package color not loaded in conjunction with
      terminal option `colourtext'%
    }{See the gnuplot documentation for explanation.%
    }{Either use 'blacktext' in gnuplot or load the package
      color.sty in LaTeX.}%
    \renewcommand\color[2][]{}%
  }%
  \providecommand\includegraphics[2][]{%
    \GenericError{(gnuplot) \space\space\space\@spaces}{%
      Package graphicx or graphics not loaded%
    }{See the gnuplot documentation for explanation.%
    }{The gnuplot epslatex terminal needs graphicx.sty or graphics.sty.}%
    \renewcommand\includegraphics[2][]{}%
  }%
  \providecommand\rotatebox[2]{#2}%
  \@ifundefined{ifGPcolor}{%
    \newif\ifGPcolor
    \GPcolortrue
  }{}%
  \@ifundefined{ifGPblacktext}{%
    \newif\ifGPblacktext
    \GPblacktexttrue
  }{}%
  \let\gplgaddtomacro\g@addto@macro
  \gdef\gplbacktext{}%
  \gdef\gplfronttext{}%
  \makeatother
  \ifGPblacktext
    \def\colorrgb#1{}%
    \def\colorgray#1{}%
  \else
    \ifGPcolor
      \def\colorrgb#1{\color[rgb]{#1}}%
      \def\colorgray#1{\color[gray]{#1}}%
      \expandafter\def\csname LTw\endcsname{\color{white}}%
      \expandafter\def\csname LTb\endcsname{\color{black}}%
      \expandafter\def\csname LTa\endcsname{\color{black}}%
      \expandafter\def\csname LT0\endcsname{\color[rgb]{1,0,0}}%
      \expandafter\def\csname LT1\endcsname{\color[rgb]{0,1,0}}%
      \expandafter\def\csname LT2\endcsname{\color[rgb]{0,0,1}}%
      \expandafter\def\csname LT3\endcsname{\color[rgb]{1,0,1}}%
      \expandafter\def\csname LT4\endcsname{\color[rgb]{0,1,1}}%
      \expandafter\def\csname LT5\endcsname{\color[rgb]{1,1,0}}%
      \expandafter\def\csname LT6\endcsname{\color[rgb]{0,0,0}}%
      \expandafter\def\csname LT7\endcsname{\color[rgb]{1,0.3,0}}%
      \expandafter\def\csname LT8\endcsname{\color[rgb]{0.5,0.5,0.5}}%
    \else
      \def\colorrgb#1{\color{black}}%
      \def\colorgray#1{\color[gray]{#1}}%
      \expandafter\def\csname LTw\endcsname{\color{white}}%
      \expandafter\def\csname LTb\endcsname{\color{black}}%
      \expandafter\def\csname LTa\endcsname{\color{black}}%
      \expandafter\def\csname LT0\endcsname{\color{black}}%
      \expandafter\def\csname LT1\endcsname{\color{black}}%
      \expandafter\def\csname LT2\endcsname{\color{black}}%
      \expandafter\def\csname LT3\endcsname{\color{black}}%
      \expandafter\def\csname LT4\endcsname{\color{black}}%
      \expandafter\def\csname LT5\endcsname{\color{black}}%
      \expandafter\def\csname LT6\endcsname{\color{black}}%
      \expandafter\def\csname LT7\endcsname{\color{black}}%
      \expandafter\def\csname LT8\endcsname{\color{black}}%
    \fi
  \fi
    \setlength{\unitlength}{0.0500bp}%
    \ifx\gptboxheight\undefined%
      \newlength{\gptboxheight}%
      \newlength{\gptboxwidth}%
      \newsavebox{\gptboxtext}%
    \fi%
    \setlength{\fboxrule}{0.5pt}%
    \setlength{\fboxsep}{1pt}%
    \definecolor{tbcol}{rgb}{1,1,1}%
\begin{picture}(2820.00,5100.00)%
    \gplgaddtomacro\gplbacktext{%
    }%
    \gplgaddtomacro\gplfronttext{%
      \csname LTb\endcsname
      \put(1903,4948){\makebox(0,0)[r]{\strut{}Quantum~~}}%
      \csname LTb\endcsname
      \put(1903,4686){\makebox(0,0)[r]{\strut{}Thermo (Rev.)~~}}%
      \csname LTb\endcsname
      \put(1903,4424){\makebox(0,0)[r]{\strut{}Thermo (Irrev.)~~}}%
      \csname LTb\endcsname
      \put(1903,4162){\makebox(0,0)[r]{\strut{}Grav. Threshold.~~}}%
    }%
    \gplbacktext
    \put(0,0){\includegraphics[width={141.00bp},height={255.00bp}]{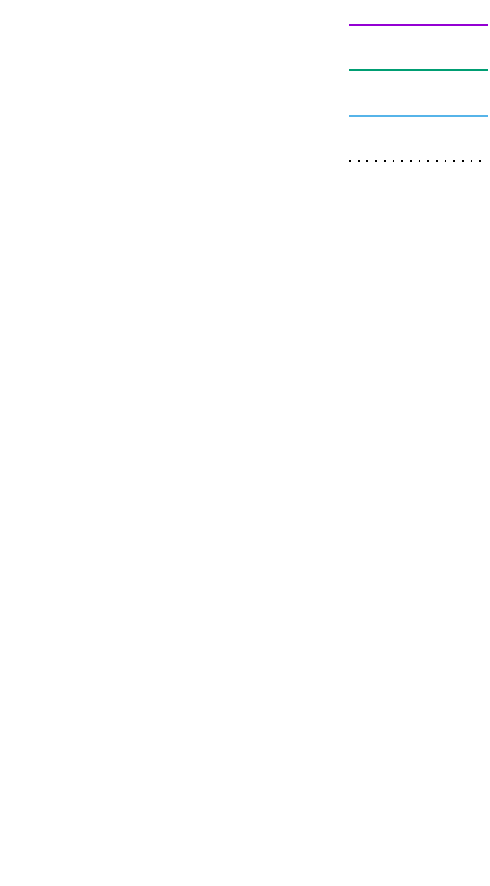}}%
    \gplfronttext
  \end{picture}%
\endgroup

\end{subfigure}}

\doublepage{%
\begin{figure}[p!]
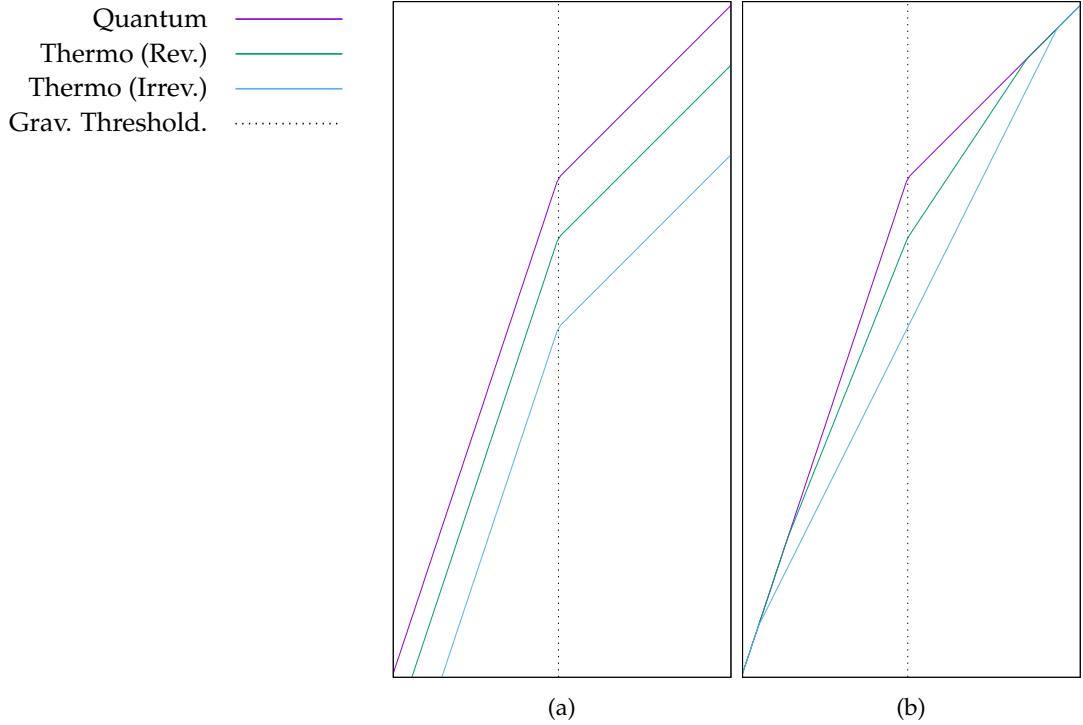

  \centering
  \parbox{\linewidth}{\hfill\begin{subfigure}[t]{5.5cm}
  \small\input{fig-lims-sub-key}
\end{subfigure}
  \begin{subfigure}[t]{\limexsize}
    \centering
    \input{fig-lims-sub-1}
    \caption{}\label{fig:lims-1}
  \end{subfigure}
  \begin{subfigure}[t]{\limexsize}
    \centering
    \input{fig-lims-sub-2}
    \caption{}\label{fig:lims-2}
  \end{subfigure}}
  \caption[A series of idealised examples of the relevant bounds on computation rate.]{%
    A series of idealised examples of the relevant bounds on computation rate.
    In all plots, the vertical axis represents rate of computation $\nu$ in arbitrary logarithmic units, and the horizontal axis the radius $r$ of the system in arbitrary logarithmic units.
    The axes have the same scaling, so that a 45\textdegree\ slope corresponds to a $\nu\propto r$ power law.
    \capnl
    (a) If power supply and heat/entropy dissipation is not limiting, then the rate of computation is proportional to the mass of the system.
    In these examples, the `technology' is fixed so that the computational systems have uniform density.
    Therefore we expect a power law $\nu\propto r^3$.
    Up to a certain threshold radius this is indeed what we see.
    The \emph{Quantum} line corresponds to the Margolus-Levitin quantum constraint on computational performance, $\nu\le mc^2/h_P$.
    The \emph{Thermo (Rev.)} line takes into account that a realistic computer might not achieve this upper bound.
    The \emph{Thermo (Irrev.)} line takes into account that the computational components of an irreversible computer can be no faster than those of a reversible one, and typically they will be slower as it takes time to erase the excess information.
    Often this erasure will be via a damping process.
    At the threshold (marked with a vertical dotted line), the system is on the cusp of gravitational collapse as its radius coincides with its Schwarzschild radius.
    Beyond this point, the mass scales linearly with radius, hence the power laws all fall to $\nu\propto r$ beyond this point.
    \capnl
    (b) In this example we assume the reversible and irreversible architectures saturate the quantum bound at the microscopic level.
    A realistic computer will be constrained by the rate at which free energy can be supplied and entropy removed.
    This scales with $r^2$, and for an irreversible computer this leads to a power law $\nu\propto r^2$ because the rate of entropy generation and free consumption per unit computational mass is bounded from below.
    At sufficiently small scales, the quantum bound is more limiting and so the power law switches to $\nu\propto r^3$; similarly, at sufficiently large scales the quantum bound is again more limiting and the power law switches to $\nu\propto r$.
    The case for reversible computers is more complicated, as analysed in this chapter.
    The power laws are $\nu\propto r^{5/2}$ below the threshold and $\nu\propto r^{3/2}$ above.
  (continued on next page)
}\label{fig:lims}
\end{figure}}{%
\begin{figure}[p!]\ContinuedFloat
  \centering
  \parbox{\linewidth}{\null
  \begin{subfigure}[t]{\limexsize}
    \centering
    \input{fig-lims-sub-3}
    \caption{}\label{fig:lims-3}
  \end{subfigure}
  \begin{subfigure}[t]{\limexsize}
    \centering
    \input{fig-lims-sub-4}
    \caption{}\label{fig:lims-4}
  \end{subfigure}
  \begin{subfigure}[t]{\limexsize}
    \centering
    \input{fig-lims-sub-5}
    \caption{}\label{fig:lims-5}
  \end{subfigure}\hfil}
  \caption[]{%
  (continued)
    The reversible architecture will always be at least as fast as the irreversible architecture.
    \capnl
    (c) This example combines the differing rates of proportionality of (a) with the thermodynamic constraints of (b).
    The dashed lines show the case where thermodynamics is not limiting.
    \capnl
    (d) This example corresponds to (c) except that the density has been increased by a factor 100.
    The dashed lines are the same as the solid lines of (c).
    Below the threshold, the Margolus-Levitin limit and reversible computers are faster (by a factor 100 and a factor 10, respectively).
    The irreversible computer is only faster in the thermodynamically non-limiting range, otherwise the thermodynamic constraint eliminates the benefit of additional computational matter.
    A second consequence of the increased density is that the Schwarzschild threshold is reached sooner.
    Nevertheless, for all architectures the computational rate is at least as great as for the lower density case.
    Therefore increasing density is always helpful at small scales, and not unhelpful at large scales.
    \capnl
    (e) This example corresponds to (b) except that gravitational time dilation has been taken into account$^\ast$.
    The dashed lines are the same as the solid lines of (b).
    The orange asymptotically-vertical dashed line shows what happens when the geometry of the computer is maintained as a sphere of uniform density, even as the density is reduced to avoid gravitational collapse.
    Namely, gravitational time dilation abrogates all computational performance from the perspective of a distant observer.
    The solid lines optimise the geometry of the system to maximise the observed computational performance from the perspective of a distant observer.
    Whilst these are reduced in comparison to (b), it is only by a constant factor of at most 3.
  }
  ~\\\parbox{\textwidth}{\hangindent=0.3cm $^\ast$\scriptsize Time dilation corrections were computed by integrating the thick shell system of differential equations using \tool{dopri5}, and the shell thickness was optimised to maximise the rate using a golden section search. The code was written in \haskell\ and is available, complete with integration and optimisation routines, at \url{https://github.com/hannah-earley/revcomp-relativistic-limits}.}
\end{figure}}


\section{Variance in Performance at Different Scales}

In order to calculate the maximum performance of a system at different sizes, we compute the optimum speed within each constraint---thermodynamic and Margolus-Levitin---with relativistic corrections, and then pick the least of these upper bounds. To apply the relativistic constraint, we substitute $M=\tfrac12 r_s c^2 / G$ for mass where $r_s$ is the Schwarzschild radius, and then multiply by our slowdown factor $f$. For better parametericity, we actually use $r_s=v_s \ell$ where $\ell$ is the system's radius/linear dimension. This yields the two bounds,
\begin{align*}
  R_{\text{thermo.}} &= \alpha \sqrt{\ell^3 v_s} f(v_s,v_0)\,, &
  R_{\text{marg.lev.}} &= \beta \ell v_s f(v_s,v_0)\,,
\end{align*}
where $\alpha$ and $\beta$ are constants of proportionality that depend on the system architecture. Notice that there is in fact some choice available in system geometry; we can pick any constant density thick shell solution, with parameters $(v_0,v_s)$. As we increase $v_s$, the factor $f$ will decrease, and so there will be an optimum pair $(v_0,v_s)$ maximising the given bound. The maximum density constraint is equivalent to setting a maximum value to the mass and hence $v_s$,
\begin{align*}
  v_s &= \frac1\ell \frac{2GM}{c^2} = \ell^2\frac{8\pi\rho G}{3c^2}(1-v_0^3) \\
    &\le v_m \equiv \ell^2\frac{8\pi\rho G}{3c^2} \\
    &\le \frac89\,,
\end{align*}
where the last constraint is to avoid gravitational collapse. For a given $v_m$ then, we have $v_0=\sqrt[3]{1-v_s/v_m}$ and we maximise the two $R$ bounds in $0<v_s\le\min(\tfrac89,v_m)$.

An illustrative example of these regimes is depicted in \Cref{fig:lims}. Notice how the solid sphere, though maximising usable computational matter and thus its local computational rate, has its observable computational rate abrogated as it approaches the point of gravitational self-collapse. This demonstrates how reducing the available matter can enable the observable computational rate to increase, though the difference between the dashed and solid lines shows that there remains a relativistic penalty in the form of gravitational time dilation. We also see that an irreversible system underperforms with respect to a reversible system across a large range of length scales, otherwise matching it at the extremes. The exact range again depends on system architecture specifics, and for some architectures the irreversible system may give transient superior performance at certain scales.

\section{Discussion}\label{sec:revi-disc}

\para{Additional Entropic Sources}

The preceding analysis has been restricted to dissipating entropy arising in the course of the computation itself. However, a physical computer will be subject to other sources of entropy in addition to this. As mentioned in the introduction, Bennett calculated that the influence of turbulence in the atmospheres of nearby stars would be sufficient to thermalise a billiard ball computer within a few hundred collisions. Thus, in general, we must aim to shield our computer from such external influences, for example by error correction procedures. Fortunately, the influences of such external sources may only interact with our system via the same boundary through which we dissipate computational entropy, and as such their rate is at most proportional to this surface area\footnote{Even though such external influences may indeed permeate the entire volume, such as gravitational waves or neutrinos, the \emph{rate} of incidence of these interaction will necessarily scale areametrically. Moreover, gravitational waves and neutrinos only penetrate the entire volume because they are so weakly interacting. A more strongly coupled interaction such as ultraviolet light or particulate matter will primarily affect only the `skin' of the system, and thus illustrates how such external influences are genuinely only areametric in strength.}. Thus, as long as the external sources are not overwhelmingly strong we should be able to suppress them at all scales with equal ease.

A more challenging source of errors and entropy comes from within. Assuming a non-vanishing system temperature, the entire system volume will generate thermodynamic fluctuations at a commensurate rate.
These fluctuations can manifest in a variety of ways, perhaps the most damaging of which are nuclear transitions.
Some nuclear transitions, such as radioactive decay leading to the production of damaging radicals, are not of great concern over long timescales as the proportion of unstable isotopes will decay to zero exponentially fast.
More problematic is the tendency of all nuclei to decay to the baryonic ground state\footnotemark, \ce{^62Ni}~\cite{nickel62} (either by fission or fusion); except for systems composed purely of \ce{^62Ni}, these spontaneous transitions must be reversed in order to maintain the intended structure.
\footnotetext{In fact, the ground state of nuclear matter depends on density, and at higher densities the equilibrium state may be as high as \ce{^118Kr}; beyond this, various forms of degenerate matter may dominate~\cite{baryon-ground-state}.}
Every event leading to a change in computational state or requiring repair invokes an entropic cost, and given that these events occur with a rate proportional to the computationally active volume, this ultimately recovers an areametric bound to our computational rate,
\begin{align*}
  V &\le \frac{P}{kT}\frac1{\dot\eta_{\text{int.fluc.}}}\,.
\end{align*}

This does not necessarily render our reversible computation performance gains unattainable, however. Providing $\dot\eta_{\text{int.fluc.}}$ is sufficiently small, we can outperform irreversible computers at scales up to
\begin{align*}
  \ell &\le \sqrt{\frac{P}{kT}\frac1{\dot\eta_{\text{int.fluc.}}}\frac1{\delta r}}
\end{align*}
where $\delta r$ is the thickness of the irreversible computational shell that can be supported by such an architecture. If this threshold size is sufficiently large as to coincide with the Schwarzschild threshold of \Cref{fig:lims} then such a reversible architecture would in fact outperform its irreversible analogue at all scales. Furthermore, it is likely that $\dot\eta_{\text{int.fluc.}}$ is smaller than it would be in an irreversible system as they have one fewer source of fluctuations to combat; namely, irreversible computers expend significant quantities of energy to ensure unidirectional operation, and any deviation from this is a potentially fatal error. In contrast, a reversible computer is intrinsically robust to such `errors', particularly the Brownian architecture described, which actively exploits this. Expending less energy has the added advantage of permitting lower operating temperatures, further reducing $\dot\eta_{\text{int.fluc.}}$.

\para{Mixed Architectures}

Because of the generality of the $R\lesssim\sqrt{AV}$ scaling law, systems of maximal computational rate can be constructed using any desired mix of architectures in order to meet the requirements of a given problem. Providing entropy and input power can be efficiently transported between the surface and the computational bulk, the entropy generation and power constraints superpose linearly. Thus in principle it is possible to have a mix of quantum and Brownian architectures within a single system. A caveat is that these two architectures may need to operate at different temperatures, and a temperature gradient introduces an additional source of entropy. Thus, the area of the boundary between these subsystems must be at most proportional to the system's bounding area in order that the entropy generated can be effectively countered.

This principle is more general: any non-equilibrium inhomogeneity in the system's structure will result in entropy generation. To quantify this, we assume that such inhomogeneities equilibrate via an uncorrelated diffusive process, and we proceed via a Fokker-Planck approach.

\let\grad\nabla
The Fokker-Planck equation (see \Cref{sec:meth-cont} for a more detailed exploration) describes the evolution of a probability distribution in space due to stochastic dynamics. Specifically, it pertains to the influence of a Langevin force with rapidly decaying correlations in time. In this limit, Brownian particles will exhibit a combination of drift and diffusion depending on the properties of the system. The evolution of the probability distribution $\varphi$ is found~\cite{risken} to obey\footnote{In \textcite{risken}, the convention is $\dot\varphi=-\partial_i\mu_i\varphi+\partial_i\partial_jD_{ij}\varphi$. We use a different convention here for convenience, wherein $\mu_i'=\mu_i-\partial_jD_{ij}$.}
\begin{align*}
  \dot\varphi &= -\grad\cdot [ \mu\varphi - D\grad\varphi ]
\end{align*}
where $\mu$ is a drift vector and $D$ is a diffusion matrix. We also identify the probability current $\mu\varphi-D\grad\varphi$. In order to establish the rate of entropy generation, we must first determine the steady state distribution. We make the assumption that the steady state probability current is everywhere zero, i.e.\ there are no persistent current flows or vortices, and therefore find that
\begin{align*}
  \frac{\nabla\varphi_0}{\varphi_0} &= D^{-1}\mu\,.
\end{align*}

Now we obtain an expression for the entropy of the system. As there are many particles, we can reinterpret $\varphi$ as a concentration of particles (as the Fokker-Planck equation does not impose any normalisation), and so each particle will have an associated entropy of $1+\varepsilon-\log\lambda$ where $\lambda\propto\varphi$. This yields intensional entropy $\eta=\varphi(1+\varepsilon-\log\lambda)$ and $\dot\eta=\dot\varphi(\varepsilon-\log\lambda)$. At equilibrium, $\dot\eta$ cannot have a leading order linear term in $\varphi$, and so $\varepsilon=\log\lambda_0$ which gives
\begin{align*}
  \eta &= \varphi\qty\Big(1-\log\frac{\varphi}{\varphi_0}) \\
  H &= \int \dd V\, \varphi\qty\Big(1-\log\frac{\varphi}{\varphi_0}) \\
  \dot H &= -\int \dd V\, \dot\varphi\log\frac{\varphi}{\varphi_0}\,.
\end{align*}

We now substitute the Fokker-Planck equation for $\dot\varphi$, obtaining
\begin{align*}
  \dot H &= \int \dd S \cdot (\mu\varphi - D\grad\varphi) \log\frac{\varphi}{\varphi_0} - \int \dd V\, (\mu\varphi - D\grad\varphi) \cdot\grad\log\frac{\varphi}{\varphi_0}  \\
    &= \int\dd V\, (D\grad\varphi-\mu\varphi)\cdot\grad\log\frac{\varphi}{\varphi_0} \\
    &= \int\dd V\, D\varphi\qty\Big(\frac{\grad\varphi}\varphi - D^{-1}\mu)\cdot\grad\log\frac{\varphi}{\varphi_0} \\
    &= \int\dd V\, D\varphi\qty\Big(\frac{\grad\varphi}\varphi - \frac{\grad\varphi_0}{\varphi_0})\cdot\grad\log\frac{\varphi}{\varphi_0} \\
    &= \int\dd V\, \varphi\qty\Big(D\grad\log\frac{\varphi}{\varphi_0})\cdot\grad\log\frac{\varphi}{\varphi_0} \\
  \dot H &= \sum_i N_i \evqty\Big{\,\qty\Big| D_i^{1/2} \grad\log\frac{\varphi_i}{\varphi_{i,0}}|^2}\,,
\end{align*}
where $N_i$ is the number of particles of diffusive species $i$. In the second line we have assumed that the probability current across the system boundary vanishes and in the last line we have generalised to multiple diffusive species.

The consequence is that a system may only sustain the use of spatial variation in its architecture if at least one of the following three conditions holds for each such species,
\begin{enumerate}
  \item The region of variation scales with the system's surface area, such as if the gradient is localised to a hemispheric plane,
  \item The diffusion rate $D$ is vanishingly small,
  \item The average relative strength of the gradient, $\grad\log(\varphi/\varphi_0)$, has scaling no greater than $\ell^{-1/2}\sim V^{-1/6}$.
\end{enumerate}
These conditions may be violated only to the extent that another condition is exploited to achieve a commensurate reduction in the entropy rate.

\para{Modes of Computation}

The quantum Zeno architecture is inherently processive, meaning that every transition it makes is useful. This is in contrast to the Brownian architecture, in which only a vanishingly small fraction of transitions lead to a net advance in computational state. There are a number of consequences to this distinction that make the quantum architecture preferable.

Firstly, a quantum architecture equivalent to a given Brownian architecture will only use $N_{\text{QZE.}}\sim\sqrt{AV}$ active computational elements compared to $N_{\text{Brown.}}\sim V$, which significantly reduces the risk of errors due to internal fluctuation. In addition, the properties of the QZE mean that the fluctuations within the computational subvolume of the quantum architecture can be rectified at all scales in contrast to the classical system. Of course, fluctuations in the remaining bulk of the volume are still problematic, but their rectification is less critical.

Secondly, whilst the two systems would have the same net rate of computation, the computational elements of the quantum system operate at the same maximum speed regardless of scale, and the degree of parallelism can be tuned between maximally parallel ($N\sim\sqrt{AV}$) and fully serial ($N=1$). This means that the quantum architecture can operate equally well for parallel and serial problems. One caveat is that the maximal attainable parallelism of the quantum system, $\sim\sqrt{AV}$, is lower than for the Brownian, $\sim V$. In practice this is not too problematic as parallel problems can be simulated serially, and because the total computation rate is independent of the degree of parallelism, the quantum system experiences no penalty for this choice. For the Brownian system, however, even a problem maximally exploiting the available parallelism will be limited by the time for each element to perform a single net operation. In addition, whilst the quantum system has fewer computational foci than the Brownian system, the remaining bulk can be used for `cold' computation such as data storage.

Thirdly, almost all parallel problems will involve inter-element communication and synchronisation. As will be discussed in \Cref{chap:revii}, these processes present significant challenges in the Brownian architecture that in the worst case would throttle the system down to $R\sim A$. This arises because synchronisation processes map to constrictions in phase space, and Brownian diffusion through such a constriction is very slow. As the quantum architecture evolves processively, it is less subject to this effect and thus its computational elements can communicate much more freely.

Fourthly, the quantum architecture permits quantum computation whereas the Brownian architecture is unlikely to be capable of supporting quantum computation. The reason for this is that the average time for each computational transition in the Brownian system is typically large, increasing as the system size does, whereas the decoherence timescale for a quantum state is fixed and typically small. Furthermore, the decoherence timescale is likely significantly smaller than for the QZE architecture as the Brownian system probably operates at a higher temperature. Altogether, these factors make it incredibly unlikely that a quantum state can be reliably maintained through a quantum computation within such an environment.

Despite all these disadvantages, the Brownian architecture is perhaps of more interest as it is ostensibly easier and cheaper to construct with current technology. The most exciting avenue for this may lie in the fields of biological and molecular computing, in which molecules such as DNA are constructed in such a way that their resulting dynamics encode a computational process~\cites{winfree-tam,cardelli-dsd}. Chemical systems are a natural substrate for Brownian dynamics, and the manipulation of DNA is becoming ever cheaper and more sophisticated. With improvements in synthetic biology, we may even be able to readily prepare self-repairing DNA computers in the near future.

We finish by discussing the timestep for a net computational transition in the Brownian architecture. This timestep is controlled by the biases, $b_i$. To maximise the system's computation rate, we take $b_i=\beta\propto1/\sqrt\ell$. If, at operating reactant concentrations, the time for any computational step (forwards or backwards) is $t_0$, then the net timestep is given by 
\begin{align*}
  \tau &= \frac{t_0}{\beta} \propto \ell^{1/2} \sim V^{1/6}
\end{align*}
which can be seen to get larger as the system gets larger. Fortunately this scaling is sublinear, and so depending on the available power, the achievable bias may be significant. For example, assuming a \SI{500}{\watt\per\meter\squared} radiative capacity and a \SI{1}{\meter} radius system consisting of fairly conservative computational particles of size $(\SI{10}{\nano\meter})^3$ operating at a gross rate of \SI{1}{\hertz}, a bias as high as $\beta\sim0.4$ is possible.

Another approach to managing modes of computation in a Brownian computer is to institute a hierarchy of bias levels at different concentrations. For example, a viable system may consist of the following tuples $\sim(N,b;R_C)$,
\begin{align*}
  \{ (V,\ell^{-1/2};V^{5/6}),~ (V^{14/15},\ell^{-2/5};V^{4/5}),~ (V^{5/6},\ell^{-1/4};V^{3/4}),~ (V^{2/3},1;V^{2/3}) \}
\end{align*}
and computations requiring faster steps may use higher biases at the expense of less computational capacity. Additionally, the highest bias level admits irreversible computation, and so we may enclose our computer with a thin irreversible shell, whilst the internal bulk consists of reversible computations at various bias levels. Going further, we could have a hot inner Brownian core, a cold outer quantum core, and an enveloping irreversible silicon shell. Additionally, a hierarchical Brownian system need not employ spatial variation; by using different species for each bias, the subsystems can be mixed homogeneously.

\section{Conclusion}

We have shown that, subject to reasonable geometric constraints on power delivery and heat dissipation, the rate of computation of a given convex region is subject to a universal bound, scaling as $R_C\lesssim\sqrt{AV}\sim V^{5/6}$, and that this can only be achieved with reversible computation. For irreversible computation, this bound falls to $A\sim V^{2/3}$. The scaling laws persist at all practical scales, only breaking down when the system size approaches the Schwarzschild scale. For typical densities ($\sim\SI{1000}{\kilo\gram\per\meter\cubed}$), this threshold scale is $\SI{4e11}{\meter}\approx\SI{2.7}{AU}$.
Recall that we are considering systems of fixed uniform density; building and maintaining a structure of this density at such sizes may not be practical without some combination of high tensile-strength materials, orbital motion\footnote{Orbital motion near the Schwarzschild regime is not sufficient. For such an object, stable circular orbits only exist for $r>\frac32r_s$ at which point the orbital velocity approaches the speed of light.}, and propulsion.
Beyond this scale, the maximum computational rate is attained by localising mass into as thin a shell as possible, reminiscent of megastructures such as Dyson spheres from the world of science fiction.
At very small scales, the rate scales with $V$ as the surface area to volume ratio is negligible. The scaling then falls to $V^{5/6}$, to $V^{1/2}$, and finally to $V^{1/3}$ as the system increases in size. At the Schwarzschild regime and beyond, the computation rate is suppressed by a factor of order unity due to gravitational time dilation.

This analysis has assumed that each computational element acts independently. In our next two chapters we shall investigate the constraints affecting cooperative reversible architectures, in particular the thermodynamic cost of synchronisation processes, such as communication and resource sharing. These costs turn out to be quite significant, even prohibitive.

\pdfsuppresswarningpagegroup=0
\endgroup

\begingroup
\pdfsuppresswarningpagegroup=1
  \def\rate{\operatorname\lambda}

\begin{chapter-summary}

  In \Cref{chap:revi}, the limits on the sustained performance of large reversible computers were investigated and found to scale as $\sqrt{AV}$ where $A$ is the convex bounding surface area of the system and $V$ its internal volume, verifying and strengthening a result of \textcite{frank-thesis}, compared to $A$ for an irreversible computer.
  This analysis neglected, however, to consider interactions between subunits of these computers. In particular, a large computer will be composed of many small computational subunits. Each of these subunits will perform independent computation, but for useful programs one often wishes to combine the results of these independent computations. To do so, the subunits must communicate---they must synchronise their individual states. In an irreversible computer with ample free energy density, this is a non-issue. In contrast, when free energy is limiting, the hidden entropic cost of synchronisation is made manifest and must be considered.

  To illustrate, consider two reversible Brownian computers which would like to communicate with one another. A reversible Brownian computer is free to wander back and forth along its configuration space, the set of possible computational states it can take, and we can illustrate this by the analogy of two wagons on intersecting train tracks as shown in \Cref{fig:cover-ii}. In order to synchronise their states, and progress past the intersection, both wagons need to arrive simultaneously. If one wagon arrives before the other, however, then it is almost certain to wander backwards instead of waiting for the other to arrive. If the probability of moving forward is $p$ and backward is $q$, then the system is said to have `computational bias' $b=p-q\ll1$, and the probability that a wagon is at the intersection point is $b/p$. The rate of synchronisation can then be found approximately by the probability that both are there simultaneously, $\sim b^2\lll1$.

  Already the rate of independent computation for each subunit is vanishingly small in the limit of large system sizes, and so synchronisation and communication events in Brownian computers are seen to progress even more slowly and appear to `freeze out' as free energy gets sparser and sparser.
  In this chapter, we show that there is indeed no way to design a better system for synchronisation of Brownian computers, evaluating the `time penalty' for a synchronisation event across a broad range of possible designs.
  We also briefly consider the applicability to non-Brownian reversible computers, conjecturing that all reversible computers whose computational elements evolve independently are subject to the same penalty.
  We conclude with a discussion of design approaches and mitigations to circumvent the synchronisation time penalty in those reversible computers subject to such a penalty.

\end{chapter-summary}

  \chapter{Performance Trade-offs for Communicating Reversible Computers}
  \label{chap:revii}

\begin{figure}[h!]
  \centering
  \includegraphics[width=.4\textwidth]{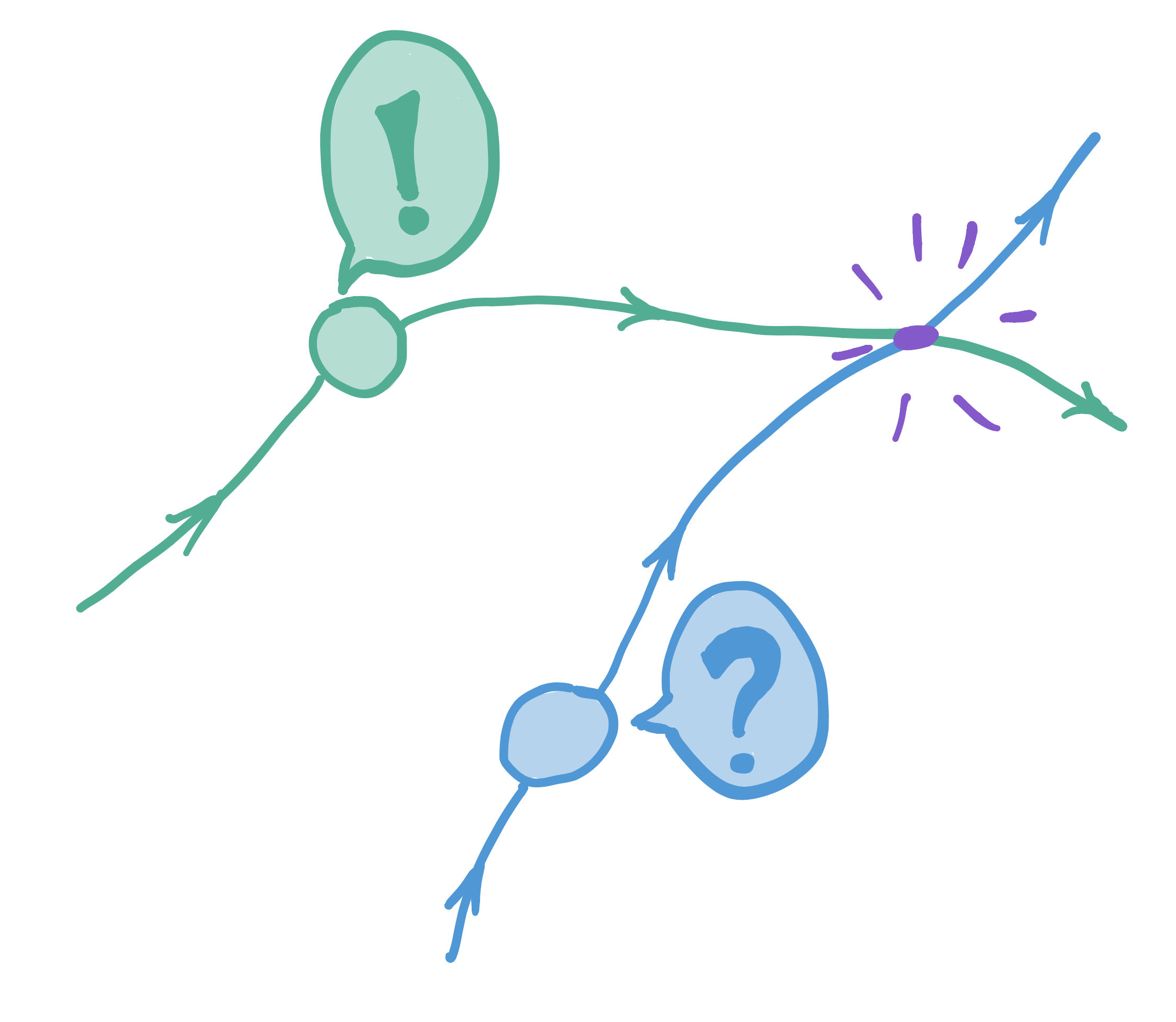}%
  \includegraphics[width=.4\textwidth]{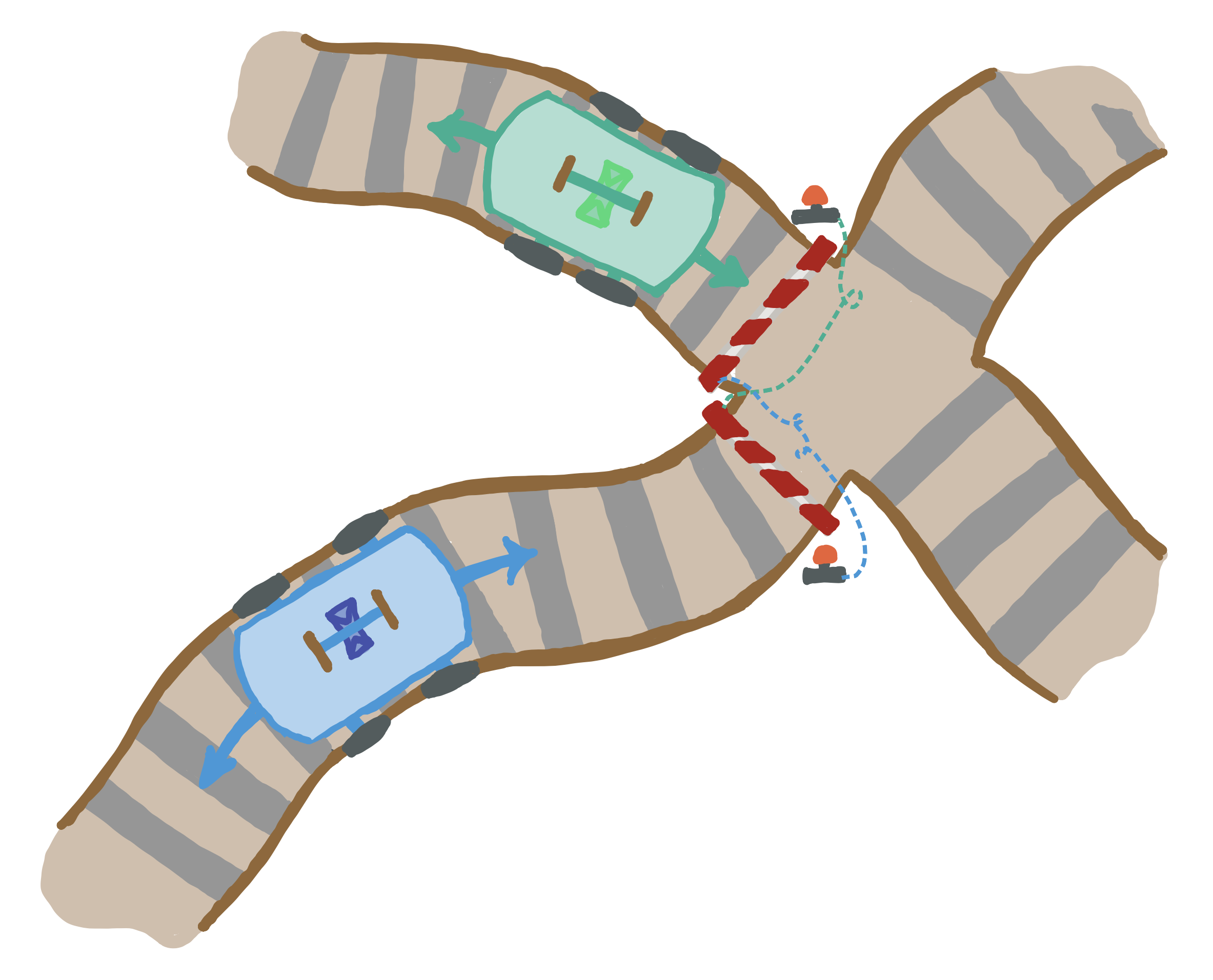}
    
  \caption[A cartoon representation of synchronisation interactions between reversible Brownian computational particles.]{This chapter concerns synchronisation interactions between reversible Brownian computational particles, which can walk backwards and forwards through phase space. This is illustrated here as two wagons on intersecting train tracks, where there is a barrier at the intersection. Each wagon controls the other's barrier, and so both must be present simultaneously in order to pass through the intersection.}
  \label{fig:cover-ii}
\end{figure}

\section{Introduction}

In \Cref{chap:revi}, we investigated how the laws of physics constrain the maximum rate of sustained computation within a given region of space.
In fact, this performance measure of operations per unit time is not the only one of interest.
As we are considering computers of arbitrary size, interaction between different subvolumes will often be of considerable import.
In this chapter and \Cref{chap:reviii}, we evaluate performance within this context.
\Cref{chap:revii} concerns the statistical dynamics of communication between two or more distinct computational units, whilst \Cref{chap:reviii} concerns the statistical dynamics of computational units interacting with resources or chemical species, of which there may be a variable number of identical units or particles.
Within these analyses we shall work with Brownian computers as our canonical computational system.
Nevertheless, we believe our results are more general, and in the conclusion of this chapter we conjecture upon the classes of reversible computational systems to which our results do and do not apply.

To proceed we must define what we mean by a `Brownian computer'.
We already introduced a notion of Brownian computer in \Cref{sec:crn} in which we used the formalism of Chemical Reaction Networks (CRNs), but the analyses of \Cref{chap:revii,chap:reviii} shall require some elaboration.
First, it will be helpful to introduce some terminology to distinguish the unique properties of computational particles from those of more traditional chemical species:

\begin{dfn}[Mona and Klona]\label{dfn:mona-klona}
  A chemical system is generally comprised of a number of species.
  Each of these species is present at some count; that is, there is some number of particles of this species in the system, and these are \emph{identical} copies of each other.
  The study of chemical thermodynamics and statistical mechanics then builds atop this.
  In contrast, chemical and Brownian computational particles are typically \emph{unique} within the system: each particle represents the state of some computer, and each evolves independently of one another.
  It may be that occasionally two or more such particles are transiently identical, but this is a coincidental occasion.
  An alternative instance of effectively unique computational particles arises in fixed lattices of simple elements, such as in \emph{amorphous} computers.
  Here, the elements at each lattice point may not be globally unique, but as they are uniquely addressable by their lattice location they take on much the same properties.
  We refer to particles which are generally uniquely distinguishable or addressable \emph{mona} (sg.\ \emph{monon}), after the Greek for unique, \emph{\gr μοναδικός}.
  Similarly, we will call indistinct particles belonging to a species \emph{klona} (sg.\ \emph{klonon}).
\end{dfn}

To get a feel for mona, consider the example of a \SI{150}{\milli\liter} solution containing a \SI{1}{\milli\molar} concentration of random nucleic acids of length 80 nucleotides.
As explained in the introduction, DNA molecules are polymers formed from an arbitrary sequence of four nucleotide monomers.
Random nucleic acids are easily obtained from DNA synthesis companies.
There are $4^{80}\approx\num{1.5e48}$ possible distinct 80-mers (nucleic acid of length 80), and the birthday paradox says that we need $\num{1.4e24}$ of these random molecules before it is likely that there are duplicates present.
Our \SI{150}{\milli\liter}, \SI{1}{\milli\molar} solution contains just shy of $\num{e20}$ such molecules, and it can be calculated that the probability that there are any duplicates in the solution is less than 1 in 350 million.
Therefore, it is almost certain that every nucleic acid in said solution is a monon.

We can therefore say that \Cref{chap:revii} studies mona-mona interactions, and \Cref{chap:reviii} mona-klona interactions.
We do not cover klona-klona interactions as these are well understood within the domains of chemistry and statistical physics.

Informally, Brownian computers are the sorts of systems embodied by the examples of the DNA-Strand Displacement (DSD) and Tile-Assembly Model (TAM) computer paradigms illustrated in the introduction and shown in \Cref{fig:dna-principle}.
Those DNA computers using DSD as their mechanism have a klonal representation of their computational state, whereas those using the TAM have a monal representation.
This is made clearer in \Cref{dfn:brownian} below:

\begin{dfn}[Brownian Computers]\label{dfn:brownian}
  A Brownian computer is a region of space (typically bounded or otherwise confined) containing particles.
  The definitions of particles and space are not important, and they may be complex macromolecules or simple elementary particles or even abstract entities in an abstract topological space.
  These particles are generally static in isolation, but when two or more particles collide they may alter their states during the interaction, in much the same way as a molecule of \ce{ATP} and a molecule of \ce{H2O} may collide, react, and beget the products \ce{ADP}, \ce{P_{i}} and \ce{H+} (particle number need not be conserved).
  These interactions must be \emph{reversible}, such that the products may collide and reform the reactants.
  To enable these interactions, particles must be endowed with velocities, and these velocities are assumed to follow a thermal distribution.

  A Brownian computer is characterised by being able to interpret the current state of the particles as a computational state:
    perhaps some particles are `computers' in their own right, and thus the evolution of their state over time is isomorphic to some model of computation, or perhaps the joint state of a number of particles must be considered.
  The former, wherein the computational particles are mona, is more powerful, whereas the latter approach, representing computational state with klona, typically results in the entire reaction volume being dedicated to a single program and thus limits the computational power.
  Multiple programs can be run if their components are orthogonal in the sense that they don't interact with components of a different program, but designing orthogonal klona is challenging and limited.
  An alternative mitigation approach is to use multiple bounded sub-volumes containing just the klona for that particular program;
    these sub-volumes may then be considered as abstract `particles' or mona, as can the combination of orthogonal sub-klona (forming a virtual sub-volume, if you will).
  Computational mona can be driven forward in computational phase space by a \emph{computational bias} (\Cref{dfn:bias}, below).
\end{dfn}

\begin{dfn}[Computational Bias]\label{dfn:bias}
  Recall from \Cref{chap:revi} that in a Brownian computer, transitions are mediated by systems of klona representing `positive' and `negative' bias.
  Physically, such klona act as a source of free energy and a real world example is given by $\ce{ATP}$ and $\ce{ADP + P_{\textrm i}}$ in biochemical systems.
  In the simplest case, there are $\oplus$ and $\ominus$ klona such that
    $\forall n.~\ce{$\ket{n}$ + $\oplus$ <=> $\ket{n+1}$ + $\ominus$}$
  where $\ket{n}$ is the $n^{\text{th}}$ state of some monon, and this is positively biased when the concentration of $\oplus$ exceeds that of $\ominus$.
  We define the bias in this case to be $b=([\oplus]-[\ominus])/([\oplus]+[\ominus])$.
\end{dfn}

To take full advantage of these molecular computers we need to let them communicate and interact.
Each interaction event will be between specific mona.
If these mona diffuse freely throughout the volume, then communicating with a specific monon will be impracticable except in the case of very small (effective) volumes. On average, the time between collisions of $n$ such mona in a system of volume $V$ will scale as $\sim\bigOO{V^{n-1}}$. Moreover, we shall see that a significant problem in reversible communication is that even after mona meet they can go backwards in phase space at which point we must wait for them to collide again; clearly in the freely diffusive case this is untenable. To improve upon this scaling, we are forced to introduce a lattice-like structure into the system. This lattice should be endowed with a coordinate system and the coordinates should be logically chosen so that it is `easy' to deterministically compute a unique path between any pair of reachable coordinates. Interacting mona should bind to this lattice, and couple their movement along it to their computational bias (\Cref{dfn:bias}). In so doing, translocation along the lattice becomes just another part of the monon's computational state---the information content of the monon that may be transformed in the course of computation.

\doublepage{%
\begin{figure}[!p]
  \centering
  \begin{subfigure}{.58\linewidth}
    \includegraphics[width=.9\linewidth]{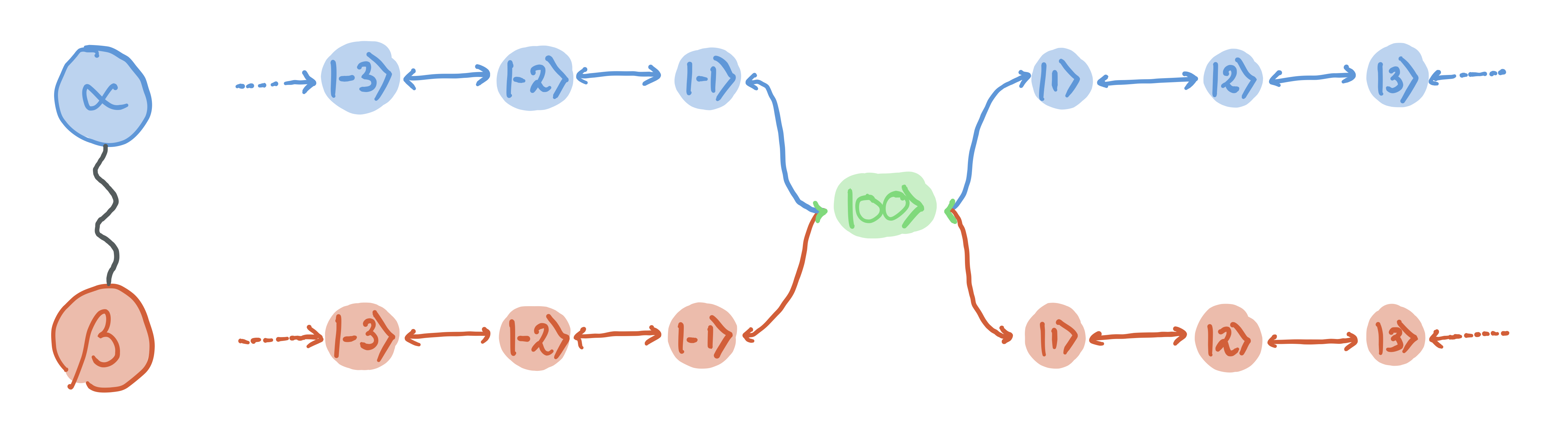}
    \caption{}\label{fig:synch-juxt}
  \end{subfigure}

  \begin{subfigure}{.207\linewidth} 
    \includegraphics[width=.97\linewidth]{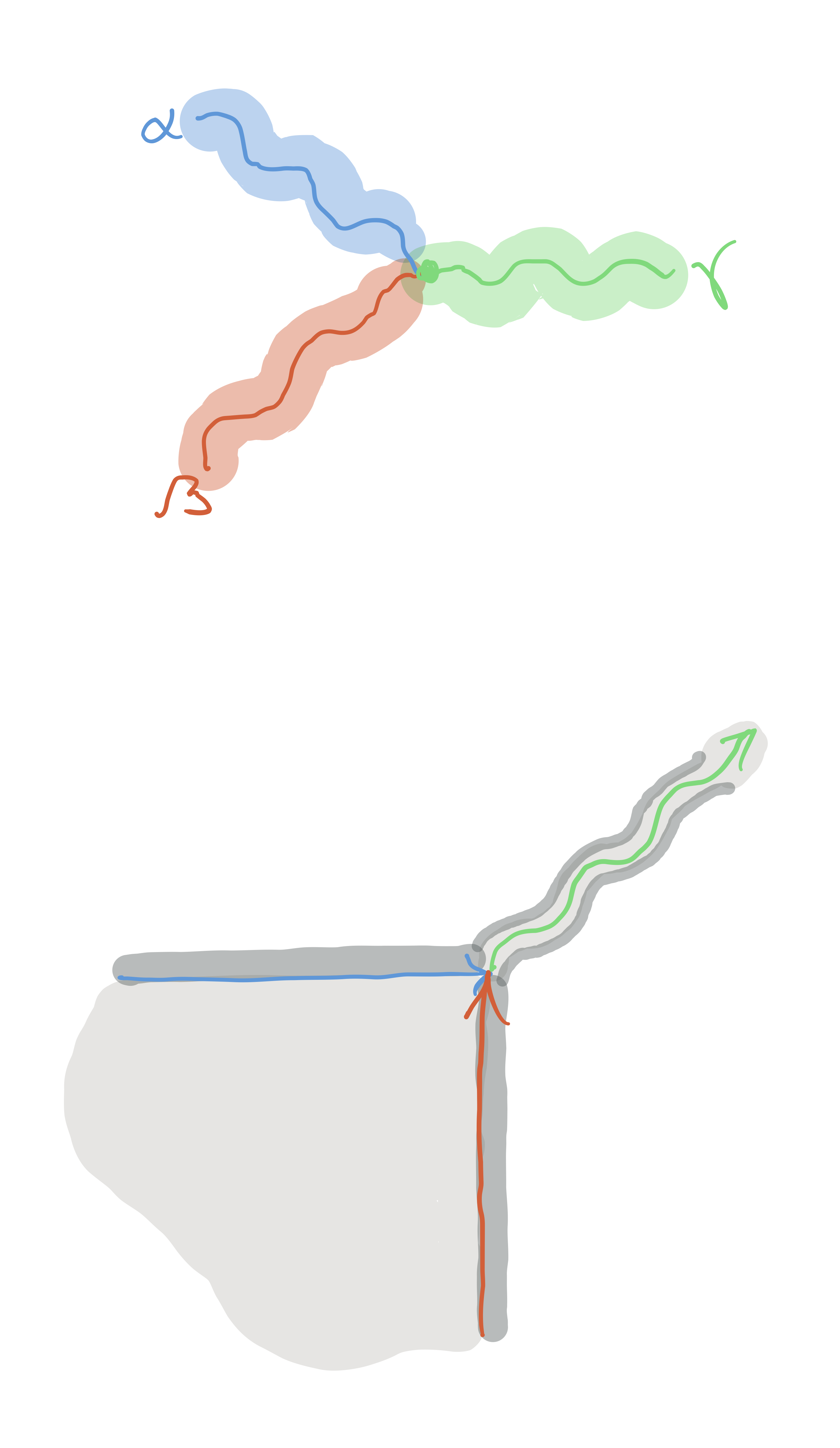}
    \caption{}\label{fig:synch-simple-1}
  \end{subfigure}%
  \begin{subfigure}{.228\linewidth} 
    \includegraphics[width=.97\linewidth]{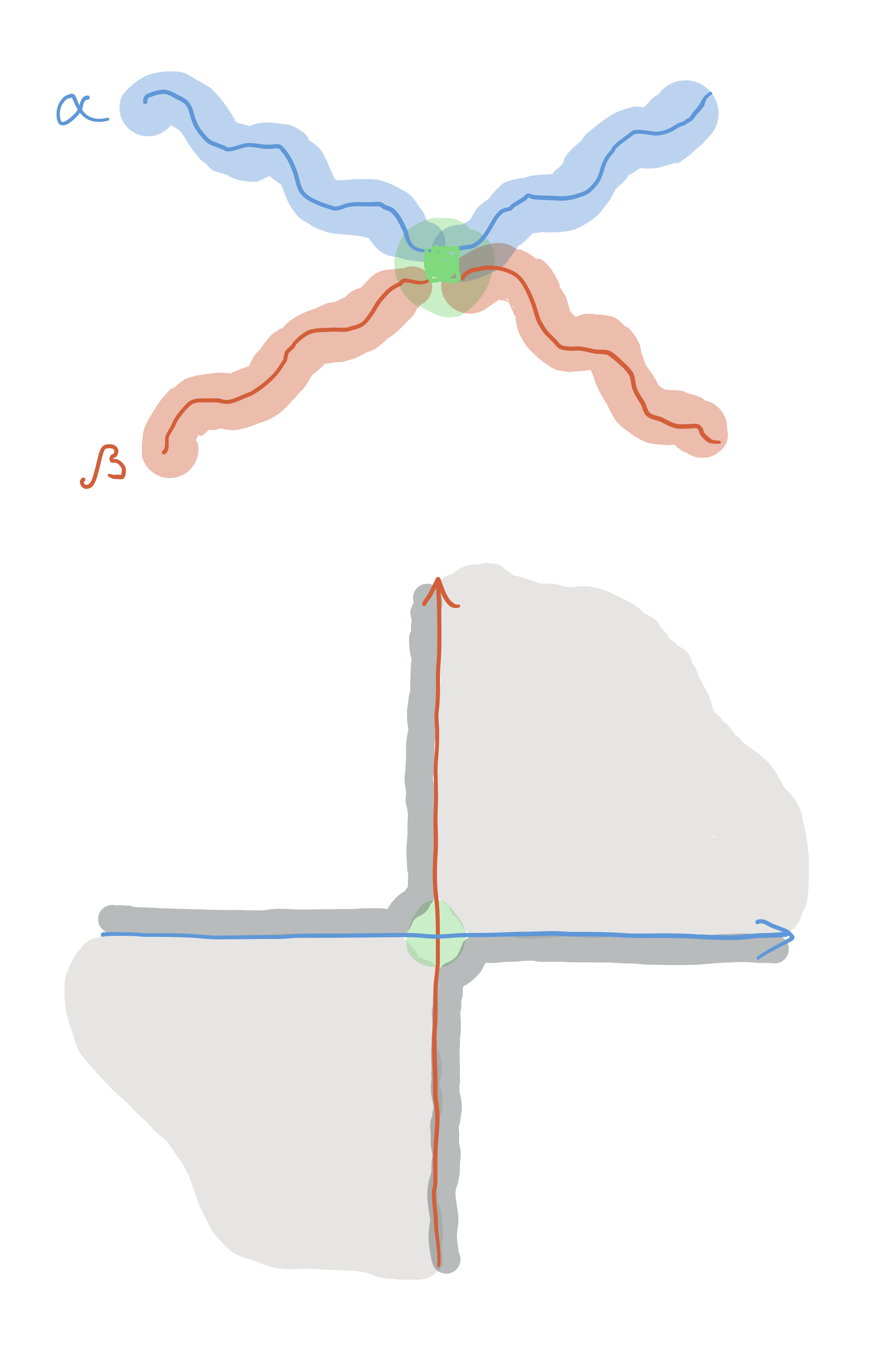}
    \caption{}\label{fig:synch-simple-2}
  \end{subfigure}%
  \begin{subfigure}{.284\linewidth} 
    \includegraphics[width=.97\linewidth]{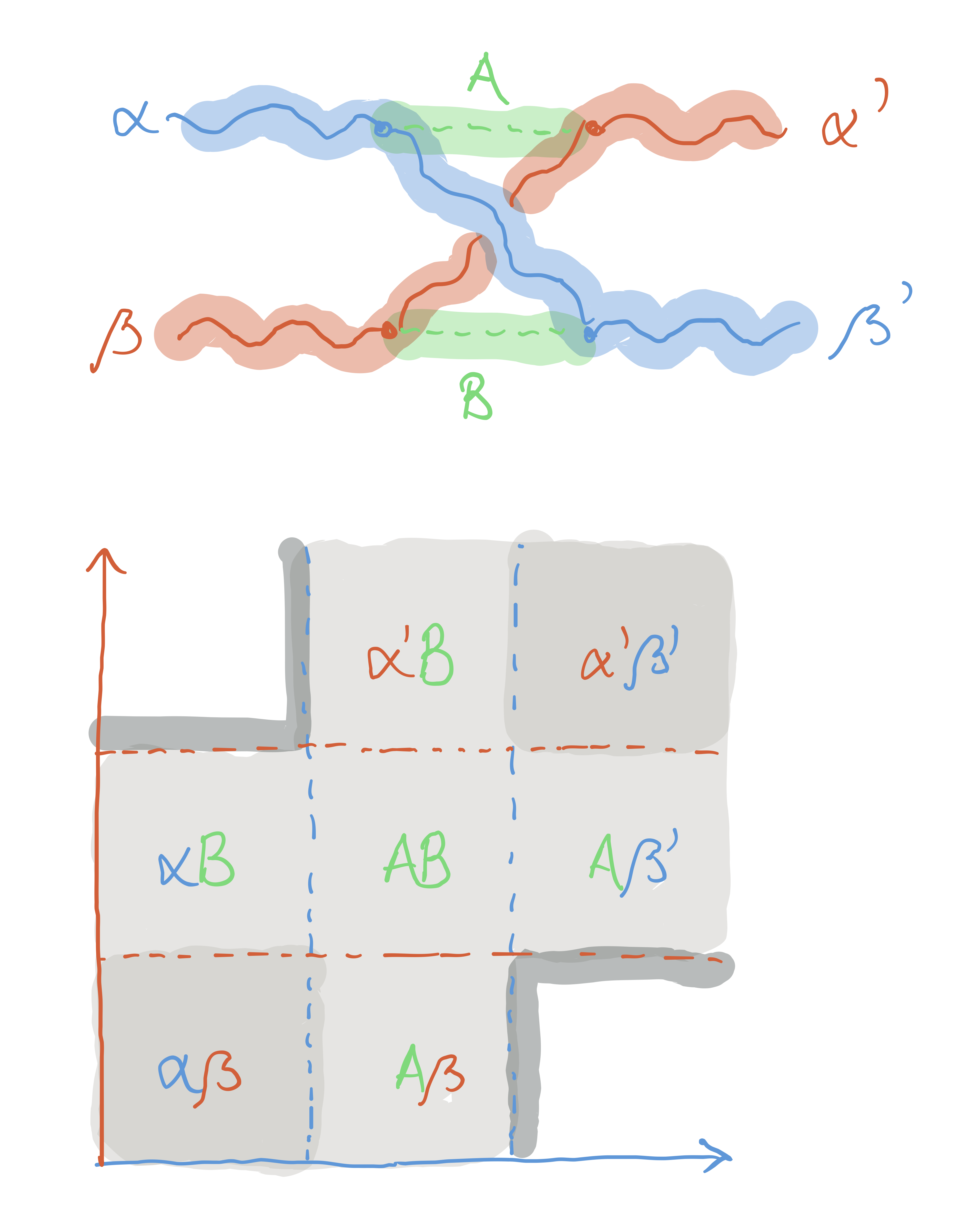}
    \caption{}\label{fig:synch-wide-1}
  \end{subfigure}%
  \begin{subfigure}{.260\linewidth} 
    \includegraphics[width=.97\linewidth]{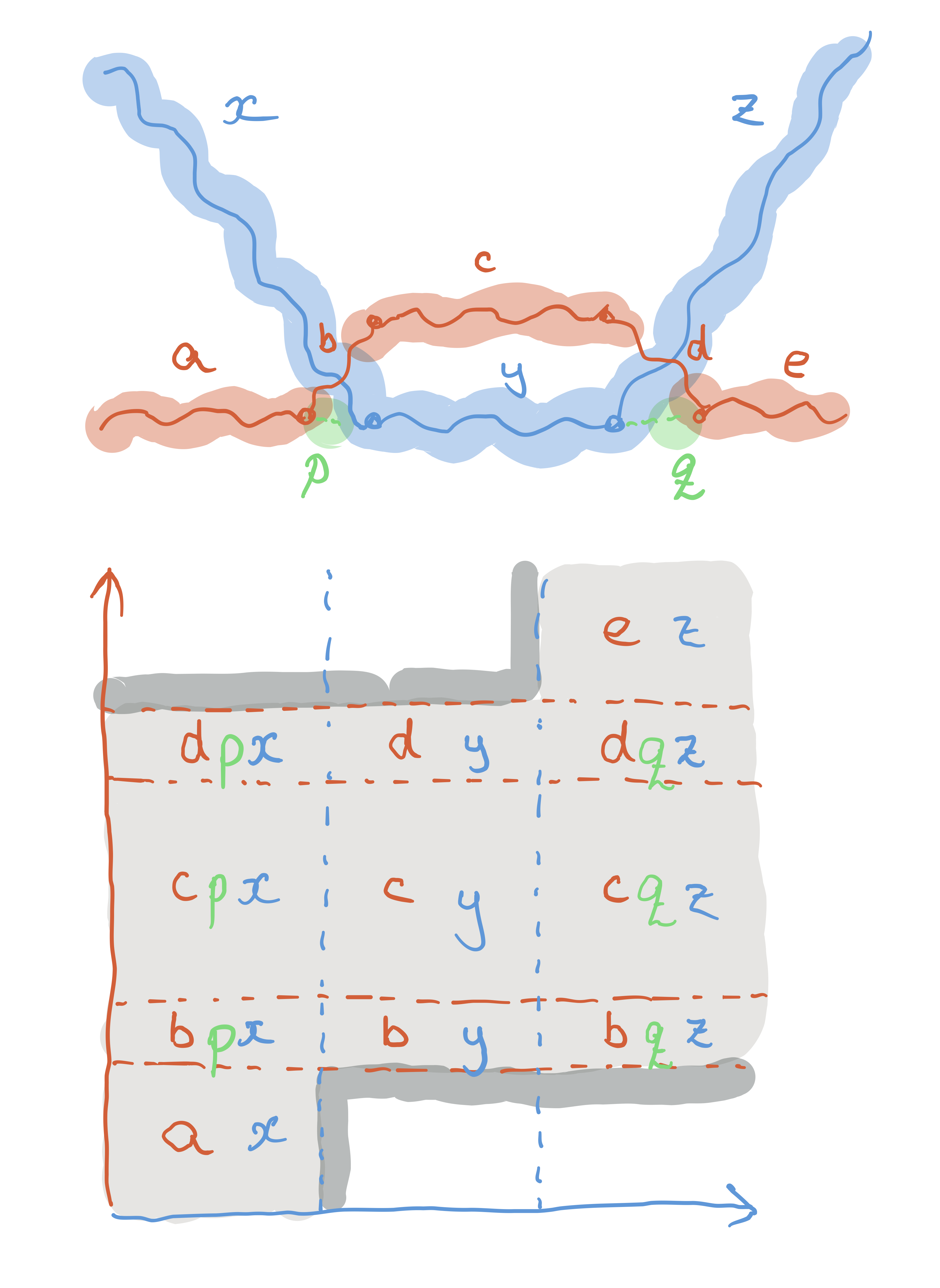}
    \caption{}\label{fig:synch-wide-2}
  \end{subfigure}

  \begin{subfigure}{.316\linewidth}
    \includegraphics[width=.97\linewidth]{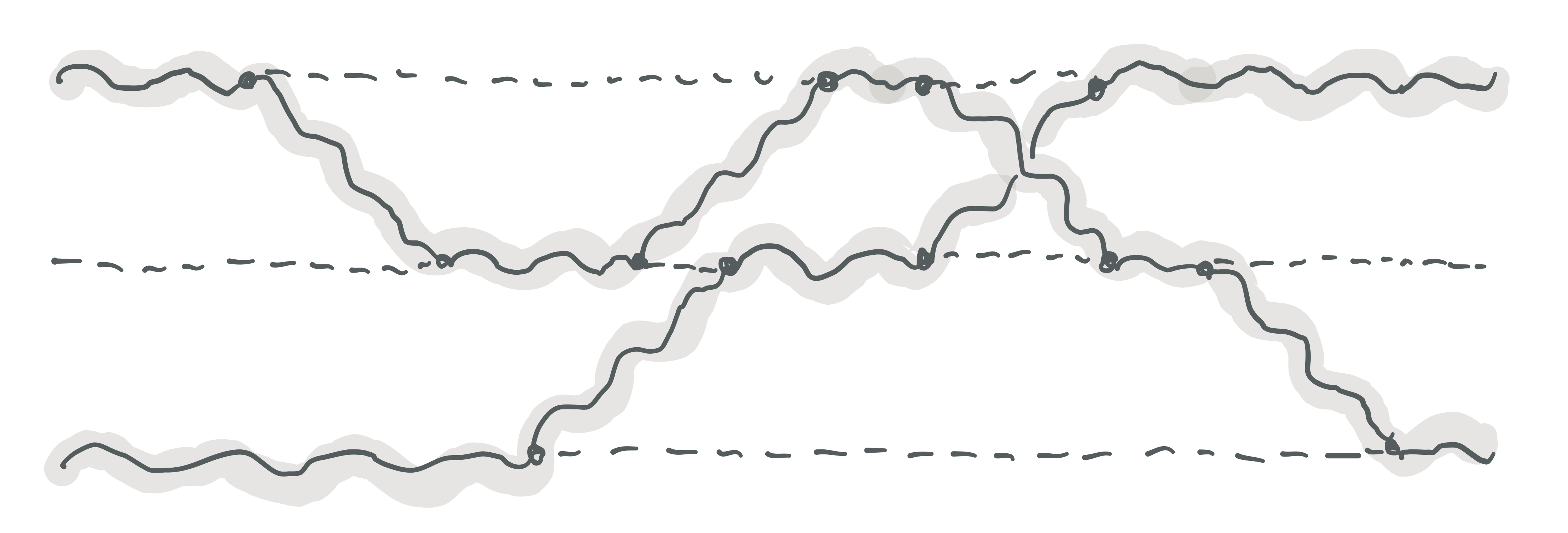}
    \caption{}\label{fig:synch-multi-1}
  \end{subfigure}%
  \begin{subfigure}{.240\linewidth}
    \includegraphics[width=.97\linewidth]{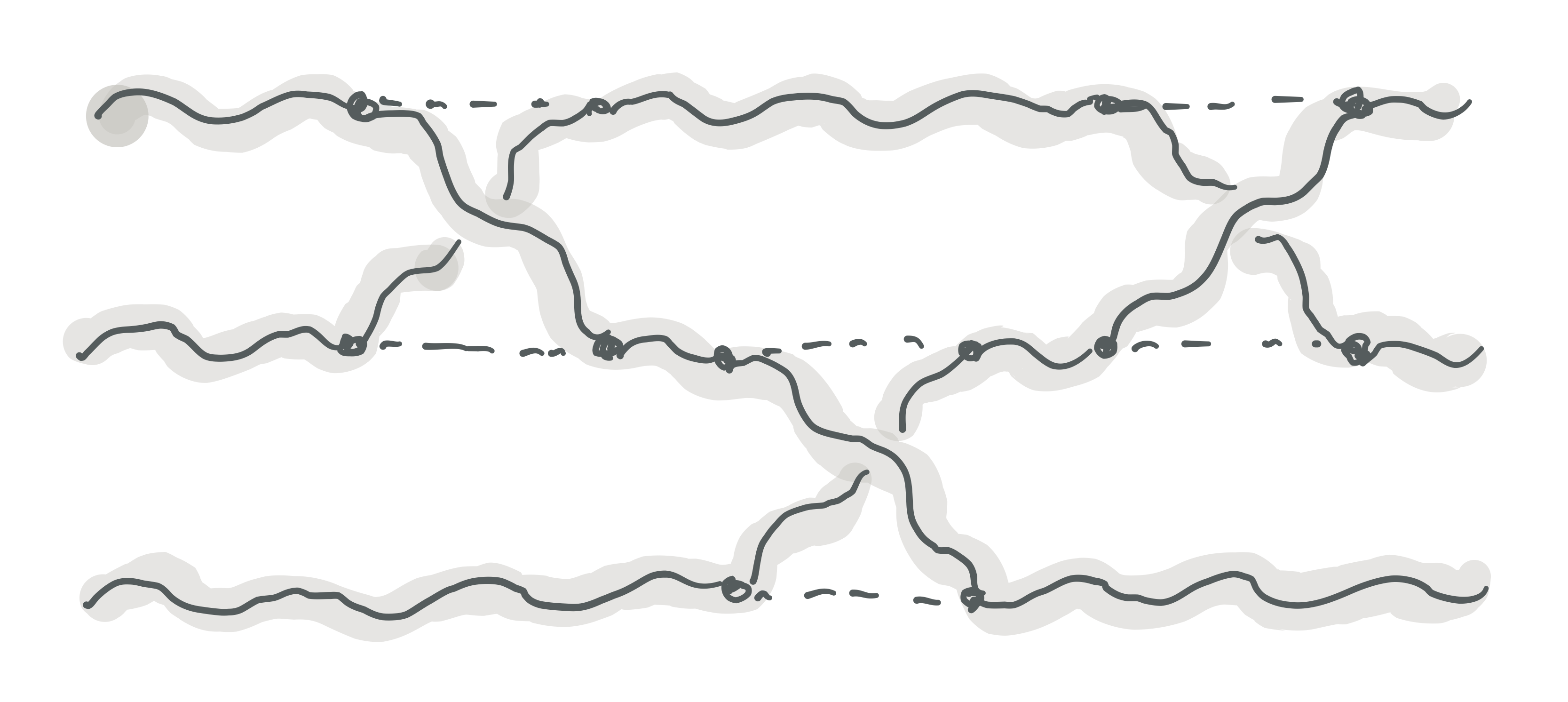}
    \caption{}\label{fig:synch-multi-2}
  \end{subfigure}

  \caption[A range of examples of \emph{constrictive} synchronisation problems, represented in different ways.]{A range of examples of \emph{constrictive} synchronisation problems, represented in different ways. Figure (a) shows a state diagram for two bound mona $\alpha$ and $\beta$; each monon is free to walk back and forth along their state space ($\mathbb{Z}$) independently except for at their \emph{synchronisation point}, $\ket{00}$, through which they must pass together. Figures (b--e) show a range of increasingly complex examples from two perspectives; the top diagrams are abstract state diagrams reminiscent of Feynman diagrams, whilst the bottom diagrams show the joint phase space of the mona. Dashed lines in state space, or `artefacts', represent static/non-evolving mona; for example, in (d) monon $\alpha$ leaves an artefact $A$---which may contain some informational payload or data packet---that is then absorbed by monon $\beta$, and vice-versa for $B$. More specifically, an artefact is a packet of data left on the lattice by a computational monon after it departs, ready to be absorbed by another computational monon.\capnl%
  Figure (b) is an asymmetric fusion/fission synchronisation, depending on the direction of time. Figure (c) is symmetric, and equivalent to (a). Figure (d) shows an approach to `widen the constriction point' of (c) using artefacts, thus making synchronisation easier. Figure (e) shows how two synchronisations close in phase space lead to an apparent `tunnel'. Finally, figures (f--g) show some more complicated examples suggestive of what a real reversible communicating system might look like.}\label{fig:synch-ex-1}
\end{figure}}{%
\begin{figure}[!p]
  \centering
  \vspace{-1.5em}
  \begin{subfigure}{\linewidth}
    \small\longce
    \begin{align*}
      \ce{ X${}_{m}$ &<=>> X${}_{m+1}$ } && \ctrRule{compute${}^\ast$} & \ce{ X${}_{-0}^{m}$ + Y${}_{-0}^{n}$ &<=>> X${}_{+0}^{m}$ + Y${}_{+0}^{n}$ } && \ctrRule{sync} \\
      \ce{ X${}_{-1}$ &<=>> X${}_{-0}^{0}$ } && \ctrRule{sync--begin} & \ce{ X${}_{+0}^{0}$ &<=>> X${}_{+1}$ } && \ctrRule{sync--end} \\
      \ce{ X${}_{-0}^{n}$ &<=>> X${}_{-0}^{n+1}$ } && \ctrRule{sync--up} & \ce{ X${}_{+0}^{n+1}$ &<=>> X${}_{+0}^{n}$ } && \ctrRule{sync--down}
    \end{align*}
    \caption{Example recessive `reaction scheme'. N.B.\ In the \ctrRule{compute} rule, $m\le-2$ or $m\ge1$.}
  \end{subfigure}
  \begin{subfigure}{\linewidth}
    \small
    \begin{gather*}
      \left.\begin{array}{r}
        \cdots~\ce{ A${}_{-3}$ <=>> A${}_{-2}$ <=>> A${}_{-1}$ <=>> A${}_{-0}^{0}$ <=>> A${}_{-0}^{1}$ <=>> A${}_{-0}^{2}$ <=>> A${}_{-0}^{3}$ } \\
        \cdots~\ce{ B${}_{-3}$ <=>> B${}_{-2}$ <=>> B${}_{-1}$ <=>> B${}_{-0}^{0}$ }
      \end{array}\right\} \\[0.5em]
        \ce{ A${}_{-0}^{3}$ + B${}_{-0}^{0}$ <=>> A${}_{+0}^{3}$ + B${}_{+0}^{0}$ } \\[0.5em]
      \left\{\begin{array}{l}
        \ce{ A${}_{+0}^{3}$ <=>> A${}_{+0}^{2}$ <=>> A${}_{+0}^{1}$ <=>> A${}_{+0}^{0}$ <=>> A${}_{+1}$ <=>> A${}_{+2}$ <=>> A${}_{+3}$ }~\cdots \\
        \ce{ B${}_{+0}^{0}$ <=>> B${}_{+1}$ <=>> B${}_{+2}$ <=>> B${}_{+3}$ }~\cdots
      \end{array}\right.
    \end{gather*}
    \caption{An example use of the above reaction scheme between two mona, \ce{A} and \ce{B}.}
  \end{subfigure}\\[1em]
  \begin{subfigure}{.2\linewidth}
    \includegraphics[width=.97\linewidth]{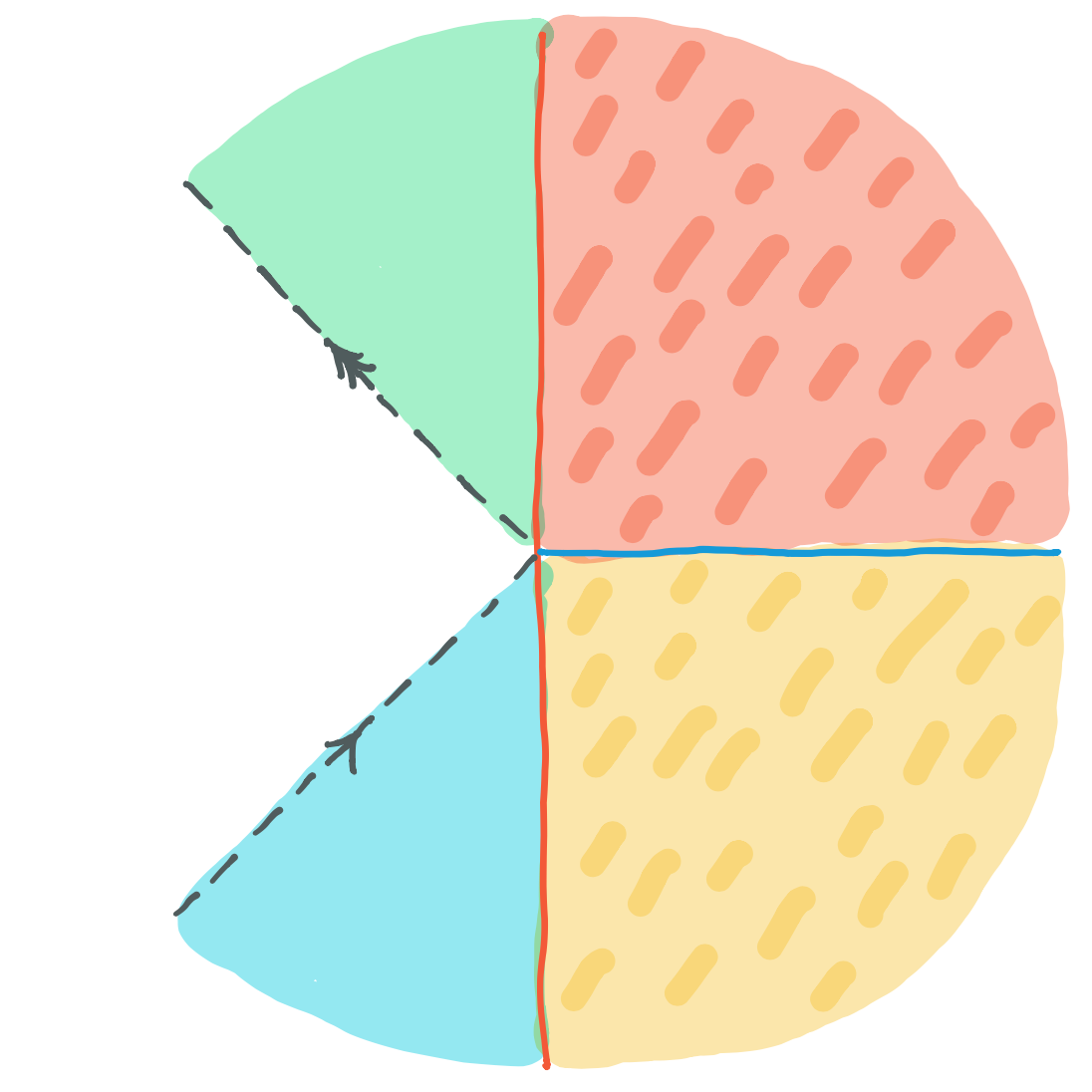}
    \caption{}\label{fig:synch-rec-1}
  \end{subfigure}~~%
  \begin{subfigure}{.2\linewidth}
    \includegraphics[width=.97\linewidth]{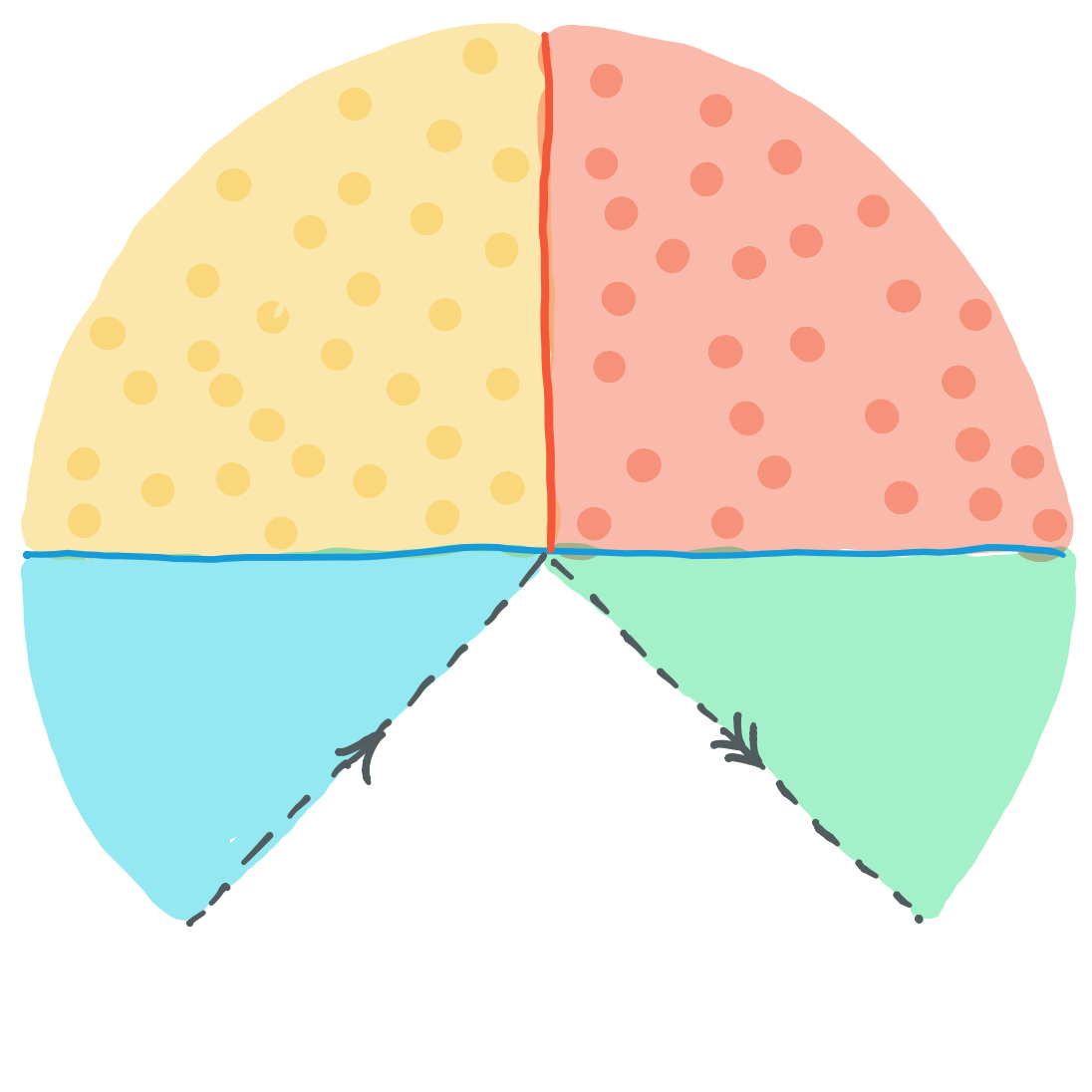}
    \caption{}\label{fig:synch-rec-2}
  \end{subfigure}~~%
  \begin{subfigure}{.2\linewidth}
    \includegraphics[width=.97\linewidth]{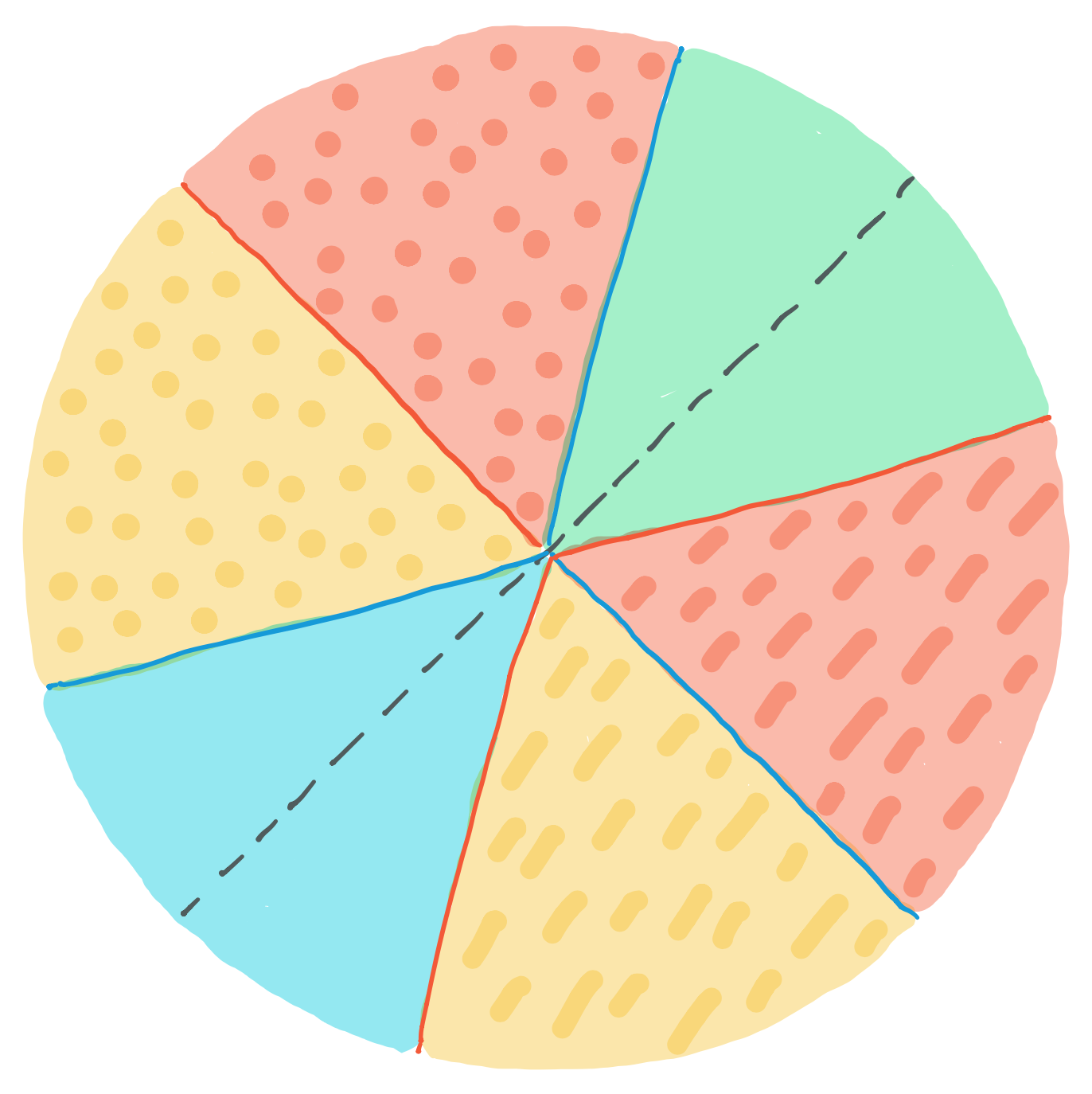}
    \caption{}\label{fig:synch-rec-domain}
  \end{subfigure}

  \caption[An example of a recessive synchronisation problem.]{%
    An example of a recessive synchronisation problem.
    In a recessive system, there can be no disallowed transitions.
    This means that if a monon \ce{X} reaches the synchronisation point prematurely, then it must continue to do `work' whilst waiting to synchronise.
    The simplest approach is for \ce{X} to start counting the net number of steps that it has been waiting to synchronise.
    When the other monon \ce{Y} arrives the synchronisation can complete, but the counter must be subsequently consumed.
    (a) An example reaction scheme implementing recessive synchronisation.
    For computation prior to synchronisation the states are labelled $\ldots, \ce{X_{-3}}, \ce{X_{-2}}, \ce{X_{-1}}$, and post synchronisation they are labelled $\ce{X_{+1}}, \ce{X_{+2}}, \ce{X_{+3}}, \ldots$.
    Transitions between these states are implemented by the rule \ctrRule{compute}.
    A special \ce{X${}_{\pm0}^{n}$} state is used for the synchronisation-competent state, where $\pm$ represents whether or not synchronisation has yet occurred and $n$ the number of steps counted.
    The four rules \ctrRule{sync--begin,up,down,end} implement the counter.
    Finally the rule \ctrRule{sync} implements the synchronisation interaction proper.
    This reaction should be assumed much faster than the counter reactions so that only one monon counts.
    (b) An example instantiation of the reaction scheme between two mona, \ce{A} and \ce{B}.
    Here, \ce{A} arrives first and counts for 3 steps.
    After synchronising, it counts back down.
    A key feature of this system is that it does not retain a memory of $n=3$; upon reversal, a higher or lower value of $n$ may be reached or it may even be the case that \ce{B} counts up instead of \ce{A}.
    This effectively sidesteps the issue of `erasing' the value of $n$.
    (c,d) Regions of phase space corresponding to the cases of \ce{A} (resp.\ \ce{B}) reaching the synchronisation point first.
    The horizontal axis corresponds to the states of \ce{A} and the vertical to those of \ce{B}.
    Therefore, in (c), the yellow quadrant corresponds to \ce{A} counting up whilst awaiting \ce{B}'s arrival, and the orange quadrant to it counting down.
    The green quadrant corresponds to both \ce{A} and \ce{B} having resumed independent computation.
    In (d) the yellow quadrant corresponds to \ce{B} counting up and the orange to it counting down.
    Note that the bias switches direction; in (c), the bias points NE in the blue and yellow quadrants, and NW in the orange and green.
    In (d) the bias points NE in the blue and yellow quadrants, and SE in the orange and green.
    (e) A representation of the entire phase space, consisting of the six `quadrants' stitched together.
  }\label{fig:synch-ex-2}
\end{figure}}

In order for two or more mona to communicate, or more generally to synchronise their joint state, we must arrange for their states to coincide at some point. Moreover, no monon will be able to proceed past this synchronisation point without the other(s). In the irreversible limit of computing, this problem is inconsequential as each computational entity can simply wait for the others. In reality, `waiting' is not a physically reversible process: If the dynamics were reversed it would not be possible for the entity to `know' how long to wait before going backwards. In an irreversible system, waiting can be simulated by raising the entropy of the environment; for example, the system could continually `forget' how long it has been since it arrived by resetting some internal state using the Landauer-Szilard principle~\cite{szilard-engine,landauer-limit}, or the system could transfer its generalised computational `momentum' to some other computational entity. In the reversible computers we are considering, there is a dearth of free energy\footnote{Free energy is the ability of the system to drive an entropically unfavourable reaction, by increasing the entropy of the environment at least as much as the reaction reduces the local entropy.}\ as we are distributing the free energy supplied at the computer's surface throughout its entire volume, and so these approaches are not generally possible.

Instead, the simplest approach is to let the mona `bounce off' the synchronisation point until the other arrives, relying on the (weak) computational bias to drive them towards this. That is, upon reaching the synchronisation point the monon cannot proceed and has a high chance of walking backwards, but the weak bias will keep it within a loose vicinity of the synchronisation point. Some initial estimates (\Cref{sec:constrict-initial}) will show that, in addition to the time for each monon to reach the synchronisation point independently, there is a `penalty' term. For computational bias $b$, each net computational step for lone mona takes time $(b\lambda)^{-1}$ where $\lambda$ is the rate at which mona collide with bias klona. If each participating monon starts at a phase distance $n_i$ from the synchronisation point, then the expected time for all to have briefly reached the synchronisation point is $\tau_0=(b\lambda)^{-1}\max_i n_i$. The penalty depends on the joint probability distribution of the mona, and how often they are all found at the synchronisation point. Analysing this in detail turns out to be surprisingly difficult, but we will succeed in showing a range of approximate, empirical and exact results, confirming the initial estimate of a penalty $\Delta\tau\gtrsim (b^d\lambda)^{-1}$ for $d$ mona. As $b\ll1$ generally, this will be very large compared to the timescales of independent computation by individual mona, and thus potentially rendering reversible communication impracticable.

In this chapter, we shall study this problem in detail for general shapes of synchronisation phase spaces and illustrate with examples of the forms of phase space that would be commonly found in reversibly communicating mona. We divide approaches to synchronisation in two: \Cref{sec:constrict} will concern the `constrictive' case in which there is a constriction in phase space, as illustrated in \Cref{fig:synch-ex-1}. \Cref{sec:recessive} will concern the `recessive' case in which the phase space is not restricted, and therefore other methods must be used to synchronise the joint state of the mona as illustrated in \Cref{fig:synch-ex-2}. In fact, there are two more related cases, dilatory and processive, which are in a sense complementary to the constrictive and recessive cases respectively. In the dilatory case, the phase space widens towards the synchronisation point rather than shrinking, whilst in the processive case the boundary of the post-synchronisation subspace (green quadrant in \Cref{fig:synch-ex-2}) is concave rather than convex. 
In both of these, the penalty will be negligibly small, zero, or even negative as the phase coordinate is pushed entropically towards larger regions of phase space. Nevertheless, they are not desirable because such asymmetry implies an increase in entropy and hence irreversibility. We can therefore make a stronger and more accurate statement: rather than forbidding dilatory cases, we require synchronisation phase spaces to be symmetric in the synchronisation point or subspace.

\section{Methodology}\label{sec:methodology}

\para{Isolated Mona}

To formalise the problem statement, we begin by considering computation by isolated mona in more detail. A monon's evolution is best understood by considering the state space of the computation it is simulating as depicted in \Cref{fig:state-space}. Ignoring any possible intermediate states, there will be an isomorphism between the states of the monon and the states of its underlying computation. As was made apparent in the previous chapter, we will want to drive computation forward by coupling the monon's state transitions to some non-equilibrium system of klona, which we term a bias source. Each transition may be coupled to a different bias source, and the strengths of the bias sources may vary over time. The definition of bias sources and their strengths is recalled in \Cref{dfn:bias}. Where possible, we will remain general in our analysis; however, \Cref{sec:crn} showed that optimal performance is achieved for bias constant in time and uniform in space, and our illustrative examples will reflect this fact.

\begin{figure}
  \def\rmsp{\nobreak\hspace{-0.1ex}}
  \centering
  \begin{subfigure}{\linewidth}
    \centering
    \includegraphics[width=.85\linewidth]{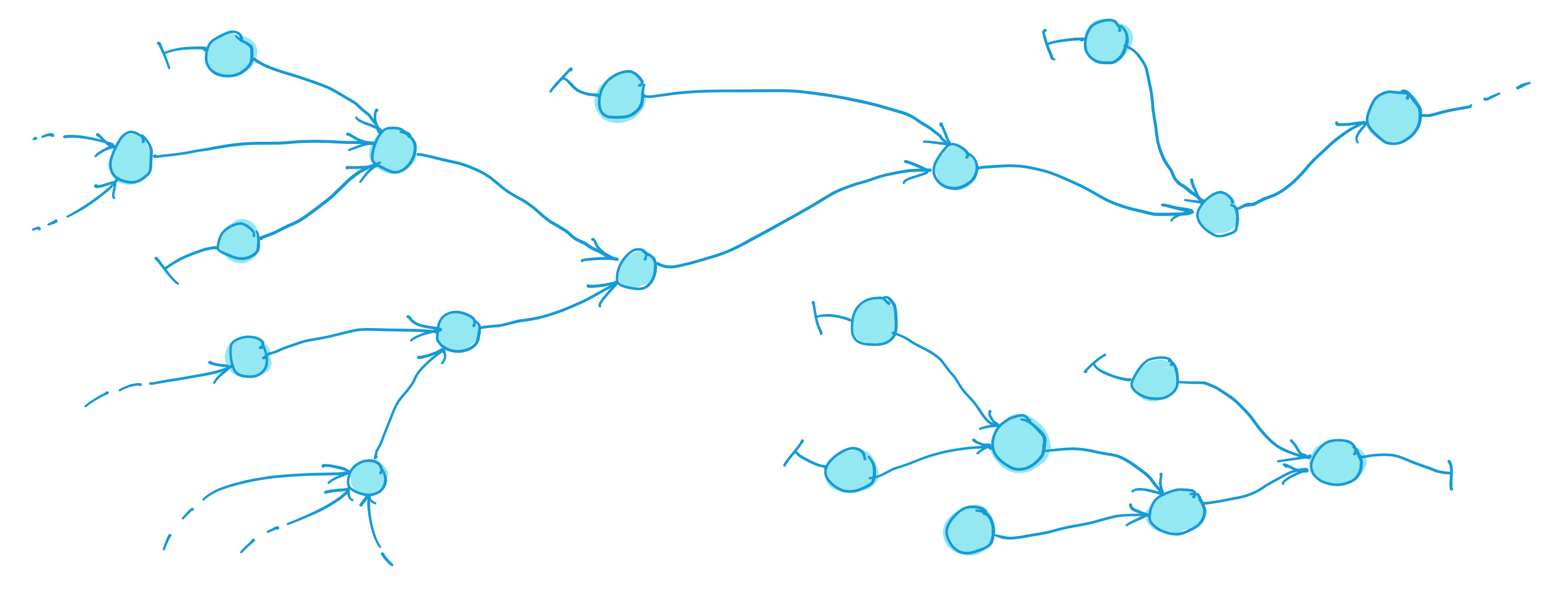}
    \caption{Examples of irreversible state spaces}
  \end{subfigure}\\[1.5em]
  \begin{subfigure}{\linewidth}
    \centering
    \includegraphics[width=.85\linewidth]{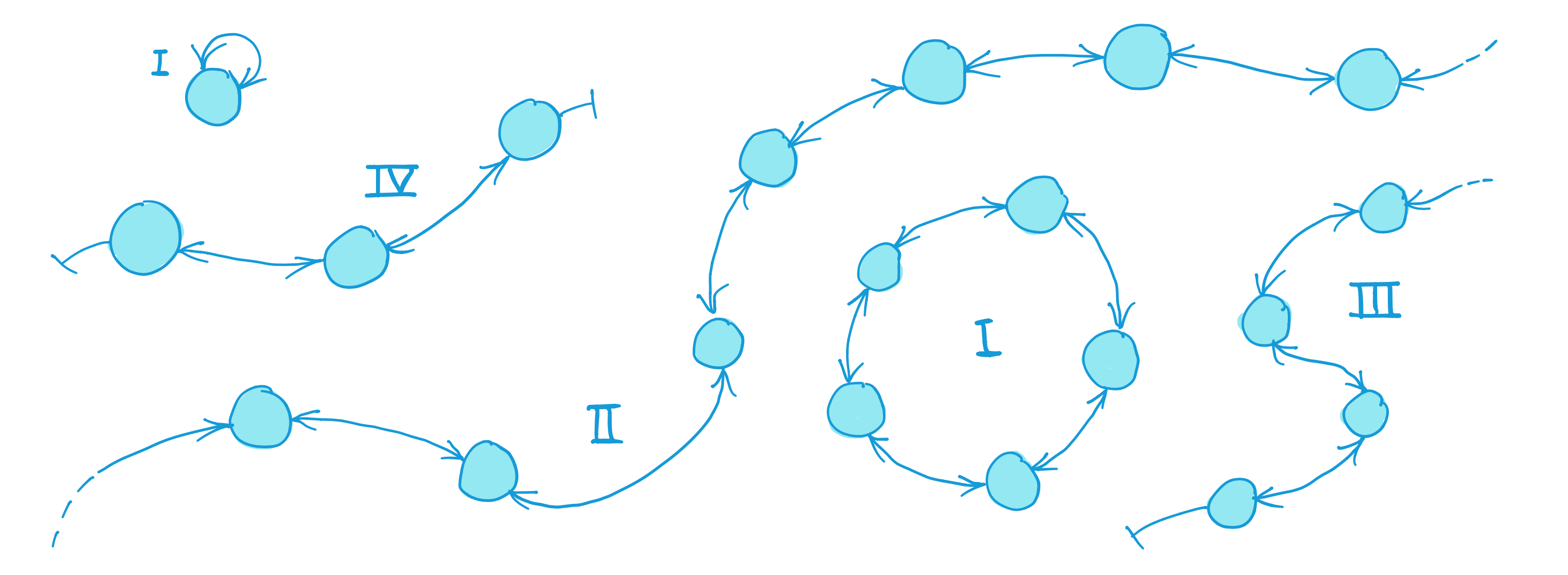}
    \caption{Examples of reversible state spaces}
  \end{subfigure}
  \caption[Some examples of the state spaces of irreversible and reversible computers.]{%
    These figures show some examples of the state spaces of (a)~irreversible and (b)~reversible computers.
    The blue nodes represent distinct computational states, and arcs represent computational transitions or operations.
    Terminal arcs, represented by $\vdash$, indicate that computation either begins or ends here; that is, no predecessor or successor states (respectively) exist.
    Recall that the dynamics of the system are free to wander forwards and backwards in state space, including away from a terminal state.\capnl
    The distinction between irreversible and reversible computers is that an irreversible state may (and often \emph{does}) have more than one predecessor state.
    For example, if two Boolean variables are consumed and replaced with their logical conjunction, $(a,b)\mapsto a\land b$, and if this is \atom{False}, then there are three possible predecessors: $(\atom{False},\atom{False})$, $(\atom{False},\atom{True})$, and $(\atom{True},\atom{False})$.
    This is loss of information, and it is impossible to determine which path was taken to reach the current state.
    Where an irreversible computer completely forgets or discards a variable---such as of type \code{uint32}---it may have a very large number of predecessor states---for \code{uint32}, $2^{32}\approx\num{4e9}$.
    For abstract models of computation such as a Register Machine (see, e.g., \textcite{register-machine}), there can even be a countable infinity of predecessors in the case of the `clear-to-zero' operation.\capnl 
    In contrast, a reversible state may have at most one predecessor and one successor state.
    We can classify reversible computers into four classes;
      (I) finite cyclic programs with no termini,
      (I\rmsp{}I) bi-infinite programs with no termini,
      (I\rmsp{}I\rmsp{}I) infinite programs with one terminus (comparable to non-halting programs in computability theory),
      (I\rmsp{}V) bi-terminating programs with two termini (i.e.\ finite linear chains).
    Bear in mind that the form of this state space is separate from that of the program logic:
      a reversible program may well contain loops and conditional branching, but for a given program and input or output, evolution is deterministic both forwards and backwards.}
  \label{fig:state-space}
\end{figure}

The dynamics of this system can therefore be seen to be a walk in a 1-dimensional phase space where at any point, there are at most two transitions available: one to the successor state, and one to the predecessor state, should they exist. Depending on whether we use a model with discrete or continuous time\footnote{Continuous time is the more apt model, however for generating function approaches discrete time must be employed. Fortunately, when the bias rates are uniform, it is easy to convert between the two by simply dividing the times by the gross bias rate, $\lambda=\lambda_\oplus+\lambda_\ominus$.}, these transitions will be labelled with probabilities or rates that depend on the strength of the bias source(s). We note that we expect the strength of the bias sources to be very small, such that the timescale of making one step of net progress is very large compared to that of a single transition. As a result, the monon will experience a time-averaged bias, and any correlations in the system will be expected to have been exponentially suppressed by diffusive processes. Therefore, we can model the monon's dynamics as a Markov process or chain; namely, the dynamics of the system depend only on its current state, rather than its history.

As mentioned in the introduction, we shall incorporate spatial position as just another part of our computational state, such that translocation along a lattice (or indeed any other environmental interaction) can be treated as just another computational transition. This abstraction allows us to greatly simplify our analysis, as we need only track position in phase space; the exact meaning of different phase coordinates is irrelevant, and so our results will remain general in a large class of synchronisation processes.

\para{Interacting Mona}

For a system of multiple particles, the state of the system is given by a coordinate in their joint phase space. In a classical physical system of $N$ particles, this manifests as a $6N$-dimensional phase space as each particle has six degrees of freedom: three for the particle's position, and three for its momentum vector. Similarly, our system of $N$ mona has as states coordinates in an $N$-dimensional phase space. When these mona do not interact, the phase space naturally decomposes into a tensor product of each monon's individual phase space. If, however, $n$ of these mona happen to interact with one another, then the subspace formed from their degrees of freedom is irreducible: it cannot be further decomposed into a tensor product. The reason is that a synchronisation manifests as a constraint on the joint state.

To make this more concrete, suppose three mona are to share some information. Let their individual states be indexed by $\mathbb Z$, such that their joint phase space is given by a subset of $\mathbb Z^3$. Suppose the interaction occurs when each monon is in state 0, i.e.\ the coordinate at which they interact is $(0,0,0)$. As the future state of each monon depends on this interaction/communication/synchronisation event, we observe that some coordinates are causally impossible. For example, $(3,-4,5)$ would imply that the first and third mona have received information from the second, despite the second having never given them this information. In fact, we find that the only valid states are given by $\{(-a,-b,-c):a,b,c\in\mathbb N\}\cup\{(a,b,c):a,b,c\in\mathbb N\}\equiv(-\mathbb N)^3\cup\mathbb N^3$. This corresponds to the negative and positive octants of $\mathbb Z^3$, and the system must pass through the origin. The equivalent case of two mona is depicted in \Cref{fig:synch-juxt,fig:synch-simple-2}.

This is not the only way for mona to interact as illustrated by other examples in \Cref{fig:synch-ex-1}. The common factor is that a synchronisation event between mona with individual phase spaces $\mathcal R_i$ corresponds to their joint state $\mathcal R$ being a strict subset of $\mathcal R' = \medotimes_i\mathcal R_i$. Our aim is to analyse the mean time to get from some region of phase space to another, and for this we shall need to specify some regions: $\mathcal I$ will be the initial region of phase space (or more generally, a distribution over phase space), $\mathcal P$ will be the region `before' synchronisation, $\mathcal P'$ the region `after', and $\mathcal S$ the interfacial region corresponding to the process of synchronisation. These phase regions must satisfy the relations $\mathcal I \subset \mathcal P$, $\mathcal P \cap \mathcal P' = \mathcal S$, and $\mathcal P \cup \mathcal P'=\mathcal R$. It should be noted that the synchronisation region $\mathcal S$ is arbitrary, although in practice there is usually a natural choice.

The question naturally arises of what happens at the boundaries of this phase space. Quite simply, transitions which would take the mona outside of the valid phase region are blocked; that is, there is no reaction with reactants $\ket{0}$ and $\oplus$ if the other mona are not all receptive to synchronisation. A more abstract but more general answer is that the boundaries of $\mathcal R$ are \emph{reflective}: the probability flux/current normal to them vanishes.

We now introduce a canonical slicing of $\mathcal P$ into hypersurfaces $\{\mathcal P_s:s\in\mathbb N\}$. The hypersurface $\mathcal P_s$ is given by the locus of nodes whose shortest path (i.e.\ the minimum number of transitions) to $\mathcal S$ is $s$. By definition, $\mathcal P_0\equiv\mathcal S$, and it can be shown that to get from a node in $\mathcal P_{s+k}$ to $\mathcal P_s$, one must pass through the hypersurfaces $\{\mathcal P_{s+k-1},\mathcal P_{s+k-2},\ldots,\mathcal P_{s+1}\}$ in turn. We will usually take $\mathcal I\subseteq\mathcal P_s$ for some `distance' $s$.

These definitions allow us to categorise the forms of synchronisation events. If $|\mathcal P_{s+1}| \ge |\mathcal P_s|$ for all $s$ within an appropriate vicinity of $\mathcal S$ where $|\mathcal P_s|$ is the number of nodes in said hypersurface, then we call this a constrictive synchronisation. If instead $|\mathcal P_{s+1}| \le |\mathcal P_s|$ we call it dilatory. These categories intersect in the special case of $|\mathcal P_s|$ constant; this can be likened to a tube, and has zero penalty. For uniformly positive bias, a dilatory synchronisation problem does not present an impediment to synchronisation and in fact may even aid the process, as the reflective boundary serves to `push' probability density towards $\mathcal S$. In practice, dilatory synchronisations only occur in contrived scenarios, if at all, and synchronisations can be assumed constrictive. If a synchronisation geometry does not fit into either of these categories, then it may need to be partitioned into constrictive and dilatory regions, each of which can be analysed separately. Alternatively, if the geometry is `on average' constrictive then one should be able to cautiously apply the relevant results.

\para{Locally Unconstrained Synchronisation}

As alluded to in the introduction and \Cref{fig:synch-ex-2}, a synchronisation problem can be modified to `fill in' the missing regions of phase space. This is achieved by introducing a `dummy' counter for each monon. Whenever a monon reaches the synchronisation subspace (i.e.\ it is receptive to synchronisation) it begins counting the time it has waited. This is not quite accurate, as the counter is just another part of the monon's computational state and so it can go down as well as up; nevertheless, on average it counts up. The consequence is that every monon can continue to evolve unimpeded. Once the last monon becomes receptive to synchronisation, the mona finally synchronise. At this point, the latent mona begin to decrement their counters. When a monon reaches zero, it may resume its original computation. Though a subtle point, this is indeed a reversible computation despite appearing to erase information, and is made clearer in \Cref{fig:synch-ex-2}. 

In return for the elimination of phase space boundaries, the phase space becomes effectively larger, and control over the state of the mona is more limited. In particular, with a weak bias the mona are free to wander into these dummy regions. We will analyse this case in \Cref{sec:recessive} in order to compare its performance to the constrictive case. As with the constrictive/dilatory cases, we can classify these problems into recessive and processive cases. Recessive cases correspond to $\mathcal P'$ being convex, and processive to it being concave. Again, it is not obvious how a processive case could be constructed, if indeed it can be. Moreover, it is unclear how to analyse a case that does not fall into these categories.

\para{Problem Statement}

With the synchronisation phase space well defined, we are now in a position to formulate a question. Our goal is to determine the mean time to get from $\mathcal I$ to $\mathcal S$, and to compare this to the unsynchronised system $\mathcal R'$. For homogeneous constant bias $b$, it is straightforward to show that unconstrained evolution from $\mathcal P_s$ to $\mathcal P_0=\mathcal S$ takes mean time $\tau=s/b\lambda$. That is, the computational 'speed' is $1/b\lambda$. More generally, we can define this quantity by $\tau = \inf\{t:\forall u>0.\evqty{x(t+u)}\in\mathcal P'\}$. That is, it is the least time for which the mean position in phase space never leaves the post-synchronisation domain. For many systems, this can be identified with the first time at which the mean position in phase space enters the post-synchronisation domain, $\tau = \inf\{t:\evqty{x(t)}\in\mathcal P'\}$.

To compute this latter quantity $\tau$ in general cases, we consider the evolution of the phase space probability distribution, $W(t)=W(\vec x;t)$. This quantity $\tau$ is well known in the literature, and is termed the Mean First-Passage Time, or MFPT. We will generally follow the approach of \textcite{risken}, in particular Chapter~8. The first observation we can make is that we need not track particles that cross the boundary between the two domains; that is, we can place an absorbing boundary at $\mathcal S$. The total extant probability, $G(t)=\int_{\mathcal P\setminus\mathcal S}\dd{\vec x}W(t)$, can then be interpreted as a survival probability, $\pr(\tau>t)$. Introducing a probability density $\varrho(t)$ for the MFPT, such that $\tau=\int_0^\infty\dd{t}t\varrho(t)$, we can see that $G(t)=\pr(\tau>t)=\int_t^\infty\dd{t'}\varrho(t')$ and hence $\tau=-\int_0^\infty\dd{t}t\dv{G}{t}=\int_0^\infty\dd{t}G(t)$.

We proceed by reducing the dimension of the problem. Surprisingly, it transpires that we can in fact eliminate the time variable. We recall the standard approach from \textcite{risken}. The $n^{\text{th}}$ moment of $\tau$ is found to be
\begin{align*}
  \evqty{\tau^n} &= -\int_0^\infty \dd{t}t^n\pdv{t}\int_{\mathcal P}\dd{\vec x}W(\vec x;t)=\int_{\mathcal P}\dd{\vec x}\underbrace{\qty\bigg(-\int_0^\infty\dd{t}t^n\pdv{t}W(\vec x;t))}_{w_n(\vec x)}.
\end{align*}
Let $\mathcal L=\mathcal L(\vec x)$ be an operator describing the time evolution of the phase distribution, i.e.\ $\dot W=\mathcal LW$. Here we require that $\mathcal L$ not depend on time, and so we will not consider time-variable biases from here on. Integrating by parts and then applying $\mathcal L$, we can relate the $w_n$ to each other:
\begin{align*}
  w_{n+1}(\vec x) &= n\int_0^\infty\dd{t}t^nW(\vec x;t), \\
  \mathcal Lw_{n+1}(\vec x) &= n\int_0^\infty\dd{t}t^n\mathcal LW(\vec x;t) \\
    &= n\int_0^\infty\dd{t}t^n\pdv{t}W(\vec x;t) \\
    &= -nw_n(\vec x).
\end{align*}
Evaluating the base case, $w_0(\vec x)=W(\vec x;0)$, we thus have an inductive definition of the $w_n$, from which we can obtain each moment of $\tau$. As these turn out to be non-trivial to evaluate, we will focus on the first moment---the mean---in this chapter, i.e.\ $\tau=\int_{\mathcal P}\dd{\vec x}w_1$ with $\mathcal Lw_1=-W(0)$ where $W(0)$ is the initial distribution on $\mathcal I$. That is, we merely need to compute the steady state distribution of the system $\mathcal P$ with absorbing boundary at $\mathcal S$, subject to forcing $W(0)$. Note that, whilst $W(0)$ is a normalised probability distribution with total density 1, $w$ is not. In fact, its total density is precisely the MFPT. 

A useful observation is that phase density is conserved by the operator $\mathcal L$. This means we can write $\mathcal L=-\vec\Delta\cdot\vec S$ for discrete phase space, where $\vec\Delta$ is the discrete vector derivative (or finite difference) and $\vec S$ is the phase current. For continuous phase space, we have instead $\mathcal L=-\vec\nabla\cdot\vec S$. Current will be very important to our analysis, and is informed by the sources and sinks due to the $W(0)$ forcing and absorbing boundary respectively.

\subsection{Techniques for Discrete Phase Space}

\para{Recurrence Relations}

There are a number of ways to represent the operator $\mathcal L$ for a discrete system. The canonical approach to describe a Markov chain is by introducing a distribution vector with indices corresponding to loci in phase space. For example, the phase space corresponding to \Cref{fig:synch-juxt} would be represented by a vector in the infinite dimensional vector space indexed by the set $\{(-i,-j):i,j\in\mathbb N\}\cup\{(i,j):i,j\in\mathbb N\}$. The operator is then an infinite matrix for this space. Whilst a convenient approach due to the wealth of techniques built around this representation, in addition to the ability to solve for the steady state as an eigenvector problem, it is too unwieldy to make much progress here.

A better representation is given by considering the local structure of $\mathcal L$, as alluded to earlier when we introduced current. For the one-dimensional phase space of a bi-infinite monon with transition rates $\rate(i\mapsto i+1)$ and $\rate(i+1\mapsto i)$, we can write the system in the following two convenient ways:
\begin{align*}
  \dot x_i &= -[\rate(i\mapsto i+1)+\rate(i\mapsto i-1)]x_i \\
           &\phantom{{}={}}+ \rate(i-1\mapsto i)\,x_{i-1} + \rate(i+1\mapsto i)\,x_{i+1}, \\
  S(i\mapsto i+1) &= \rate(i\mapsto i+1)\,x_i - \rate(i+1\mapsto i)\,x_{i+1}.
\end{align*}
Note how $\dot x_i \equiv - S(i\mapsto i+1) + S(i-1\mapsto i)$, as expected. Also keep in mind that the expression for $\dot x_i$ needs to be modified to incorporate sources and sinks. These representations are just recurrence relations, and methods to solve these include elimination, induction, and generating functions. To demonstrate, we solve for the MFPT from $\ket{-s}$ to $\ket{0}$.

To set up the problem, we place an absorbing boundary at $\ket{0}$ such that $x_0=0$ and we force the system by placing a negative unit-strength source at $\ket{-s}$ such that $S(-s\mapsto-s+1) - S(-s-1\mapsto-s)=1$. The boundary condition at $-\infty$ is $x_{-\infty}=0$, as is the current. Therefore, we find that $S(-s-k-1\mapsto-s-k)=0$ for $k>0$ and $S(-s+k\mapsto-s+k+1)=1$ for $s>k\ge0$. That is, the current transports density from the source to the sink. Writing out the current equations,
\begin{align*}
  1 = S(-1\mapsto0) &= \rate(-1\mapsto0)\,x_{-1}, \\
  1 = S(-2\mapsto-1) &= \rate(-2\mapsto-1)\,x_{-2} - \rate(-1\mapsto-2)\,x_{-1}, \\
  1 = S(-3\mapsto-2) &= \rate(-3\mapsto-2)\,x_{-3} - \rate(-2\mapsto-3)\,x_{-2}, \\
  &~\,\vdots \\
  0 = S(-s-k-1\mapsto-s-k) &= \rate(-s-k-1\mapsto-s-k)\,x_{-s-k-1} \\&\phantom{{}={}}- \rate(-s-k\mapsto-s-k-1)\,x_{-s-k}, \\
  &~\,\vdots,
\end{align*}
we can find the densities in the vicinity of the boundary;
\begin{align*}
  x_0 &= 0, \\
  x_{-1} &= \frac1{\rate(-1\mapsto0)}, \\
  x_{-2} &= \frac1{\rate(-2\mapsto-1)}\qty\Big[1+\frac{\rate(-1\mapsto-2)}{\rate(-1\mapsto0)}], \\
  x_{-3} &= \frac1{\rate(-3\mapsto-2)}\qty\Big[1+\frac{\rate(-2\mapsto-3)}{\rate(-2\mapsto-1)}\qty\Big[1+\frac{\rate(-1\mapsto-2)}{\rate(-1\mapsto0)}]], \\
  &~\,\vdots.
\end{align*}
More generally, we can write
\begin{align*}
  x_{-n} &= \begin{cases}
    \frac{1}{\rate(-1\mapsto0)}\qty\big(\sum_{k=1}^n\prod_{r=2}^{k}\frac{\rate(-r+1\mapsto-r)}{\rate(-r\mapsto-r+1)}), & n \le s, \\
  \frac{1}{\rate(-1\mapsto0)}\qty\big(\sum_{k=1}^s\prod_{r=2}^{k}\frac{\rate(-r+1\mapsto-r)}{\rate(-r\mapsto-r+1)})\qty\big(\prod_{r=s+1}^{n}\frac{\rate(-r+1\mapsto-r)}{\rate(-r\mapsto-r+1)}), & n > s.
  \end{cases}
\end{align*}
That is, we have exact expressions for the entire steady state distribution. Adding these up will yield the MFPT, $\tau=\sum_{i=0}^\infty x_i$. This expression is not particularly convenient; we can instead rewrite the recurrence relation as
\begin{align*}
  x_{-n} &= \frac{S(-n\mapsto-n+1)}{\rate(-n\mapsto-n+1)} + \frac{\rate(-n+1\mapsto-n)}{\rate(-n\mapsto-n+1)}x_{-n+1},
\end{align*}
from which we find
\begin{align*}
  \tau &= \sum_{n=1}^{s}\frac{1}{\rate(-n\mapsto-n+1)} + \tau\sum_{n=1}^\infty\frac{\rate(-n+1\mapsto-n)}{\rate(-n\mapsto-n+1)}\frac{x_{-n+1}}{\tau} \\
  &= \qty\Big(1-\evqty\Big{\frac{\rate(-n+1\mapsto-n)}{\rate(-n\mapsto-n+1)}})^{-1}\sum_{n=1}^{s}\frac{1}{\rate(-n\mapsto-n+1)},
\end{align*}
where the expectation value is taken over the steady state probability distribution. Even if the steady state distribution is unknown, we can nevertheless use this to obtain bounds on $\tau$. In the specific case of homogeneous bias, $\rate(n\mapsto n+1)=p\lambda$ and $\rate(n+1\mapsto n)=q\lambda$, the bounds coincide and thus we can obtain an exact expression without knowledge of the exact distribution (though the distribution also reduces to a reasonably simple form), finding
\begin{align*}
  \tau &= \qty\Big(1-\frac qp)^{-1}\frac{s}{p\lambda} = \frac{s}{b\lambda},
\end{align*}
as expected.

\para{Generating Functions}

When the values of $\lambda$ take on a particularly nice form, such as constant uniform, generating functions provide a tempting approach. A generating function is used to represent a series as coefficients of a power series. For example, the series $\{0,1,2,3,\ldots\}$ has generating function $t(1-t)^{-2}$ because its series expansion is $t+2t^2+3t^3+4t^4+\cdots$. Using multiple variables, complicated structures such as walks can be encoded in a generating function\footnote{For introductions to generating functions and enumerative combinatorics, see \textcite{gen-fun,enum-comb}. For a review of techniques used to analyse walks in two or more dimensions with generating functions, see \textcite{bm-walks}.}. We will seek a generating function, $W$, whose terms correspond to walks starting at each position $\ket{s}$, and end when they reach $\ket{0}$ for the first time. For convenience, we are walking down towards 0 from positive states, rather than up from negative states. It turns out that the easiest way to construct these walks is in reverse: we start from $\ket{0}$ and add either an `up' step or a `down' step, providing we never go down to $\ket{-1}$. Stated another way, a walk is either the empty walk (with generating function $1$), another walk $w$ followed by an `up' step ($wu$), or another (positive) walk $w_+$ followed by a `down' step ($w_+d$). This way, the walks described are only those right up to the point they cross the origin.

To get meaningful statistical information, we label each step with its probability. In the reversed walk, an `up' step has probability $p$ and a `down' step $q$. This leads to the implicit functional equation
\begin{align*}
  W &= 1 + pxtW + q\bar xt(W-W|_{x=0})
\end{align*}
where $\bar x\equiv 1/x$, $x$ is a variable that records the final position of the walk, and $t$ is a variable that records the number of steps in the walk. To prevent a downward step at $\ket{0}$, we ignore these walks when constructing the next step. Expanding $W$, we can see that it has the terms we expect:
\begin{align*}
  W &= t^0(x^0) + t^1(px^1) + t^2(ppx^2 + pqx^0) + t^3(pppx^3 + ppqx^1 + pqpx^1) + \cdots
\end{align*}
To proceed, we seek a closed form expression. Rewriting, we have $(x-t(px^2+q))W = x - qtW|_{x=0}$. It would appear that we are now stuck, as there is an unknown $W|_{x=0}$ which depends on the complete generating function $W$. Fortunately there is a powerful technique that can be used here, the Kernel method~\cite{gf-kernel}. We require that $W$ have no negative powers of $x$, and therefore the \emph{kernel}, $x-t(px^2+q)$, must be a factor of the right hand side of the equation. Among other things, this means that when the kernel vanishes, so must the right hand side. The kernel has two roots, $X_\pm$, although only one of these turns out to be appropriate (in the sense of not containing negative powers of $t$):
\begin{align*}
  X_\pm &= \frac{1\pm\sqrt{1-4pqt^2}}{2pt}, & 0 &= X_- - qtW|_{x=0}.
\end{align*}
Substituting into our functional equation, we obtain
\begin{align*}
  W &= \frac{1-\bar xX_-}{1-(px+q\bar x)t}.
\end{align*}
In fact, this is not quite the generating function we want. We want the final step to be into an absorbing boundary, so we define the `true' generating function to be $V = pxtW$. A little thought shows that the coefficient of $x^s$ in $V$ is the probability generating function of the first passage time; that is, the coefficient of $t^nx^s$ is the probability that a walk starting at $\ket{s}$ takes $n$ steps to reach $\ket{0}$. In order to find the MFPT, we want to calculate $\sum_{n=0}^\infty n[t^nx^s]V$ where $[t^nx^s]V$ is the coefficient of $t^nx^s$ in V. In fact, we can obtain a generating function encoding this information by taking the $t$-derivative of $V$, exploiting the fact that $\partial_t t^n=nt^{n-1}$. If we then set $t=1$, we sum up all the $t$ coefficients and therefore obtain the generating function for the MFPTs, parametrised by s. That is, $\dot V|_{t=1}$. Evaluating this gives
\begin{align*}
  \dot V|_{t=1} &= \frac1b\frac{x}{(1-x)^2} = \sum_{s=0}^\infty \frac sbx^s,
\end{align*}
i.e.\ the MFPT for $\ket{s}$ is $s/b$ as expected for discrete time. Converting to continuous time recovers $s/b\lambda$.

Unfortunately, these techniques do not extend easily to multiple dimensions. Nevertheless, there is a growing selection of techniques applicable to two-dimensional quadrant walks. See \textcite{bm-walks} for a review. We had partial success in applying these techniques, and some of our results are summarised in \Cref{app:gf}, but ultimately we were unable to use these to obtain expressions for the MFPT. Instead, we focussed on obtaining lower and upper bounds by making use of the recurrence relations and properties such as detailed balance; these results are presented in \Cref{sec:constrict-discrete}.

\subsection{Techniques for Continuous Phase Space}\label{sec:meth-cont}

\para{Fokker-Planck Equation}

The phase spaces we are interested in are discrete, but for completeness we shall also consider continuous phase spaces. Many of the definitions for discrete phase space generalise naturally such as the phase regions $\mathcal R'$, $\mathcal R$, $\mathcal P$, $\mathcal P'$, $\mathcal I$, and $\mathcal S$. The distribution $W$ is to be interpreted as a density, but otherwise $W$, $G$, $w_n$, $\varrho$, etc.\ generalise as expected. The boundaries of $\mathcal R$ are enforced by a no-flux boundary condition. The hypersurface slicing is less obvious, and will require us to first define the dynamics of the system.

Being stochastic, the dynamics are conventionally described by a Langevin stochastic differential equation. If $\vec\xi(t)$ is the phase coordinate, then we can write $\dv{t}\vec\xi(t)=\vec\eta(\vec\xi;t)$ where $\vec\eta(\vec\xi;t)$ is the noise term, a family of vectors of random variables indexed by the time coordinate. It is standard to use the decomposition $\vec\eta=\vec h+\vec\theta$ where $\vec h=\evqty{\vec\eta}$ is a deterministic `forcing' term, and $\vec\theta=\vec\eta-\evqty{\vec\eta}$ is a noise term with zero mean. Moreover, we expect the noise term to be uncorrelated in time, i.e.\ our dynamics should have the Markov property. Following the approach of \textcite{risken} in Chapter~3.4, we find we can write $\vec\theta(\vec\xi;t)=\mathbf g(\vec\xi;t)\vec\Gamma(t)$ where $\vec\Gamma$ is a vector of independent unit-variance zero-mean Gaussian variables, and $\mathbf g$ is a noise strength matrix. Using the Kramers-Moyal expansion, we can derive an equation for the evolution of the probability density $W$,
\begin{align*}
  \dot W &= -\vec\nabla\cdot\underbrace{(\vec\mu-\vec\nabla\cdot\mathbf D)W}_{\vec S},
\end{align*}
where $\vec\mu$ is the drift coefficient and $\mathbf D$ the diffusion matrix, both deriving from $\vec h$ and $\mathbf g$. This equation is known as the \emph{Fokker-Planck} equation (FPE).

We can now define hypersurface slicing for continuous phase space, remaining robust against coordinate transforms. Geodesic paths can be found by integrating the drift coefficient, $\dot{\vec x}=\vec\mu(\vec x)$, from some initial coordinate $\vec x_0$. A hypersurface $\mathcal P_s$ is defined by the locus of points which take the same time $t$ to reach $\mathcal S$. The label $s$ is arbitrary as long as it increases monotonically with $t$; typically $s$ will be chosen to correspond to geodesic path lengths in the `physically relevant' coordinate system, e.g.\ $s=b\lambda t$ in the case of uniform constant bias.

To relate the drift coefficient and diffusion matrix to the computational bias, we need to compare the statistical properties of the two systems. For constant bias, this will be straightforward and will result in a constant drift coefficient and diffusion matrix. Where the bias varies, more care is required to ensure the desired properties are replicated in the continuous case. 

We match the statistical properties for a single degree of freedom. In the discrete case, the relevant distribution is given by the difference between two Poisson distributions. One Poisson distribution represents the number of $\oplus$ tokens received, and the other the number of $\ominus$ tokens. This is known as the \emph{Skellam} distribution, and our parameters are $p\lambda t$ and $q\lambda t$, yielding mean $b\lambda t$ and variance $\lambda t$. In the continuous case, the FPE is given by $\dot W=-\mu W'+DW''$ and its exact solution can be obtained by a routine application of Fourier transforms:
\begin{align*}
  \dot{\widetilde W} &= -ik\mu\widetilde W - k^2D\widetilde W \\
  \partial_t \log\widetilde W &= -(k^2D+ik\mu) \\
  W &= \frac1{\sqrt{2\pi}}W_0 \ast \frac1{\sqrt{2\pi}}\int_{\mathbb R}\dd{k} \exp\!\qty\Big(ikx-(k^2D+ik\mu)t) \\
    &= \frac1{\sqrt{2\pi}}W_0 \ast \frac1{\sqrt{2\pi}}\int_{\mathbb R}\dd{k} \exp\!\qty\Big(-Dt\qty\Big(k+\frac{i(x-\mu t)}{2Dt})^2-\frac{(x-\mu t)^2}{4Dt}) \\
    &= W_0 \ast \frac1{\sqrt{4\pi Dt}}\exp\!\qty\Big(-\frac{(x-\mu t)^2}{4Dt}).
\end{align*}
That is, the initial distribution $W_0$ is convolved with a Gaussian of mean $\mu t$ and variance $2Dt$. Comparing with the Skellam distribution, we identify $\mu=b\lambda$ and $D=\tfrac12\lambda$.

\section{Constrictive Case}\label{sec:constrict}

Using these techniques, we will be able to obtain a variety of exact, approximate, and numeric results for the constrictive case for general geometries. In practice, we are more interested in truncated simplicial geometries (\Cref{dfn:simplex}) as their construction is more `natural'. The construction is natural because it comprises $d$ mona which evolve independently at all times except for the moment of synchronisation.

\begin{dfn}[Truncated Simplex]\label{dfn:simplex}
  We refer to our primary geometry of interest as a `truncated simplex'. A simplex is a generalisation of the concept of a triangle to arbitrary dimensions. That is, in dimension one it is a line segment, and in dimension three it is a tetrahedron. The pre-synchronisation subspace $\mathcal P$ of \Cref{fig:synch-simple-2} can be considered to be a right simplex in two dimensions, i.e.\ a right triangle, of infinite extent. It is defined by the set $\{(x,y):x\le0\land y\le0\}$ (restricted to $\mathbb Z^2$ for the discrete case or $\mathbb R^2$ for the continuous). The pre-synchronisation subspace $\mathcal P$ of \Cref{fig:synch-wide-1}, meanwhile, is a \emph{truncated} simplex in two dimensions. Its set definition is $\{(x,y):x\le0\land y\le0\land |x+y|\ge w\}$ where $w$ is the side-length of the constriction surface $\mathcal S$. In $d$ dimensions, the set is given by
  \[ \qty\bigg{\vec x:\bigwedge_{i=1}^dx_i\le0\land\qty\bigg|\sum_{i=1}^d x_i|\ge w}. \]
\end{dfn}

It will be useful to determine the sizes of hypersurfaces $|\mathcal P_n|$ for truncated simplices. These hypersurfaces are in fact regular $d-1$ simplices, which in $\mathbb R^d$ have hypervolume
\[ \frac{\sqrt{d}}{(d-1)!}\qty\bigg(\frac{n+w}{\sqrt{2}})^{d-1}. \]
In $\mathbb Z^d$ we can use generating functions. The size of $|\mathcal P_n|$ is equivalent to the number of ways $n+w-1$ can be made from the sum of $d$ natural numbers. This is conveniently given by the coefficient of $x^{n+w-1}$ in $(1-x)^{-d}$, 
\[ \binomqty\bigg{d+n+w-2}{n+w-1}. \]

\subsection{Initial Estimates}\label{sec:constrict-initial}

\para{Information Erasure Model}
We shall begin first with some quick initial estimates in order to better understand the qualitative behaviour of our results. Consider the lateral evolution of some initial distribution $\mathcal I$; whilst this initial distribution can in principle be chosen arbitrarily, it will tend to diffuse to fill the entire hypersurface. Moreover, the timescale of this lateral diffusion will be much shorter than that of the approach of the mean position towards the interfacial region $\mathcal S$ in the case of vanishing bias. To clarify, the bias between two adjacent nodes $x$ and $y$ is given by
\[ b(x\mapsto y) = \frac{\rate(x\mapsto y) - \rate(y\mapsto x)}{\rate(x\mapsto y) + \rate(y\mapsto x)}. \]
Therefore, without significant loss of generality, we assume a substantially delocalised initial distribution. This distribution need not be uniform, but it will have an entropy that scales roughly proportional to $\log|\mathcal P_s|$, where $|\mathcal P_s|$ is the size of the initial hypersurface. In order for an interaction to occur, the distribution must pass through the interfacial hypersurface $\mathcal S\equiv \mathcal P_0$ whereupon it will have a distributional entropy proportional to $\log|\mathcal P_0|$. As the geometry is constrictive, $|\mathcal P_0|\le|\mathcal P_s|$ and therefore this process requires an erasure of information. For constant bias $b$, information can be erased (see \Cref{chap:revi}) at a maximum rate $2b^2\lambda$ where $\lambda$ is the gross transition rate. We can therefore bound the synchronisation time from below by
\[ \frac1{2db^2\lambda} \log\frac{|\mathcal P_s|}{|\mathcal P_0|} \]
which, for a truncated simplex, is
\begin{align*}
    \frac1{2db^2\lambda} \sum_{k=0}^{d-2}\log\!\qty\Big(1+\frac{s}{w+k}) 
  = \underbrace{\frac1{4b^2\lambda} \log\!\qty\Big(1+\frac{s}{w})}_{d=2}.
\end{align*}
Note the factor $d$ in the denominator which arises because each monon is able to independently erase information at the given rate.

This neglects the time for the individual mona to reach the synchronisation surface, which is $s/b\lambda$. For $s\lesssim1/b$, the erasure time will dominate and so the MFPT can be approximated as
\[ \frac{s}{b\lambda} + \frac1{4b^2\lambda} \log\!\qty\Big(1+\frac{s}{w}), \]
which shows that there is a `penalty' term for synchronisation.
When $s\gtrsim1/b$ it is less clear as it is possible that the mona could perform some or all of the erasure along their journey leading to no penalty at all. From a practical point of view, any penalty in this case is negligible as the journey time will be larger, but it is instructive to understand what occurs in this case in order to understand how the system erases information. If the penalty were eliminated, then it would suggest that the lateral distribution of the mona should collimate into a narrow `beam' in anticipation of the synchronisation surface. This is clearly nonsensical, as the tendency is for the mona to diffuse laterally. A simple way to model the interaction is to divide the $d$ mona into 1 `latent' monon and $d-1$ `precocious' mona. The latent monon is identified as the last monon to arrive, and as such the precocious mona can be assumed to have reached a steady state distribution. For uniform constant bias, each monon's steady state distribution is geometric with $\pr(\mathcal P_n)=\frac bp(\frac qp)^n$ and entropy $\log\frac pb+1+\bigOO{\frac bq}$. For the latent monon to pass through the constriction point, each of the precocious mona's distributions must be erased. As these erasures could in principle occur simultaneously, this gives a lower bound for such a penalty as
\[ \frac{d-1}{2db^2\lambda}\log\frac{p}{bw}. \]
Putting the two results together, we see that the lower bound on the penalty should increase logarithmically with distance before reaching a plateau as $s\gtrsim1/b$.

\para{Quasi-Steady State Approximation}
Whilst the above is reasonable from an information theory perspective, it does not provide a mechanism for information erasure. A simple mechanistic model for synchronisation can be given by the Quasi-Steady State Approximation. This approximation consists of two phases; first, the initial distribution $\mathcal I$ evolves and approaches the interfacial surface $\mathcal S$, which is taken to be reflective. We then assume the distribution to have roughly attained its steady state form, at which point the second phase begins. In  the second phase, there is a steady leak of density through $\mathcal S$ per the transition rates of the system which is assumed to not significantly affect the form of the distribution in $\mathcal P$. 

For constant uniform bias, the steady state is given by $\pr(\mathcal P_n) = |\mathcal P_n|(\frac{q}{p})^n / \sum_{k=0}^\infty |\mathcal P_k|(\frac{q}{p})^k$. Evaluating the normalisation constant is non-trivial for arbitrary $|\mathcal P_k|$, but we find approximately that $\pr(\mathcal P_0)\sim b^dw^{d-1}$. The leak rate is $p\lambda\pr(\mathcal P_0)$, and therefore the Quasi-Steady State approximation gives an exponential decay process with half life $\sim b^{-d}w^{1-d}\lambda^{-1}$. This is the approximate penalty term, giving an overall MFPT
\begin{align*}
  \tau &\sim \frac{s}{b\lambda} + \frac{1}{b^dw^{d-1}\lambda}
\end{align*}
and showing that, whilst in principle an order $1/b^2$ penalty term is possible via information erasure, in practice a $d$-dimensional synchronisation will incur a substantially greater penalty for $d>2$ (as $b\ll1$). Fortunately this can be mitigated, either by replacing the synchronisation with a series of $2$-dimensional interactions or by making the information erasure explicit. Explicating the information erasure in such a case can be achieved using the `naive' approach of making the interfacial hypersurface sticky at the cost of erasing at least \SI{1}{\bit} of information (representing the state \code{mononIsMoving}).

\subsection{Discrete Phase Spaces}\label{sec:constrict-discrete}

Whilst the preceding approximations are reasonable, they make significant simplifications that neglect much of the specific dynamics of the system. In order to be more confident in their conclusions, we shall proceed by making use of the analytical techniques discussed in \Cref{sec:methodology}. We begin with the discrete case, being that this corresponds to the underlying geometry of the phase spaces of interest.

\subsubsection{Exact results}\label{sec:constrict-exact}

Recall that the MFPT can be obtained as the total phase density at steady state of the modified system, wherein the initial distribution $\mathcal I$ is instead supplied as a constant forcing term. Equivalently, this modified system can be understood as the `teleportation' of density absorbed at $\mathcal S$ back into the system according to $\mathcal I$. This is because, at steady state, the rate of density loss at the absorbing boundary must exactly balance the rate of the forcing term. Observe now that, if $s=1$ (i.e.\ $\mathcal I\subseteq \mathcal P_1$), then the `teleportation' of density from $\mathcal S$ to $\mathcal I$ resembles a reflective boundary. An unforced system with reflecting boundaries and whose underlying Markov chain is reversible has the convenient property that the current is everywhere vanishing. Such a reversible Markov system with vanishing current permits us to readily obtain the steady state using the principle of \emph{detailed balance}: the densities of each pair of states $x$ and $y$ are such that $[x]\rate(x\mapsto y)=[y]\rate(y\mapsto x)$.

Before continuing, we must first clarify the constraint on $\mathcal I$. We begin by introducing a labelling for the states in each hypersurface $\mathcal P_n$ as $\{(n,c) : c\in\mathcal P_n\}$. It is important to note that the steady state densities in the initial hypersurface, $[(s,c)]$, do not generally correspond in a simple way to the initial distribution $[(s,c)]_0\sim\mathcal I$. In the reflective system, the effective forcing term is given simply by the reflection rates,
\begin{align*}
  [(1,c)]_0 &= \sum_{c'\in\mathcal S} [(0,c')]\rate((0,c')\mapsto(1,c)) \\
            &\equiv [(1,c)] \sum_{c'\in\mathcal S}\rate((1,c)\mapsto(0,c')).
\end{align*}
That is, the only permissible initial distribution is prescribed by the steady state distribution of $\mathcal P_1$ and is in proportion to the forward transition rates. As shall be seen shortly, for the example geometry shown in \Cref{fig:synch-wide-1} this yields the initial distribution
\begin{align*}
  [(1,c)]_0 &= \frac1w\begin{cases}
    \frac12, & \text{$c$ on boundary}, \\
    1,       & \text{otherwise},
  \end{cases}
\end{align*}
which, for large $w$, is effectively uniform. Finally, to ensure that the density sum yields the MFPT, we pick the normalisation such that $[\mathcal I]=\sum_{cc'}[(1,c)]\rate((1,c)\mapsto(0,c'))=\sum_{cc'}[(0,c')]\rate((0,c')\mapsto(1,c))=1$. Bear also in mind that the reflecting boundary is an artefact of this representation, with its density in the original system being zero; therefore, in computing the MFPT, we must remember to neglect its contribution.

\para{Truncated Simplices} In addition to being a relatively simple geometry, the truncated simplex systems have constant uniform bias. This property gives rise to a particularly simple steady state. Suppose the forward and backward rates are given by $p\lambda/d$ and $q\lambda/d$ respectively where $p-q=b$ and $d$ is the dimension of the simplex. More precisely, given that each non-boundary state $(n,c)$ is adjacent to $d$ states $(n-1,c'_i)$ in the forward direction and $d$ states $(n+1,c'_i)$ in the backward direction for $i=1\ldots d$, the rates may be written as $\rate((n,c)\mapsto(n-1,c'_i))=p\lambda/d$ and $\rate((n,c)\mapsto(n+1,c'_i))=q\lambda/d$. Using detailed balance, we can deduce therefore that $[(n,c)]=\frac qp[(n-1,c'_i)]$ for \emph{any} $i=1\ldots d$ and hence $[(n-1,c'_i)]=\frac pq[(n,c)]$. As each pair of adjacent hypersurfaces in a truncated simplex forms a connected subgraph, one can show inductively that $[(n,c)]=[(n,c')]$ for all $c,c'\in\mathcal P_n$ (including boundary states). Combining the results thus far, we obtain $[(n+m,c)] = (\frac qp)^m [(n,c')]$. Finally, we find the normalisation as $\forall c.\,1 = [(0,c)] |\mathcal S| q\lambda$ which leads to the complete description of the steady state distribution and hence the MFPT,
\begin{align*}
  [(n,c)] &= \frac{(\frac qp)^n}{q\lambda |\mathcal S|}, &
  \tau &= \sum_{n=1}^\infty \frac{(\frac qp)^n}{q\lambda}\frac{|\mathcal P_n|}{|\mathcal S|},
\end{align*}
where we have neglected the contribution of the reflective boundary density, $[\mathcal S]$ as this is just an artefact of our representation.

Using the expression for the hypersurface sizes, we can compute the MFPT for an arbitrary dimension $d$. The factor $|\mathcal P_n|/|\mathcal S|$ reduces to $\prod_{k=0}^{d-2}(1+\frac n{w+k})$ and therefore the MFPT is given by $\frac1{q\lambda}\sum_{n=1}^\infty(\frac qp)^n\prod_{k=0}^{d-2}(1+\frac n{w+k})$. Enumerating for dimensions $d=1,2,3$, we find the $d$-dimensional MFPTs $\tau_d$,
\begin{equation}\begin{alignedat}{4}
  \tau_1 &= \frac1{q\lambda}\sum_{n=1}^\infty\qty\Big(\frac qp)^n&&=\frac1{b\lambda}, \\
  \tau_2 &= \frac1{q\lambda}\sum_{n=1}^\infty\qty\Big(\frac qp)^n\qty\Big(1+\frac nw)&&=\frac1{b\lambda} + \frac1w\frac p{b^2\lambda}, \\
  \tau_3 &= \frac1{q\lambda}\sum_{n=1}^\infty\qty\Big(\frac qp)^n\qty\Big(1+\frac nw)\qty\Big(1+\frac n{w+1})&&=\frac1{b\lambda} + \frac{2w+1}{w(w+1)}\frac{p}{b^2\lambda} + \frac1{w(w+1)}\frac{p}{b^3\lambda}.
\end{alignedat}\label{eqn:mfpts-simplex}\end{equation}
Obtaining an expression for general $d$ is non-trivial, but we can see that the leading order term will be $\sim b^{-d}w^{1-d}\lambda^{-1}$ whenever $w\lesssim 1/b$. When $d\lesssim w$, we can also obtain an approximate expression,
\begin{align*}
\tau_d &\approx \frac1{q\lambda}\sum_{n=1}^\infty\qty\Big(\frac qp)^n\qty\Big(1+\frac nw)^{d-1} \\
&= \frac1{q\lambda}\sum_{k=0}^{d-1}\binomqty\Big{d-1}{k}\qty\Big[\qty\Big(\frac{t}{w}\pdv{t})^k\frac1{1-t}]_{t=\frac qp} \\
&\approx \frac1{q\lambda}\sum_{k=0}^{d-1}\binomqty\Big{d-1}{k}\qty\Big[\qty\Big(-\frac{1}{w}\pdv{s})^k\frac1{s}]_{s=2b} \\
&\approx \sum_{k=0}^{d-1}\qty\Big(\frac d{2w})^k\frac{1}{b^{k+1}\lambda}.
\end{align*}

\para{General Geometries} Other geometries can be evaluated in the same way as for truncated simplices. We will now generalise the preceding argument as far as possible. First, we use detailed balance to compute the average state density for each hypersurface,
\begin{align*}
  \frac{[\mathcal P_n]}{|\mathcal P_n|} &= \frac1{|\mathcal P_n|}\sum_{c\in\mathcal P_n}[(n-1,c')]\frac{\rate((n-1,c')\mapsto(n,c))}{\rate((n,c)\mapsto(n-1,c'))} \\
  &= \frac{[\mathcal P_{n-1}]}{|\mathcal P_{n-1}|} \underbrace{ \frac1{|\mathcal P_n|}\sum_{c\in\mathcal P_n}\frac{[(n-1,c')]}{[\mathcal P_{n-1}]/|\mathcal P_{n-1}|}\frac{\rate((n-1,c')\mapsto(n,c))}{\rate((n,c)\mapsto(n-1,c'))} }_{\hat t_n},
\end{align*}
where for each $c\in\mathcal P_n$, $c'\in\mathcal P_{n-1}$ is any node adjacent to $c$ (i.e.\ such that the transition rates between them are non-zero). In so doing, we have generalised the factor $t=\frac qp$, used for constant uniform bias, to $\hat t_n$. These densities can then be summed to obtain an expression for the MFPT, 
\begin{align*}
  \tau &= [\mathcal S] \sum_{n=1}^\infty \frac{|\mathcal P_n|}{|\mathcal S|}\prod_{k=1}^n \hat t_k \\
  &= \evqty\Big{\sum_{c'}\rate((0,c)\mapsto(1,c'))}_{c\in\mathcal S}^{-1} \sum_{n=1}^\infty \frac{|\mathcal P_n|}{|\mathcal S|}\prod_{k=1}^n \hat t_k,
\end{align*}
where we have used the normalisation condition to solve for $[\mathcal S]$. 

We can now try to extract general scaling behaviour for a broad class of systems. As we are considering constrictive geometries, $|\mathcal P_n|$ is a monotonically increasing function as $n\to\infty$. Therefore, we must have $\prod^n \hat t_k\to0$ as $n\to\infty$ or else the MFPT will diverge (the physical reason being that $\hat t_k>1$ implies negative bias). Moreover, the $\prod^n\hat t_k$ must, in the limit, decrease faster than $|\mathcal P_n|$ grows. Provided that (most) of the $\hat t_k$ are less than 1 this is usually true as it ensures exponential decay whereas we expect $|\mathcal P_n|$ to grow only polynomially. It is reasonable to assume that regions of $\hat t_k>1$ (negative bias) are brief and rare. Given this assumption, the $\prod^n\hat t_k$ can be characterised by a decay length $\ell_n = \min \{m : -\sum_{k=n}^{n+m}\log\hat t_k\ge 1\}$, i.e.\ the distance from $n$ over which it decays by a factor $e$. This decay length is an estimator of the bias, $\ell_n \sim 1/2b$.

Consider the function $n^{p-1} t^n$ for $p,n>0$ and $0<t<1$; the function increases monotonically until its one turning point, before decaying inexorably towards zero. The turning point can be shown to occur at $n=(p-1)\ell$ where $\ell=-1/\log t$; consequently, the integral under the curve is dominated by the range $[0,(p-1+\bigOO{1})\ell)$. The integral can therefore be approximated as $\sim((p-1+\bigOO{1})\ell)^{p}/p\sim p!\ell^{p}$ by assumption that within this range $t^n\sim\bigOO{1}$. Returning to our general MFPT expression, and using the fact that $|\mathcal P_n|$ is monotonic, we can come to the same approximate conclusion. Interpolating quantities to be continuous in $n$, the turning point occurs when $\partial_n\log|\mathcal P_n|=-\log\hat t_n$. Expanding $|\mathcal P_n|$ as a polynomial in $n$ over this range, we can approximate it by its leading order term as $|\mathcal P_n|=an^{p-1}+\bigOO{n^{p-2}}$, and therefore obtain the same range $\sim[0,p\ell)$. Finally then, the MFPT can be approximated to leading order (up to a constant multiplicative factor of order unity) as
\[\tau\sim a\evqty\Big{\sum_{c'}\rate((0,c)\mapsto(1,c'))}_{c\in\mathcal S}^{-1}\frac{p!}{2^p}\frac{1}{b^{p}}\]
and we identify $p$ as the effective dimensionality $d$ of the system.

\subsubsection{Bounded Results}

To extend to the case of $s>1$ we proceed inductively. The MFPT from a node $x$ is given by the recurrence relation 
$\tau(x\mapsto\mathcal S) = (\sum_{x'}\rate(x\mapsto x'))^{-1} + \sum_{x'}\rate(x\mapsto x')\tau(x'\mapsto\mathcal S)$
where the $x'$ are adjacent to $x$. To rewrite this inductively, first pick a contiguous hypersurface $\mathcal X$ lying between $\mathcal I$ and $\mathcal S$ such that all paths from $\mathcal I$ to $\mathcal S$ pass through $\mathcal X$ at least once. It can then be seen that $\tau(i\mapsto\mathcal S)$ where $i\in\mathcal I$ can be written as $\tau(i\mapsto \mathcal X) + \sum_{x\in\mathcal X}p_{x|i}\tau(x\mapsto\mathcal S)$ where $\sum_{x\in\mathcal X}p_{x|i}=1$. This itself can be proven inductively by assumption that any node in the region before $\mathcal X$ can be written thus and using the trivial base case of $\tau(x\mapsto\mathcal X)=0$ where $x\in\mathcal X$. We then introduce the distributional quantity $\tau(\mathcal I\mapsto\mathcal S)=\sum_{i\in\mathcal I}\pr(i|\mathcal I)\tau(i\mapsto\mathcal S)$ to write this as $\tau(\mathcal I\mapsto\mathcal S)=\tau(\mathcal I\mapsto\mathcal X_{\mathcal I})+\tau(\mathcal X_{\mathcal I}\mapsto\mathcal S)$ where $\mathcal X_{\mathcal I}$ is the \emph{unique} distribution of states in $\mathcal X$ as they are first reached from $\mathcal I$.

Using this inductive form, we can rewrite the MFPT as the sum of a series of $s=1$ MFPTs,
$\tau(\mathcal I_s \mapsto \mathcal S) = \sum_{k=0}^{s-1}\tau(\mathcal I_{k+1} \mapsto \mathcal I_{k})$
where $\mathcal I_k$ is the unique distribution over $\mathcal P_k$ as first reached from $\mathcal I_{k+1}$. This immediately yields a first approximation to the MFPT for arbitrary $s$ by assuming that $\mathcal I_k$ is similar to the reflective distribution used in \Cref{sec:constrict-exact}. Calculating for the general geometry restricted to uniform constant bias, we find
\begin{equation}\begin{aligned}
  \tau &\approx \sum_{k=1}^s \frac1{q\lambda}\sum_{n=1}^\infty \frac{|\mathcal P_{k+n}|}{|\mathcal P_{k-1}|}\qty\Big(\frac{q}{p})^n \\
  &= \frac1{p\lambda}\sum_{n=0}^\infty \qty\Big(\frac qp)^n \sum_{k=0}^{s-1}\frac{\mathcal P_{k+n+1}}{\mathcal P_k} \\
  &\equiv \frac1{p\lambda}\opn{\mathcal Z}\!\qty\Big{\sum_{k=0}^{s-1}\frac{\mathcal P_{k+n+1}}{\mathcal P_k}}\!\qty\Big[\frac pq]
\end{aligned}\label{eqn:mfpt-lb}\end{equation}
where $\mathcal Z$ is the unilateral $\mathcal Z$-transform.

Unfortunately we have been unable to generally determine $\mathcal I_{n-1}$ from $\mathcal I_n$, but we can extract sufficient information to be able to bound the MFPT from above and below. In the reflective approach for $s=1$, recall that the steady state distribution is given by $[(n+m,c)]=(\frac qp)^m[(n,c')]$ for all $c,c'$. Notice also that the transitions to $\mathcal S$ are precisely those that are to be `absorbed', and those from $\mathcal S$ are precisely those to be `teleported'. The consequence is that the distribution on $\mathcal S$ is the unique distribution $\mathcal I_0$ as earlier defined, induced by $\mathcal I=\mathcal I_1$. Consider then the case $s=2$ where the initial distribution $\mathcal I$ is this `reflective' distribution. The MFPT can be decomposed into $\tau(\mathcal I_2\mapsto\mathcal I_1)+\tau(\mathcal I_1\mapsto\mathcal S)$; we know the first term as it is just that given by our $s=1$ calculation, but the second term requires us to know the $s=1$ MFPT for a uniform initial distribution.

The `reflective' distribution is defined by $[(n,c)]\propto\text{number of forward transitions}$; for example, the distribution for the Truncated 2-Simplex is given ratiometrically as $1:2:2:\cdots:2:1$. The uniform $1:1:1:\cdots:1:1$ distribution can thus be considered as a normalised superposition of the reflective distribution and a pure-boundary distribution $1:0:0:\cdots:0:1$. This is useful because the MFPT is linear, and therefore we can apply the principle of linear superposition. Now, for uniform constant bias in a constrictive geometry, the boundary nodes within a given hypersurface $\mathcal P_n$ must have maximal MFPT because the expected transition direction from the boundary nodes is backwards whereas for all other nodes it is forwards. As a result, the MFPT for a uniform distribution is greater than for a reflective distribution. This applies for each subsequently induced hypersurface distribution, as the effective negativity of the boundary transitions causes it to be `sticky' and attract density to itself.

\para{Lower Bound}
Therefore, for an initially reflective distribution, the estimate of the MFPT given by \Cref{eqn:mfpt-lb} is in fact a lower bound. Moreover, it is a lower bound for any distribution lying between the reflective distribution and the limiting distribution induced by the aforementioned boundary attraction phenomenon. 

\para{Upper Bound}
A trivial upper bound is given by the sum of the MFPTs for pure-boundary distributions, but this is needlessly loose. To obtain a tighter upper bound, we appeal to the distributional superposition,
\begin{align*}
  \tau(\mathcal I_n^{\text{refl.}}) &= \tau(\mathcal I_n^{\text{refl.}}\mapsto\mathcal I_{n-1}^{\text{unif.}}) + \alpha\tau(\mathcal I_{n-1}^{\text{refl.}}) + (1-\alpha)\tau(\mathcal I_{m-1}^{\text{bound.}}) \\
  &= \frac1{p\lambda}\sum_{r=0}^\infty\frac{|\mathcal P_{n+r}|}{|\mathcal P_{n-1}|}\qty\Big(\frac qp)^r + \alpha\tau(\mathcal I_{n-1}^{\text{refl.}}) + (1-\alpha)\tau(\mathcal I_{m-1}^{\text{bound.}}).
\end{align*}
The value of $\alpha$ can be found by taking the ratiometric vector for the reflective distribution, say $1:2:\cdots:2:1$, and then finding the boundary distribution which makes this up to uniform, i.e.\ $1:0:\cdots:0:1$. The uniform distribution would then be $2:2:\cdots:2:2$, and so the proportion of density in the reflective component is $(1+2+\cdots+2+1)/(2+2+\cdots+2+2)$, or the normalised average $\evqty{f_x}/\max_x f_x$ where $f_x$ is the number of forward transitions for node $x$.

The boundary MFPT is harder to find, but we can bound it from above. First consider the exact $s=1$ MFPT when $w=1$; in this case, $\mathcal P_1$ is formed entirely from boundary nodes and therefore its MFPT is precisely the boundary MFPT. In fact, this boundary MFPT is an upper bound on all the other boundary MFPTs as it corresponds to the special case of no interior nodes; when interior nodes do exist, they will have lower MFPTs than the boundary and density in the boundary will diffuse into these interior nodes, thus reducing the boundary MFPT. These definitions of $\alpha$ and boundary MFPTs will become clearer with a concrete example.

\para{Truncated Simplices}
Specialising to truncated simplices, we can quantify these lower and upper bounds. The $s=1$ MFPTs have already been obtained earlier for $d=1,2,3$ in \Cref{eqn:mfpts-simplex}, and so the MFPT lower bounds can be obtained thus,
\begin{align*}
  \tau_1 &= \frac{s}{b\lambda}, \\
  \tau_2 &\ge \frac{s}{b\lambda} + \frac{p}{b^2\lambda}\Delta\psi, \\
  \tau_3 &\ge \frac{s}{b\lambda} + \frac{p}{b^2\lambda}\qty\Big(2\Delta\psi + \frac1{s+w} - \frac1w) + \frac{p}{b^3\lambda}\qty\Big(\frac1w-\frac1{s+w}),
\end{align*}
where $\Delta\psi=\sum_{k=w}^{w+s-1}\frac1k=\psi(s+w)-\psi(w)\approx\log(1+\frac sw)$ with $\psi$ the digamma function. The factor $\Delta\psi$ closely resembles the information loss in the synchronisation interaction for 2-simplices, and the coefficient of $p/b^2\lambda$ in $\tau_3$ resembles that for 3-simplices, as expected. Notice that these lower bounds do not plateau as $s\to1/b$ as conjectured in \Cref{sec:constrict-initial}; whilst it is in principle possible from an information theoretic perspective to reach this plateau penalty, the passive dynamics of these interactions preclude this possibility. Nevertheless, at the point where the plateau becomes relevant, the penalty becomes insignificant in comparison to the $s/b\lambda$ for $\tau_2$ and therefore this is of little concern. For $\tau_3$ and above, penalties of order $\bigOO{b^{-3}}$ and above are present and remain significant beyond $s/b\lambda$.

\begin{figure}
  \centering
  \includegraphics[width=.8\linewidth]{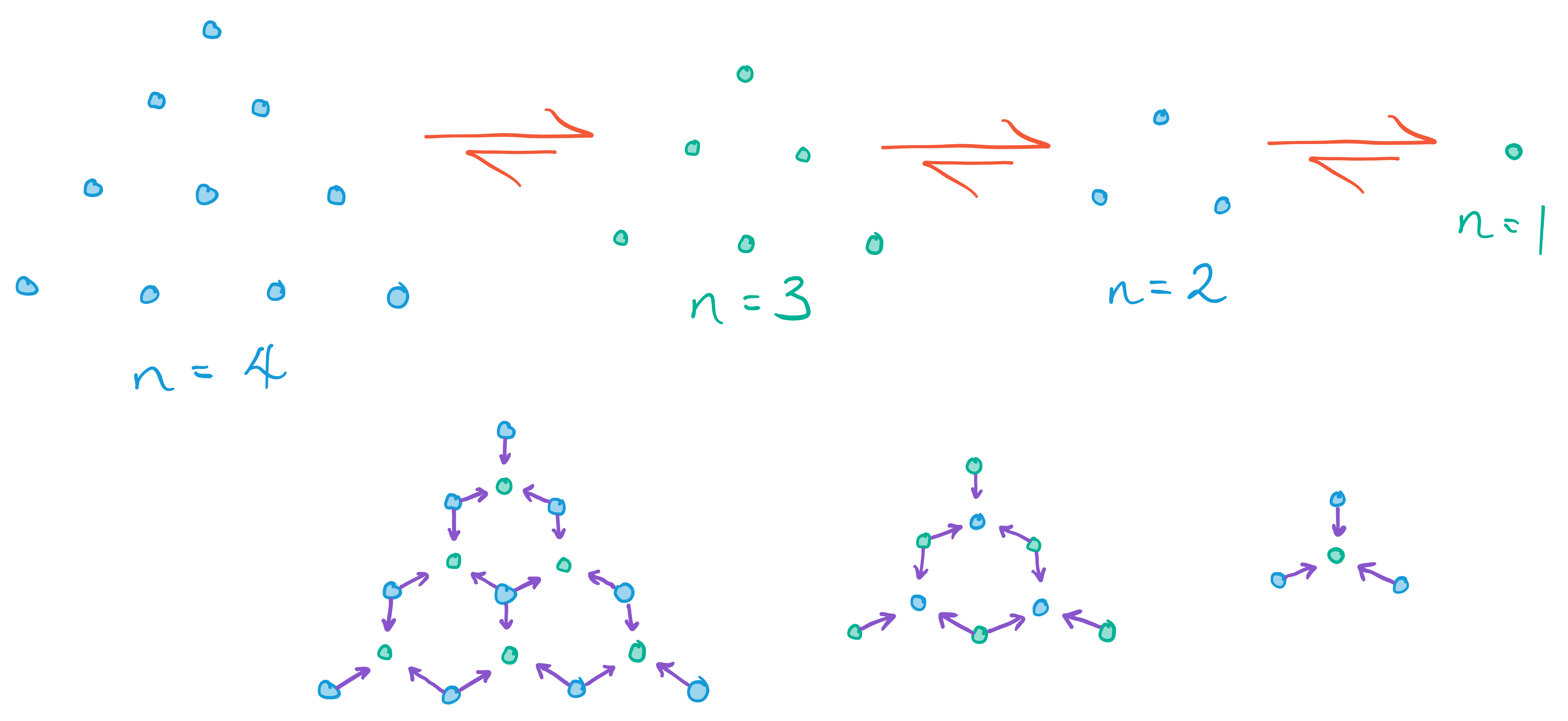}
  \caption[An illustration of a series of hypersurfaces in the $d=3$ simplex to demonstrate the different classes of boundary nodes.]{An illustration of a series of hypersurfaces in the $d=3$ simplex to demonstrate the different classes of boundary nodes. The top row shows the hypersurfaces on their own, whilst the bottom row shows adjacent hypersurfaces superimposed and the forward transitions between them.}
  \label{fig:simplex3-boundary}
\end{figure}

For the upper bound, we first obtain the upper bound on the boundary MFPT as $\frac{s}{b\lambda}+\frac{ps}{b^2\lambda}$ for $d=2$ and $\frac{s}{b\lambda}+\frac32\frac{ps}{b^2\lambda}+\frac12\frac{ps}{b^3\lambda}$ for $d=3$. One needs to be particularly careful for $d=3$ and above as there are different `classes' of boundary nodes: for $d=3$ (as shown in \Cref{fig:simplex3-boundary}), there are `vertex' boundary nodes which have only one forward transition, and `edge' boundary nodes which have two forward transitions, whilst interior nodes have three. More generally, in dimension $d$ there are $d-1$ classes of boundary nodes with $1,2,\ldots,d-1$ forward transitions respectively (and interior nodes with $d$ forward transitions). The vertex nodes will have the highest MFPTs and correspond precisely to the $w=s=1$ MFPT, and are therefore what we shall use for MFPT bounds in higher dimensions. The values of $\alpha$ can be found by counting transitions, and are $\alpha_1=1$, $\alpha_2=1-(n+w-1)^{-2}$ and $\alpha_3=1-2(n+w)^{-1}$.

Using these values of the boundary MFPTs and superposition fractions, we can write a recurrence relation for each dimension in $\tau_d(n)$,
\begin{alignat*}{6}
  \tau_1(w+k+1) &= \frac1{b\lambda} +{} &1&\cdot\tau_1(w+k) +{} &0&\cdot\qty\Big(\frac{k}{b\lambda}), \\
  \tau_2(w+k+1) &\le \frac1{b\lambda} +{} &\frac{w+k-1}{w+k}&\cdot\tau_2(w+k)+{}&\frac1{w+k}&\cdot\qty\Big(\frac{k}{b\lambda}+\frac{pk}{b^2\lambda}), \\
  \tau_3(w+k+1) &\le \frac1{b\lambda} +{} &\frac{w+k-1}{w+k+1}&\cdot\tau_3(w+k)+{}&\frac2{w+k}&\cdot\qty\Big(\frac{k}{b\lambda}+\frac32\frac{pk}{b^2\lambda}+\frac12\frac{pk}{b^3\lambda}
  ),
\end{alignat*}
with initial conditions
\begin{align*}
  \tau_1(w+1) &= \frac1{b\lambda}, \\
  \tau_2(w+1) &= \frac1{b\lambda} + \frac1w\frac{p}{b^2\lambda}, \\
  \tau_3(w+1) &= \frac1{b\lambda} + \frac{2w+1}{w(w+1)} \frac{p}{b^2\lambda} + \frac1{w(w+1)}\frac{p}{b^3\lambda}.
\end{align*}
These can finally be solved to yield
\begin{align*}
  \tau_1 &= \frac{s}{b\lambda}, \\
  \tau_2 &\le \frac{s}{b\lambda}+\frac{\frac12s(s+1)}{s+w-1}\frac{p}{b^2\lambda}, \\
  \tau_3 &\le \frac{s}{b\lambda}+\frac{\frac12s[2s^2+(3s+1)(w+1)-4s]}{(w+s)(w+s-1)}\frac{p}{b^2\lambda} + \frac{\frac16s(2s^2+3(s-1)(w-1)+4)}{(w+s)(w+s-1)}\frac{p}{b^3\lambda},
\end{align*}
and, in the limit of large $s$, these simplify to
\begin{align*}
  \tau_1 &= \frac{s}{b\lambda}, &
  \tau_2 &\le \frac{s}{b\lambda} + \frac{s+1}{2b^2\lambda}, &
  \tau_3 &\le \frac{s}{b\lambda}+\frac{2s+3w}{2b^2\lambda}+\frac{2s+3w}{6b^3\lambda}.
\end{align*}
In summary, the MFPT penalties for $d=1,2,3$ truncated simplices are bounded above and below by
\begin{align*}
  \Delta\tau_1 &= 0, \\
  \frac{p}{b^2\lambda}\Delta I_2 \le \Delta\tau_2 &\le \frac{p}{b^2\lambda}\frac{\frac12s(s+1)}{s+w-1}, \\
  \frac{p}{b^2\lambda}\Delta I_3 + \frac{p}{b^3\lambda}\frac s{w(s+w)} \le \Delta\tau_3 &\le \frac{p}{b^2\lambda}\frac{\frac12s[2s^2+(3s+1)(w+1)-4s]}{(w+s)(w+s-1)} \\&+ \frac{p}{b^3\lambda}\frac{\frac16s(2s^2+3(s-1)(w-1)+4)}{(w+s)(w+s-1)},
\end{align*}
where $\Delta I_2=\psi(s+w)-\psi(w)$ and $\Delta I_3=2\Delta I_2+\frac1{s+w}-\frac1w$. Expressions for higher dimensions can be obtained using the same approach, but the key point to note is that the terms in each order of $b^{-1}$ are present and significant in both the lower and upper bounds and therefore a synchronisation interaction in dimension $d$ will be subject to a penalty of order $\bigOO{b^{-d}}$, with the coefficient of the $\bigOO{b^{-2}}$ penalty term bounded below by the information loss in the interaction.

\subsubsection{Numerical Simulation}
\label{sec:simulation}

\begin{figure}
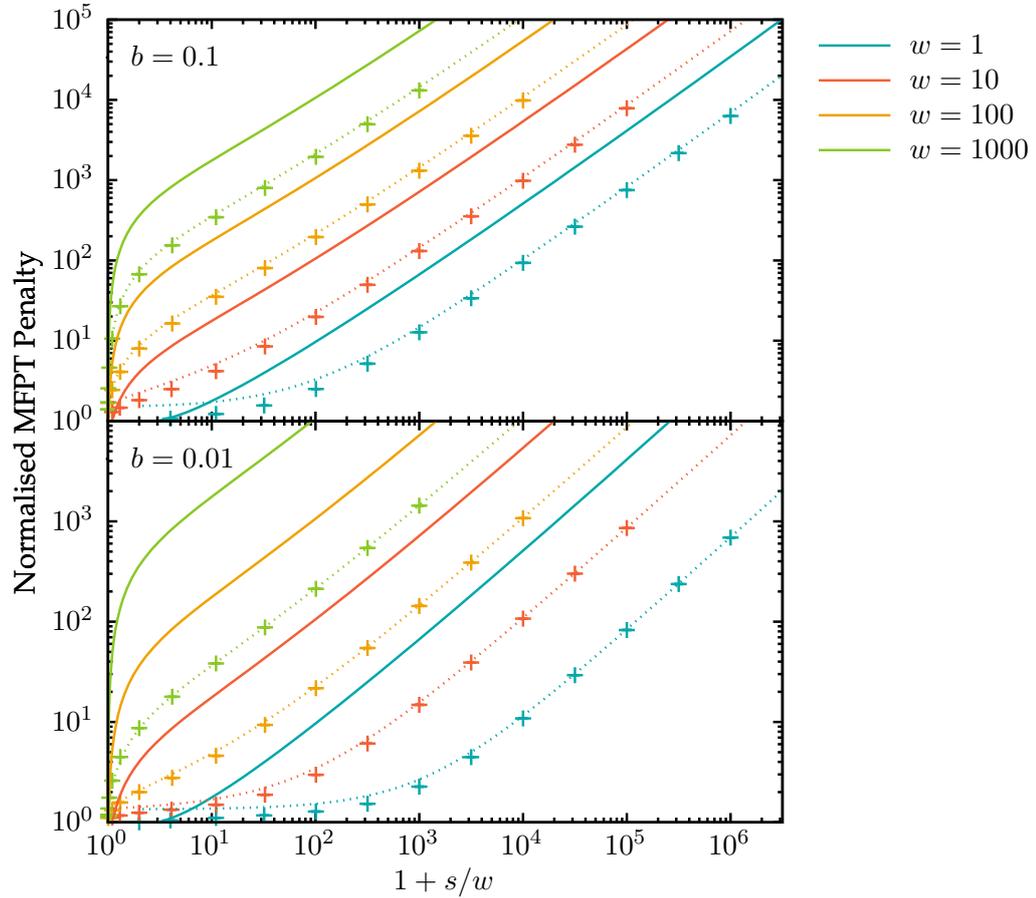

  \captionsetup{singlelinecheck=off,type=figure}
  \centering
  \begin{minipage}[t]{10.5cm}
    \strut\vspace*{-\baselineskip}\newline
    \input{fig-penalties}
  \end{minipage}\begin{minipage}[t]{2.75cm}
    \strut\vspace*{-\baselineskip}\newline
    \\[-3pt]\input{fig-penalties-key}
  \end{minipage}
    
  \caption[Simulation results for the $d=2$ truncated simplex case, with varying initial hyperplane distances $s$ and constriction widths $w$, and subject to a uniform bias $b\in\{0.1,0.01\}$.]{Simulation results for the $d=2$ truncated simplex case, as depicted in \Cref{fig:synch-wide-1}, with varying initial hyperplane distances $s$ and constriction widths $w$, and subject to a uniform bias $b\in\{0.1,0.01\}$. The \emph{Normalised MFPT Penalty} corresponds to $(\tau-\frac s{b\lambda})/(\frac p{b^2\lambda}\Delta I_2)$, i.e.\ it is the MFPT penalty divided by the lower bound as computed in \Cref{sec:constrict-discrete}. As a result, the horizontal axis corresponds to the lower bound. The plotted points show simulation data with (negligible) error bars, the solid lines show the (normalised) upper bound, and the dashed lines show the empirical approximation given by \Cref{eqn:mfpt-empirical}.}
  \label{fig:penalties}
\end{figure}

In the course of investigating this synchronisation problem, a suite of tools\footnote{\url{https://github.com/hannah-earley/revcomp-synch-rw-mfpt}}\ for computing MFPTs was developed, as well as an assortment of other tools\footnote{\url{https://github.com/hannah-earley/revcomp-synch-rw-distribution}}\footnote{\url{https://github.com/hannah-earley/revcomp-synch-rw-distribution-2}}. Computing MFPTs in regimes of low bias and unbounded state space is challenging because the probability density tends to diffuse very far from $\mathcal S$. Attempting to compute it in a deterministic way via the explicit steady state distribution has a high spatial complexity as one must make sure not to prematurely truncate the state space, or else risk underestimating the MFPT. A rough estimate gives a spatial complexity of $\bigOO{1/b^2 + (s+w)/b}$, and the time complexity will be the product of this spatial complexity and the number of simulation steps to convergence, which will be at least $\bigOO{1/b}$ to allow particles from the initial distribution to explore the entire state space.

The alternative is to exploit the Markov property of the system in order to use a stochastic Markov-Chain Monte-Carlo simulation approach, or MCMC. In an MCMC approach we essentially apply the stochastic dynamics of the system to an ensemble of instances of the system. In this case, each instance is a single phase particle, whose coordinate represents the joint state of the underlying mona. Therefore, for an ensemble of size $m$, the spatial complexity will be $\bigOO{m}$. Again, the goal is to obtain the steady state distribution of the forced absorbing-boundary system; to ensure conservation of density, we use the teleporting form in which there is a one-to-one correspondence between each particle incident on the absorbing boundary and a particle injected into the initial distribution.

The naive approach is to record the number of time-steps between injection and absorption for each particle. Unfortunately this approach is not robust enough here as the high-diffusion negligible-drift regime leads to substantial variance in the FPT distribution such that a significant proportion of particles in the ensemble will get `lost' in the phase space for an unboundedly long time. Therefore the mean of the particles' injection-absorption times will typically be divergent, or the program will run indefinitely as it waits for each particle to be absorbed. A possible solution is to exclude particles which have not returned after some predefined number of time-steps, but this will lead to an underestimate of the MFPT.

The correct solution, it transpires, is to instead record the absorption current: that is, within some time-step window $\Delta t$, count the number of particles $n$ incident on the absorbing boundary in order to obtain an unbiased estimate of the current $S=n/m\Delta t$ where $m$ is the ensemble size. This process can be repeated as many times as one likes to obtain a series of measurements of the current, from which basic statistics can be used to find an improved estimate of the mean current and its standard error.

There remains, however, a caveat common to all MCMC simulations. Assuming our particles are distributed according to the desired steady state, then their distribution will remain so as they evolve under the stochastic dynamics of the system. The problem is in ensuring that the initial distribution of the system is the steady state, despite not knowing what that is. To address this, one typically initialises the system to an approximation of the steady state and then lets the distribution `burn in' by running the simulation for some number of initial time-steps, after which it is hoped that the approximation will have converged to the steady state. The poorer the quality of the initial approximation to the steady state, the longer this burn-in phase will take. Moreover, the length of the burn-in time can be hard to ascertain and so often one must resort to a heuristic judgement.

Oft forgotten can be the quality of the random number generator used to simulate the stochastic dynamics. In an earlier iteration of this tool, we neglected to consider this and as a result used the standard \code{rand} function: a \emph{linear congruential generator} with period $2^{32}$. As the simulation times required to obtain good quality MFPT estimates for many of our parameters exceeded this period, the data obtained was invalid. It can be hard to detect this, and we only realised when a debug trace of the intermediate outputs showed a noticeable and unexpected periodicity. As a result, our MFPT suite now uses the PCG family of PRNGs~\cite{pcg} which have an internal state space size and period of $2^{64}$ whilst also being very fast.

Whilst the dynamics of this system are very simple, some several quadrillion time-steps were required to obtain the data shown in \Cref{fig:penalties} and so a sophisticated toolchain was built around the MCMC routine to improve robustness and automatability. The core program, \tool{./walk}, is written in \cpp\ and makes use of \tool{OpenMP} for parallelisation. Simulating an ensemble of random walkers is an \emph{embarrassingly parallel} problem, meaning that performance can generally scale linearly with the number of available CPUs due to the minimal amount of inter-thread synchronisation necessary. Unfortunately the \tool{GCC} and \tool{LLVM} compilers\footnote{In our experience, \tool{LLVM}'s optimisation was far better than that of \tool{GCC} in this case.} were not able to fully optimise the inner loop responsible for executing the stochastic dynamics, and so some hand tuning was necessary to improve the assembly output and bring down the iteration time to $\lesssim\SI{5}{\nano\second}$ on the machines used. To ensure resumable operation and cooperation with cluster scheduling systems such as \tool{slurm}, the program supports checkpointing (both intermittent and in response to trappable signals). The program is also reasonably modular, making it straightforward to adapt to new random walk systems.

To coordinate simulating the large number of parameters across different topologies of networked systems and to minimise need for manual intervention, a batched job runner was developed in the form of a \python\ script. The job runner has three modes of operation; it can take a description of a set of jobs and generate specific instructions for each job, it can connect to a distributed job queue to request and run these jobs, and it can query the status of all the jobs and job runners (optionally sending updates by email at regular intervals). A key feature is that it can analyse the output of \tool{./walk} to determine whether the error has converged to a sufficiently small value, whereupon it will move onto the next job. Finally, a number of tools were created to analyse the resulting data, including the ability to infer and inspect the approximate steady state distribution over the phase space.

The results of these simulation tools are shown in \Cref{fig:penalties}. As well as serving as a sanity check on the derived MFPT bounds, we were also able to obtain an estimated empirical equation for the true MFPT as
\begin{align}
  \tau &= \frac{s}{b\lambda}\qty\Big(1+p\frac{3+s}{w+s}) + \frac43\frac{p}{b^2\lambda}\Delta I_2,\label{eqn:mfpt-empirical}
\end{align}
where the $(3+s)/(w+s)$ term is inspired by the exact MFPT for $b=1$, given by $\tau=s(1+\frac12\frac{3+s}{w+s})$. In particular these tools were very helpful in getting an intuitive feel for the parametric dependence of the MFPT when analytic approaches seemed unyielding.

\subsection{Continuous Phase Spaces}\label{sec:constrict-continuous}

The approach for continuous phase spaces will largely follow that for discrete, except that we shall make use of the Fokker-Planck equation as introduced in \Cref{sec:meth-cont}. In order to solve for the reflective steady state, we shall use the fact that its current vanishes. Recall that the continuous current is given by $\vec S=(\vec\mu-\vec\nabla\cdot\mathbf D)W$ where $W$ is the phase density, $\vec\mu$ the drift coefficient and $\mathbf D$ the diffusion matrix. When $\mathbf D$ is non-singular, this can be rewritten as $-\mathbf D(-\mathbf D^{-1}\vec\mu'+\vec\nabla)W=-\mathbf De^{-\varphi}\vec\nabla e^\varphi W$ where $\vec\mu'=\vec\mu-[\vec\nabla\cdot\mathbf D]$ and $\varphi$ is defined by $\vec\nabla\varphi=-\mathbf D^{-1}\vec\mu'$. If the current is identically zero, then it immediately follows that the density is given by $W=Ne^{-\varphi}$ where $N$ is a normalisation constant. This is fairly standard, and a similar result is obtained in Chapter~6 of \textcite{risken}.

To connect this reflective-boundary distribution with the absorbing-boundary distribution, we must find the absorption current. We assume that our reflective-boundary distribution is approximately equivalent to an absorbing-boundary distribution with forcing applied at the $\mathcal P(-\delta x)$ hypersurface, where $x$ is a coordinate orthogonal to the hypersurfaces; in the limit $\delta x\to0$ this becomes exact, in analogy with the discrete case. The true absorbing-boundary distribution has vanishing density at $\mathcal P(0)=\mathcal S$ by definition, and therefore we can recover our desired distribution by introducing an appropriate decay from $\mathcal P(-\delta x)$ to $\mathcal S$. In the limit $\delta x\to0$ any higher order polynomial terms in this decay will vanish, and therefore the decay can be assumed linear. The absorption current is then given by
\begin{align*}
  \vec S(0) &= \vec\mu W|_{x=0} - \sum_i \vec D_i\lim_{\delta x_i\to 0} \frac{W|_{x_i}-W|_{x_i-\delta x_i}}{\delta x_i} = \vec D_x \frac{W(-\delta x)}{\delta x} = \vec D_x \frac{Ne^{-\varphi}}{\delta x}.
\end{align*}
Using the fact that the current is along the $x$ direction, we can find the normalisation constant subject to the constraint that the total absorption current is 1 as $N = \delta x / \int_{\mathcal S}\dd[d-1]{\vec x} D_{xx}e^{-\varphi}$. Therefore we find that the MFPT from $s=-\delta x$ is given by $\delta\tau=\delta x\int_{\mathcal P}\dd[d]{\vec x} e^{-\varphi} / \int_{\mathcal S}\dd[d-1]{\vec x} D_{xx}e^{-\varphi}$, from which we find the lower bound for the MFPT from arbitrary $s$,
\begin{align*}
  \tau(s) &\ge \int_{-s}^0 \dd{x} \frac{\int_{-\infty}^x \dd{x'}\int_{\mathcal P(x')}\dd[d-1]{\vec x} e^{-\varphi}}{\int_{\mathcal P(x)}\dd[d-1]{\vec x} D_{xx}e^{-\varphi}}
    = \int_0^\infty\dd{t}\int_{-s}^0\dd{x}\frac{[\evQty{e^{-\varphi}}A]_{x-t}}{[\evQty{D_{xx}e^{-\varphi}}A]_x}
\end{align*}
where $\evQty{e^{-\varphi}}$ is averaged over the given hypersurface and $A=\int_{\mathcal P(x)}\dd[d-1]{\vec x}$ is the area of that hypersurface. We see that the MFPT for general phase space is easier to express in the continuum limit, and in the particular case of $\varphi$ linear in $x$, i.e.\ $\varphi=\phi-\mu x/D$, this reduces to a Laplace transform:
\begin{align*}
  \tau &\ge \mathcal L_t\qty\Big{\int_{-s}^0\dd{x}\frac{[\evQty{e^{-\phi}}A]_{x-t}}{[\evQty{D_{xx}e^{-\phi}}A]_x}}\!\qty\Big(\frac{\mu}{D}).
\end{align*}

\para{Truncated Simplices}

For the Truncated Simplex example in $d$ dimensions, the drift coefficient is constant uniform in the $x$ direction and the diffusion matrix is the constant uniform and isotropic $\mathbf D=D\mathbf 1$. Therefore, $\phi=0$ and we have from earlier that
\[A(-n) = \frac{\sqrt{d}}{(d-1)!}\qty\Big(\frac{n+w}{\sqrt{2}})^{d-1}.\]
Hence, the MFPT lower bound for arbitrary $d$ may be obtained thus:
\begin{align*}
  \tau &\ge \frac1D\mathcal L_t\qty\Big{\int_{-s}^0\dd{x}\qty\Big(\frac{w-x+t}{w-x})^{d-1}}\!\qty\Big(\frac{\mu}{D}) \\
  &\ge \frac1D\mathcal L_t\qty\Big{\int_0^s\dd{x}\sum_{k=0}^{d-1}\binomqty\Big{d+1}{k}\qty\Big(\frac{t}{w+x})^{k}}\!\qty\Big(\frac{\mu}{D}) \\
  &\ge \frac1D\sum_{k=0}^{d-1}\binomqty\Big{d+1}{k}\mathcal L_t\qty\Big{t^k}\!\qty\Big(\frac{\mu}{D})\int_0^s\dd{x}\qty\Big(\frac{1}{w+x})^{k} \\
  &\ge \frac{s}{\mu}+(d-1)\frac{D}{\mu^2}\log(1+\frac sw)+\sum_{k=2}^{d-1}\frac{D^k}{\mu^{k+1}}\frac1{k-1}\frac{(d-1)!}{(d-1-k)!}\qty\Big[\frac1{w^{k-1}}-\frac1{(w+s)^{k-1}}].
\end{align*}

\section{Recessive Case}\label{sec:recessive}

We turn now to the recessive case, for which we shall assume a continuous phase space for simplicity. We shall also restrict our attention to recessive systems whose diffusion matrices are uniform and isotropic, $\mathbf D=D\mathbf 1$, across $\mathcal R=\mathcal P\cup\mathcal P'$ and whose drift vectors are constant uniform in $\mathcal P$; the condition on $\vec\mu$ in $\mathcal P'$ will become apparent later, but for now we shall assume it takes the same value as in $\mathcal P$. Finally, we adopt an orthogonal coordinate system $(u;\vec v)$ where $u$ is chosen such that $\vec\mu=\mu\hat u$. It may appear that $u$ indexes the hypersurfaces, but this is in fact not the case; a coordinate transform that would render $u$ a hypersurface index would apply a shear and therefore make $\mathbf D$ anisotropic.

\begin{figure}
  \captionsetup{singlelinecheck=off,type=figure}
  \centering
  \begin{minipage}[t]{12cm}
    \input{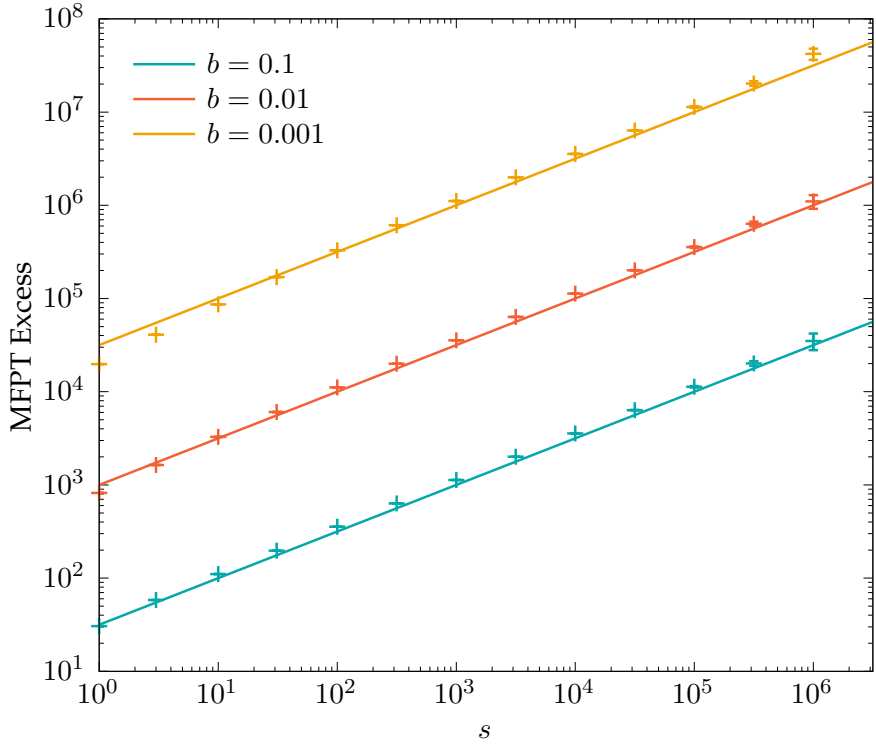}
  \end{minipage}
    
  \caption[Simulation results for the $d=2$ recessive case, starting on the line $x=y$ with varying initial distances $s$, and subject to a uniform bias $b\in\{0.1,0.01,0.001\}$.]{Simulation results for the $d=2$ recessive case, as depicted in \Cref{fig:synch-ex-2}, starting on the line $x=y$ and with varying initial distances $s$, and subject to a uniform bias $b\in\{0.1,0.01,0.001\}$. The \emph{MFPT Excess} corresponds to $\tau-\frac{2s}{b\lambda}$ (the coefficient $2$ was chosen empirically). The plotted points show simulation data with error bars, whilst the solid lines show $\frac{\sqrt{s}}{b^{3/2}\lambda}$ hence demonstrating an empirical MFPT of $\frac{2s}{b\lambda}+\frac{\sqrt{s}}{b^{3/2}\lambda}$.}
  \label{fig:gessel}
\end{figure}

\para{MFPT Approach}
Unfortunately the reflective steady state approach used for the constrictive case is inapplicable here as the unbounded transverse width of the space means the steady state is everywhere vanishing. Moreover, ignoring the zero normalisation, the uniform initial distribution is undesirable. For the canonical example shown in \Cref{fig:synch-ex-2}, initial distributions $\mathcal I$ of interest are typically centered on the line $x=y$ (corresponding to the white dashed line in the Figure). Nonetheless, we can apply\footnote{Our simulation suite is specialised to the case of walks in a single quadrant (with a possibly truncated boundary). Whilst it would not be too difficult to handle more general geometries, this particular three-quadrant example can be easily transformed into a single quadrant walk. First, we identify that the geometry is symmetric in the line $x=y$ and therefore the three quadrants are mapped to one and a half, with a reflective boundary placed at $x=y$. That is, $\mathcal P=\{(-x,y):x\in\mathbb N\land y\in\mathbb Z\land y>-x\}$. Next, we apply a shear $(-x,y)\mapsto(-x,y+x)$ to map the region to the top-left quadrant. The transformed transitions turn out to correspond to those of Gessel walks, a summary of which is presented by \textcite{bm-gessel} along with an investigation of their generating function.}\ the simulation suite introduced in \Cref{sec:simulation}. The results for an initial delta distribution on the line $x=y$ and a distance $s$ from $\mathcal S$ are shown in \Cref{fig:gessel} and reveal an empirical MFPT of $\tau\approx\frac{2s}{b\lambda}+\frac{\sqrt{s}}{b^{3/2}\lambda}$ and therefore a penalty of order $b^{-3/2}$.

This penalty would appear to be a significant improvement over that for the constrictive case of $b^{-2}$. Unfortunately this penalty is invalid because in fact the MFPT is not an appropriate quantity to use here. The reason is that the time quantity we are seeking is $\inf\{t:\forall u>0.\evqty{n(t+u)}>0\}$ where $n$ is a hypersurface index, and the use of the MFPT is predicated upon the idea that an initial distribution within $\mathcal P'$ has $\evqty{n(t)}>0$ for all positive $t$. It is straightforward to see that this property does not apply in the recessive case: Consider a simplified view of \Cref{fig:synch-rec-domain} as just the four quadrant domain $\mathcal R=\mathbb R^2$ such that $\mathcal P=\{(x,y):x<0\wedge y<0\}$ and $\mathcal P'=\{(x,y):x>0\land y>0\}$ and with $\vec\mu=\mu\frac1{\sqrt 2}(\hat x+\hat y)$; the hypersurfaces are given by $\mathcal P(\sqrt{2}n)=\{(n,n+y):y\in\mathbb R_+\}\cup\{(n+x,n):x\in\mathbb R_+\}$, i.e.\ $n(\vec x)=\sqrt 2\min(x,y)$. Now, the evolution of a phase particle may be decomposed into that along the line $x=y$ and perpendicular to this line. Calling these coordinates $u$ and $v$ respectively, i.e.\ $(u,v)=\frac1{\sqrt 2}(x+y,x-y)$ and $n=u-|v|$, we see that the distribution is given by the product of two Gaussians, with probability density
\begin{align*}
  \frac1{\sqrt{4\pi Dt}}\exp\!\qty\Big(-\frac{(u+s-\mu t)^2}{4Dt}) \cdot \frac1{\sqrt{4\pi Dt}}\exp\!\qty\Big(-\frac{v^2}{4Dt}).
\end{align*}
Now, the expected hypersurface is given by
\begin{align*}
  \evqty{n} &= \evqty{u}-\evqty{|v|} = -s + \mu t - \sqrt{\frac{4Dt}{\pi}}
\end{align*}
and therefore an initial distribution $\delta(n_0,0)$ will tend to go backwards in $n$-space for a time $t=\frac{D}{\pi\mu^2}$, retreating as far back as $\evqty{n}=n_0-\frac{D}{\pi\mu}$ before it begins to advance; at time $t=\frac{2D}{\pi\mu^2}$ it returns to its initial position $\evqty{n}=n_0$ whereupon it begins to make net progress. The consequence is that there is a time penalty for all $n_0<+\frac{D}{\pi\mu}$, and hence the appropriate boundary to consider for the MFPT calculation is given by $n=+\frac{D}{\pi\mu}$. 

Of course, in the case of $\vec\mu$ constant uniform across $\mathcal R$ we are able to use the far simpler approach of considering the exact probability distribution, as detailed above. We can also extend easily to arbitrary boundary shape. Unfortunately, as mentioned, $\vec\mu$ may well differ between $\mathcal P$ and $\mathcal P'$. For our canonical example in $(u,v)$ coordinates, 
\begin{align*}
  \vec\mu &= \begin{cases}
    \mu\hat u, & (u,v)\in\mathcal P, \\
    -(\opn{sign}v)\mu\hat v & (u,v)\in\mathcal P'.
  \end{cases}
\end{align*}
That is, in $\mathcal P'$ the drift direction is towards the $u$ axis. Recall also that we have ignored the six-quadrant structure, and so there is in a sense a lot more `space' in $\mathcal P'$. Nevertheless, we will be able to obtain lower and upper bounds.

\para{Probability Distribution Approach}

Ignoring the caveat about $\vec\mu$ for now, we proceed for a general boundary shape $u=f(v)$. For the canonical example, $f(v)=|v|$. Without loss of generality, we pick $f(0)=0$ and we note that a recessive geometry implies $f$ increases monotonically away from $v=0$. We then have $\evqty{n}=\evqty{u-f(v)}$; decomposing $f$ into its even and odd parts, $f_e(v)=\tfrac12[f(v)+f(-v)]$ and $f_o(v)=\tfrac12[f(v)-f(-v)]$ respectively, we see that $\evqty{n}\equiv\evqty{u-f_e(v)}$ whenever the initial distribution is even because the transverse dynamics respect evenness. As we are considering the initial distribution $\delta(-s,0)$, we discard the odd part of $f$ without loss of generality. 

If $f$ has a power series expansion about $v=0$, it will take the form $f(v)=\sum_{p=1}^\infty f_{2p}v^{2p}$. Series representations do not exist for many boundary shape functions of interest however, in particular $|v|$ has no such expansion. Consider instead a smooth function $g(v)$ such that $f(v)=g(|v|)$; if $g$ exists, it has series expansion $\sum_{p=1}^\infty f_p v^p$ and hence $f$ will have series expansion $\sum_{p=1}^\infty f_p |v|^p$. We can then compute $\evqty{f}$ as $\sum_{p=0}^\infty \alpha_pf_p(2Dt)^{p/2}$ where $\alpha_{2p}=(2p-1)!!$ and $\alpha_{2p+1}=\sqrt{2/\pi}(2p)!!$, and $n!!=n(n-2)(n-4)\cdots$ is the double factorial.

For the case $f=f_1|v|$, $\evqty{n}=-s+\mu t-2f_1\sqrt{Dt/\pi}$ and we will find a penalty of order $1/\mu^2$ as earlier. If instead $f$ is quadratic, then $\evqty{n}=-s+(\mu-2f_2D)t$ and there are two subcases: if $\mu>2f_2D$ then the rate of computation is reduced to $\mu-2f_2D$ within a large vicinity of the synchronisation interaction; if $\mu\le2f_2D$ then computation halts or goes backwards, and the synchronisation never\footnote{It is possible to escape if the phase geometry eventually changes to a permissible form, but the penalty will be substantial.} occurs. As we are considering systems with arbitrarily low bias, and hence $\mu$, this effectively means that quadratic boundaries are untenable. For $f$ cubic and above, $\evqty{n}$ briefly increases before decreasing forever. From these, we can deduce that $f$ should ideally be (absolute) linear, though it is admissible to also have a small quadratic component if $\mu$ is bounded from below by $2f_2D$.

We now show how to find lower and upper bounds for $\vec\mu$ per our canonical example, that is $\vec\mu(\mathcal P)=\mu\hat u$ and $\vec\mu(\mathcal P')=-(\opn{sign}v)\mu\hat v$. In both regions, the action of the drift vector alone is to increase $n$. The direction of $\vec\mu$ in $\mathcal P'$ acts to pull the distribution further away from $\mathcal P$ than would be the case if $\vec\mu$ pointed $u$-wards; therefore, picking constant uniform $\vec\mu$ will retard the process of synchronisation and thereby result in an upper bound. Alternatively we can advance the process by picking a uniform superposition of the two, $\vec\mu=\mu(\hat u-(\opn{sign}v)\hat v)$, letting the extra `space' in the $v$ axis apply to $\mathcal P$ too. This extra `space' would appear to complicate matters further, but we can avoid this complication by realising that the action of this $v$-wards drift is simply to increase the value of $n$ and therefore it is equivalent to setting $\vec\mu=2\mu\hat u$. Solving for the bounds finally yields
\begin{align*}
  \tau &\ge \frac{s}{2\mu-2f_2D}+\frac{2f_1^2D}{\pi(2\mu-2f_2D)^2}\qty[1+\sqrt{1+\frac{\pi s}{f_1^2D(2\mu-2f_2D)}}] \\
  &\le \frac{s}{\mu-2f_2D}+\frac{2f_1^2D}{\pi(\mu-2f_2D)^2}\qty[1+\sqrt{1+\frac{\pi s}{f_1^2D(\mu-2f_2D)}}],
\end{align*}
which for our canonical example reduces to
\begin{align*}
  \frac{s}{2\mu}+\frac{D}{2\pi\mu^2}\qty[1+\sqrt{1+\frac{\pi s}{2\mu D}}] \le \tau
  &\le \frac{s}{\mu}+\frac{2D}{\pi\mu^2}\qty[1+\sqrt{1+\frac{\pi s}{\mu D}}].
\end{align*}
This shows that the penalty is at least of order $\mu^{-2}$ and, for small $s$, is as high as $\mu^{-5/2}$. In other words, the penalty is worse than for the constrictive case. Interestingly however, the penalty is almost exactly the same for higher dimensions if the boundary is replaced by a hypersurface of revolution of $f$; in fact, in dimension $d$ we find
\begin{align*}
  \alpha_p=\frac{(p+d-3)!!}{(d-3)!!}\begin{cases}
    \phantom{\sqrt{\frac{2}{\pi}}}\mathllap{\ooalign{\phantom{$\frac{2}{\pi}$}\cr\hss$1$\hss}} & \text{$p,d$ even} \\
    \sqrt{\frac{2}{\pi}} & \text{$p$ odd} \\
    \phantom{\sqrt{\frac{2}{\pi}}}\mathllap{\frac{2}{\pi}} & \text{$p$ even, $d$ odd}
  \end{cases}
\end{align*}
but otherwise the form of the synchronisation times remains the same. As a result, using a recessive synchronisation geometry is an alternative way of efficiently synchronising in dimensions $d=3$ and above. Of course, a sequence of $d=2$ constrictive synchronisations will limit the penalty to order $\mu^{-2}$ and so constrictive synchronisation geometries remain the most effective approach.

\section{Constrictive Generating Functions}
\label{app:gf}

Finally, we present some limited successes towards generating functions for the MFPT for arbitrary initial distribution. 
We shall consider the truncated simplex for $d=2$ and $w=1$, as shown in \Cref{fig:synch-simple-2}. To do so, we must construct a function to enumerate all possible walks from each initial position and weight them by their probabilities. A walk starts at some initial position (for convenience we use the upper right quadrant), takes a number of up, down, left and right steps, and terminates at the origin. It is also subject to the constraint that it must not reach the origin until its final step, and if the walker is on a reflective boundary then a forbidden step (left or down, depending on which boundary) is replaced by a null step with the same probability but which leaves its position unchanged.

Constructing a generating function subject to these constraints is easier in reverse. Let the variables of the generating function be $x$, $y$ and $t$, such that a term $\rho x^my^nt^s$ indicates that a walk starting at $(m,n)$ reaches the origin in $s$ steps with probability $\rho\equiv\pr(s|m,n)$. As all walks almost surely reach the origin in a finite time, $\sum_{s=0}^\infty\pr(s|m,n)=1$ and so if $W(x,y;t)=\sum_{mns} \pr(s|m,n)x^my^nt^s$ then $W(x,y;1)=\sum_{mn}x^my^n\equiv\frac1{1-x}\frac1{1-y}$. To construct such a $W$ from walks in reverse, we start from the origin and walk anywhere in the plane for any number of steps. As we shall be handling the origin specially, we start with the last step which must be from either $(0,1)$ or $(1,0)$. The probability that the next step from either of these reaches the origin is $\frac12p$, and therefore these two particular walks are given by the initial term $\frac12pt(x+y)$. Given a walk $w$ of length $s$ starting at $(m,n)$, we can construct a walk of length $s+1$ by prepending it with each of the four possible steps. Ignoring the boundary conditions, this means we can generate a longer walk as $\frac12pt(x+y)w + \frac12qt(\bar x+\bar y)w$ where $\bar x\equiv x^{-1}$. Note that $\frac12ptyw$ corresponds to prepending a downward step with probability $\frac12p$, but we increase the power of $y$ because the power of $y$ is the \emph{starting} position. If $w$ is of the form $\rho x^mt^s$ or $\rho y^nt^s$, then we replace respectively the steps $\frac12qt\bar y$ or $\frac12qt\bar x$ with $\frac12pt$ (this is not a typo; to see that $p$ is correct it is helpful to work through an example walk). Lastly we must ensure walks avoid the origin, or at least that we ignore those walks which prematurely reach the origin. We do so by `trapping' such walks so that they get stuck there by excluding these walks from the $\frac12pt(x+y)w$ transitions. Assembling these, we obtain the functional equation
\begin{align}\begin{aligned}
  W(x,y;t) &= \tfrac12pt(x+y) && \text{(last step)} \\
   &+ \tfrac12qt\bar x[W(x,y;t)-W(0,y;t)] + \tfrac12ptW(0,y;t) && \text{(right step)} \\
   &+ \tfrac12qt\bar y[W(x,y;t)-W(x,0;t)] + \tfrac12ptW(x,0;t) && \text{(up step)} \\
   &+ \tfrac12pt(x+y)[W(x,y;t)-W(0,0;t)], && \text{(left/down step)}
\end{aligned}\label{eqn:functional-constrict-raw}\end{align}
which can be written more usefully in the form
\begin{align}
  \tfrac2p\bar tKW &= (x+y)(1-S) + (1-r\bar y)X(x) + (1-r\bar x)X(y) \label{eqn:functional-constrict} \\
  K &= 1-\tfrac12pt(x+r\bar x+y+r\bar y) \nonumber
\end{align}
where $r=\frac qp$, $S=W(0,0;t)$, $X(x)=W(x,0;t)\equiv W(0,x;t)$, and $K\equiv K(x,y;t)$ is known as the kernel.

\para{Orbital Sum Approach}
We now apply the orbital sum approach as detailed by \textcite{bm-walks}. The group of the kernel is defined to be the set of maps $(x,y)\mapsto(x',y')$ that leave the kernel unchanged, with the group operation taken to be function composition, $\circ$. For this kernel it is generated by the orthogonal involutions $\chi:(x,y)\mapsto(r\bar x,y)$ and $\upsilon:(x,y)\mapsto(x,r\bar y)$, and hence is finite with order 4. Explicitly, the elements of the group are $G=\{1,\chi,\chi\upsilon,\upsilon\}$. We also define the `sign' of each element, $\sigma_\gamma$, as $\sigma_1=\sigma_{\chi\upsilon}=+1$ and $\sigma_\chi=\sigma_\upsilon=-1$. When finite, this group has the interesting property that the `orbital sum' of a function of just $x$ or $y$ under this group vanishes, i.e.\ $\sum_{\gamma\in G}\sigma_\gamma\gamma(f(x))=\sum_{\gamma\in G}\sigma_\gamma\gamma(f(y))=0$.

Looking at our functional equation~\ref{eqn:functional-constrict}, all of the terms on the right hand side are functions of both $x$ and $y$. However, by dividing through by $(1-r\bar x)(1-r\bar y)$ we can make two of these terms functions of only one of $x$ and $y$. We need to be careful however, as generating functions have ring structure but don't generally admit an operation corresponding to division. Moreover, we are dividing through by series in negative powers of $x$ and $y$, and so we need to consider the Laurent series structure. Specifically our series belong to the ring $\mathbb R(x,y)\llb t\rrb$ of formal power series in $t$ whose coefficients are Laurent polynomials in $x$ and $y$, and the series $W$ belongs to the ring $\mathbb R[x,y]\llb t\rrb$ of formal power series in $t$ whose coefficients are polynomials in $x$ and $y$. To perform the division, we observe that the series $-\bar rx(1-\bar rx)^{-1}=-\sum_{n=1}^\infty(\bar rx)^n$ annihilates $(1-r\bar x)$ under multiplication and therefore is an appropriate multiplicative inverse to use.

Multiplying through then and taking the orbital sum, we find
\begin{align*}
  \sum_{\gamma\in G}\sigma_\gamma\gamma\qty\Big( \frac{\bar rx}{1-\bar rx}\frac{\bar ry}{1-\bar ry}\tfrac2p\bar tKW) &= (1-S) \sum_{\gamma\in G}\sigma_\gamma\gamma\qty\Big( (x+y)\frac{\bar rx}{1-\bar rx}\frac{\bar ry}{1-\bar ry}).
\end{align*}
Recalling that $K$ is invariant with respect to the group, we can rewrite as
\begin{align*}
  \sum_{\gamma\in G}\sigma_\gamma\gamma\qty\Big( \frac{\bar rx}{1-\bar rx}\frac{\bar ry}{1-\bar ry}W) &= \tfrac12pt(1-S)\underbrace{\frac{1}{K} \sum_{\gamma\in G}\sigma_\gamma\gamma\qty\Big( (x+y)\frac{\bar rx}{1-\bar rx}\frac{\bar ry}{1-\bar ry})}_{R(x,y;t)}.
\end{align*}
To proceed, observe that $\frac{\bar rx}{1-\bar rx}\frac{\bar ry}{1-\bar ry}W$ remains part of the ring $\mathbb R[x,y]\llb t\rrb$. Moreover, every term has a positive power in $x$ and $y$. Therefore, the orbital sum yields the sum of series in the rings $\mathbb R[x,y]\llb t\rrb$, $\mathbb R[\bar x,y]\llb t\rrb$, $\mathbb R[x,\bar y]\llb t\rrb$ and $\mathbb R[\bar x,\bar y]\llb t\rrb$, and so if we extract the positive part consisting of only the terms with positive powers of $x$ and $y$, then we can find an expression for $W$
\begin{align*}
  \frac{\bar rx}{1-\bar rx}\frac{\bar ry}{1-\bar ry}W &= \tfrac12pt(1-S)[x^>y^>]R \\
  W &= \tfrac12pt(1-r\bar x)(1-r\bar y)(1-S)R^>
\end{align*}
where $[x^>y^>]R\equiv R^>$ is the positive part of $R$. There is a limitation here: if $W$ can be written as $f(x,y;t)+g(t)\frac1{1-x}\frac1{1-y}$ then the orbital sum of the $g$ term in this scheme will vanish because $\frac1{1-x}\frac{\bar rx}{1-\bar rx}$ is invariant under the group, as is the $y$ part. Moreover, any other functions which vanish under the orbital sum will also be excluded from our apparent $W$. This means, for example, that evaluating the above expression for $W(x,y;1)$ will yield 0 because $W(x,y;1)$ is in fact $\frac1{1-x}\frac1{1-y}$.

The MFPT starting from $(m,n)$ is given by $\sum_s s\pr(s|m,n)$, and its generating function is $T(x,y)=\sum_{mns} s\pr(s|m,n)x^my^n$. It is easy to show, therefore, that $T(x,y)=\dot W(x,y;1)$ and so we find
\begin{align*}
  T(x,y) &= \tfrac12p(1-r\bar x)(1-r\bar y)(-\dot S(1))R^>(x,y;1),
\end{align*}
assuming that $T$ has no terms which vanish under the orbital sum. It remains to find $S$, and to determine whether or not $T$ has any such terms.

\para{Obstinate Kernel Method}
Another approach is given by the obstinate kernel method, wherein we exploit roots of the kernel. A pair of such roots is given by $(x,Y_\pm(x))$ where $Y_\pm=B\pm\sqrt{B^2-r}$, $B=\frac1{pt}-\tfrac12z$, and $z=x+r\bar x$ is a symmetric function of $x$ and $\chi(x)\equiv r\bar x$. It will be helpful to note the elementary symmetric polynomials of $Y_\pm$ as $\tfrac12(Y_++Y_-)=B$ and $Y_+Y_-=r$. Only one of these roots, $Y_-$, is a power series in $t$; the other includes a term in $t^{-1}$, and so we shall substitute $Y_-$ into the functional equation~\ref{eqn:functional-constrict} to get
\begin{align*}
  0 &= (x+Y_-)(1-S) + (1-Y_+)X(x)+(1-r\bar x)X(Y_-).
\end{align*}
To solve for $X(x)$, and thence $W(x,y)$, we eliminate $X(Y_-)$ by making use of the kernel's group; in particular applying $\chi$ yields a second equation in $X(Y_-)$,
\begin{align*}
  0 &= (r\bar x+Y_-)(1-S) + (1-Y_+)X(r\bar x)+(1-x)X(Y_-),
\end{align*}
from which we find
\begin{align*}
  Q(x)-Q(r\bar x) &= (1-S)(r\bar x-x)\qty\Big(\frac{1-Y_--z}{1-Y_+})
\end{align*}
with $Q(x)\equiv(1-x)X(x)$. Extracting the positive part then gives
\begin{align*}
  Q(x)-S &= (1-S)[x^>](r\bar x-x)\qty\Big(\frac{1-Y_--z}{1-Y_+}).
\end{align*}
From $Q$, $X$ can be obtained and thereby $W$ and $T=\dot W|_{t=1}$. Alternatively, the MFPTs for arbitrary positions can be computed using the recurrence relations and the boundary MFPTs $\dot X|_{t=1}$. In order to do so, we shall need to obtain an expression for $S$. Letting $U(t)=[x^1]X(x;t)$, we see that $[x^1](Q-S)=U-S$ and we compute $[x^0y^0]W$ using \Cref{eqn:functional-constrict-raw} to find $(1-pt)S=qtU$. Therefore we can obtain an expression for $S$,
\begin{align*}
  \frac{S}{1-S} &= \frac{qt}{1-t}[x^1](r\bar x-x)\qty\Big(\frac{1-Y_--z}{1-Y_+}).
\end{align*}
Using the symmetric variable $x$ helps to evaluate the positive parts of the antisymmetric series $\qty\big(r\bar x-x)\qty\big(\frac{1-Y_--z}{1-Y_+})$, but we were unable to proceed beyond this point.

\section{Applicability to Other Architectures}

Finally, we comment on the applicability of these results to non-Brownian reversible computers in the limit of vanishing free energy density.
The concepts of mona and synchronisation phase space can be readily extended to other models of communicating reversible computer, but the time evolution of the phase space distribution may differ.
In many architectures, the probability of a backward transition approaches 0.
Nevertheless, it is still possible for the joint distribution of mona in phase space to expand over time.
In particular, if the eigenstates of the Hamiltonian change over time due to thermal fluctuations, then the rate of computation at different \emph{independent} loci will fluctuate and so synchronisation will be lost.
Moreover, universal phenomena such as the Heisenberg uncertainty principle and general relativistic time dilation effects suggest that even a system at absolute zero would experience desynchronisation.

An important assumption in the above, and indeed this chapter, was that the mona evolve independently except during synchronisation events.
As argued by \textcite{frank-coupled-sync} in a thought experiment, however, it is possible to couple the mona in a system with a potential that increases with increasing desynchronisation.
In particular he imagined a lattice of clockwork computational units driven by a series of motors, with adjacent axles connected by bicycle chains.
At low speeds $|v|$, the power supplied to counter frictional losses in the system would scale with $\sim N|v|^2$ where $N$ is the number of mona, which recovers the $V^{5/6}$ scaling law.
Any tension due to desynchronisation between adjacent mona would not substantially increase this power requirement, not least because adjacent axles experience equal and opposite tensions.
Therefore the introduction of such potentials may be used to constrain the size of the joint phase distribution of adjacent mona; that is, mona remain locally in approximate synchronisation.
Globally, remote mona may accumulate more significant desynchronisation, scaling with $\sqrt{RT}$ where $R$ is linear dimension of the system and $T$ the temperature, but this would not pose any problems to long-distance communication as long as message packets remain in local synchronisation at all points along their routes.

As a corollary, we extend Frank's thought experiment to a general conjecture about which classes of reversible computer are and are not subject to a synchronisation penalty:
\begin{cnj}[Inter-monon synchronisation potentials are a necessary condition for cheap communication]
  \label{cnj:sync-potential}
  By the above arguments, we conjecture that the joint phase space of independently evolving mona cannot remain localised, and hence that the synchronisation penalties derived in this chapter apply for all such reversible computer architectures.
  Conversely, mona whose evolution is coupled by a potential that increases with increasing desynchronisation will remain localised, and therefore can in principle pass through synchronisation windows with arbitrarily low penalty.
  If the distribution can be constrained to be compact, then the penalty may be eliminated.
\end{cnj}
However, this is predicated on communicating mona being carefully programmed.
Communicating mona must anticipate each other's computational states, such that they are both simultaneously (up to the width of the synchronisation window) receptive to communication.
If a monon $A$ attempts to send information to another monon $B$ that is not expecting to receive such information, then this will incur an entropic synchronisation penalty by the mechanisms discussed in this chapter.
Alternatively, a radically different approach to eliminating communication penalties is given by fully serial architectures, such as is \emph{theoretically} permitted by a Quantum Zeno architecture (\Cref{sec:qze}).

An interesting example of an architecture where careful ordering is intrinsic to the computation itself is the BARC\footnote{BARC was historically referred to as ABRC, the Asynchronous Ballistic Reversible Computing model.} (Ballistic Asynchronous Reversible Computing) model of \textcite{frank-barc}.
BARC systems consist of devices with an internal state, and wires between them along which pulses can travel. Devices are designed to interact with at most one pulse at a given time, the effect of which may be to change the device's internal state and to re-emit the pulse along a different wire.
As the order of signals is important to the computation, any well-designed BARC computer will satisfy the synchronisation simultaneity condition.
In this case, the synchronisation window depends on pulse spacing.
Additionally it is an architecture without backwards transitions.
However, the architecture---as presented---does not couple pulses to each other with a synchronisation potential, and so it is subject to the conclusions of this chapter as the ordering between free signals is subject to change.
Nevertheless, this desynchronisation can be kept impressively minimal through the high efficiencies and engineering tolerances afforded by superconducting circuits (see \textcite{frank-barc-fluxon} for progress towards a Single Flux Quanta implementation of the BARC model).

\paragraph{Petri Nets}

Another model of asynchronous concurrency that is quite popular is that of Petri Nets (see~\cite{petri} for a review).
A Petri Net is a weighted directed graph with two types of nodes, `places' and `transitions'.
A place contains tokens, and the configuration of tokens in all places---a `marking'---represents the state of the system.
Tokens can be consumed, generated, and moved by the activation of a transition.
A transition node consumes some tokens from each place connected to it by an incoming edge, and deposits some tokens to each place connected by an outgoing edge.
Edges may only be drawn between a transition and a place, and the weight of the edge dictates the number of tokens involved in the transition.
A transition may only activate if there are enough input tokens.

Whilst a Petri Net is not necessarily reversible in the sense of having a corresponding inverse transition for every transition, it is always possible to construct an `inverse Petri Net' by inverting each edge's direction.
From this perspective, a Petri Net along with its sequence of transitions form a logically reversible system.
It is not immediately obvious where entropy costs manifest in a Petri Net, and so it may appear that they offer a route to avoid the conclusions of this chapter.
In order to accommodate entropy costs, the sequence of transitions needs closer attention.
Without loss of generality, assume discrete time where at each time step a single\footnote{Multiple transitions can occur simultaneously only if they do not share any places, but then they can always be assigned an arbitrary order and treated as sequential.} transition is attempted.
If every transition is always successful, then there is no entropy increase.
If, however, a transition may fail, then reversibility requires us to record a binary outcome of success or failure.
That is, the input is a stream of symbols---each corresponding to a particular transition---and the output is a stream of symbol-outcome pairs.
This is where entropy may increase, and is another example of a breakdown of simultaneity in synchronisation.
Moreover, such a system---in which the transition order is sequential---does not allow for separate independent entities, let alone spatial structure.
Therefore the notion of synchronisation/communication is not particularly relevant.
To accommodate this, multiple transition streams can be introduced.
For avoidance of conflict, the streams must satisfy the constraint that no two transitions sharing a place fire at the same time.
Now the issues of synchronicity of signals and events become manifest: if these transition streams can be coupled to each other to maintain some local synchronisation then low cost communication can be supported, otherwise there will be a synchronisation penalty.
Therefore, whilst Petri Nets are a convenient framework for modelling asynchronous systems, they are not particularly convenient for modelling the possibility of desynchronisation and the resultant resynchronisation entropy costs.

\section{Conclusion}

The primary result of this chapter is that communication between distinct reversible computers in a Brownian system (and arguably any reversible architecture whose mona evolve independently) with a limiting supply of free energy is very expensive in comparison to independent computation.
A single computational step takes time $1/b\lambda$ for gross computational rate $\lambda$ and bias $b=(\dot G/k_BT\lambda N)^{1/2}$, where $\dot G$ is the rate of supply of free energy, $N$ is the number of mona, $k_B$ is Boltzmann's constant, and $T$ is the system's average temperature. In contrast, an interaction event such as communication takes at least a time $\sim1/b^2\lambda$. As the bias scales as $\sim\sqrt{A/V}\sim R^{-1/2}$ for a system of convex bounding surface area $A$, internal volume $V$ and radius $R$, communication will tend to `freeze' out as the system gets bigger. This cost is unavoidable as it is due to the implicit erasure of information involved in synchronising divergent computational states; however, if synchronisation is rare then the temporal cost can be reduced. Suppose that the average proportion of mona wishing to communicate at a given time is $\nu\ll1$, then the free energy can be divided into two supplies: $\dot G_{\text{indep.}}$ and $\dot G_{\text{synch.}}$ for independent computation and synchronisation events respectively, as intimated in \Cref{sec:revi-disc}. A separate set of bias klona can then be introduced for the synchronising mona,
\begin{align*}
  b_{\text{indep.}} &= \sqrt{\frac{\dot G_{\text{indep.}}}{k_BT}\frac{1}{\lambda N}}, &
  b_{\text{synch.}} &= \sqrt{\frac{\dot G_{\text{synch.}}}{k_BT}\frac{1}{\lambda\nu N}}.
\end{align*}
In order for the synchronisation to be viable we therefore require $1/b_{\text{synch.}}^2 \lesssim 1/b_{\text{indep.}}$ and hence find
\begin{align*}
  \nu &\lesssim \frac{\dot G_{\text{synch.}}}{k_BT}\sqrt{\frac{k_BT}{\dot G_{\text{indep.}}}\frac{1}{\lambda N}}.
\end{align*}
If we want to maintain our asymptotic independent computation rate of $\sim\sqrt{AV}$, then we will have $\dot G_{\text{synch.}}\sim\dot G_{\text{indep.}}\sim A$ and $N\sim V$, giving $\nu\lesssim\sqrt{A/V}\sim b$.
That is, the number of synchronising mona $\nu N$ can be as high as $\sim\sqrt{AV}$. In fact, this subpopulation of mona are also permitted to perform arbitrary operations subject to entropic costs, with the proviso that their individual net transition rates will be $b\lambda\sim R^{-1/2}$. The consequence is that, for a very large reversible Brownian computer, any parallel computation should be structured to minimise the need to synchronise the joint state of the system.
Equivalently, the proportion of computational steps that are permitted to synchronise is $\lesssim b$.
If one wishes to increase the proportion of synchronising mona, then the total computational rate must necessarily be reduced or an architecture satisfying \Cref{cnj:sync-potential} must be employed.

\para{Impact on Algorithms}

The consequence of this `freezing' out of synchronisation is that the most natural computational problems for these computers are those which are `embarrasingly parallel'.
An embarrasingly parallel problem is one which can be easily divided into many tasks with very infrequent synchronisations between them.
Some examples of such problems are Monte Carlo algorithms, brute force searches, and ray tracing.
In contrast, highly serial problems, or parallel problems involving significant synchronisation, are not well accommodated.
Serial programs suffer because the rate of computation by individual mona is asymptotically vanishing.
Highly synchronous parallel programs suffer because of the time penalty for synchronisations.
In limiting cases where no parallelisation is possible or where synchronisation cannot be made arbitrarily infrequent, the rate of these computations will fall to $\sim A$ as with irreversible computers.
Nevertheless, as discussed in \Cref{sec:revi-disc}, hybrid computational systems supporting both high-bias and low-bias subsystems can be constructed.
Therefore serial/highly synchronous problems can be executed on a thin outer shell of the computer, and embarrassingly parallel problems can make use of the inner bulk of the computational matter.

A common example of highly synchronous parallel problems is that of simulating systems involving local interactions on a lattice, such as cellular automata, Ising models, and molecular dynamics simulations.
Indeed, in \textcite{frank-thesis} (Section 6.3.2.1), he examines local computation by a three dimensional array of computing elements.
For the class of non-Brownian computers under consideration, he finds an asymptotic benefit for reversible architectures over irreversible architectures.
Unfortunately the results of this chapter show that problems like these are not well-supported by Brownian architectures, and so this asymptotic benefit is not seen in this case.
However, the thought experiment of \textcite{frank-coupled-sync} and \Cref{cnj:sync-potential} do support the realisation of such asymptotic benefit in architectures in which the computing elements are appropriately coupled.

\endgroup

\begingroup

\begin{chapter-summary}

  This chapter concludes our analysis of the limits the laws of physics place on the sustained performance of reversible computers.
  \Cref{chap:revi} concerned aggregate performance in terms of computational operations per unit time, but neglected to consider interactions among computational sub-units or between computational sub-units and shared resources such as memory or chemical species.
  \Cref{chap:revii} extended the analysis to consider the former set of interactions.
  In this chapter we extend the analysis to consider the latter set, which in general will be governed by the laws of thermodynamics and statistical mechanics.
  When the free energy available to the computational subunits is limiting, this becomes a challenging endeavour as perturbations to these subsystems may have a corresponding entropy cost for which our available free energy is insufficient.

  In the first half we pay particular attention to resource distribution:
  Our Brownian computational subunits---themselves being small and particulate---will likely lack adequate internal memory for many programs of interest;
    as such, it can be expected that a realistic Brownian system will make heavy use of a shared pool of memory resources that can be dynamically distributed amongst the computational subunits according to their need.
  It is found that most schemes imaginable fail to function effectively in the limit of vanishing `computational bias' $b$, which measures the net fraction of transitions which are successful and which falls as the system grows in size.

  Driving thermodynamically unfavourable reactions, such as resource distribution, is a very general problem for such systems and can be solved by supplying a sufficient excess of free energy.
  In order to make these interactions viable, we propose an abstract chemical scheme to dynamically infer how much free energy is needed, and to supply it by accumulating just enough from the (weak) ambient free energy.
  The scheme takes any reaction, arbitrarily far from equilibrium, and sequesters some of the reactants and products such that the extant reactants and products are in equilibrium.
  The scheme rapidly adjusts to changes in the equilibrium position, and makes use of molecular computation for its inference algorithm.
  Unfortunately there is a cost in applying this scheme: the reaction is slowed by a factor inversely proportional to the `computational bias', $b\ll1$ (the net proportion of computational transitions that are successful).
  In fact, this overhead is comparable to that found for synchronisation interactions in \Cref{chap:revii}.

\end{chapter-summary}

  \chapter{Performance Trade-offs for Reversible Computers Sharing Resources}
  \label{chap:reviii}

\begin{figure}[h!]
  \centering
  \includegraphics[width=.9\textwidth]{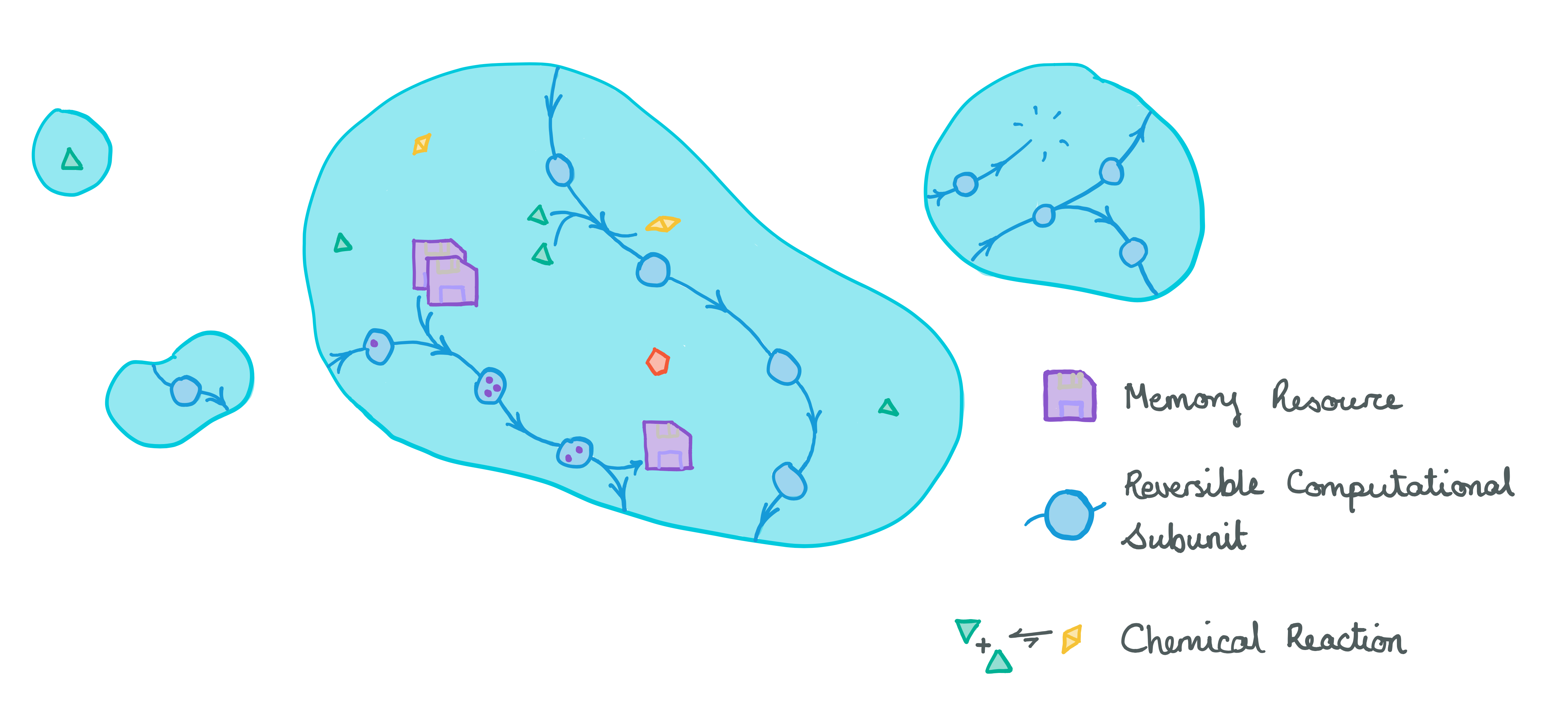}
    
  \caption[A cartoon representation of resource sharing interactions amongst Brownian computational particles.]{This chapter concerns interactions between reversible Brownian computational particles and populations of other particles such as species involved in chemical reactions or shared memory resources. Populations of particles are subject to well known statistical mechanical laws, and these laws can cause problems when computational particles attempt to interact with these populations.}
  \label{fig:cover-iii}
\end{figure}

\section{Introduction}

This is the final chapter investigating the influence of the laws of physics on computation, particularly as they apply to large reversible Brownian computers (see \Cref{dfn:brownian}).
As discussed in \Cref{chap:revii}, the net computational rate (in terms of primitive computational operations per unit time) studied in \Cref{chap:revi} is not the only performance metric of interest.
Therefore, in \Cref{chap:revii} we evaluated computational performance of reversible Brownian computers
(and argued for the applicability of the results to other reversible computers whose computational units evolve independently)
whose individual computational components were permitted to interact with one another, i.e.\ to communicate and to perform concurrent/parallel computation.
A key property of these interactions is that computational `particles' can be uniquely referenced, and therefore the traditional thermodynamic principles that apply to chemical systems were inapplicable.
In this chapter, we now consider the interface between traditional chemical species, of which there may be a variable number of identical units or particles, and the computational particles.
See \Cref{fig:cover-iii} for an illustration.
For convenience, we shall refer to these as \emph{klona} and \emph{mona} respectively, following the formulation introduced in \Cref{dfn:mona-klona}. Recalling the definitions, we call particles which are uniquely distinguishable or addressable \emph{mona} (sg.\ \emph{monon}), after the Greek for unique, \emph{\gr μοναδικός}. Similarly, we call indistinct particles belonging to a species \emph{klona} (sg.\ \emph{klonon}).
Klona may be likened to traditional chemical species such as glucose and ATP.
Mona are less obvious; an example would be a \SI{150}{\milli\liter} solution containing a \SI{1}{\milli\molar} concentration of random nucleic acids of length 80 nucleotides.

Interactions between mona and klona are of great practical importance for Brownian computers; whilst mona by themselves may be capable of arbitrary computation, this computation only becomes useful when used to actuate some other system or when it otherwise produces some observable output. Examples may include the activation of a fluorophore upon reaching some logical condition, or directing the translation of a particular messenger RNA into a protein. Moreover, interaction with klona may be essential for computation: general computation requires access to an unbounded amount of memory, but by their very nature mona are limited in size and thus memory. Therefore it is important for mona to be able to recruit more memory as required, and release it for use by other mona when the need passes. Resources other than memory may also be considered, such as structural monomers for use in constructing and deconstructing molecular machines or computers.

In this chapter we shall be most interested in the case of resource distribution, and will consider various such schemes in \Cref{sec:resource-scheme}. One approach is a centralised store which mona can access, either by travelling there themselves or by sending a `courier' monon on their behalf to fetch and retrieve (or carry and return) a resource. This approach is undesirable. If the couriers float freely in solution, then reaching the central store and returning will both take a time $\sim\bigOO{V}$; if on the other hand we introduce a lattice as in \Cref{chap:revii}, then the average distance will be proportional to the radius $R$ and, given $b\sim R^{-1/2}$, the average time will be $\sim\bigOO{b^{-3}}$. As the time for a single net computational step is $\sim\bigOO{b^{-1}}$, both of these are untenable. Instead we shall seek a decentralised approach by making use of klona; intuitively it must become harder to access a resource as its scarcity increases, but we would like the timescale to do so to not be substantially more costly than $\sim (b\ce{[X]})^{-1}$ where $[\ce{X}]$ is the available concentration of the resource.

We will unfortunately find that the reactions involved in all possible schemes are unfavourable except in a very limited range of resource availability. In \Cref{sec:drive-unfav} we shall return to considering general mona-klona interactions, and demonstrate how mona can dynamically drive any unfavourable reaction with klona by inferring the number of bias tokens required to balance the entropic cost of the reaction. Alas, there will turn out to be an overhead in achieving this even in the case when the driven reaction is at equilibrium and hence cost-free.

\section{Resource Distribution Schemes}
\label{sec:resource-scheme}

A resource distribution scheme for a resource $\ce{X}$ is a set of klona from which $\ce{X}$ can be extracted and into which it can be deposited. Formally, resources are conserved discrete quantities analogous to charges in physics. Letting the set of resource species be indexed by $\mathcal R$ and the set of resource klona by $\Sigma$, then to each klonon $\sigma\in\Sigma$ is associated a resource charge $\vec q_\sigma\in\mathbb N^{\mathcal R}$: a vector over the field of naturals and indexed by the resource species. In any reaction, the net resource charge must be conserved; for example, for the reaction $\ce{ $\ket{n}$ + $\sum_{\sigma\in\Sigma}\nu_{\sigma}\sigma$ <=> $\ket{n+1}$ + $\sum_{\sigma\in\Sigma}\nu'_{\sigma}\sigma$ }$ involving the transition of a monon between states $\ket{n}$ and $\ket{n+1}$ and the interaction of said monon with resource klona according to the stoichiometries $\ce{ $\vec\nu$ <=> $\vec\nu'$ }$, this conservation law can be expressed as $\vec q_{\ket{n+1}} - \vec q_{\ket{n}} = \sum_{\sigma\in\Sigma} (\nu_\sigma-\nu_\sigma')\vec q_\sigma$.

\para{Free Klona}

The simplest possible scheme is that in which the resource is freely present in solution, as shown in \Cref{fig:rr-free}. That is, a system with $\mathcal R=\{\ce{X}\}$, $\Sigma=\{\ce{X}\}$ and $(\vec q_{\ce{X}})_{\ce{X}}=1$.
This system will be a viable resource pool so long as the free energy change involved in both extraction and deposition of $\ce{X}$ is less than the available computational free energy, $\log\frac pq=2\arctanh b$. Note that our notion of `free energy' uses a different convention than is normally taken, namely we are giving the change in information entropy of the universe $\Delta h$, whereas normally the free energy change is expressed as $\Delta G=-kT\Delta h$. The free energy change vanishes when the system is in equilibrium, which for $\ce{X}$ an ideal gas can be calculated (via the Sackur-Tetrode equation) to occur at a concentration
\begin{align*}
  [\ce{X}]_{(\text{eq.})} &= e^{5/2}\qty\Big(\frac{4\pi mu}{3h_P^2})^{3/2},
\end{align*}
where $h_P$ is Planck's constant, $m$ is the mass of $\ce{X}$, and $u$ is its average internal energy. As we interact with the system, the concentration will be drawn away from this equilibrium position and as a result the absolute free energy difference will increase. For this scheme, the free energy difference corresponds to that due to adding or removing a particle of $\ce{X}$ and can therefore be identified with the \emph{Chemical Potential}, $\mu_{\ce{X}}$. In units of information, this can be expressed as $\mu_{\ce{X}}=\log\frac{[\ce{X}]\gamma_{\ce{X}}}{[\ce{X}]\gamma_{\ce{X}}|_{(\text{eq.})}}$ where $\gamma_{\ce{X}}$ is the activity coefficient of $\ce{X}$ that quantifies the deviation of the properties of $\ce{X}$ from an `ideal' substance. If $Q$ is the concentration of resource $\ce{X}$ (here trivially equal to $[\ce{X}]$) then we can now determine how far from equilibrium the resource system can be brought, and therefore how much of the resource $Q$ is accessible to the computational mona. For a small change $\delta Q$, we find
\begin{align*}
  \qty\Big|\frac{\delta Q}{Q}| < 2b\qty\Big|1+\pdv{\log\gamma_{\ce{X}}}{\log Q}|^{-1}  + \bigO{b^2}.
\end{align*}
We don't expect the activity coefficient to depend strongly on $Q$, and certainly not to the extent of $\pdv{\log\gamma}{\log Q}=-1$, which would imply a range of isentropic equilibria, and therefore we can only usefully access a proportion $\bigOO{b}\ll1$ of the available resource $\ce{X}$ under this scheme.

\doublepage{
\begin{figure}[p]
  \centering
  \fbox{\begin{subfigure}{\linewidth}
    \centering
    \includegraphics[width=\linewidth]{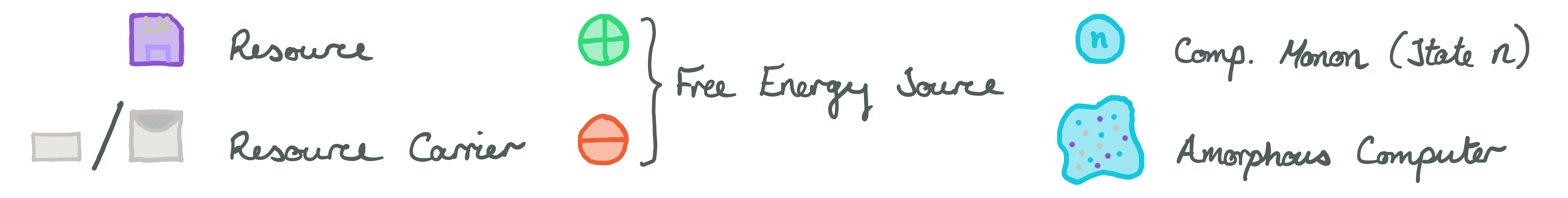}
  \end{subfigure}}
  \begin{subfigure}{\linewidth}
    \centering
    \rule{0pt}{0.4\linewidth}%
    \includegraphics[width=\linewidth]{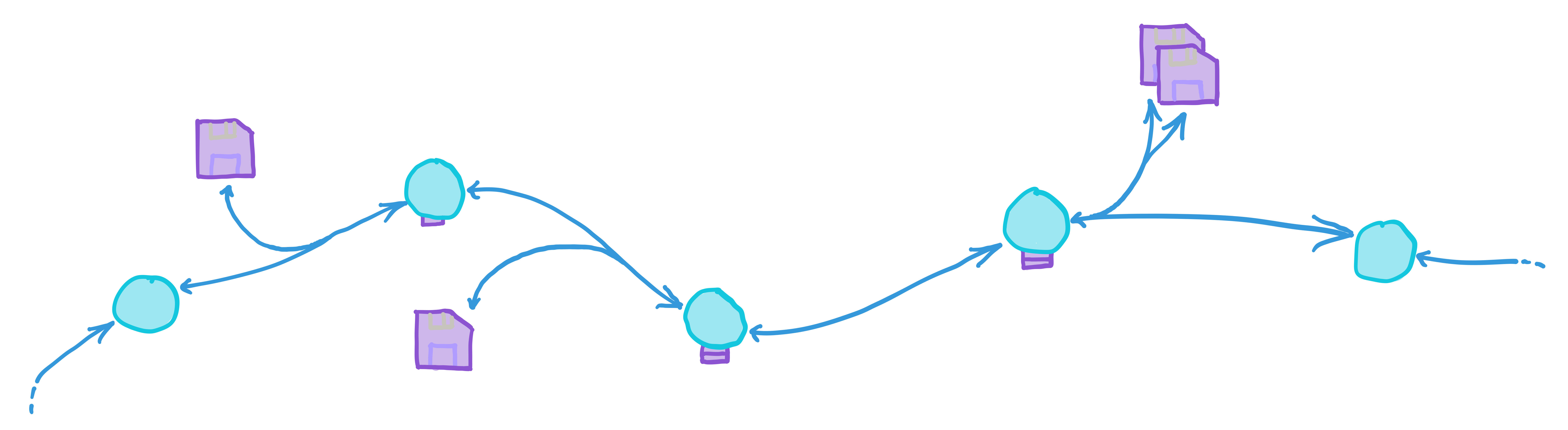}
    \caption{A possible state diagram for a monon interacting with a resource pool; the particular resource distribution scheme is left abstract.}\label{fig:rr-ex}
  \end{subfigure}
  \begin{subfigure}{\linewidth}
    \centering
    \rule{0pt}{0.4\linewidth}%
    \includegraphics[width=\linewidth]{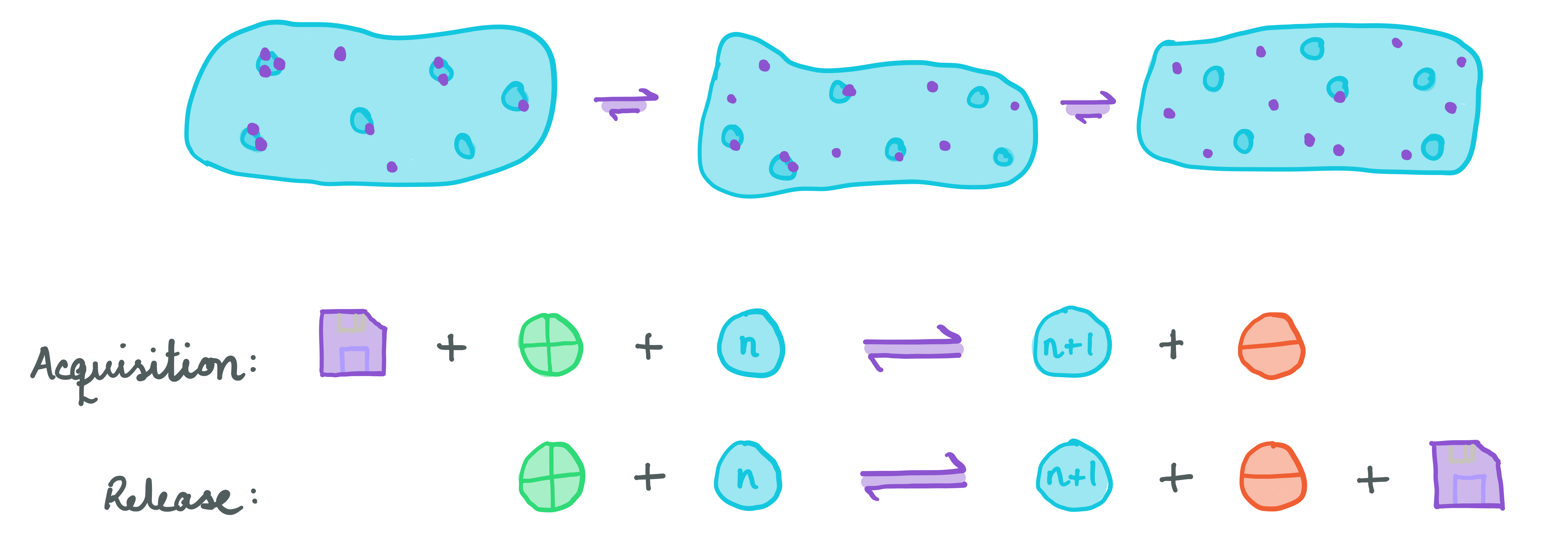}
    \caption{An illustration of the statistical states and reactions for the Free Klona resource distribution scheme.\\[2\baselineskip]}\label{fig:rr-free}
  \end{subfigure}
  \caption[An assortment of illustrative examples of some of the resource distribution schemes for a Brownian amorphous computer.]{%
    An assortment of illustrative examples of some of the resource distribution schemes for a Brownian amorphous computer as considered in this chapter.%
  }\label{fig:rr}
\end{figure}}{%
\begin{figure}[p]
  \ContinuedFloat
  \centering
  \fbox{\begin{subfigure}{\linewidth}
    \centering
    \includegraphics[width=\linewidth]{rr3-key}
  \end{subfigure}}
  \begin{subfigure}{\linewidth}
    \centering
    \rule{0pt}{0.4\linewidth}%
    \includegraphics[width=\linewidth]{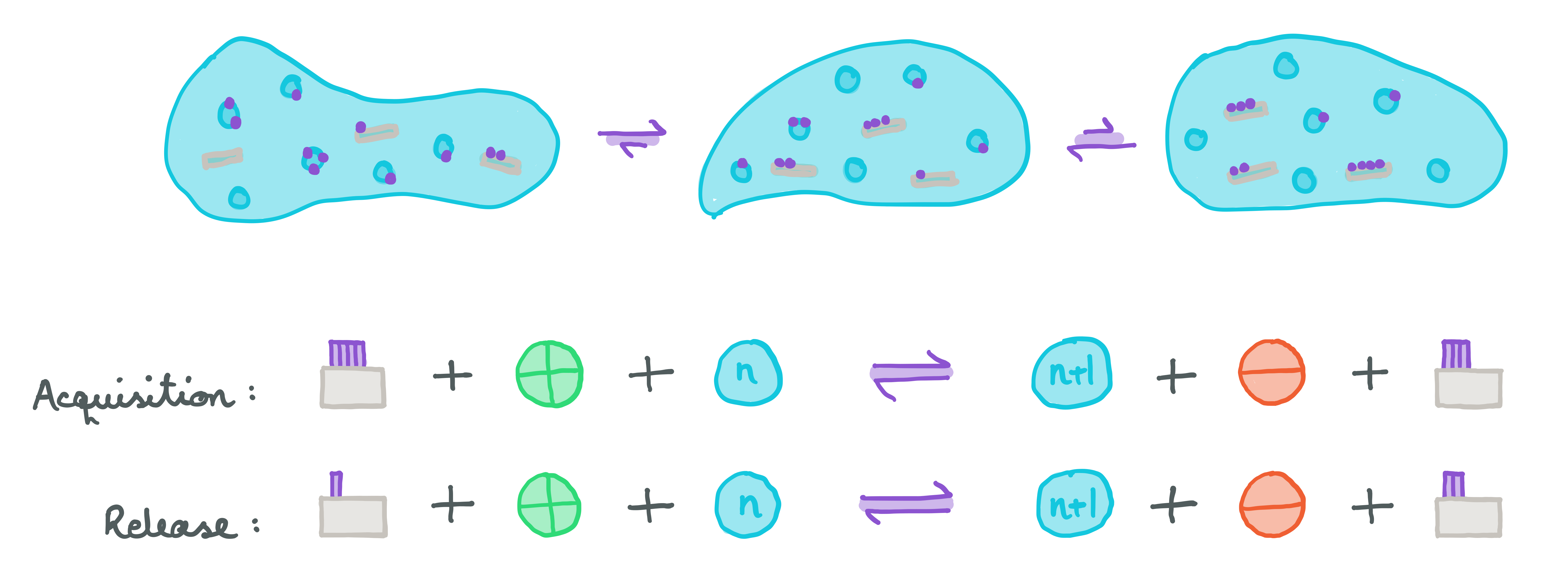}
    \caption{An illustration of the statistical states and reactions for the Bounded Carrier resource distribution scheme.\\}\label{fig:rr-bounded}
  \end{subfigure}
  \begin{subfigure}{\linewidth}
    \centering
    \rule{0pt}{0.4\linewidth}%
    \includegraphics[width=\linewidth]{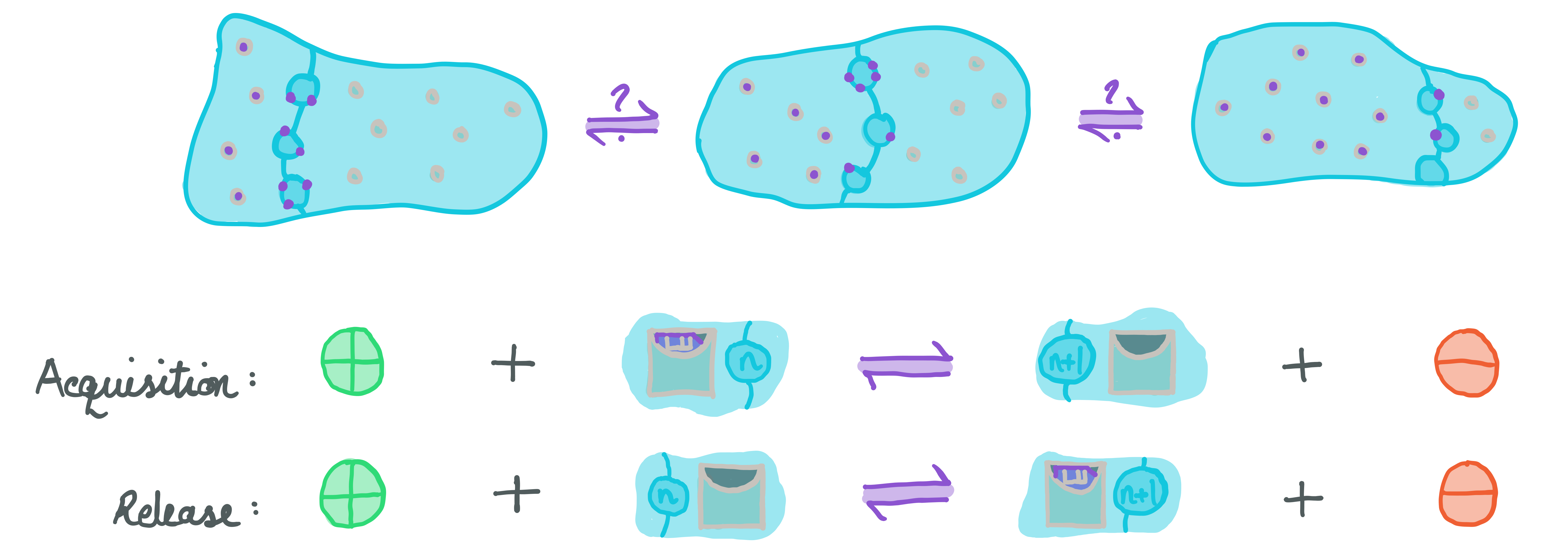}
    \caption{An illustration of the statistical states and reactions for the Isobaric Compartment resource distribution scheme. The question marks serve to indicate that the apparent isergonicity of the resource reactions is in question, and as we find, is unfortunately not the case.}\label{fig:rr-isobaric}
  \end{subfigure}
  \caption[]{%
    (continued) An assortment of illustrative examples of some of the resource distribution schemes for a Brownian amorphous computer as considered in this chapter.%
  }
\end{figure}}

\para{Bounded Carriers}

Perhaps the primary issue with free klona is that there is a change in particle number, implying a substantial entropy change when resources are traded with the pool. A possible improvement is suggested in ensuring total particle number remains constant by binding the resource to carrier klona, as shown in \Cref{fig:rr-bounded}. This may manifest, for example, as $\Sigma=\{\ce{X^0},{X^1}\}$ with $(\vec q_{\ce{X^0}})_{\ce{X}}=0$ and $(\vec q_{\ce{X^1}})_{\ce{X}}=1$. More generally, we have polymeric klona $\Sigma=\{\ce{X^k}:k\in[0,n]\cap\mathbb N\}$ with $(\vec q_{\ce{X^k}})_{\ce{X}}=k$ for some integer $n>1$. Once again, we shall consider a perturbation to the equilibrium distribution of the pool; it will be convenient to introduce a partition function $\mathcal Z(\beta)$ where the thermodynamic variable $\beta$ correlates with the amount of resource stored by the pool. The partition function will take the form $\mathcal Z=\sum_{k=0}^ne^{-\varepsilon_k-\beta k}$ where $\varepsilon_k$ contains all the thermodynamically pertinent information about $\ce{X^k}$ klona, such as its energy and internal degrees of freedom. In fact, we will take $\varepsilon_k=0$ as any $k$-dependence of the $\varepsilon_k$ would only serve to reduce the effective value of $n$ by narrowing the distribution\footnote{Assuming that the equilibrium distribution is unimodal, this corresponds to a reduction in variance. If the distribution is multimodal then instead there is a reduced `local' variance about each mode.}.

The free energy change for a reaction is given by the logarithm of the ratios of reaction rates (see, e.g., \Cref{sec:crn} for a derivation). Therefore in this case it is given by $\Delta h = \log\sum_{k=1}^n e^{-\beta k}-\log\sum_{k=0}^{n-1}e^{-\beta k}$,
and so at equilibrium (when $\Delta h=0$) we find $\beta=0$. For a perturbation from equilibrium, this reduces (exactly) to $\delta\Delta h = -\delta\beta$. Now, the average $\ce{X}$-charge of each klonon is given by $\bar k=-\partial_\beta\log\mathcal Z$, and the charge density by $Q=\sum_{k=0}^n k[\ce{X^k}]\equiv \bar k[\ce{X}]$. The proportional resource accessibility is then given by
\begin{align*}
  \qty\Big|\frac{\delta Q}{Q}| &= \qty\Big|\frac{-\delta\beta\var k}{\bar k}| < \tfrac13b(n+2)
\end{align*}
using $\bar k=\tfrac12n$ and $\var k=\tfrac1{12}n(n+2)$. This suggests taking $n\gtrsim b^{-1}$ to ensure full resource availability. Unfortunately, doing so would mean the total system charge scales as $V/b\sim V^{7/6}$ which clearly cannot be sustained in the limit of large systems. One potential resolution is to reduce the carrier concentration to $[\ce{X}]\sim b^{1/2}$ and then take $n\sim b^{-1/2}$ so that the total charge scales with $V$; whilst this would slow down resource reactions commensurately, mona could use this additional time to accumulate $\sim b^{-1/2}$ tokens in order to raise their free energy budget to $\Delta h\lesssim b^{1/2}$, thereby obtaining a resource accessibility of $|\delta Q/Q|\sim 1$ at a time penalty of $\sim b^{-3/2}$.

\para{Unbounded Carriers}

It is also possible to consider polymeric klona of unbounded size, i.e.\ the limit $n\to\infty$, in which case the partition function is given by $\mathcal Z=(1-e^{-\beta})^{-1}$. It should be noted that in this limit there is no equilibrium distribution of the resources, as $\mathcal Z$ diverges as $\beta\to0$. In fact, distributions only exist for $\beta>0$ whilst the free energy constraint bounds it from above and hence $0<\beta<2b$. This range corresponds to $\bar k>(2b)^{-2}-\bigOO{1}$ and therefore a total enclosed charge $\gtrsim V^{4/3}$; for lower densities, the only favourable reaction is resource deposition. Therefore, whilst conceptually simpler, unbounded carriers are to be avoided.

\para{Isobaric Compartments}

The problem inherent to seemingly all schemes is that a resource reaction must necessarily reduce the quantity of reactants and increase the quantity of products, and this will in turn bias the resource reactions in the opposite direction per \emph{Le Chatelier}'s principle. As our last attempt to circumvent these issues, we introduce a scheme exploiting varying partial volumes in order to maintain equal concentrations of reactants and products. This `Isobaric Compartment' scheme is illustrated in \Cref{fig:rr-isobaric} and consists of dividing the system into two compartments by means of a massless and infinitely flexible membrane, impermeable to the carrier species $\ce{X^0}$ and $\ce{X^1}$. Referring to the two compartments by $\Gamma_\oplus$ and $\Gamma_\ominus$, we then establish the invariants that $\Gamma_\oplus$ only contains $\ce{X^1}$ and $\Gamma_\ominus$ only contains $\ce{X^0}$. A monon wishing to exchange resources with the pool must first bind to the membrane, and then to maintain the invariant it will translocate the carrier species between the two compartments. At equilibrium, the concentrations of carrier species in each compartment must be equal and thus it would appear that we have been successful in eliminating the free energy cost for this mona-klona interaction.

\newcommand{\eff}{{\text{eff.}}}
\newcommand{\eqm}{{\text{eq.}}}
Unfortunately, whilst the carrier species concentrations are equal, concentration is not the correct quantity for determining the reaction rates and thence the free energy difference. From a collision theory perspective, the reaction rates are given by the frequency with which mona and klona collide (in the correct orientation and conformation). Consider the case where all but one carrier reside within $\Gamma_\oplus$: if the total number of carriers is $N$ then the mean volume of $\Gamma_\ominus$ will be $V/N$ and the mean concentration $1/(V/N)=N/V$. This would suggest that a monon bound to the interfacial membrane will encounter the $\ce{X^0}$ klonon as often as it does an $\ce{X^1}$ klonon, but this clearly cannot be the case. If we imagine the $\ce{X^0}$ klonon as residing within a small bubble of volume $V/N$, then whilst its local concentration is indeed the same as for the $\ce{X^1}$ klona, the bubble itself is free to migrate anywhere along the membrane. Consequently, in a system in which particles can bind to the boundary of the system in addition to exploring their volume, the effective concentration becomes $N/(\alpha + \gamma)$ where $\alpha$ is the effective `volume' of the boundary (proportional to its area) and $\gamma$ is the enclosed conventional volume. For our isobaric system, $\gamma_\oplus=N_\oplus\gamma_0$ (and similarly for $\gamma_\ominus$) where the constant $\gamma_0=V/(N_\oplus+N_\ominus)$ is the volume available to each particle. We can therefore write $[\ce{X^0}]_\eff=[\ce{X^0}]/(\alpha/V + [\ce{X^0}]\gamma_0)$ where $[\ce{X^0}]=N(\ce{X^0})/V$ is the `true' concentration; redefining terms, we let $[\ce{X^0}]_\eff=[\ce{X^0}]/(\alpha + [\ce{X^0}]\gamma)$.

The system can be generalised to carriers for $n>1$ by letting the membrane be semi-permeable to $\{\ce{X^k}:0<k<n\}$, as these klona can participate as both reactants and products in either reaction. We can then write $\mathcal Z_\ominus=\sum_{k=0}^{n-1}e^{-\beta_\ominus k}$ for the partition function of $\Gamma_\ominus$ and similarly for $\Gamma_\oplus$. The effective concentrations of $\ce{X^k}$ in each compartment must be equal for $0<k<n$, i.e.\ $(e^{-\beta_\ominus k}/\mathcal Z_\ominus)[\ce{X}]_{\ominus\eff}=(e^{-\beta_\oplus k}/\mathcal Z_\oplus)[\ce{X}]_{\oplus\eff}$. As this must apply for each $0<k<n$, $(\beta_\oplus-\beta_\ominus)k$ must be constant with respect to $k$ and therefore $\beta_\oplus=\beta_\ominus=\beta$. Moreover, $[\ce{X}]_{\oplus\eff}/[\ce{X}]_{\ominus\eff}=\mathcal Z_\oplus/\mathcal Z_\ominus=e^{-\beta}$. Equilibrium occurs once again for $\beta=0$ with corresponding `true' concentrations $[\ce{X}]_\oplus=[\ce{X}]_\ominus=\tfrac12[\ce{X}]_{\text{total}}$. The free energy difference is then given by
\[\delta\Delta h=-\delta\beta=\frac{\delta[\ce{X}]_\oplus}{\tfrac12[\ce{X}]}\frac{2\alpha}{\alpha+\tfrac12\gamma[\ce{X}]}.\]
Finally, the resource charge is $Q=\bar k_\oplus[\ce{X}]_\oplus+\bar k_\ominus[\ce{X}]_\ominus$ and hence the proportional availability is given by
\begin{align*}
  \qty\Big|\frac{\delta Q}{Q}| &= \qty\Big|\frac{-\delta\beta\tfrac12[\ce{X}](\var_\oplus k+\var_\ominus k) + \delta[\ce{X}]_\oplus}{\tfrac12n[\ce{X}]}| \\
  &< 2b \qty\Big(\frac{\tfrac16(n^2-1)}{n} + \frac{\alpha+\tfrac12\gamma[\ce{X}]}{2\alpha n}) \\
  &< b ( \tfrac13n+\tfrac1{6n}  + \tfrac{1}{4\alpha n} ) ,
\end{align*}
recalling that $\gamma\equiv1/[\ce{X}]$.

As with the Bounded Carrier scheme we can let $n\sim1/b$ to ensure full resource availability, but we can also choose to let $\alpha\sim b$ with $n=1$. The constant $\alpha$ is proportional to the ratio of the membrane area to the system volume. A naive implementation would of course have $\alpha\sim b^2$ which, whilst ensuring full resource availability, would limit the number of mona able to interact with the resource scheme proportional to $A$. Instead, consider replacing the system's volume with a space-filling tubule (or perhaps a tubule-lattice) and partitioning the tubule along its length; this system is still an instance of the Isobaric Compartment scheme, but allows an $\alpha$ as large as $\sim1$ (if we wanted to allow all mona to interact with the resource system simultaneously). By making the tubule diameter $\sim b^{-1}$, a value of $\alpha\sim b$ can be achieved in which case the proportion of mona able to simultaneously interact with the resource system will be $\sim V^{5/6}$. This is, however, a time penalty of $\sim b^{-2}$ and hence is strictly worse than the Bounded Carrier scheme.

\section{Driving Unfavourable Reactions}
\label{sec:drive-unfav}

We conjecture that the results in the previous section are general; that is, no resource scheme can achieve better resource availability under the same constraints. In order to make use of a greater proportion of the resource pool, more free energy must be supplied and, ideally, the amount of free energy supplied should be the minimum necessary to avoid wasting the supply of negentropy. As alluded to in the introduction, we shall now consider general mona-klona interactions of arbitrary free energy cost.

\para{Reactions of Known or Bounded Cost}

Consider an arbitrary reaction $\ce{ $\sum_{i}\nu_{i}$X_i <=> $\sum_{i}\nu'_{i}$X_i }$ with forward reaction rate $\alpha$ and backward rate $\beta$. Recalling that the free energy change for the forward reaction is given by $\Delta h=\log\frac\alpha\beta$, it can be seen that $n>-\Delta h/2\arctanh b$ computational bias tokens are needed to drive the reaction forward in the case $\Delta h<0$; similarly, if $\Delta h>0$ then $n>\Delta h/2\arctanh b$ tokens are needed to drive the reaction backward. In addition, if we wish to couple such an unfavourable reaction to our computational mona then we should ensure by some means that the reaction does not occur in uncoupled isolation, such as raising the activation energy of the reaction or by employing a reaction mechanism that requires interaction with the mona. This kind of coupling abounds in nature, wherein a variety of bias systems are exploited. In fact, the $\ce{ATP}:\ce{ADP + P_i}$ bias system is itself generated from a number of other bias sources. One particular example involves the conversion of the free energy of a transmembrane proton gradient for which $n\gtrsim 4$; the molecular machine responsible for this coupling, ATP synthase, is a molecular motor which translocates 11 protons\footnote{The exact number varies between species, and even between different organelles of the same species, but is typically between 10 and 14.} across the membrane per revolution, in the process converting 3 molecules of $\ce{ADP}$ to $\ce{ATP}$, an entropically unfavourable process.

When $n$ is known and fixed (or, at least, bounded from above) an optimal approach is provided by the simple scheme that performs $n$ transitions between each unfavourable reaction. Indexing the mona states, $\ce{M_k}$, modulo $n+1$, the scheme can be written thus
\begin{align*}\begin{gathered}
  \ce{M_0 <=>>[$p\lambda$][$q\lambda$] $\cdots$ <=>>[$p\lambda$][$q\lambda$] M_k <=>>[$p\lambda$][$q\lambda$] $\cdots$ <=>>[$p\lambda$][$q\lambda$] M_n},\\
  \ce{M_n + $\sum_{i}\nu_{i}$X_i <=>>[$\alpha'$][$\beta'$] M_0 + $\sum_{i}\nu'_{i}$X_i}.
\end{gathered}\end{align*}
The reactions between $\ce{M_k}$ and $\ce{M_{k+1}}$ serve to increase the ratio $[\ce{M_n}]/[\ce{M_0}]$ such that it counteracts the ratio $\alpha'/\beta'=\alpha/\beta<1$. To solve for the steady state dynamics of this system, we write the currents
\begin{align*}
  S(\ce{M_k}\mapsto\ce{M_{k+1}}) &= p\lambda[\ce{M_k}] - q\lambda[\ce{M_{k+1}}] \\
  S(\ce{M_n}\mapsto\ce{M_0}) &= \alpha'[\ce{M_n}] - \beta'[\ce{M_0}]
\end{align*}
and use the fact that, at steady state, the currents must be equal. Moreover, the current will correspond to the net rate at which the target reaction is driven, and therefore we can find the optimal choice of $n>|\Delta h/2\arctanh b|$ to maximise this rate (as for excessively large $n$, the chain of intermediate monon states will dilute the proportion of transitions performing the desired reaction).

We first show that the system makes optimal use of the consumed bias tokens by calculating the value of $n$ for which the current vanishes, and hence for which the reaction is exactly balanced. For vanishing current, the principle of detailed balance applies and the steady state distribution is obtained simply by \smash{$[\ce{M_k}]/[\ce{M_0}]=(\frac pq)^k$} and \smash{$[\ce{M_0}]/[\ce{M_n}]=\alpha/\beta$}. This yields \smash{$n=\log\frac\beta\alpha/\log\frac pq=-\Delta h/2\arctanh b$}, precisely the threshold value of $n$. The attentive reader may point out that this value of $n$ is not necessarily integral, and therefore a cyclic chain of length $n+1$ may not exist; this is true, but we shall be seeking larger values of $n$ to increase the rate of the reaction and for which picking the nearest integral value shall be acceptable. Alternatively, one could consider the superposition of two such chains for $\lfloor n\rfloor$ and $\lceil n\rceil$ in the appropriate proportions.

To solve for non-vanishing current, we perform two telescopic sums
\begin{align*}
  nS &= \sum_{k=0}^{n-1} S(\ce{M_k}\mapsto\ce{M_{k+1}})  &
  \frac1pS \sum_{k=0}^{n-1} t^{-k} &= \sum_{k=0}^{n-1}( t^{-k} \lambda[\ce{M_k}] - t^{-k-1} \lambda[\ce{M_{k+1}}] )  \\
  &= b\lambda[\ce{M}] + q\lambda[\ce{M_0}] - p\lambda[\ce{M_n}], &
  S \frac{p^n-q^n}{p-q} &= p^n\lambda[\ce{M_0}] - q^n\lambda[\ce{M_n}],
\end{align*}
where $t\equiv\frac pq$, and then substitute $S(\ce{M_n}\mapsto\ce{M_0})$ for the current:
\begin{align*}\begin{gathered}
  b\lambda[\ce{M}] = (p\lambda+n\alpha')[\ce{M_n}] - (q\lambda+n\beta')[\ce{M_0}], \\
  \qty\Big(q^n\lambda + \frac{p^n-q^n}{p-q}\alpha')[\ce{M_n}] = \qty\Big(p^n\lambda + \frac{p^n-q^n}{p-q}\beta')[\ce{M_0}].
\end{gathered}\end{align*}
Hence, solving for $[\ce{M_0}]$ and $[\ce{M_n}]$, we can obtain the steady state current thus:
\begin{align*}
  S &= b\lambda[\ce{M}] \frac{q^n\beta' - p^n\alpha'}{(q^{n+1}-p^{n+1})\lambda + n(q^n\beta'-p^n\alpha')+(q\alpha'-p\beta')\frac1b(p^n-q^n)}.
\end{align*}
After some rearrangement, we find that the current is maximised when
\begin{align*}
  (t^{n/2}-t^{-n/2}\gamma)^2 &= \frac{q\log t}{b}(\gamma-1)(t\gamma-1)+\frac{q\lambda\log t}{\alpha}(t\gamma+1)
\end{align*}
where $\gamma\equiv\beta/\alpha$. For small bias, this reduces to $n=2n_0 + \bigOO{b}$ where $n_0=|\Delta h|/2\arctanh b$ is the threshold value. When $n$ is large, this gives a current $S \to [\ce{M}]\min(2b^2\lambda\tfrac\alpha\beta, 2b\alpha)$. This assumes that the intermediate monon transitions are not useful; if the intermediate transitions in fact perform useful work then this current can be multiplied by $n$ and, in the limit of large $n$, the current tends to $b\lambda$: the net transition rate for uncoupled mona.

\para{Reactions of Unknown or Unbounded Cost}

The preceding scheme suffers from a number of shortcomings due to the fact that $n$ is fixed; when $n_0$ is lower than that designed for, the system expends more bias tokens than necessary, whereas when it is greater the system will stall. Moreover the system is asymmetric with respect to direction of the driven reaction: if $\beta>\alpha$ then the forward reaction should be driven with $n$ tokens whilst the reverse reaction needs no additional bias, but this complicates the act of running a subroutine in reverse as each bias chain for the forward reaction must be removed, and a bias chain must be introduced for each instance of the reverse reaction. This is further complicated if the intermediate transitions along the chain perform useful computation. Whilst these problems are not insurmountable, it is possible to solve them all with a revised scheme in which $n$ is permitted to vary.

A scheme with variable $n$ is incompatible with the specific implementation of that for fixed $n$, and so we first migrate the mona chains into a distinct system of klona consisting of the species $\{\ce{X}^{(n)}_{k}:n,k\in\mathbb N\land 0\le k\le n\} \cup \{\ce{X}^{(-n)}_{-k}:n,k\in\mathbb N\land 0\le k\le n\}$ and subject to reactions \smash{$\ce{ X$^{(n)}_{k}$ <=>>[$p$][$q$] X$^{(n)}_{k+1}$ }$}, with steady state $[\ce{X}^{(n)}_{k}]/[\ce{X}^{(n)}_{0}]=(\frac pq)^k$. To demonstrate practicability, an abstract molecular realisation is provided in \Cref{app:seq-klona}. The idea is to sequester some of the reactants and products in an inactive state, with those remaining in an active state being at equilibrium (for an isenthalpic reaction, this manifests as equal concentrations). Representing the reactants by $\ce{A}$ and the products by $\ce{B}$, this effective sequestration behaviour is achieved by the reaction $\ce{ M$_i$ + A{:}X$^{(n_0)}_{n_0}$ <=>>[$p$][$q$] M$_{i+1}$ + B{:}X$^{(n_0)}_0$ }$ with $n_0=-\Delta h/2\arctanh b$. For fixed $n_0$, this scheme already solves all the aforementioned problems as the complexed system \smash{$\ce{A{:}X$^{(n_0)}_{n_0}$ <=> B{:}X$^{(n_0)}_0$}$} is an equilibrated abstraction over the underlying $\ce{A <=> B}$ reaction, and hence can be coupled directly to the computational mona system.

In the course of system operation, the value of $n_0$ will almost certainly vary. Introducing the sequestration klona for all values of $n\in\mathbb Z$, and letting $\alpha' \equiv \alpha\sum[\ce{ X$^{(n)}_n$ }]$ and $\beta' \equiv \beta\sum[\ce{ X$^{(n)}_0$ }]$, we see that the desired equilibrium condition is given by $\alpha'=\beta'$. Therefore, if the values of $\alpha'$ and $\beta'$ deviate from equilibrium then we should counteract this deviation by perturbing the $n$-distribution of the sequestration klona. Ideal kinetics are attained if this $n$-distribution converges to $\alpha'=\beta'$ exponentially fast, i.e.\ $\dot n\propto \beta'-\alpha'$. This target kinetics arises because $[\ce{ X$^{(n)}_n$ }]=(\frac pq)^n[\ce{ X$^{(n)}_0$ }]$ and $\frac pq>1$, hence increasing $n$ increases $\alpha'/\beta'$. A first attempt at implementation might be given by $\forall n\in\mathbb Z. \ce{ M$_i$ + A{:}X$^{(n)}_{n}$ <=>>[$p$][$q$] M$_{i+1}$ + B{:}X$^{(n-1)}_0$ }$, but this doesn't quite replicate the desired kinetics. Instead we separate the concerns of the computational bias and the reaction coupling thus:
  \begin{equation}\begin{aligned}
    \left.{\longce\begin{array}{c}
      \ce{ M_i <=>[$p$][$p$] M$_{i+\frac13}$ } \\
      \ce{ M_i <=>>[$p$][$q$] M$_{i+\frac13}$ }
    \end{array}}\right\} &&
    \ce{ M$_{i+\frac13}$ + A{:}X^{(n)}_{n} <=> M$_{i+\frac23}$ + B{:}X^{(n-1)}_0 } && 
    \left\{{\longce\begin{array}{c}
      \ce{ M$_{i+\frac23}$ <=>[$q$][$q$] M_{i+1} } \\
      \ce{ M$_{i+\frac23}$ <=>>[$p$][$q$] M_{i+1} }
    \end{array}}\right.
  \end{aligned}\label{eqn:dyn-coupling-scheme}\end{equation}
Assuming that at steady state $[\ce{M$_{i+\frac13}$}]=[\ce{M$_{i+\frac23}$}]\equiv\mu$, the coupled reaction then replicates the kinetics $\dot n=k\mu(\beta' - \alpha')$. For symmetry, we have introduced two intermediate mona with their inter-monon transitions used for coupling to the computational bias. In order to ensure that only a single token's worth of bias is consumed in the reaction, we also introduce two leak reactions such that the total free energy change is precisely that for a single token: $\Delta h=\log\frac{p+p}{q+p} + \log\frac{p+q}{q+q}=2\arctanh b$.

To verify that this system has the desired properties, we now solve for the steady state distribution. Employing detailed balance, we find for each $n$ that $[\ce{X$^{(n)}_{n}$}]/[\ce{X$^{(n)}_{n}$}]$
\begin{align*}
  \frac{[\ce{X$^{(n)}_{n}$}]}{[\ce{X$^{(n-1)}_{0}$}]} &= \frac\beta\alpha \equiv \qty\Big( \frac pq )^{n_0}  & &\implies &
  \frac{[\ce{X$^{(n)}_{0}$}]}{[\ce{X$^{(0)}_{0}$}]} &= \qty\Big( \frac pq )^{-\frac12(n+\frac12-n_0)^2+\frac12(\frac12-n_0)^2},
\end{align*}
where we have used the intra-$\ce{X$^{(n)}_{k}$}$ steady state distribution, $[\ce{X}^{(n)}_{k}]/[\ce{X}^{(n)}_{0}]=(\frac pq)^k$. This steady state distribution resembles a discrete Gaussian in $n$, and hence is relatively tightly centered about $n_0$; this is important to ensure good control over the distribution and that perturbing its mean value is practicable. As a sanity check, we can show that this distribution indeed gives $\beta'=\alpha'$:
\begin{align*}
  \frac{\beta'}{\alpha'} &= \qty\Big( \frac pq )^{n_0} \frac{\sum_{n\in\mathbb Z}( \frac pq )^{-\frac12(n+\frac12-n_0)^2+\frac12(\frac12-n_0)^2}}{\sum_{n\in\mathbb Z}( \frac pq )^{n}( \frac pq )^{-\frac12(n+\frac12-n_0)^2+\frac12(\frac12-n_0)^2}} \\
  &= \frac{\sum_{n\in\mathbb Z}( \frac pq )^{-\frac12(n+\frac12-n_0)^2+\frac12(\frac12+n_0)^2}}{\sum_{n\in\mathbb Z}( \frac pq )^{-\frac12(n-\frac12-n_0)^2+\frac12(\frac12+n_0)^2}} \\
  &= 1.
\end{align*}
Finally we seek to characterise the rates for the coupled reaction; first, we shall require the weight for each of the $\{\ce{X$^{(n)}_{k}$}:k\}$ subsystems:
\begin{align*}
  \sum_{k=0}^n[\ce{X$^{(n)}_{k}$}] &= \frac{[\ce{X$^{(n)}_{0}$}]}{t-1}\begin{cases}
    t^{n+1}-1 & n \ge 0 \\
    t-t^{n} & n \le 0
  \end{cases} 
\end{align*}
where $t=\frac pq$. Next, we use $t^nt^{-\frac12(n+\frac12-n_0)^2+\frac12(\frac12-n_0)^2}=t^{-\frac12(n-\frac12-n_0)^2+\frac12(\frac12+n_0)^2}$ to find the normalisation,
\begin{align*}
  [\ce{X$^{(0)}_{0}$}]^{-1} = 1 &+ \textstyle\frac1{t-1}[t\sum_{n<0}-\sum_{n>0}\,]t^{-\frac12(n+\frac12-n_0)^2+\frac12(\frac12-n_0)^2} \\
  &+ \textstyle\frac1{t-1}[t\sum_{n>0}-\sum_{n<0}\,]t^{-\frac12(n-\frac12-n_0)^2+\frac12(\frac12+n_0)^2}.
\end{align*}
This is non-trivial for general $n_0$, so we restrict to the two limiting cases of $n_0=0$ and $|n_0|\gg1$. In the former case, the normalisation simplifies to $[\ce{X$^{(0)}_{0}$}]=\frac{b}{1+b}$ and so we have $\beta'=\alpha'=\frac{\alpha b}{1+b}\sum_{n\in\mathbb Z}t^{\frac18-\frac12(n-\frac12)^2}$. The sum is bounded from below by the integral $2\int_0^\infty t^{-\frac12(s+\frac12)^2}\dd{s}$, and from above by $2\int_0^\infty t^{-\frac12(s-\frac12)^2}\dd{s}$, from which we obtain $\alpha'=\alpha\sqrt{\pi b}+\bigOO{b}$. In the latter case, suppose $n_0\gg1$ is positive (i.e.\ $\beta\gg\alpha$). The normalisation will then be dominated by the sums over $n>0$ and so has asymptotically exact approximation
\begin{align*}
  [\ce{X$^{(0)}_{0}$}]^{-1} &=\textstyle \frac{t}{t-1}\sum_{n\in\mathbb Z}t^{-\frac12(n-\frac12-n_0)^2+\frac12(\frac12+n_0)^2}-\frac1{t-1}\sum_{n\in\mathbb Z}t^{-\frac12(n+\frac12-n_0)^2+\frac12(\frac12-n_0)^2} \\
  &=\textstyle \frac 1b[pt^{\frac12(\frac12+n_0)^2}-qt^{\frac12(\frac12-n_0)^2}]\sum_{n\in\mathbb Z}t^{-\frac12(n+\frac12-n_0)^2},
\end{align*}
hence $\alpha'$ reduces to $2b\alpha + \bigOO{b^2}$ in the case of $\beta\gg\alpha$, or $2b\beta + \bigOO{b^2}$ in the case of $\alpha\gg\beta$. In the ideal case, in which only the excess $\ce{B}$ (when $\beta>\alpha$) is sequestered, we should have $\alpha'=\beta'=\alpha$. The consequence of the factor $(\sqrt{\pi b},2b)$ is to further suppress the current through the coupled transition. As the uncoupled monon transitions have current $\propto b$, this results in an overhead of $\sim b^{-3/2}$ in the best case and $\sim b^{-2}$ in the worst. This overhead is comparable to the time penalties we found for communication and synchronisation interactions in \Cref{chap:revii}, as for the coupled reaction to proceed both the monon and the complexed klonon must be in receptive states. Drawing further from the analogy, one may wonder whether reactions involving multiple klona are subject to even greater overheads; fortunately, this is not the case: provided only one sequestration klonon participates, by choosing a single representative from each of the reactant klona and product klona to engage in complexation with the sequestration klona, then the results stand.

\section{Molecular Counter Implementation}
\label{app:seq-klona}

Here we present example implementations of the counter klona $\ce{X^{(n)}_k}$ used in \Cref{sec:drive-unfav}. These implementations are abstract, but point towards a viable molecular realisation.

\para{Unary Representation}

\begin{listing}
  \centering
  {\longce\fbox{\begin{minipage}{.85\textwidth}\begin{align*}
    \ce{ $\llb\ctrTok\sigma\mkrAnyR\tokSucc\ctrTok\sigma'\rrb$ &<=>>[$p$][$q$] $\llb\ctrTok\sigma\tokSucc\mkrAnyR\ctrTok\sigma'\rrb$ } && \ctrRule{intra--bias} \\
    \ce{ M$_{i+\frac13}$ + A{:}$\llb\ctrTok\sigma\tokSucc\tokSucc\mkrPosR\rrb$ &<=> $\llb\ctrTok\sigma\tokSucc\mkrPosL\rrb${:}B + M$_{i+\frac23}$ } && \ctrRule{pos--step} \\
    \ce{ M$_{i+\frac13}$ + A{:}$\llb\tokSucc\mkrPosR\rrb$ &<=> $\llb\mkrZero\rrb${:}B + M$_{i+\frac23}$ } && \ctrRule{pos--base} \\
    \ce{ M$_{i+\frac13}$ + A{:}$\llb\mkrZero\rrb$ &<=> $\llb\mkrNegL\tokSucc\rrb${:}B + M$_{i+\frac23}$ } && \ctrRule{neg--base} \\
    \ce{ M$_{i+\frac13}$ + A{:}$\llb\mkrNegR\tokSucc\ctrTok\sigma\rrb$ &<=> $\llb\mkrNegL\tokSucc\tokSucc\ctrTok\sigma\rrb${:}B + M$_{i+\frac23}$ } && \ctrRule{neg--step} \\
  \end{align*}\end{minipage}}}
  \caption[The abstract molecular reaction scheme implementing unary counters.]{%
    The abstract molecular reaction scheme implementing $\ce{ X^{(n)}_k <=> X^{(n)}_{k+1} }$ and $\ce{ X^{(n)}_n <=> X^{(n-1)}_0 }$ for the unary representation case.
    As described in \Cref{app:seq-klona}, the species $\ce{ X^{(n)}_k }$ is represented by a polymer of $n$ copies of the $\tokSucc$ monomer.
    These polymers are delimited by $\llb$ and $\rrb$, and are intercalated with a `marker' monomer, $\protect\mkrPosR$ or $\protect\mkrNegR$.
    The position of the marker monomer corresponds to the value of $k$ as explained in the text, and the \ctrRule{intra--bias} reaction tends to push it in the direction of its arrow.
    Note that $\protect\mkrAnyR$ represents either marker monomer and $\sigma$, $\sigma'$ represent arbitrary strings of $\protect\tokSucc$ monomers.}
  \label{lst:ctr-unary}
\end{listing}

The simplest implementation of counters is given by representing numbers in unary, i.e.\ a Peano axiomatic approach. For a natural integer $n\in\mathbb N$, we propose a polymeric representation consisting of $n$ repeated $\tokSucc$ monomers and represented thus,
\begin{align*}
  \llb\underbrace{\tokSucc\tokSucc\cdots\tokSucc\tokSucc}_{n}\rrb.
\end{align*}
We can then introduce a subcounter $0\le k\le n$ by incorporating a marker monomer, $\mkrPosR$, placed between the monomers and which can migrate between the $n+1$ possible positions. Moreover, we can handle non-positive values of $n$ and $k$ via the alternative marker $\mkrNegR$. Whilst not strictly necessary, we also differentiate the case $n=k=0$ with a third marker monomer, $\mkrZero$. Example polymers are given for $n=-4,0,4$:
\begin{align*}
  \ce{X^{(4)}_0} &= \llb\mkrPosR\tokSucc\tokSucc\tokSucc\tokSucc\rrb, & & & \ce{X^{(-4)}_{0}} &= \llb\tokSucc\tokSucc\tokSucc\tokSucc\mkrNegR\rrb, \\
  \ce{X^{(4)}_1} &= \llb\tokSucc\mkrPosR\tokSucc\tokSucc\tokSucc\rrb, & & & \ce{X^{(-4)}_{-1}} &= \llb\tokSucc\tokSucc\tokSucc\mkrNegR\tokSucc\rrb, \\
  \ce{X^{(4)}_2} &= \llb\tokSucc\tokSucc\mkrPosR\tokSucc\tokSucc\rrb, & \ce{X^{(0)}_0} &= \llb\mkrZero\rrb, & \ce{X^{(-4)}_{-2}} &= \llb\tokSucc\tokSucc\mkrNegR\tokSucc\tokSucc\rrb, \\
  \ce{X^{(4)}_3} &= \llb\tokSucc\tokSucc\tokSucc\mkrPosR\tokSucc\rrb, & & & \ce{X^{(-4)}_{-3}} &= \llb\tokSucc\mkrNegR\tokSucc\tokSucc\tokSucc\rrb, \\
  \ce{X^{(4)}_4} &= \llb\tokSucc\tokSucc\tokSucc\tokSucc\mkrPosR\rrb, & & & \ce{X^{(-4)}_{-4}} &= \llb\mkrNegR\tokSucc\tokSucc\tokSucc\tokSucc\rrb.
\end{align*}
Notice the mirrored interpretation of marker positions as values of $k$ for positive and negative values of $n$; this is intentional and serves to simplify the implementation of the intra-$n$ biasing reactions. It should also be noted that these polymers are achiral and so their mirror images, e.g.\ $\ce{X^{(4)}_1} = \llb\tokSucc\tokSucc\tokSucc\mkrPosL\tokSucc\rrb$, are equivalent representations of the same polymer. This property greatly simplifies the implementation of the $\ce{ X^{(n)}_n <=> X^{(n-1)}_0 }$ reactions. Altogether, this implementation can be achieved with five reactions as listed in \Cref{lst:ctr-unary}. In fact, this scheme can be implemented with just four reactions by replacing $\mkrZero$ with $\mkrPosR$ and combining reactions \ctrRule{pos-step} and \ctrRule{pos-base}, but we prefer the given implementation for its symmetry. More concretely, these reactions engender the following bi-infinite Markov chain:
{\mhchemoptions{arrow-min-length=1.5em}\begin{multline*}
  \cdots \ce{<=>} \llb\mkrPosR\tokSucc\tokSucc\rrb
  \ce{<=>>} \llb\tokSucc\mkrPosR\tokSucc\rrb \ce{<=>>} \llb\tokSucc\tokSucc\mkrPosR\rrb
  \ce{<=>} \llb\mkrPosR\tokSucc\rrb \ce{<=>>} \llb\tokSucc\mkrPosR\rrb
  \ce{<=>} \llb\mkrZero\rrb \\ \llb\mkrZero\rrb
  \ce{<=>} \llb\mkrNegL\tokSucc\rrb \ce{<<=>} \llb\tokSucc\mkrNegL\rrb
  \ce{<=>} \llb\mkrNegL\tokSucc\tokSucc\rrb \ce{<<=>} \llb\tokSucc\mkrNegL\tokSucc\rrb
  \ce{<<=>} \llb\tokSucc\tokSucc\mkrNegL\rrb \ce{<=>} \cdots.
\end{multline*}}

\para{Binary Representation}

{\def\AutLabel#1{{\small\textsc{#1}}}
\def\autLabel#1{\makebox[0pt][c]{\AutLabel{#1}}}

\begin{figure}
  \centering
  \begin{tikzpicture}[shorten >=1pt,node distance=2.5cm,on grid,auto]
    \node[state,initial] (succ) {\autLabel{inc}};
    \node[state,accepting] (halt) [right=of succ] {\autLabel{halt}};
    \path[->] (succ) edge [loop above] node {$1\mapsto0$; \emph{left}} ()
                     edge [bend left, above] node {$0\mapsto1$} (halt)
                     edge [bend right, below] node {$\varnothing\mapsto1$} (halt);
  \end{tikzpicture}
  \caption[The state diagram for a Turing Machine implementing incrementation of a natural number in binary positional notation.]{The state diagram for a Turing Machine~\cite{turing-machine} implementing incrementation of a natural number in binary positional notation. The alphabet is $\{\varnothing,0,1\}$, with $\varnothing$ indicating a blank square; numbers may be provided with any number of leading 0s, and the initial position should be the least significant digit.}
  \label{fig:tm-succ}

  \vspace{1.5em}
  \begin{tikzpicture}[shorten >=1pt,node distance=2cm,on grid,auto]
     \node[state,accepting] (init) {\autLabel{init}};%
     \node[state] (carry) [right=1.5cm of init] {\autLabel{car}};%
     \node[state] (inc) [right=1.5cm of carry] {\autLabel{inc}};%
     \node[state] (end1) [above=of inc] {\autLabel{end\smash{$_1$}}};%
     \node[state] (end2) [right=of end1] {\autLabel{end\smash{$_2$}}};%
     \node[state] (dne1) [below=of inc] {\autLabel{end\smash{$_1'$}}};%
     \node[state] (dne2) [right=of dne1] {\autLabel{end\smash{$_2'$}}};%
     \node[state] (inc2) [right=of inc] {\autLabel{inc\smash{$'$}}};%
     \node[state] (inc3) [right=of inc2] {\autLabel{inc\smash{$''$}}};%
     \node[state] (carry2) [right=1.5cm of inc3] {\autLabel{car\smash{$'$}}};%
     \node[state] (carry3) [right=1.5cm of carry2] {\autLabel{car\smash{$''$}}};%
     \node[state,accepting] (term) [right=1.5cm of carry3] {\autLabel{term}};
     \node (anchor) [right=1.5cm of term] {};
    \path[->] (init) edge node {$\varnothing$} (carry)
              (carry) edge[bend left, above] node {\emph{left}} (inc)
              (inc) edge[bend left, below] node {$1 \mapsto 0$} (carry)
                    edge node [right] {$\varnothing \mapsto 1$} (end1)
                    edge node [right] {$0 \mapsto 1$} (dne1)
              (end1) edge node [above] {\emph{left}} (end2)
              (end2) edge node [right] {$\varnothing$} (inc2)
              (dne1) edge node [below] {\emph{left}} (dne2)
              (dne2) edge [bend left] node [right] {$0$} (inc2)
                     edge [bend right] node [left] {$1$} (inc2)
              (inc2) edge node [above] {\emph{right}} (inc3)
              (inc3) edge node {$1$} (carry2)
              (carry2) edge[bend left, above] node {\emph{right}} (carry3)
              (carry3) edge [bend left, below] node {$0$} (carry2)
                      edge node {$\varnothing$} (term);
  \end{tikzpicture}
\\[-3em]
  \begin{tikzpicture}[shorten >=1pt,node distance=2cm,on grid,auto]
     \node[state,accepting] (init) {\autLabel{init}};%
     \node[state] (carry) [left=1.5cm of init] {\autLabel{car}};%
     \node[state] (inc) [left=1.5cm of carry] {\autLabel{inc}};%
     \node[state] (end1) [above=of inc] {\autLabel{end\smash{$_1$}}};%
     \node[state] (end2) [left=of end1] {\autLabel{end\smash{$_2$}}};%
     \node[state] (dne1) [below=of inc] {\autLabel{end\smash{$_1'$}}};%
     \node[state] (dne2) [left=of dne1] {\autLabel{end\smash{$_2'$}}};%
     \node[state] (inc2) [left=of inc] {\autLabel{inc\smash{$'$}}};%
     \node[state] (inc3) [left=of inc2] {\autLabel{inc\smash{$''$}}};%
     \node[state] (carry2) [left=1.5cm of inc3] {\autLabel{car\smash{$'$}}};%
     \node[state] (carry3) [left=1.5cm of carry2] {\autLabel{car\smash{$''$}}};%
     \node[state,accepting] (term) [left=1.5cm of carry3] {\autLabel{term}};
     \node (anchor) [left=1.5cm of term] {};
    \path[->] (carry) edge node [above] {$\varnothing$} (init)
                      edge [bend left, below] node {$1 \mapsfrom 0$} (inc)
              (inc) edge [bend left, above] node {\emph{right}} (carry)
              (end1) edge node [left] {$\varnothing \mapsfrom 1$} (inc)
              (end2) edge node [above] {\emph{right}} (end1)
              (dne1) edge node [left] {$0 \mapsfrom 1$} (inc)
              (dne2) edge node [below] {\emph{right}} (dne1)
              (inc2) edge node [left] {$\varnothing$} (end2)
                     edge [bend left] node [left] {$0$} (dne2)
                     edge [bend right] node [right] {$1$} (dne2)
              (inc3) edge node [above] {\emph{left}} (inc2)
              (carry2) edge node [above] {$1$} (inc3)
                       edge [bend left, below] node {$0$} (carry3)
              (carry3) edge [bend left, above] node {\emph{left}} (carry2)
              (term) edge node [above] {$\varnothing$} (carry3);
  \end{tikzpicture}
  \\[1em]
  \caption[The state diagram for a Reversible Turing Machine (and its reverse) implementing incrementation (resp.\ decrementation) of a natural number in binary positional notation.]{The state diagram for a Reversible Turing Machine~\cite{bennett-tm} (and its reverse) implementing incrementation (resp.\ decrementation) of a natural number in binary positional representation. The alphabet is $\{\varnothing,0,1\}$, with $\varnothing$ indicating a blank square; numbers must be provided in the described prefix-free form, and the initial position should be the (empty) square to the immediate right of the least significant digit. The box-shaped `subroutine' is used to handle the case when a new digit must be prepended; as this is a reversible branch, we must ensure that the converged state \AutLabel{inc$'$} is able to uniquely determine the branch from whence it came.}
  \label{fig:rtm-succ}
\end{figure}}

Whilst particularly simple, the unary representation is very space inefficient; worse, incrementing and decrementing $n$, for the purpose of modulating access to a sparse resource/disequilibrium reaction, will itself require interaction with a resource pool in order to acquire and release $\circ$ monomers. Fortunately an exponentially more compact representation exists in the form of positional notations such as binary, with spatial complexity logarithmic in $|n|$ and hence needing far fewer interactions with its respective monomer pool.

Unlike with the unary case, we cannot exploit the structure of the representation of $n$ to represent $k$. Instead we shall need to store both explicitly. We will also require the ability to recognise the states $k=0$ and $k=n$, both for determining whether the coupled reaction $\ce{A <=> B}$ should proceed and for ensuring the reactions that generate the intra-$n$ distribution do not produce illegal values of $k$. These two states correspond respectively to all the digits of $k$ being $0$ or being identical to those of $n$. Since $|k|<|n|$, its minimal representation will be no larger than that of $n$ and hence a convenient structure for encoding the pair $(n,k)$ with the ability to easily check the condition $n=k$ is given by a `zipped' double-stranded polymer which we denote by $\llb\begin{smallarray}{c}n\\k\end{smallarray}\rrb$; for example, in base 2 the pair $(11,6)$ would be given by
\[ \begin{cB}1 & 0 & 1 & 1 \\ 0 & 1 & 1 & 0 \\\end{cB}. \]
In contrast to the unary case, this representation \emph{does} have an intrinsic polarity, with the leftmost digits being the most significant. To ensure logical reversibility of the forthcoming reaction scheme, the polymer representation must be unique for each $(n,k)$ pair, but positional representations admit a countably infinite equivalence class for each integer (e.g.\ 3 has base-10 representations $\{3,03,003,\ldots\}$). We resolve this ambiguity by trimming all leading zeros of $n$ and trimming $k$ to the same length. This has the consequence that $(0,0)$ has the `empty' representation $\llb\rrb$; it is certainly possible to introduce a special case for $(0,0)$ of $\llb\begin{smallarray}{c}0\\0\end{smallarray}\rrb$ for {\ae}sthetic reasons, but this would significantly complicate the reaction scheme for no other benefit. 

Thus far, we have only considered the representation for non-negative $k,n\in\mathbb N$; there are a number of approaches one could take to extend the representation, from simply incorporating a \emph{sign} trit, $\{-,0,+\}$, similar to the approach for the unary case, to a two's complement approach. The two's complement approach is particularly attractive as it would require only minimal adjustments to the non-negative reaction scheme. The two's complement representation of negative numbers is realised by taking the equivalence class for the integer, e.g.\ $\{11,011,0011,\ldots\}$ for 3, and complementing each bit to obtain, e.g., $\{00,100,1100,\ldots\}$. That is, each negative integer has an infinite prefix of 1s instead of 0s as for the non-negative case. This approach is used in most conventional CPUs because no special-casing for negative numbers is necessary: adding a negative number $m$, with two's complement $m'$, to a positive number $n$ is equivalent to adding the positive numbers $m'$ and $n$. As a simple justification for this, consider decrementing $1000\cdots000$. We will need to `borrow' a 1 from each digit to the left, eventually obtaining $111\cdots111$. In the limit of infinitely many 0s, this results in infinitely many 1s. For brevity we shall restrict our attention to the non-negative case, but extending to the full domain is a straightforward---if tedious---exercise.

Whilst recognising the cases $k=0$ and $k=n$ is not difficult in our representation, comparing $\sim\log n$ digits each time we wish to execute the coupled reaction is not ideal. By a slight increase in the complexity of the $\ce{ X^{(n)}_k <=>> X^{(n)}_{k+1} }$ and $\ce{ X^{(n)}_n <=> X^{(n-1)}_0 }$ algorithms, we can render these checks trivial. In particular, we augment the polymer with two bits of state for each monomer pair which we call $Q$, for e$Q$uality, and $Z$, for $Z$ero. As biochemical precedent, compare with phosphorylation sites. The function of these states is best understood by example; for $n=11$, the $Q$-$Z$ states for each of $k=0,1,8,11$ are given by
\begin{align*}
  \begin{cB}1 & 0 & 1 & 1 \\ 0 & 0 & 0 & 0 \\\ctrBinInv{1-4}\end{cB}, &&
  \begin{cB}1 & 0 & 1 & 1 \\ 0 & 0 & 0 & 1 \\\ctrBinInv{1-3}\end{cB}, &&
  \begin{cB}1 & 0 & 1 & 1 \\\ctrBinInv{1-2} 1 & 0 & 0 & 0 \\\end{cB}, &&
  \begin{cB}1 & 0 & 1 & 1 \\\ctrBinInv{1-4} 1 & 0 & 1 & 1 \\\end{cB}.
\end{align*}
Namely, a prefix of zeroes is marked by the $Z$ state (bottom lines) and a prefix matching $n$ is marked by the $Q$ state (middle lines), and these are easily achieved by implementing two local invariants for each digit pair $\begin{smallarray}{c}x\\y\end{smallarray}$:
\begin{enumerate}
  \item[$(Q)$] The $Q$ state is \emph{on} if and only if $y=x$ and the $Q$ state of the digit pair to its left (if it exists) is also \emph{on}.
  \item[$(Z)$] The $Z$ state is \emph{on} if and only if $y=0$ and the $Q$ state of the digit pair to its left (if it exists) is also \emph{on}.
\end{enumerate}
It can be seen that these two states are orthogonal as the leading bit of $n$ must be 1, with the possible exception of $(0,0)$ where it makes sense to define $Q$ and $Z$ to be both \emph{on} despite the lack of lack of monomers on which to mark said states.

\begin{listing}
  \centering
  \fbox{\begin{minipage}{.5\linewidth}\centering
  \def\quo#1{{\text{`}#1}}%
  \def\inc{\infix{\atom{Inc}}}%
  \def\eq#1#2{\ooalign{\hss\phantom{$#1$}\hss\cr\hss$#2$\hss}}%
  \vspace{.5em}
  \begin{align*}
    & \rlap{${[]}$}\phantom{[\quo{\eq10}~{x}~{\cdot}~{\cursivebS}]}~\inc~{[\quo{\eq01}]}{;} \\
    & {[\quo{\eq10}~{x}~{\cdot}~{\cursivebS}]}~\inc~{[\quo{\eq01}~{x}~{\cdot}~{\cursivebS}]}{;} \\
    & {[\quo{\eq01}~\phantom{x}~{\cdot}~{\cursivebS}]}~\inc~{[\quo{\eq10}~\phantom{x}~{\cdot}~{\cursivebS'}]}{:} \\
    &\qquad {\cursivebS}~\inc~{\cursivebS'}{.}
  \end{align*}
  \vspace{-.2em}
  \end{minipage}}
  \caption[A recursive \alethe\ program implementing incrementation of a natural number in binary positional notation.]{A recursive $\aleph$/\alethe\ (see \Cref{chap:aleph}) program implementing incrementation of a natural number in binary positional representation. Numbers are provided in the described prefix-free form as a list of digits $\{0, 1\}$ with the least-significant bit first. The second case pattern matches on two bits in order to ensure its output pattern is distinct from that of the first case, and this is where a program will stall should a number be provided which is not in prefix-free form.}
  \label{lst:alethe-succ}
\end{listing}

Before implementing a reaction scheme, we first require an algorithmic understanding of reversibly incrementing/decrementing integers in positional notation. As it is possible (and indeed trivial) to do so for integers in unary notation as demonstrated by \Cref{lst:ctr-unary} we should expect it to also be possible for positional notation, and indeed it is. The irreversible algorithm for incrementing a binary number is very simple (\Cref{fig:tm-succ}): if there are any trailing 1 bits then we carry the 1 (flipping the intermediate 1 bits to 0) until we reach the first 0, which we flip to 1. If we reach the end of the number, we prepend a new 1 bit. Making this algorithm reversible is conceptually simple, but expressing it in the form of a Reversible Turing Machine (RTM, as introduced by \textcite{bennett-tm}) is somewhat involved. An example implementation is provided in \Cref{fig:rtm-succ}, and can be broken down into three stages: (1) carrying; (2) incrementing the least significant 0 bit, or prepending a fresh 1 bit, followed by the logic to reversibly join these two branches; (3) `reverse' carrying, in which we return to the starting position and use the fact that the suffix we are processing takes the form $100\cdots000$ to ensure we can logically reverse this stage. It can be readily checked that this scheme as presented is reversible, and indeed we provide the mechanical reversal of the RTM which may be seen to implement the operation of reversibly decrementing a positive natural number. To demonstrate that the algorithm is in fact conceptually simple, an equivalent recursive program written in a higher level reversible language is shown in \Cref{lst:alethe-succ}.

Finally, the reaction schemes for $\ce{ X^{(n)}_k <=>> X^{(n)}_{k+1} }$ and $\ce{ X^{(n)}_n <=> X^{(n-1)}_0 }$ are made manifest in \Cref{lst:ctr-bin-k,lst:ctr-bin-n} respectively. In addition to implementing the reversible incrementation procedure in an abstract molecular sense, these schemes must also incorporate the relevant logic for ensuring $k$ is appropriately bounded and that the $Q$ and $Z$ invariants are preserved (although these may be transiently broken within the intermediate states). Furthermore the logic for the second reaction (which we choose to implement in reverse, i.e.\ $\ce{ X^{(n)}_0 <=> X^{(n+1)}_{n+1} }$, per the freedom reversibility affords us) must reversibly copy the value of $n$ into $k$, and thus the second scheme must read the entire length of the polymer. As for the case of $n\equiv2^p-1$ in $n\mapsto n+1$ (resp. $n+1\mapsto n$), the scheme recruits (resp.\ releases) a pair of binary monomers ($\mathbb B_2$) which will be provided by some resource pool mediated by its own sequestration klona.
An additional complication is the prevention of multiple initiation events: as the scheme operates via local rules, it is not possible to tell whether a molecule is in a stable state $\ce{ X^{(n)}_k }$ or if it is currently undergoing a transition. To address this, a `blocked' terminus is used to represent intermediate states. Symbolically, this takes the form $\llb\,\cdot\,\ooalign{\hss\hspace{-0.2ex}$\bullet$\hss\cr$\rrbracket$}$ instead of $\llb\,\cdot\,\rrb$.

\doublepage{%
\begin{listing}[!p]
  \centering
  {\longce\fbox{\begin{minipage}{0.85\textwidth}\begin{align*}
    \begin{cB}\sigma \\ \tau \\\ctrBinVar{1-1}\end{cB} &\ce{<=>>} \begin{cB*}\sigma & \ctrBinKL \\ \tau & \\\ctrBinVar{1-1}\end{cB*} && \ctrRule{$k$--init} \\
    \begin{cB*}\sigma & x & \ctrBinKL & \sigma' \\ \tau & 1 & & \tau' \\\ctrBinVar{1-1}\end{cB*} &\ce{<=>>} \begin{cB*}\sigma & \ctrBinKL & x & \sigma' \\ \tau & & 0 & \tau' \\\ctrBinVar{1-1}\end{cB*} && \ctrRule{$k$--carry} \\
    \begin{cB*}\sigma & x & 1 & \ctrBinKL & \sigma' \\\ctrBinInv{1-2} \tau & x & 0 & & \tau' \\\end{cB*} &\ce{<=>>} \begin{cB*}\sigma & x & 1 & \ctrBinKR & \sigma' \\\ctrBinInv{1-3} \tau & x & 1 & & \tau' \\\end{cB*} && \ctrRule{$k$--inc$_1$} \\
    \begin{cB*}\sigma & x & y & \ctrBinKL & \sigma' \\ \tau & 0 & 0 & & \tau' \\\ctrBinInv{1-3}\end{cB*} &\ce{<=>>} \begin{cB*}\sigma & x & y & \ctrBinKR & \sigma' \\ \tau & 0 & 1 & & \tau' \\\ctrBinInv{1-2}\end{cB*} && \ctrRule{$k$--inc$_2$} \\
    \begin{cB*}\sigma & x & y & \ctrBinKL & \sigma' \\\ctrBinVar{1-1} \tau & z & 0 & & \tau' \\\end{cB*} &\ce{<=>>} \begin{cB*}\sigma & x & y & \ctrBinKR & \sigma' \\\ctrBinVar{1-1} \tau & z & 1 & & \tau' \\\end{cB*} && \ctrRule{$k$--inc$_3$} \\
    \begin{cB*}1 & \ctrBinKL & \sigma' \\ 0 & & \tau' \\\ctrBinInv{1-1}\end{cB*} &\ce{<=>>} \begin{cB*}1 & \ctrBinKR & \sigma' \\\ctrBinInv{1-1} 1 & & \tau' \\\end{cB*} && \ctrRule{$k$--inc$_4$} \\
    \begin{cB*}\sigma & x & \ctrBinKR & 0 & \sigma' \\\ctrBinInv{1-2} \tau & x & & 0 & \tau' \\\end{cB*} &\ce{<=>>} \begin{cB*}\sigma & x & 0 & \ctrBinKR & \sigma' \\\ctrBinInv{1-3} \tau & x & 0 & & \tau' \\\end{cB*} && \ctrRule{$k$--carry$_1'$} \\
    \begin{cB*}\sigma & x & \ctrBinKR & 1 & \sigma' \\\ctrBinInv{1-2} \tau & x & & 0 & \tau' \\\end{cB*} &\ce{<=>>} \begin{cB*}\sigma & x & 1 & \ctrBinKR & \sigma' \\\ctrBinInv{1-2} \tau & x & 0 & & \tau' \\\end{cB*} && \ctrRule{$k$--carry$_2'$} \\
    \begin{cB*}\sigma & x & \ctrBinKR & z & \sigma' \\ \tau & y & & 0 & \tau' \\\end{cB*} &\ce{<=>>} \begin{cB*}\sigma & x & z & \ctrBinKR & \sigma' \\ \tau & y & 0 & & \tau' \\\end{cB*} && \ctrRule{$k$--carry$_3'$} \\
    \begin{cB*}\sigma & \ctrBinKR \\\ctrBinVar{1-1} \tau & \\\end{cB*} &\ce{<=>>} \begin{cB}\sigma \\\ctrBinVar{1-1} \tau \\\end{cB} && \ctrRule{$k$--term} \\
  \end{align*}\end{minipage}}}
  \caption[Part I of the abstract molecular reaction scheme implementing binary counters.]{The abstract molecular reaction scheme implementing $\ce{ X^{(n)}_k <=>> X^{(n)}_{k+1} }$ for the binary representation case. To only use a single bias token's worth of free energy, one can make use of a similar scheme to that used in \Cref{eqn:dyn-coupling-scheme}, and perhaps making some of the reactions in this scheme unbiased. Alternatively one can accept the use of many tokens, which has the advantage of increasing the rate of the coupled reaction due to the effectively higher bias, but has the disadvantage of slower convergence of the sequestration klona to steady state---potentially interfering with correct operation of the coupled reaction during periods of high flux. In this scheme, the $\langle$ and $\rangle$ symbols represent klona that mark the current state and progress of the overall reaction. It is of note that, in the \ctrRule{$k$--carry} rule, there is no provision for $Q$ to be set on the bit-pair to the left of $\begin{smallarray}{c}x\\1\end{smallarray}$, as such a state would imply we are trying to increment $k\ge n$ and is thus illegal.}
  \label{lst:ctr-bin-k}
\end{listing}}{%
\begin{listing}[!p]
  \centering
  {\longce\fbox{\begin{minipage}{0.85\textwidth}\begin{align*}
    \ce{ M$_{i+\frac23}$ + B{:} }\begin{cB}\sigma \\ \tau \\\ctrBinInv{1-1}\end{cB} &\ce{<<=>} \begin{cB*}\sigma & \ctrBinNL \\ \tau & \\\ctrBinInv{1-1}\end{cB*}\ce{ {:}{\abmTS} } && \ctrRule{$n$--init} \\
    \ce{ {\abmTS}{:} }\begin{cB*}\sigma & 1 & \ctrBinNL & \sigma' \\\ctrBinInv{4-4} \tau & 0 & & \tau' \\\ctrBinInv{1-2}\end{cB*} &\ce{<=>} \begin{cB*}\sigma & \ctrBinNL & 0 & \sigma' \\\ctrBinInv{3-4} \tau & & 0 & \tau' \\\ctrBinInv{1-1}\end{cB*}\ce{ {:}{\abmTS} } && \ctrRule{$n$--carry}\\
    \ce{ {\abmTS}{:} }\begin{cB*}\sigma & x & 0 & \ctrBinNL & \sigma' \\\ctrBinInv{5-5} \tau & 0 & 0 & & \tau' \\\ctrBinInv{1-3}\end{cB*} &\ce{<=>} \begin{cB*}\sigma & x & \ctrBinPL & 1 & \sigma' \\\ctrBinInv{4-5} \tau & 0 & & 1 & \tau' \\\ctrBinInv{1-2}\end{cB*}\ce{ {:}{\abmTS} } && \ctrRule{$n$--inc$_1$}\\
    \ce{ {\abmTS}{:} }\begin{cB*}\ctrBinNL & \sigma' \\\ctrBinInv{2-2} & \tau' \\\end{cB*} &\smash{\ceArrAdd{<=>}{$\mathbb B_2$}} \begin{cB*}\ctrBinPL & 1 & \sigma' \\\ctrBinInv{2-3} & 1 & \tau' \\\end{cB*}\ce{ {:}{\abmTS} } && \ctrRule{$n$--inc$_2$}\\
    \ce{ {\abmTS}{:} }\begin{cB*}\sigma & x & \ctrBinPL & y & \sigma' \\\ctrBinVar{5-5} \tau & 0 & & 0 & \tau' \\\ctrBinInv{1-2}\end{cB*} &\ce{<=>} \begin{cB*}\sigma & \ctrBinPL & x & y & \sigma' \\\ctrBinVar{5-5} \tau & & 0 & 0 & \tau' \\\ctrBinInv{1-1}\end{cB*}\ce{ {:}{\abmTS} } && \ctrRule{$n$--prefix$_1$}\\
    \ce{ {\abmTS}{:} }\begin{cB*}\ctrBinPL & \sigma' \\\ctrBinVar{2-2} & \tau' \\\end{cB*} &\ce{<=>} \begin{cB*}\ctrBinPR & \sigma' \\\ctrBinVar{2-2} & \tau' \\\end{cB*}\ce{ {:}{\abmTS} } && \ctrRule{$n$--prefix$_2$}\\
    \ce{ {\abmTS}{:} }\begin{cB*}\sigma & \ctrBinPR & x & \sigma' \\\ctrBinInv{1-1}\ctrBinVar{4-4} \tau & & 0 & \tau' \\\end{cB*} &\ce{<=>} \begin{cB*}\sigma & x & \ctrBinPR & \sigma' \\\ctrBinInv{1-2}\ctrBinVar{4-4} \tau & x & & \tau' \\\end{cB*}\ce{ {:}{\abmTS} } && \ctrRule{$n$--prefix$_3$}\\
    \ce{ {\abmTS}{:} }\begin{cB*}\sigma & \ctrBinPR & 1 & \sigma' \\\ctrBinInv{1-1}\ctrBinVar{4-4} \tau & & 1 & \tau' \\\end{cB*} &\ce{<=>} \begin{cB*}\sigma & 1 & \ctrBinNR & \sigma' \\\ctrBinInv{1-2}\ctrBinVar{4-4} \tau & 1 & & \tau' \\\end{cB*}\ce{ {:}{\abmTS} } && \ctrRule{$n$--inc$'$}\\
    \ce{ {\abmTS}{:} }\begin{cB*}\sigma & \ctrBinNR & 0 & \sigma' \\\ctrBinInv{1-1}\ctrBinInv{3-4} \tau & & 0 & \tau' \\\end{cB*} &\ce{<=>} \begin{cB*}\sigma & 0 & \ctrBinNR & \sigma' \\\ctrBinInv{1-2}\ctrBinInv{4-4} \tau & 0 & & \tau' \\\end{cB*}\ce{ {:}{\abmTS} } && \ctrRule{$n$--carry$'$}\\
    \ce{ {\abmTS}{:} }\begin{cB*}\sigma & \ctrBinNR \\\ctrBinInv{1-1} \tau & \\\end{cB*} &\ce{<=>>} \begin{cB}\sigma \\\ctrBinInv{1-1} \tau \\\end{cB}\ce{  {:}A + M$_{i+\frac13}$ } && \ctrRule{$n$--term} \\
  \end{align*}\end{minipage}}}
  \caption[Part II of the abstract molecular reaction scheme implementing binary counters.]{The abstract molecular reaction scheme implementing $\ce{ X^{(n)}_0 <=> X^{(n+1)}_{n+1} }$ for the binary representation case. This overall reaction is unbiased, which raises the concern that the sequestration klona could spend an undesirable amount of time in the intermediate states above. To limit this possibility, the intermediate states should be made more unfavourable and this is achieved by the complementary biases in the rules \ctrRule{$n$--init} and \ctrRule{$n$--term}. These biases need not be provided by the bias tokens, and in fact that is unideal as the bias is vanishingly small; instead it can be implemented by introducing an enthalpy change, such that there is a chance of thermal activation of the reaction, whilst the klona still prefer to exist in the non-intermediate states. As in the scheme for the intra-$n$ reaction, klona are introduced to mark the state and progress of the overall reaction; in this case, $\{$, $\}$, $($, and $)$, with the $\{$ and $\}$ klona corresponding primarily to incrementation of $n$ and the $($ and $)$ klona to copying of the prefix of $n$ into $k$.}
  \label{lst:ctr-bin-n}
\end{listing}}

\section{Conclusion}

In this chapter we studied the interaction of mona, such as molecular computers, with klona, shared resources, in a Brownian system with a limiting supply of free energy.
As in \Cref{chap:revii}, we found this to be very expensive in comparison to independent computation. In general the cost, in terms of transition time, to couple mona to an arbitrary system of klona is of order $\bigOO{b^{-2}}$ (an overhead of $\bigOO{b^{-1}}$) but, in contrast to the mona-mona interactions of \Cref{chap:revii}, lighter overheads of order $\bigOO{b^{-1/2}}$ or better can be achieved under certain conditions. For example, whilst we found that all resource distribution schemes we could devise (some particular examples of which were illustrated in \Cref{sec:resource-scheme}) did not admit a fractional resource availability above $\bigOO{b}$ if time overheads were to be avoided, by allowing a time overhead of $\bigOO{b^{-1/2}}$ (thence a time penalty of $\sim\bigOO{b^{-3/2}}$) the full resource pool was rendered accessible. We conjectured that these results are in fact general bounds on the capabilities of any such system.

To drive a general out-of-equilibrium system of klona, one needs to supply sufficient free energy to overcome the entropy cost of the reaction. When this cost is known or bounded, such as in biological systems, the definition of `sufficient' can be precomputed and built in to the specifications of the system. In general, however, the cost may be unknown or unbounded, or the favourable direction of the reaction may switch over time. We proposed a scheme in \Cref{sec:drive-unfav} to dynamically infer this cost, as well as automatically applying the correct amount of free energy tokens to pay for it, thereby exposing an equilibrated interface to the disequilibrium reaction. A more detailed specification of the (abstract) molecular realisation of such a scheme is presented in \Cref{app:seq-klona}. This equilibrated interface can then be readily coupled to our monon transitions, however the cost must still be paid in terms of transition time overhead and indeed it is, as the concentration of the interface is commensurately reduced. Unfortunately our proposed scheme has a minimum overhead of $\bigOO{b^{-1/2}}$ even in the case where the underlying reaction is at or near equilibrium (where it is theoretically possible to achieve zero overhead). It is unknown whether a better scheme exists, whilst still retaining the ability to drive reactions that are arbitrarily far from equilibrium.

As any non-zero overhead to a class of interactions will lead to this class `freezing out' in very large reversible computers, their frequency must be minimised---the proposed scheme notwithstanding. Once again referencing the conclusions of \Cref{chap:revii}, these interactions thus minimised can make use of a subpopulation bias klona of greater free energy to drive these unfavourable interactions more effectively.

\endgroup

\begingroup

\begin{chapter-summary}[.7]

  Reversible computation is somewhat neglected in the field of molecular computation.
  Motivated by a need for a model of reversible computation appropriate for a Brownian molecular architecture, this chapter introduces the \textAleph-calculus as a step towards this goal.
  This novel model is declarative, concurrent, and term-based---encapsulating all information about the program data and state within a single structure in order to obviate the need for a \emph{von Neumann}-style discrete computational `machine', a challenge in a molecular environment.
  The name is inspired by the Greek for `not forgotten' (\emph{\gr ἀλήθεια}), due to the emphasis on (reversibly) learning and un-learning knowledge of different variables.
  To demonstrate its utility for this purpose, as well as its elegance as a programming language, a number of examples are presented;
    two of these examples, addition/subtraction and squaring/square-rooting, are furnished with designs for abstract molecular implementations, although more work needs to be done in order to realise something resembling \textAleph\ experimentally.
  The example of addition/subtraction is reproduced in \Cref{fig:cover-aleph}.
  A natural by-product of these examples and accompanying syntactic sugar is the design of a fully-fledged programming language, \alethe, which is also presented along with an interpreter.
  Efficiently simulating \textAleph\ on a deterministic computer necessitates some static analysis of programs within the \alethe\ interpreter in order to render the declarative programs sequential.
  Finally, work towards a type system appropriate for such a reversible, declarative model of computation is presented.

\end{chapter-summary}

  \chapter[head={The \textslAleph-Calculus}]{The \textbfAleph-Calculus}
  \textbf{\textsf{\Large A declarative model of reversible programming with relevance to Brownian computers}}
  \label{chap:aleph}
  \vspace{1em}

\begin{figure}[!h]
  \centering
  \includegraphics[width=10cm]{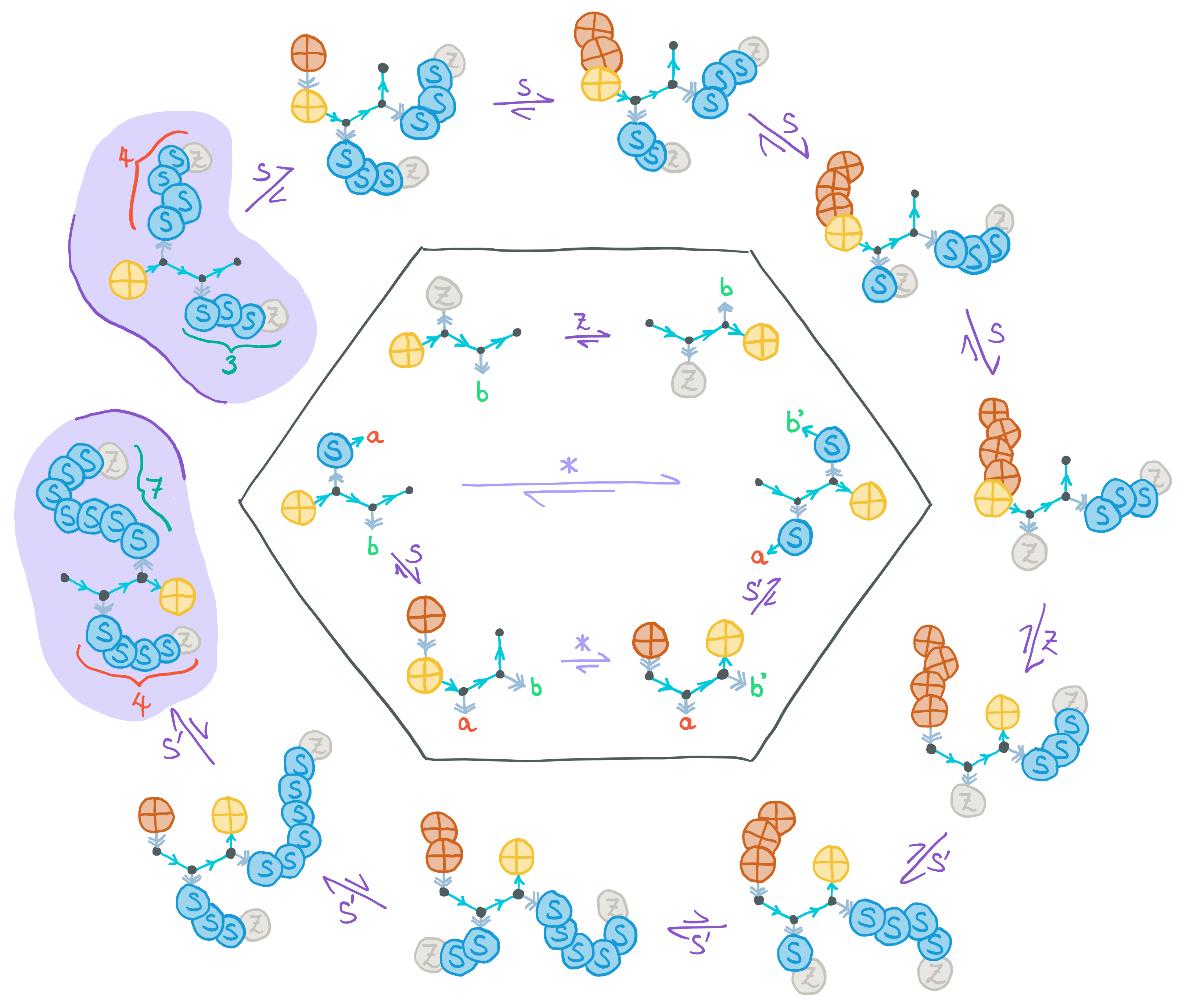}\vspace{1em}
  \caption[An abstract molecular realisation of an \textAleph\ definition of addition, here reversibly computing the sum of 4 and 3.]{An abstract molecular realisation of an \textAleph\ definition of addition (resp.\ subtraction), here reversibly computing the sum of 4 and 3 (resp.\ the difference between 7 and 4).}
  \label{fig:cover-aleph}
\end{figure}

\section{Introduction}

In this chapter we present a new model of reversible computing, the \textAleph-Calculus, and an associated programming language \alethe\ together with interpreter, whose properties we believe to be novel. The name is inspired by the Greek meaning `not forgotten' (\emph{\gr ἀλήθεια}), as the semantics of \textAleph\ revolve around the transformation of knowledge: `unlearning' knowledge of one or more variables in order to `learn' knowledge about one or more other variables, but in a reversible fashion such that nothing is ever truly forgotten. \textAleph\ is declarative and concurrent and, whilst perhaps a little too abstract and high level for this purpose, is motivated by a need for a reversible model for molecular programming and DNA computing. Crucial to such applications is that almost all information about not just the program's data, but the program state too, should be encoded within a computation `term'. This is in contrast to the imperative and functional languages reviewed in the introduction, wherein state information such as the instruction counter (indicating where in the program the current execution context is) is implicitly stored in a special register or other hidden state of the processor.
In \textAleph, the program itself is (reversibly) mutated during the course of its execution, such that the distinction between `program' and `processor' vanishes. The reason for this design is that a \emph{von Neumann}-style architecture, in which there is a discrete processing unit that interacts with a memory unit to execute a program, is unsuitable for a molecular context as it is---in some sense---too `bulky' and difficult to engineer. We assert that the proposed design is far more suitable to molecular implementation, or at least serves as a step towards this goal, and it is hoped the reader will be convinced of this by the examples and semantics illustrated herein.

\begingroup

\section[%
  head={\textslAleph\ by Example},%
  tocentry={\textAleph\ by Example}%
]{\textbfAleph\ by Example}
\label{sec:ex1}

Before expositing the formal semantics of \textAleph, it is illuminating to introduce a number of examples of \textAleph\ in order to gain an intutition for its syntax and features.
We begin with a reversible definition of addition.
An immediate obstacle with reversible arithmetic is that binary operations such as addition are not invertible: it is not possible, given the answer $7$, to infer which two numbers were originally added.
As shown in \Cref{fig:bennett-algos}, one possible resolution of this (due to Bennett~\cite{bennett-tm,bennett-pebbling}) is to embed an irreversible computation $x\mapsto f(x)$ injectively as $x\mapsto(x,f(x))$.
For addition, Bennett's algorithm would yield $+:(x,y)\leftrightarrow(x,y,x+y)$.
For example, adding 3 and 4 would give as output $(3,4,7)$.
Nevertheless we are not satisfied with this embedding, and wish to better exploit reversible computational architectures by programming directly with reversible primitives.
The benefits of doing so are that one often finds that far less temporary information need be generated than Bennett's algorithms might suggest, and also the injective embedding, $x\mapsto(x,f(x))$, retains the input which is excessive except in the trivial case of $f$ being constant;
  for any other function $f$ having any correlation at all between its outputs and inputs, only a partial image of $x$ need be preserved.
Moreover it is often the case that one can imagine a suitable and more preferable injective or bijective embedding.
For example, a more useful embedding of addition might take the form $+:(x,y)\leftrightarrow(x,x+y)$.
In defining these injections more carefully, one then often finds that the residual information of the input can in fact be made use of;
  for example, a Peano arithmetic implementation of $+:(x,y)\leftrightarrow(x,x+y)$ over the naturals readily begets the (domain-restricted) injection $\forall x>0.\times:(x,y)\leftrightarrow(x,xy)$ without generating any temporary data whatsoever.
Furthermore, when redundancy is eliminated as in these two cases one obtains the converse operation for free:
  that is, running $+:(x,y)\leftrightarrow(x,x+y)$ in reverse yields subtraction, and $\forall x>0.\times:(x,y)\leftrightarrow(x,xy)$ yields division, whilst their Bennett embeddings cannot do the same.

A little thought shows these must fail in certain cases. For the Bennett-style reversible embedding of addition, we saw that the forward direction of this program maps, e.g., $(3,4)$ to $(3,4,7)$, and likewise the reverse direction maps $(3,4,7)$ to $(3,4)$.
Clearly the forward direction is an injection, but what about the reverse?
Suppose we attempt to feed in $(3,4,5)$: as $3+4\ne5$, and as the forward direction is injective, there can be no pair $(a,b)$ that maps to $(3,4,5)$.
In fact, even our less redundant embedding $+:(x,y)\leftrightarrow(x,x+y)$ suffers from non-injectivity of its inverse, for example there is no value $y\in\mathbb N$ satisfying $+:(5,y)\leftrightarrow(5,2)$.
Yet another example occurs for the function $+':(x,y)\leftrightarrow(x-y,x+y)$ defined over the integers, in which there is no pair $(x,y)$ satisfying $+':(x,y)\leftrightarrow(3,6)$ (there is in $\mathbb Q$, however; namely, $(\frac92,\frac32)$).
Therefore, whilst we might expect reversible computers to compute bijective operations, this appears to be violated in even the simple example of addition.
In fact, no violation of reversibility has occurred, and the `primitive' steps of any reversible computer \emph{are} bijective.
What we have encountered here is simply a (co)domain error, the same as if we were to ask a conventional computer to evaluate $1/0$.
That is, whereas conventional computers compute \emph{partial} functions, reversible computers compute \emph{partial} bijections (or partial isomorphisms).
Exactly what happens when such an error is encountered depends on both the algorithm and architecture in question;
  for example, attempting to divide 1 by 0 may cause a computer to immediately complain, or it may enter an infinite loop if the algorithm is `repeatedly subtract the divisor until the dividend vanishes'.
We shall have more to say about how \textAleph\ handles such errors in due course.

\def\monus{\mathbin{\ooalign{\hss\raisebox{0.5ex}{$\cdot$}\hss\cr\phantom{$+$}\cr$-$}}}
\doublepage{%
\begin{listing}[!p]
  \centering
  \begin{sublisting}{\textwidth}
    \centering
    \begin{align*}
      & {+}~{\Z}~{b}~{\unit} = {\unit}~{\Z}~{b}~{+}{;} && \ruleName{add--base} \\
      & {+}~({\S} {a})~{b}~{\unit} = {\unit}~({\S} {a})~({\S} {b'})~{+}{:} && \ruleName{add--step} \\
      &\qquad  {+}~{a}~{b}~{\unit} = {\unit}~{a}~{b'}~{+}{.} && \ruleName{add--step--sub}
    \end{align*}
    \caption{The definition of reversible natural addition in \textAleph.}
    \label{lst:ex-add-def}
  \end{sublisting}
  \begin{sublisting}{\textwidth}
    \centering
\begingroup%
\newcommand{\subLeft}[1][1]{\smash{\raisebox{-0.65em}{\scaleto{\begin{tikzpicture}%
    \draw[<->] (0,0) to [out=270,in=90, looseness=1] (-#1,-0.5);%
  \end{tikzpicture}}{1.78em}}}\hspace{4.1em}}%
\newcommand{\subRight}[1][1]{\hspace{7.1em}\smash{\raisebox{-0.65em}{\scaleto{\begin{tikzpicture}%
    \draw[<->] (#1,-0.5) to [out=90,in=270, looseness=1] (0,0);%
  \end{tikzpicture}}{1.78em}}}}%
\def\bindings#1{\{#1\}}%
\def\bindingsSub#1{\smash{\underline{\bindings{#1}}}}%
\def\subterm#1{\smash{\overline{#1}}}%
\def\midsp{\phantom{\bindings{b:3}}}%
\begin{align*}
  {+}~{4}~{3}~{\unit} \leftrightsquigarrow \bindingsSub{a:3, b:3} &\midsp \bindingsSub{a:3, b':6} \leftrightsquigarrow {\unit}~{4}~{7}~{+} && \ruleName{add--step} \\
  \subLeft &\subRight && \ruleName{add--step--sub} \\
  \subterm{{+}~{3}~{3}~{\unit}} \leftrightsquigarrow \bindingsSub{a:2, b:3} &\midsp \bindingsSub{a:2, b':5} \leftrightsquigarrow \subterm{{\unit}~{3}~{6}~{+}} && \ruleName{add--step} \\
  \subLeft &\subRight && \ruleName{add--step--sub} \\
  \subterm{{+}~{2}~{3}~{\unit}} \leftrightsquigarrow \bindingsSub{a:1, b:3} &\midsp \bindingsSub{a:1, b':4} \leftrightsquigarrow \subterm{{\unit}~{2}~{5}~{+}} && \ruleName{add--step} \\
  \subLeft &\subRight && \ruleName{add--step--sub} \\
  \subterm{{+}~{1}~{3}~{\unit}} \leftrightsquigarrow \bindingsSub{a:0, b:3} &\midsp \bindingsSub{a:0, b':3} \leftrightsquigarrow \subterm{{\unit}~{1}~{4}~{+}} && \ruleName{add--step} \\
  \subLeft[0.25]\hspace{-1.8em} & \hspace{-1.95em}\subRight[0.28] && \ruleName{add--step--sub} \\
  \subterm{{+}~{\Z}~{3}~{\unit}} \leftrightsquigarrow{} &\bindings{b:3} \leftrightsquigarrow \subterm{{\unit}~{\Z}~{3}~{+}} && \ruleName{add--base}
\end{align*}%
\endgroup
    \caption{The \textAleph\ execution path when reversibly adding $4$ to $3$ or, in reverse, subtracting $4$ from $7$. The $\leftrightsquigarrow$~arrows refer to pattern matching/substitution, whilst the solid arrows refer to instantiation/consumption of `sub-terms'.}
    \label{lst:ex-add-43}
  \end{sublisting}
  \begin{sublisting}{\textwidth}
    \centering
\begingroup%
\newcommand{\subLeft}[1][1]{\smash{\raisebox{-0.65em}{\scaleto{\begin{tikzpicture}%
    \draw[<->] (0,0) to [out=270,in=90, looseness=1] (-#1,-0.5);%
  \end{tikzpicture}}{1.78em}}}\hspace{4.3em}}%
\def\bindings#1{\{#1\}}%
\def\bindingsSub#1{\smash{\underline{\bindings{#1}}}}%
\def\subterm#1{\smash{\overline{#1}}}%
\def\nomatch{\mathrel{\reflectbox{$\rightsquigarrow$}\hspace{-0.7ex}/}}%
\def\nomatchlong{\mathrel{\rlap{$\nomatch$}\phantom{\leftrightsquigarrow}}}%
\begin{align*}
  {\unit}~{5}~{2}~{+} \leftrightsquigarrow \bindingsSub{a:4, b':1} & \phantom{\bindings{b:3} \bindingsSub{a:3, b':6} \leftrightsquigarrow {\unit}~{4}~{7}~{+}} && \ruleName{add--step} \\
  \subLeft &&& \ruleName{add--step--sub} \\
  \subterm{{\unit}~{4}~{1}~{+}} \leftrightsquigarrow \bindingsSub{a:3, b':0} &&& \ruleName{add--step} \\
  \subLeft &&& \ruleName{add--step--sub} \\
  \subterm{{\unit}~{3}~{\Z}~{+}} \nomatchlong \phantom{\bindingsSub{a:1, b':3}} &&& \text{\small\emph{no matching rule}}
\end{align*}%
\endgroup
    \caption{The \textAleph\ execution path when attempting to (erroneously) subtract $5$ from $2$. The recursive algorithm identifies that $2-5\equiv 0-3$, but there is no matching definition for this and therefore computation `stalls' on this sub-term. This is usually addressed in the natural numbers by employing the saturating option of `monus', i.e.\ $2\monus5=0$, but it is not reversible.}
    \label{lst:ex-add-52}
  \end{sublisting}
  \caption{The definition of, and example applications of, reversible addition in \textAleph.}
  \label{lst:ex-add}
\end{listing}}{%
\begin{figure}[!p]
  \centering
  \includegraphics[width=\textwidth]{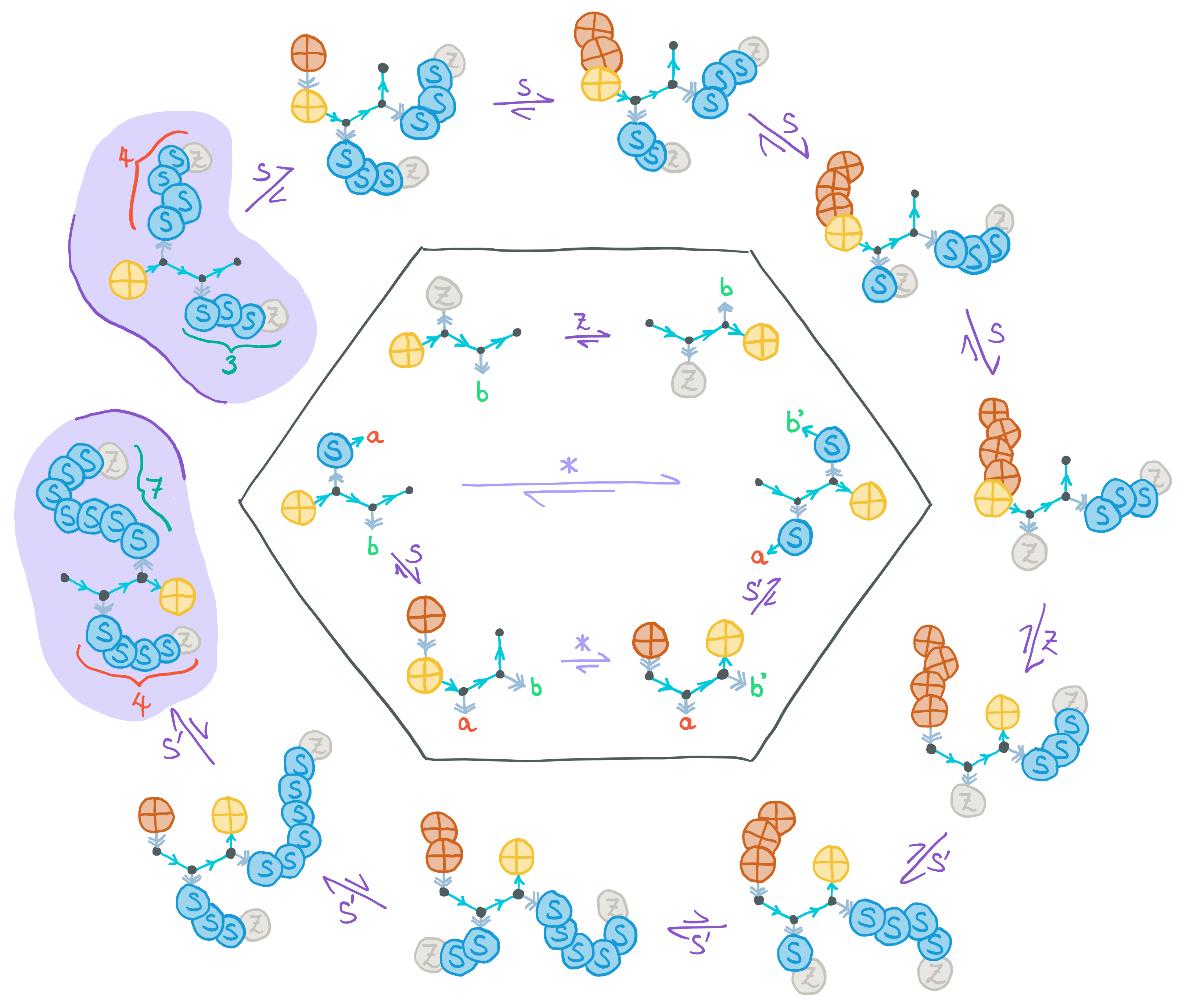}
  \caption[An abstract molecular translation of the reversible addition definition, as well as the example addition of $4$ and $3$.]{An abstract molecular translation of the reversible addition definition given in \Cref{lst:ex-add-def}, as well as the example addition of $4$ and $3$ following \Cref{lst:ex-add-43}. The abstract molecular model espoused here is defined by a fixed set of atoms (e.g.\ $\{+,\S,\Z,\bullet\}$) connected by two kinds of bond. Atoms joined by single-headed bonds are analogous to \textAleph\ terms, whilst double-headed bonds corresponds to nesting of composite terms. The bonds are represented by arrows because there is an intrinsic polarity/directionality to the molecules; this is not necessary, and can be replaced by auxiliary atoms such as \atom{L} and \atom{R}, but it does simplify our representation. Atoms are rendered by circles, whilst the reaction definitions also use variables written as un-circled letters. The $\bullet$ `atom' is special in that it is a placeholder for a nested composite term. Reaction arrows for `elementary' reactions are labelled by rule names, and starred reaction arrows indicate an effective reaction composed of multiple elementary steps. The arrows are drawn biased in anticipation of a biasing mechanism to be discussed later on within this section.}
  \label{fig:ex-plus}
\end{figure}}

\para{Addition}

Our examples will mostly concern natural numbers because they are particularly amenable to inductive and recursive definitions of both their structure and operations over them, such as addition and multiplication. The standard approach to this is the Peano axiomatic formulation, in which a natural number is defined to either be $\Z\equiv 0$, or $\S n\equiv n+1$ whenever $n$ is a natural number---i.e.\ the \emph{successor} of $n$. For example, 4 is constructed as $\S(\S(\S(\S\Z)))$. In fact there are seven more axioms in order to clarify such subtle points as uniqueness of representation, conditions for equality, and non-negativity. Addition can then be defined recursively by the base case $\Z+b=b$ and the inductive step $\S a+b=\S(a+b)$, and multiplication by $a\cdot\Z=\Z$ and $a\cdot\S b=a+a\cdot b$. To render addition reversible, we write $+:(\Z,b)\leftrightarrow(\Z,b)$ and $+:(\S a,b)\leftrightarrow(\S a,\S(a+b))$, realising the proposed embedding $+:(x,y)\leftrightarrow(x,x+y)$. Multiplication is embedded similarly, but there is a domain-restriction imposed in that $a$ must not be zero (or else $b$ cannot be uniquely recovered); $b$ may, however, be zero. In \textAleph, addition is written thus
  \begin{align*}
    & {+}~{\Z}~{b}~{\unit} = {\unit}~{\Z}~{b}~{+}{;} && \ruleName{add--base} \\
    & {+}~({\S} {a})~{b}~{\unit} = {\unit}~({\S} {a})~({\S} {b'})~{+}{:} && \ruleName{add--step} \\
    &\qquad  {+}~{a}~{b}~{\unit} = {\unit}~{a}~{b'}~{+}{.} && \ruleName{add--step--sub}
  \end{align*}
and is perhaps best understood by example. In \Cref{lst:ex-add}, we perform the example addition of $4$ and $3$ and, as promised, an example failure mode in which we attempt to subtract $5$ from $2$.

\textAleph\ can thus be seen, in a loose sense, to be a term-rewriting system. It is `loose' in the sense that its `sub-terms' exist `separate' from their parent term. In the example of \Cref{lst:ex-add}, an addition term is written ${+}~{a}~{b}~{\unit}$ and is mapped by the transition rules to ${\unit}~{a}~{c}~{+}$ where $c\equiv a+b$; here $+$ is an `atom', $a$, $b$ and $c$ are terms representing natural numbers (as composite terms formed from nested applications of the atoms $\S$ and $\Z$), and $\unit$, or `unit', is the empty term and is used by convention in \textAleph\ to avoid certain ambiguities. A program corresponds to a series of definitions of transition rules which pattern match on terms, and then substitute the variables into an output pattern. This matching process is subject to certain constraints that ensure reversibility. Moreover, a rule may specify sub-rules that indicate how to transform knowledge of some variable, e.g.~$b$, to knowledge of another, e.g.~$b'$, and is the primary mechanism of composition in \textAleph. Inherent to the semantics of the calculus is a secondary composition mechanism, which was the only mechanism available in an earlier iteration of this calculus (see \Cref{app:sigma}): composite terms at any level are all subject to the same transition rules. This behaviour is more reminiscent of functional programming languages, but is somewhat clumsy in practice; nevertheless it is not without utility in \textAleph---in particular, it is well suited for when a continuation-passing style approach is favoured. The sub-rule mechanism, on the other hand, is more reminiscent of declarative programming languages. 

When there is no matched rule, such as in \Cref{lst:ex-add-52}, this simply means that there is no successor state and so computation cannot continue. It is in fact very similar in nature to the case when computation succeeds, as then we obtain a term which has a predecessor state via some rule, but lacks a rule to generate a successor state. To properly distinguish these two conditions, we must explicitly mark `true' halting states; for addition, this is written as
\begin{align*}
  & \haltSym~{+}~{\blank}~{\blank}~{\unit}{;}\quad \haltSym~{\unit}~{\blank}~{\blank}~{+}{;}
\end{align*}
where $\haltSym$ (`bang') indicates that any term of the specified forms is a valid computational output. The former corresponds to the output of a subtraction, and the latter to the output of an addition. This subtlety is further contextualised by \Cref{fig:state-space}.

We conclude this first example by alluding to a possible molecular implementation, as depicted in \Cref{fig:ex-plus}. Here we see the importance of \textAleph\ being interpretable as a term-rewriting system, as the entirety of the state of the computation must be encoded within a single macromolecule (or molecular complex). Strictly speaking, this is not a requirement as DNA Strand Displacement systems~\cite{dsd} achieve computation without this requirement; in payment for this, though, the entire reaction volume is dedicated to the same (typically analogue) program. The precise mechanism of the reactions is omitted from the figure, instead expressing the model as an abstract chemical reaction network. Finding possible reaction mechanisms to implement \textAleph, or similar calculi, in real chemical systems will be the subject of future work. Whilst the abstract molecular formalism is attractive form an explanatory perspective, \textAleph\ provides a more concise formalism that is also easier to manipulate and to explicate its semantics.

\doublepage{%
\begin{listing}[!p]
  \centering
  \begin{sublisting}{\textwidth}
    \centering
    \begin{align*}
      & \haltSym~{\square}~{\blank}~{\unit}{;}\quad \haltSym~{\unit}~{\blank}~{\square}{;} \\
      & {\square}~{n}~{\unit} = {\square}~{n}~{\Z}~{\square}{;} && \ruleName{square--init} \\
      & {\square}~{\Z}~{m}~{\square} = {\unit}~{m}~{\square}{;} && \ruleName{square--term} \\
      & {\square}~({\S}{n})~{m}~{\square} = {\square}~n~({\S}{m''})~{\square}{:} && \ruleName{square--step} \\
      &\qquad {+}~{n}~{m}~{\unit} = {\unit}~{n}~{m'}~{+}{.} && \ruleName{square--step--sub$_1$}\\
      &\qquad {+}~{n}~{m'}~{\unit} = {\unit}~{n}~{m''}~{+}{.} && \ruleName{square--step--sub$_2$}
    \end{align*}
    \caption{The definition of $\square$ in \textAleph.}
    \label{lst:ex-square-def-aleph}
  \end{sublisting}
  \begin{sublisting}{\textwidth}
    \centering
    \includegraphics[width=.7\textwidth]{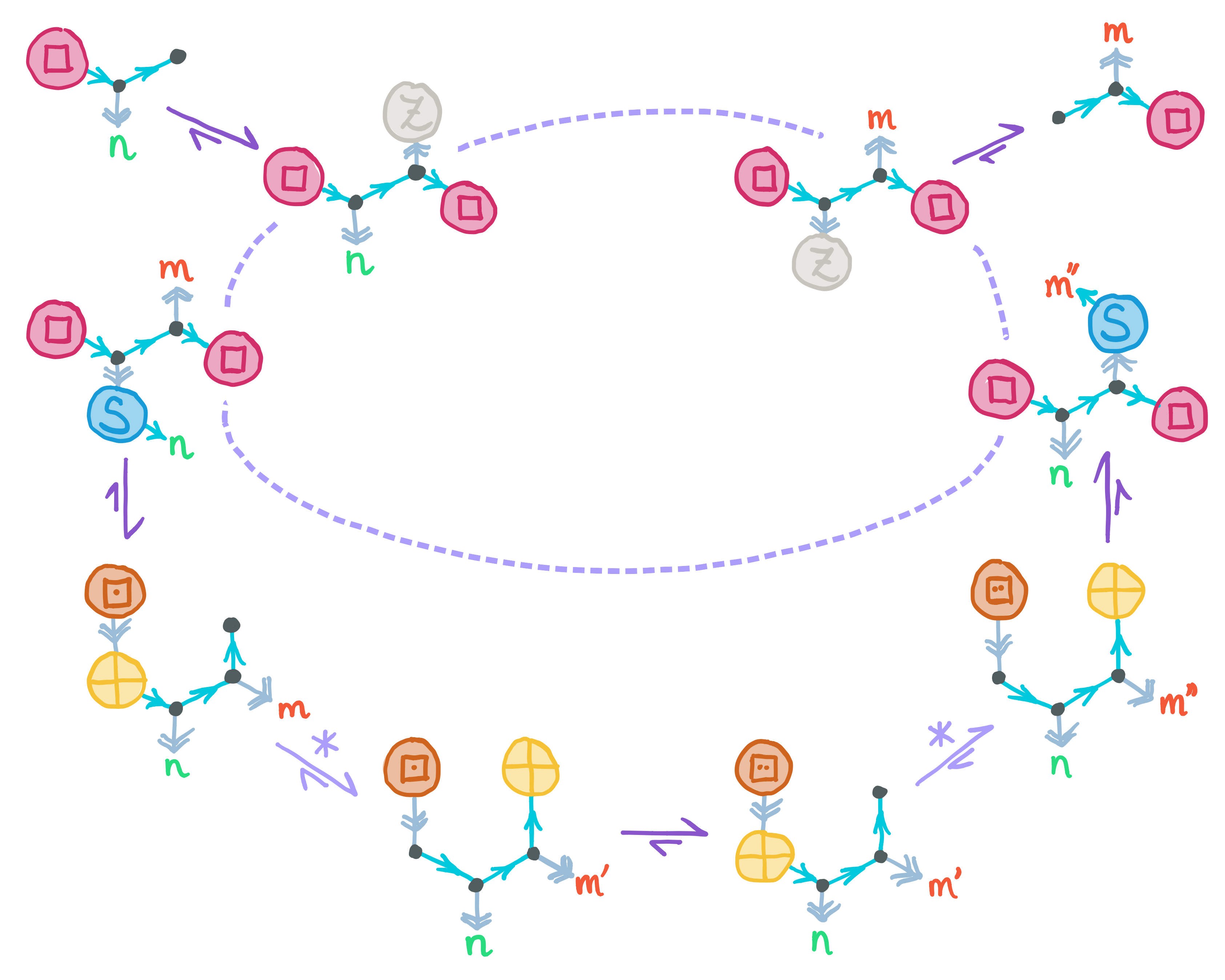}
    \caption{The definition of $\square$ in the abstract molecular formalism introduced in \Cref{fig:ex-plus}. The dashed lines indicate that there exists a term matching both patterns, and this is characteristic of conditionals and looping in \textAleph. Note that computation remains unambiguous and deterministic: by \Cref{fig:state-space}, each intermediate state has both a predecessor and a successor, and of the two matching patterns one will correspond to the reverse direction and one to the forward direction. For example, in the middle of a loop one may encounter the intermediate species ${\square}~{2}~{5}~{\square}$ and this term matches both the patterns ${\square}~({\S}{n})~{m}~{\square}$ and ${\square}~n~({\S}{m''})~{\square}$, but if it matches the former the computation will proceed forward whilst if the latter the computation will proceed backward. It will be shown later how the semantics keep track of this, and how we ensure the program really is unambiguous.}
    \label{lst:ex-square-def-mol}
  \end{sublisting}
  \caption{A reversible definition of squaring of natural numbers, $\square:n\leftrightarrow n^2$, both in \textAleph\ and in an abstract molecular formalism.}
  \label{lst:ex-square-def}
\end{listing}}{%
\begin{listing}[!p]
  \centering
\begingroup%
\newcommand{\subb}{\smash{\raisebox{-0.65em}{\scaleto{\begin{tikzpicture}%
    \draw[<->] (0,0) to [out=270,in=90, looseness=1.2] (-0.4,-0.22);%
  \end{tikzpicture}}{0.78em}}}}
\def\bindings#1{\{#1\}}%
\def\over#1{\smash{\overline{#1}}}%
\def\under#1{\smash{\underline{#1}}}%
\def\overunder#1{\smash{\overline{\underline{#1}}}}%
\def\hl{\\[-1em]}%
\def\nl{\\[1.5em]}%
\def\subl{\hl&\hspace{0.8em}\subb\\}%
\def\subr{\hl&\hspace{4em}\subb\\}%
\begin{align*}
  \haltSym\quad{\square}~{3}~{\unit} \leftrightsquigarrow{} &\bindings{n:3} \leftrightsquigarrow {\square}~{3}~{\Z}~{\square} \quad\cdots && \ruleName{square--init} \nl
  \cdots\quad{\square}~{3}~{\Z}~{\square} \leftrightsquigarrow{} &\under{\bindings{n:2,m:\Z}} && \ruleName{square--step} \subl
  \multicolumn{2}{c}{$\over{{+}~{2}~{\Z}~{\unit}}\leftrightarrow\under{{\unit}~{2}~{2}~{+}}$} && \ruleName{square--step--sub$_1$} \subr
  &\overunder{\bindings{n:2,m':2}} && \text{\small\emph{checkpoint}} \subl
  \multicolumn{2}{c}{$\over{{+}~{2}~{2}~{\unit}}\leftrightarrow\under{{\unit}~{2}~{4}~{+}}$} && \ruleName{square--step--sub$_2$} \subr
  &\over{\bindings{n:2,m'':4}} \leftrightsquigarrow {\square}~{2}~{5}~{\square}\quad\cdots && \ruleName{square--step} \nl
  \cdots\quad{\square}~{2}~{5}~{\square} \leftrightsquigarrow{} &\under{\bindings{n:1,m:5}} && \ruleName{square--step} \subl
  \multicolumn{2}{c}{$\over{{+}~{1}~{5}~{\unit}}\leftrightarrow\under{{\unit}~{1}~{6}~{+}}$} && \ruleName{square--step--sub$_1$} \subr
  &\overunder{\bindings{n:1,m':6}} && \text{\small\emph{checkpoint}} \subl
  \multicolumn{2}{c}{$\over{{+}~{1}~{6}~{\unit}}\leftrightarrow\under{{\unit}~{1}~{7}~{+}}$} && \ruleName{square--step--sub$_2$} \subr
  &\over{\bindings{n:1,m'':7}} \leftrightsquigarrow {\square}~{1}~{8}~{\square}\quad\cdots && \ruleName{square--step} \nl
  \cdots\quad{\square}~{1}~{8}~{\square} \leftrightsquigarrow{} &\under{\bindings{n:\Z,m:8}} && \ruleName{square--step} \subl
  \multicolumn{2}{c}{$\over{{+}~{\Z}~{8}~{\unit}}\leftrightarrow\under{{\unit}~{\Z}~{8}~{+}}$} && \ruleName{square--step--sub$_1$} \subr
  &\overunder{\bindings{n:\Z,m':8}} && \text{\small\emph{checkpoint}} \subl
  \multicolumn{2}{c}{$\over{{+}~{\Z}~{8}~{\unit}}\leftrightarrow\under{{\unit}~{\Z}~{8}~{+}}$} && \ruleName{square--step--sub$_2$} \subr
  &\over{\bindings{n:\Z,m'':8}} \leftrightsquigarrow {\square}~{\Z}~{9}~{\square}\quad\cdots && \ruleName{square--step} \nl
  \cdots\quad {\square}~{\Z}~{9}~{\square} \leftrightsquigarrow{} &\bindings{m:9} \leftrightsquigarrow {\unit}~{9}~{\square}\quad\haltSym && \ruleName{square--term}
\end{align*}%
\begin{center}\fbox{\begin{minipage}{0.8\textwidth}\vspace{-0.5em}\begin{align*}
  \haltSym\quad{\square}~{3}~{\unit}\leftrightarrow
  {\square}~{3}~{\Z}~{\square}\leftrightarrow
  {\square}~{2}~{5}~{\square}\leftrightarrow
  {\square}~{1}~{8}~{\square}\leftrightarrow
  {\square}~{\Z}~{9}~{\square}\leftrightarrow
  {\unit}~{9}~{\square}\quad\haltSym
\end{align*}\vspace{-1.2em}\end{minipage}}\end{center}%
\endgroup
  \caption[An example application of the $\square$ definition.]{An example application of the $\square$ definition from \Cref{lst:ex-square-def}.}
  \label{lst:ex-square}
\end{listing}}

\para{Squaring}

Eliminating redundancy in the definition of addition yielded its inverse for free, but one may still object that the additional output is `garbage' and not of any utility. What will happen if one uses this addition subroutine many times in a larger program? Naively we may expect this garbage to accumulate, requiring either active dissipation of the additional entropy or the application of Bennett's algorithm to clean it up. In fact, by retraining one's thought process from the irreversible programming paradigm to the reversible paradigm, it is often possible to make use of this garbage data. We demonstrate this with the example of finding the square of a natural number. The candidate function, $\square:n\leftrightarrow n^2$, is an injection and so clearly meets our requirement of partial bijection. Therefore, we have good reason to believe that it is possible to implement it. The obvious approach of using multiplication will not work because its reversible embedding will take a form not dissimilar to $\times:(m,n)\leftrightarrow(m,m\cdot n)$, and so would yield $\square:n\leftrightarrow(n,n^2)$. Whilst this suffices for realising the square of a number, it retains too much redundancy in its outputs.

Often a helpful tactic is to consider an inductive approach. For the square numbers this is encapsulated by $(n+1)^2-n^2=2n+1$, from which can be obtained the identity $n^2\equiv\sum_{k=0}^{n-1}2k+1$. Again, we need to be clever: in order to achieve our desired (partial) bijection, we need to completely consume our input value of $n$. This can be done by evaluating the sum in reverse. Instantiate a new variable, $m=0$; as $m$ is set to a known value, this is reversible. Then, perform the following loop until $n$ reaches 0: decrement $n$, add $2n+1$ to $m$, repeat. At the end of this loop, $n$ will have reached a unique value (and can thus be reversibly destroyed) and $m$ will have been set to the square of the original value of $n$. In addition, this can be implemented with our addition subroutine in two steps, by adding $n$ twice to $m$ (retaining the value of $n$) and finally incrementing $m$. This is implemented in \textAleph\ in \Cref{lst:ex-square-def}, and the example of squaring $3$ (equivalently, taking the square root of $9$) is presented in \Cref{lst:ex-square}.

\section[%
  head={\textslAleph\ by Example 2: Parallelism \& Concurrency},%
  tocentry={\textAleph\ by Example 2: Parallelism \& Concurrency}%
]{\textbfAleph\ by Example 2: Parallelism \& Concurrency}
\label{sec:ex2}

Having introduced the essence of \textAleph\ in the previous section, we now dive deeper into some more advanced features and examples of \textAleph.

\para{Sugar}

To reduce boilerplate and increase clarity in longer programs, it is helpful to introduce some syntactic sugar (shorthands). More sugar will be introduced later for the definition of the programming language \alethe\ (which is really just sugared \textAleph), but for now only a light sprinkling is required.

Many rules take the rote form
\begin{align*}
  & \haltSym~{f}~{x_1}~{x_2}~\cdots~{x_m}~{\unit}{;}\quad
    \haltSym~{\unit}~{y_1}~{y_2}~\cdots~{y_n}~{f}{;} \\
  & {f}~{x_1}~{x_2}~\cdots~{x_m}~{\unit}={\unit}~{y_1}~{y_2}~\cdots~{y_n}~{f}{:} \\
  &\qquad \cdots
\end{align*}
which we can abbreviate with an infix form as
\begin{align*}
  & {x_1}~{x_2}~\cdots~{x_m}~\infix{f}~{y_1}~{y_2}~\cdots~{y_n}{:} \\
  &\qquad \cdots
\end{align*}
where the halting patterns are implied. In the special case of $f$ a single symbol, such as $+$, $\times$ or $\square$, we omit the backticks and write, e.g., ${4}~{3}~{+}~{4}~{7}$. Note that, if $f$ is a composite term, then we additionally assert $\haltSym~{f}$.

We have already seen sugar for numeric data, e.g.\ $4\equiv\S(\S(\S(\S\Z)))$. Another common data type is that of lists. As is standard in the functional world, we opt for singly linked lists implemented as composite pairs. If $\Cons~{x}~{y}$ represents the pair $(x,y)$ and $\Nil$ the empty list, then the list $[{5}~{2}~{4}~{3}]$ has corresponding representation $\Cons~{5}~(\Cons~{2}~(\Cons~{4}~(\Cons~{3}~\Nil)))$. We also introduce sugar for matching on a partial prefix of a list, i.e.\ $[{x}~{y}~{\cdot}~\cursivezS]$ corresponds to $\Cons~{x}\ (\Cons~{y}\ \cursivezS)$.

\para{Parallelism}

With this sugar thus defined, we can rewrite the somewhat clumsy definition of recursively mapping a function $f$ over a list,
\begin{align*}
  & \haltSym~{\atom{Map}}~{\blank}{;}\quad \haltSym~{({\atom{Map}}~{\blank})}~{\blank}~{\unit}{;}\quad \haltSym~{\unit}~{\blank}~{({\atom{Map}}~{\blank})}{;} \\
  & {({\atom{Map}}~{f})}~{\Nil}~{\unit} = {\unit}~{\Nil}~{({\atom{Map}}~{f})}{;} \\
  & {({\atom{Map}}~{f})}~{({\Cons}~{x}~\cursivexS)}~{\unit} = {\unit}~{({\Cons}~{y}~\cursiveyS)}~{({\atom{Map}}~{f})}{:} \\
  &\qquad {f}~{x}~{\unit} = {\unit}~{y}~{f}{.} \\
  &\qquad {({\atom{Map}}~{f})}~\cursivexS~{\unit} = {\unit}~\cursiveyS~{({\atom{Map}}~{f})}{.}
\end{align*}
more concisely and clearly as 
\begin{align*}
  & []~\infix{\atom{Map}~{f}}~[]{;} \\
  & [{x}~{\cdot}~\cursivexS]~\infix{\atom{Map}~{f}}~[{y}~{\cdot}~\cursiveyS]{:} \\
  &\qquad {x}~\infix{f}~{y}{.} \\
  &\qquad \cursivexS~\infix{\atom{Map}~{f}}~\cursiveyS{.}
\end{align*}
Notice that the order in which the sub-rules are executed does not make a difference to the final result, due to referential transparency; in fact \textAleph, being declarative, does not ascribe any importance to the ordering of the statements. Moreover, should a rule not be necessary for the final computation (or if there are multiple routes to the answer) then that rule will not necessarily be executed. A rule may even be evaluated more than once; for example, Bennett's algorithm may be implemented as
\begin{align*}
  & {x}~\infix{\atom{Bennett}~{f}}~{x}~{y}{:} \\
  &\qquad {x}~\infix{f}~\textit{garbage}~{y}{.}
\end{align*}
wherein the function $f$ will be evaluated once in the forward direction, its output $y$ will be copied, and then $f$ will be evaluated again in the reverse direction to consume the garbage. This arbitrarity in rule ordering and execution is important to its ability to operate in a stochastic system such as a molecular context, although in practice a compilation pass that chooses and enforces an optimal execution plan is important for efficiency.

The implicit duplication of variables that may occur means that one possible interpretation of \atom{Map} is
\begin{align*}
  & []~\infix{\atom{Map}~{f}}~[]{;} \\
  & [{x}~{\cdot}~\cursivexS]~\infix{\atom{Map}~{f}}~[{y}~{\cdot}~\cursiveyS]{:} \\
  &\qquad \infix{\atom{Dup}~{f}}~{f'}{.} \\
  &\qquad {x}~\infix{f}~{y}{.} \\
  &\qquad \cursivexS~\infix{\atom{Map}~{f'}}~\cursiveyS{.}
\end{align*}
from which we can see that not only is the order of the sub-rules arbitrary, but that they can be evaluated in parallel as the following example makes clear:
\begingroup%
\def\bindings#1{\{#1\}}%
\def\over#1{\smash{\overline{#1}}}%
\def\under#1{\smash{\underline{#1}}}%
\newcommand{\sub}[2][1]{\smash{\raisebox{-0.65em}{\scaleto{\begin{tikzpicture}%
    \draw[<->] (0,0) to [out=270,in=90, looseness=#1] (#2,-0.5);%
  \end{tikzpicture}}{1.78em}}}}%
\begin{align*}
  \haltSym\quad{({\atom{Map}}~{\square})}~{[3~5~8]}~{\unit} \leftrightsquigarrow \bindings{ \under{f:\square, x:3\vphantom{[]}} &, \under{f':\square, \cursivexS:{[5~8]}} } \\
  \sub{-1}\hspace{2.2em} & \hspace{3em}\sub{0.5} \\
  \over{{\square}~{3}~{\unit}}=\under{{\unit}~{9}~{\square}} \quad&\quad
  \over{{(\atom{Map}~{\square})}~{[5~8]}~{\unit}}=\under{{\unit}~{[25~64]}~{(\atom{Map}~{\square})}} \\
  \sub{0.2}\hspace{1.5em} & \hspace{4.8em}\sub[0.5]{-2.3} \\
  \bindings{ \over{f:\square, y:9\vphantom{[]}} &, \over{f':\square, \cursiveyS:{[25~64]}} } \leftrightsquigarrow {\unit}~{[9~25~64]}~{({\atom{Map}}~{\square})}\quad\haltSym
\end{align*}%
\endgroup

There remains a subtle point to be made: variables can be implicitly duplicated if they are used by multiple rules, but there is then a contract made with these rules that they really do return the variable unchanged. It is not possible to ensure this statically, however, and so it is entirely possible that the copies of some variable may diverge. In this case, running \atom{Dup} in reverse will fail, and hence computation will stall. If duplication is not used, then computation will occur linearly and the changed value may be fed into subsequent rules unnoticed. In this case, computation may run to completion but yield an incorrect result due to the logic error.

\begin{listing}[!p]
  \centering
  \begingroup
\def\nl{\\[1em]}
\begin{minipage}{.35\linewidth}\begin{align*}
    & \atom{False}~\infix{\atom{Not}}~\atom{True}{;} \\
    & \atom{True}~\infix{\atom{Not}}~\atom{False}{;}
  \nl
    & \infix{{<}~{m}~{\Z}}~\atom{False}{;} \\
    & \infix{{<}~{\Z}~{({\S}~{n})}}~\atom{True}{;} \\
    & \infix{{<}~{({\S}~{m})}~{({\S}~{n})}}~{b}{:} \\
    &\qquad \infix{{<}~{m}~{n}}~{b}{.}
  \nl
    & \infix{{\le}~{m}~{n}}~{b'}{:} \\
    & \quad\infix{{<}~{n}~{m}}~{b}{.}
      \quad{b}~\infix{\atom{Not}}~{b'}{.} \\
    & \infix{{>}~{m}~{n}}~\mathrlap{{b}{:}}\phantom{{b'}{:}} \\
    & \quad\infix{{<}~{n}~{m}}~{b}{.} \\
    & \infix{{\ge}~{m}~{n}}~{b'}{:} \\
    & \quad \infix{{<}~{m}~{n}}~{b}{.}
      \quad {b}~\infix{\atom{Not}}~{b'}{.}
\end{align*}\end{minipage}\hfill%
\begin{minipage}{.55\linewidth}\begin{align*}
    & \haltSym~{\blank}~\infix{\atom{InsertionSort}~{p}}~{\blank}~{\blank}{;} \\
    & {({\atom{InsertionSort}}~{p})}~\cursivexS~{\unit} = \atom{IS}~{p}~\cursivexS~[]~[]~\atom{IS}{;} \\
    & \atom{IS}~{p}~[{x}~{\cdot}~\cursivexS]~\cursivenS~\cursiveyS~\atom{IS} = \atom{IS}~{p}~\cursivexS~[{n}~{\cdot}~\cursivenS]~\cursivezS~\atom{IS}{:} \\
    &\qquad {x}~\cursiveyS~\infix{\atom{Insert}~{p}}~{n}~\cursivezS{.} \\
    & \atom{IS}~{p}~[]~\cursivenS~\cursiveyS~\atom{IS} = {\unit}~\cursivenS~\cursiveyS~{(\atom{InsertionSort}~{p})}{;}
  \nl
    & {x}~[]~\infix{\atom{Insert}~{p}}~{\Z}~[{x}]{;} \\
    & {x}~[{y}~{\cdot}~\cursiveyS]~\infix{\atom{Insert}~{p}}~{n}~[{z}~{z'}~{\cdot}~\cursivezS]{:} \\
    &\qquad \infix{{p}~{x}~{y}}~{b}{.} \\
    &\qquad {x}~[{y}~{\cdot}~\cursiveyS]~{b}~\infix{\atom{Insert'}~{p}}~n~[{z}~{z'}~{\cdot}~\cursivezS]{.}
  \nl
    & {x}~[{y}~{\cdot}~\cursiveyS]~\atom{True}~\infix{\atom{Insert'}~{p}}~{\Z}~[{x}~{y}~{\cdot}~\cursiveyS]{;} \\
    & {x}~[{y}~{\cdot}~\cursiveyS]~\atom{False}~\infix{\atom{Insert'}~{p}}~{({\S}~n)}~[{y}~{\cdot}~\cursivezS]{:} \\
    &\qquad {x}~\cursiveyS~\infix{\atom{Insert}~{p}}~{n}~\cursivezS{.}
\end{align*}\end{minipage}
\\[1em]
\fbox{\begin{minipage}{.8\textwidth}\vspace{-0.2em}\begin{align*}
  & {[3~2~0~7~6~4~5~1]}~\infix{\atom{InsertionSort}~{<}}~{[1~4~3~3~3~0~0~0]}~{[0~1~2~3~4~5~6~7]}{.} \\
  & {[3~2~0~7~6~4~5~1]}~\infix{\atom{InsertionSort}~{\ge}}~\mathrlap{\underbrace{\phantom{[6~2~2~1~0~2~1~0]}}_{\text{garbage}}}{[6~2~2~1~0~2~1~0]}~{[7~6~5~4~3~2~1~0]}{.}
\end{align*}\vspace{-1.2em}\end{minipage}}
\endgroup
  \caption[An \textAleph\ implementation of insertion sort with arbitrary comparator and applied to an example list.]{An \textAleph\ implementation of the comparison operators $<$, $\le$, $>$ and $\ge$, and of insertion sort that can make use of these comparison operators. Lastly, the list $[{3}~{2}~{0}~{7}~{6}~{4}~{5}~{1}]$ is sorted into both ascending and descending order.}
  \label{lst:ex-sort}
  \begingroup%
\def\nl{\\[1.5em]}%
\begin{minipage}{.6\linewidth}\begin{align*}
  & \cursivexS~\infix{\atom{InsertionSort}~{p}}~{\cursivenS'}~\cursiveyS{:} \\
  &\qquad \cursivenS~\infix{\atom{Reverse}}~{\cursivenS'}{.}
    ~\comment{\cmtSLopen~list reversal} \\
  &\qquad \llap{\haltSym~}\atomLocal{Go}~{p}~\cursivexS~{[]}~{[]} =
              \atomLocal{Go}~{p}~\cursivexS~{[{n}~{\cdot}~\cursivenS]}~{\cursiveyS'}{.} \\
  &\qquad \atomLocal{Go}~{p}~{[{x}~{\cdot}~\cursivexS]}~\cursivenS~\cursiveyS = 
          \atomLocal{Go}~{p}~\cursivexS~{[{n}~{\cdot}~{\cursivenS'}]}~{\cursiveyS'}{:} \\
  &\qquad\qquad {x}~\cursiveyS~\infix{\atom{Insert}~{p}}~{n}~{\cursiveyS'}{.} \\
  &\qquad\qquad \cursivenS~\infix{\atom{Map}~{(\atom{IS$_1$}~{n})}}~{\cursivenS'}{.}
\end{align*}\end{minipage}\hfill%
\begin{minipage}{.35\linewidth}\begin{align*}
  & {m}~\infix{\atom{IS$_1$}~{n}}~{m'}{:} \\
  &\qquad \infix{{<}~{m}~{n}}~{b}{.} \\
  &\qquad \infix{{<}~{m'}~{n}}~{b}{.} \\
  &\qquad {m}~\infix{\atom{IS$_2$}~{b}}~{m'}{.} \\
  & {m}~\infix{\atom{IS$_2$}~\atom{True}}~{m}{;} \\
  & {m}~\infix{\atom{IS$_2$}~\atom{False}}~{(\S{m})}{;}
\end{align*}\end{minipage}
\\[1em]
\fbox{\begin{minipage}{.95\textwidth}\vspace{-0.2em}\begin{align*}
  & {[77~2~42~68~41~36~8~36]}~\infix{\atom{InsertionSort}~{<}}~{[7~0~5~6~4~2~1~3]}~{[2~8~36~36~41~42~68~77]}{.}
\end{align*}\vspace{-0.8em}\end{minipage}}%
\endgroup
  \captionsetup{singlelinecheck=off}
  \caption[The \alethe\ implementation of insertion sort from the standard library.]{The \alethe\ implementation of insertion sort from the standard library. In contrast to the implementation in \Cref{lst:ex-sort}, this definition yields more useful `garbage' data in that $\cursivenS'$ contains the permutation which maps $\cursivexS$ to $\cursiveyS$, as shown in the boxed example. Specifically, this corresponds to the following permutation (in standard notation):
    \begin{equation*}\qty\Big(\,\begin{matrix}
      0 & 1 & 2 & 3 & 4 & 5 & 6 & 7 \\
      7 & 0 & 5 & 6 & 4 & 2 & 1 & 3
    \end{matrix}\,).\end{equation*}}
  \label{lst:ex-sort2}
\end{listing}

\para{Sorting}

A larger example program, which is capable of sorting a list using an arbitrary comparison function, is presented in \Cref{lst:ex-sort}. Whilst of poor algorithmic complexity, insertion sort is employed for simplicity. A more efficient implementation using merge sort is available in the accompanying standard library\footnote{\url{https://github.com/hannah-earley/alethe-examples}.}. Clearly sorting is an irreversible process, as its purpose is to discard information regarding the original ordering of the list. The presented sorting algorithm puts a little effort towards increasing the utility of the garbage, in that the garbage data is a list saying where the original list item was placed into the sorted list \emph{at the moment of insertion}. The insertion sort included in the standard library (\Cref{lst:ex-sort2}) applies some additional processing to make this garbage data correspond to the permutation that maps the original list to the sorted list.

\para{Concurrency}

In addition to automatic parallelisation of independent sub-rules, \textAleph\ also supports defining transitions between separate terms. The motivation for this is for \textAleph\ to be able to fully exploit molecular architectures, but its utility is general and it is intended that a future version of the \alethe\ interpreter will support concurrency.

\subpara{Biasing Computation}

In \Cref{fig:ex-plus} and \Cref{lst:ex-square-def-mol}, abstract molecular implementations of addition and squaring were introduced respectively. Notably, the reaction arrows---whilst reversible---were biased in the forward direction. From a thermodynamic perspective, the direction in which a reaction occurs cannot be specified in an absolute sense, and depends on the conditions of the reaction volume at a given point in time. Moreover, at equilibrium the net direction of each reaction is null: they make no net progress. These concerns were discussed in more detail in \Cref{chap:revi}, but suffice it to say that trying to arrange the conditions of the system such that the computational terms themselves are inherently biased is impracticable, and will also limit the number of transition steps that can be executed. Moreover, it makes it difficult to run sub-rules in a reverse direction (such as is needed for Bennett's algorithm).

The solution is to separate the concerns of computation and biasing said computation. In biochemical systems, this is (mostly) achieved by using a common free energy carrier in the form of $\ce{ATP}$ and $\ce{ADP + P_{\textrm{i}}}$, which are held in disequilibrium such that the favourable reaction direction is $\ce{ATP + H2O <=>> ADP + P_{\textrm{i}} + H+}$. Other biochemical reactions are then coupled to one or more copies of this hydrolysis reaction. For our purposes, we generalise this to assume that free energy is available in the form of two terms $\oplus$ and $\ominus$, where the concentration of $\oplus$ in the reaction volume is greater than that of $\ominus$. The `computational bias' quantifying the average net proportion of computational transitions that are successful is found to be $b=([\oplus]-[\ominus])/([\oplus]+[\ominus])$.

With a bias system thus defined, we can couple any computational transition to it by simply having the input side consume a $\oplus$ term and the output side release a $\ominus$ term. The more transitions that couple to the bias source, the faster and more robustly the computation proceeds. For example, the square definition can be amended like so:
\begin{align*}
  & \haltSym~{\square}~{\blank}~{\unit}{;}\quad \haltSym~{\unit}~{\blank}~{\square}{;} \\
  & \partyBag\Big{&\oplus \\ &{\square}~{n}~{\unit}} = \partyBag\Big{&\ominus \\ &{\square}~{n}~{\Z}~{\square}}{;} && \ruleName{square--init} \\
  & \partyBag\Big{&\oplus \\ &{\square}~{\Z}~{m}~{\square}} = \partyBag\Big{&\ominus \\ &{\unit}~{m}~{\square}}{;} && \ruleName{square--term} \\
  &\partyBag\Big{&\oplus \\ &{\square}~({\S}{n})~{m}~{\square}} =
   \partyBag\Big{&\ominus \\ &{\square}~n~({\S}{m''})~{\square}}{:} && \ruleName{square--step}\\
  &\qquad {+}~{n}~{m}~{\unit} = {\unit}~{n}~{m'}~{+}{.} && \ruleName{square--step--sub$_1$}\\
  &\qquad {+}~{n}~{m'}~{\unit} = {\unit}~{n}~{m''}~{+}{.} && \ruleName{square--step--sub$_2$}
\end{align*}
The braces indicate that these definitions are concurrent, in that the transition will draw (release) two terms from (into) the reaction volume. As written above, the scheme is fairly basic in that it only drives computation in one direction. An improved scheme may make use of a hidden variable in each term to indicate its preferred direction of computation. If this direction is reversed, then the transition will instead draw a $\ominus$ term and release a $\oplus$ term. If a sub-rule then needs to be run in the opposite direction to its parent, then the hidden variable need simply be inverted. This scheme can be made even more sophisticated by allowing the hidden variable to refer to alternate bias sources in case there are multiple to which the transition could couple (as was recommended in \Cref{chap:revii,chap:reviii}).

\subpara{Communication}

A more canonical application of concurrency is that of communication between different computers. There are many possible communication schemes, and an overview and analysis of schemes appropriate for Brownian computers was presented in \Cref{chap:revii}, but we shall consider the simple example of an open communication channel between two computers, \atom{Alice} and \atom{Bob}. \atom{Alice} has a list, the contents of which she wishes to send to \atom{Bob}. A naive approach to this may resemble the \textAleph\ definition
\begin{align*}
  & \atom{Alice}~[{x}~{\cdot}~\cursivexS] = \partyBag\Big{&\atom{Alice}~\cursivexS\\&\atom{Courier}~{x}}{;} &
  & \partyBag\Big{&\atom{Bob}~\cursiveyS\\&\atom{Courier}~{y}} = \atom{Bob}~[{y}~{\cdot}~\cursiveyS]{;}
\end{align*}
where the \atom{Courier} terms are used to convey list items between the two computers. Unfortunately, this does not quite realise the desired behaviour in a Brownian context: Suppose \atom{Alice} starts with the list $[1~2~3~4~5]$. She will generate the courier terms $(\atom{Courier}~1)$, $(\atom{Courier}~2)$, \ldots, $(\atom{Courier}~5)$, as expected, but these terms will be delievered to \atom{Bob} via a diffusive approach. Except for the case of a one-dimensional system, \atom{Bob} will almost certainly receive these couriers in a random order, completely uncorrelated from the original order. This may be avoided if delivery over the channel is substantially faster than the rates of term fission/fusion reactions, but this is not likely in a chemical system. Moreover, this being a reversible Brownian system, \atom{Alice} will from time to time re-absorb a courier that she previously dispatched, thereby increasing the opportunity for the list order to be shuffled.

The take-home message, here, is that whilst the local dynamics of a concurrent reversible computation system---such as \textAleph---may well be reversible, the global dynamics are not necessarily reversible. In fact, this is a common property of microscopically reversibly systems, and is the basis of thermodynamics: the manifestation of macroscopically irreversible dynamics from microscopically reversible physics. It is doubtful that a model of concurrent reversible programming could preclude the possibility of such increases in entropy without severely restricting the system: likely any attempt to do so would effectively prevent the use of concurrency. Nevertheless, the programmer has a high degree of control here and so can, in principle, avoid entropy generation except where desired. In the above case of sending an ordered list, one could for example explicitly sequence the couriers. This suggests that a system of reversible molecular computers could have exceptionally low entropy generation and could realise very strange behaviours by getting as close to the implementation of a Maxwell D{\ae}mon as physics allows. Conversely, one is also free to exploit the thermodynamic properties of the system; for example, one could realise (to the extent the laws of physics permit) a true random number generator.

\subpara{Resources}

Our last example of concurrency concerns the distribution of conserved resources amongst computers. If our computers are to respect the reversibility of the laws of physics, they should also respect mass conservation and so will need to contend with their finity. Amorphous computing presents one approach to achieving powerful computation from limited computational subunits, but it does so by making extensive use of communication which precludes the ability to perform meaningfully reversible computation. We instead suppose that the computers can exchange resources, such as memory units and structural building blocks, with the environment. Whilst this also has thermodynamic consequences, as shown in \Cref{chap:reviii}, it at least separates the concerns of computation from resource access and thus preserves the ability to perform reversible computation. 

A host of resource distribution schemes, and thermodynamic analyses thereof, were presented in \Cref{chap:reviii}, but we content ourselves here with demonstrating how computers may interact with resources free in solution. In particular, we amend the definition of addition to be mass-conserving:
\begin{align*}
  & {+}~{\Z}~{b}~{\unit} = {\unit}~{\Z}~{b}~{+}{;} && \ruleName{add--base} \\
  & \partyBag\Big{&\S\\&{+}~({\S} {a})~{b}~{\unit}} = {\unit}~({\S} {a})~({\S} {b'})~{+}{:} && \ruleName{add--step} \\
  &\qquad  {+}~{a}~{b}~{\unit} = {\unit}~{a}~{b'}~{+}{.} && \ruleName{add--step--sub}
\end{align*}

\para{Effects and Contexts}

As briefly mentioned, composite terms are subject to the same transition rules. Suppose \atom{Lena} is learning \textAleph\ and experiments with this, creating the following trivial wrapper around $\square$:
\begin{align*}
  & \haltSym~{\blank}~\infix{\atom{MySquare}}~{\blank}{;} \\
  & {\atom{MySquare}}~{n}~{\unit} = {\atom{MySquare}}~{({\square}~{n}~{\unit})}~{\atom{MySquare}}{;} \\
  & {\atom{MySquare}}~{({\unit}~{m}~{\square})}~{\atom{MySquare}} = {\atom{MySquare}}~{m}~{\unit}{;}
\end{align*}
This definition, if a little contrived, will function as intended. Suppose, now, that she wants to inspect what happens within the loop of $\square$, and so writes the following:
\begin{align*}
  & \haltSym~{\blank}~\infix{\atom{MySquare'}}~{\blank}{;} \\
  & {\atom{MySquare'}}~{n}~{m}~{\unit} = {\atom{MySquare'}}~{({\square}~{n}~{m}~{\square})}~{\atom{MySquare'}}{;} \\
  & {\atom{MySquare'}}~{({\square}~{n'}~{m'}~{\square})}~{\atom{MySquare'}} = {\atom{MySquare'}}~{n'}~{m'}~{\unit}{;}
\end{align*}
Now, what happens if we instantiate the term ${\atom{MySquare'}}~{3}~{\Z}~{\unit}$? Well, the final term could be any of ${\unit}~{3}~{\Z}~{\atom{MySquare'}}$, ${\unit}~{2}~{5}~{\atom{MySquare'}}$, ${\unit}~{1}~{8}~{\atom{MySquare'}}$ or ${\unit}~{\Z}~{9}~{\atom{MySquare'}}$. Whilst this arguably achieves \atom{Lena}'s aim of inspecting the execution of $\square$, there is a problem in that \textAleph\ claims all of these terms are halting states yet \Cref{fig:state-space} asserts there should be a maximum of two halting states. Despite being contrived, this example shows that we need to introduce another constraint into \textAleph\ to prevent unexpected entropy generation: Any composite term or sub-term can only be created in a halting state, and can only be consumed in a halting state. Moreover, the creation and consumption patterns must be unambiguous as to which of the term's two halting states they refer (if, indeed, there are two). Otherwise, each instance of this ambiguity would multiply the state space and thus there would be an exponential increase in the size of the state space over time.

With the constraint on composite terms clarified, one can then ask whether all terms must follow the same set of rules. It turns out there is an important case where distinguishing between composite terms and top-level terms is useful; specifically, some cases of effectful computation. Suppose that the reaction volume has been endowed with a spatial lattice along which terms can travel, and imagine a term $(\atom{Cons}~\atom{Charlie}~\atom{Dan})$ with such a wanderlust. If the composite terms \atom{Charlie} and \atom{Dan} want to travel to different locations, then where will the term end up? It is likely that a tug-of-war will occur, and so it makes sense to restrict the effecting of translocation to top-level terms.

This distinction is achieved by introducing term contexts: a top-level term has a top-level context, which is some label (itself a term) that may contain information, whilst composite terms have `one-hole contexts'. For example, a term attached to a lattice may have context \atom{Lattice} while a term free in solution might have context \atom{Free}. Concretely, the aforementioned tug-of-war may be resolved by giving \atom{Charlie} control by making him the top-level term, rendered as $\atom{Lattice}:\atom{Charlie}~{(\atom{Dan})}$. Meanwhile, \atom{Dan} is rendered as $(\atom{Lattice}:(\atom{Charlie}~{\bullet})):\atom{Dan}$ where $\bullet$ indicates that this is a one-hole context. One-hole contexts can be defined for many data structures, but for a tree they correspond to removing some sub-tree of interest, replacing it with a `hole'; see \textcite{ohc-diff} for some interesting properties of one-hole contexts. Rules may pattern match on top-level contexts and thus consume the information held within, or even create and destroy top-level terms, but they cannot match on one-hole contexts as this would risk altering their structure; instead, one-hole contexts can only be matched by `opaque variables'. An opaque variable is a special variable found only in context patterns, which can match against a top-level context or a one-hole context, whilst regular variables in a context pattern can only match against top-level contexts. That is, the following are allowed
\begin{align*}
  & \partyBag{\atom{Lattice}&:\atom{Charlie}~{x}} = 
  \partyBag\Big{\atom{Lattice}&:\atom{Charlie}~{\unit} \\ \atom{Free}&:{x}}{;} \\
  & \partyBag{{\gamma}&:\atom{Dan}~{[\atom{Dan's Stuff}]}} =
  \partyBag\Big{{\gamma}&:\atom{Dan}\\\atom{Lattice}&:{[\atom{Dan's Stuff}]}}{;}
\end{align*}
whilst this is not
\begin{align*}
  & \partyBag{(\atom{Lattice}:(c~{\bullet}))&:\atom{Dan}} =
  \partyBag{(\atom{Lattice}:({\bullet}~\atom{Dan}))&:c}{;}
\end{align*}
This is not too onerous a restriction, as if one wishes to manipulate the structure of the one-hole context one can simply match against a higher level context.

In the earlier \textAleph\ definitions, contexts were missing from the rule patterns. This can be seen as another example of syntactic sugar, with a missing context implying an opaque variable context (i.e.\ $\gamma:$) such that the rule can match against any term at any level. For concurrent definitions, however, all participating terms must be contextualised as otherwise it is unclear how to assign the results to the bound one-hole contexts.

It is as yet unclear how translocation along a lattice, or other effects, is actually achieved. Typically effects will be introduced as additional computational primitives, as by definition their actions are not `computational'. As these computational primitives must be instantiated as top-level terms, we require a way to interface between computational terms and effector terms and this warrants a continuation-passing-style approach. For example, walking along a lattice from coordinate $\alpha$ to coordinate $\beta$ may be realised thus:
\begin{align*}
  & \haltSym~\atom{Charlie}'~{\blank}{;} \quad \haltSym~\atom{Charlie}''~{\blank}{;} \\
  & \partyBag{\atom{Latt}&:\atom{Charlie}~{\beta}~{x}} =
    \partyBag{\atom{Latt}&:\atom{Walk}~{\beta}~{(\atom{Charlie}'~x)}}{;} \\
  \mbox{\emph{(primitive)}}\quad &
    \partyBag{\atom{Latt}&:\atom{Walk}~{\beta}~c} \leftrightarrow
    \partyBag{\atom{Latt}&:\atom{Walk}'~{\alpha}~c}{;} \\
  & \partyBag{\atom{Latt}&:\atom{Walk}'~{\alpha}~{(\atom{Charlie}''~x)}} =
    \partyBag{\atom{Latt}&:\atom{Charlie}'''~{\alpha}~{x}}{;}
\end{align*}
That is, \atom{Walk} is given a coordinate to travel to as well as a continuation (which is free to perform additional computation during the walk, if desired). \atom{Walk} then replaces the destination coordinate with the origin coordinate to ensure reversibility. Finally, we define a transition from $\atom{Walk}'$ in order to return control to the continuation. This reveals a subtle point, that the effectful primitives are intentionally not marked as halting so that control can be transferred to and fro' them.

\endgroup

\section{The Calculus}

The \textAleph-calculus thus introduced has a very simple definition. In BNF notation, it is
\begin{equation*}\begin{aligned}
  \ruleName{pattern term} && \tau &::= \bnfPrim{atom} ~|~ \bnfPrim{var} ~|~ (\,\tau^\ast\,) \\
  \ruleName{party} && \pi &::= \tau:\tau^\ast ~|~ \bnfPrim{var}':\tau^\ast \\
  \ruleName{definition} && \delta &::= \{\pi^\ast\}=\{\pi^\ast\}:\pi\rlap.^\ast ~|~ \haltSym~{\tau}
\end{aligned}\end{equation*}
where \bnfPrim{atom} is an infinite set of atomic symbols (conventionally starting with an uppercase letter or a symbol), \bnfPrim{var} an infinite set of variables (conventionally starting with a lowercase letter), and $\bnfPrim{var}'$ is an orthogonal infinite set of variables (conventionally rendered in Greek) used for opaquely matching one-hole-contexts. A program is a series of definitions, $\delta^\ast$, and a physical term is simply a pattern term without variables. Notice that the form of the sub-rules differs from the examples: a sub-rule can be separated into the instantiation of a term according to an input pattern, the evolution of that term, followed finally by the consumption of its final halting state according to the output pattern. In this view, these sub-terms are identified by one-hole-contexts---specifically, opaque variables. This formulation not only allows the semantics to represent automatic parallelisation and a non-deterministic sub-rule ordering, but also begets an additional feature whereby sub-rules can instantiate top-level terms. The sub-rule forms in the examples can then be seen as sugar, i.e.\ $s=t.$ is equivalent to $\lambda:s.~\lambda:t.$ where $\lambda\in\bnfPrim{var}'$ is a fresh opaque variable. To illustrate, recall the definition of natural addition from \Cref{lst:ex-add-def}; desugared, this takes the form
  {\def\ruleLabel#1{\llap{\tiny$#1.$}~}\begin{align*}
    \ruleLabel{1}& \haltSym~{+}~{a}~{b}~{\unit}{;} \\
    \ruleLabel{2}& \haltSym~{\unit}~{a}~{b}~{+}{;} \\
    \ruleLabel{3}& \{\alpha:{+}~{\Z}~{b}~{\unit}\} = \{\alpha:{\unit}~{\Z}~{b}~{+}\}{;} && \ruleName{add--base} \\
    \ruleLabel{4}& \{\alpha:{+}~({\S} {a})~{b}~{\unit}\} = \{\alpha:{\unit}~({\S} {a})~({\S} {b'})~{+}\}{:} && \ruleName{add--step} \\
    \raisebox{0pt}[1.9em]{}%
      {\begin{aligned}\ruleLabel{4a}&\\\ruleLabel{4b}&\end{aligned}}&%
      {\begin{aligned}
        &\qquad  \beta: {+}~{a}~{b}~{\unit}{.} \\
        &\qquad  \beta: {\unit}~{a}~{b'}~{+}{.}
      \end{aligned}} && \ruleName{add--step--sub}
  \end{align*}}%
That is, there are four definitions: two `halting', and two `computational'. On the left side of definition 4, we have a bag of one party. This party has as context pattern the opaque variable $\alpha$, and its pattern term is a composite pattern term of length 4, consisting of the atom $+$, the composite pattern term of length 2 consisting of the atom $\S$ and the variable $a$, the variable $b$, and the empty composite pattern term (`unit'). There is also an opaque variable $\beta$, confined to the inner scope of definition 4. During execution of the sub-rule in the forward direction, the variables $a$ and $b$ will first be consumed in order to generate the sub-term ${+}~{a}~{b}~{\unit}$, which will be bound to the opaque variable $\beta$. This sub-term will then evolve to its other halting state, whereupon it will match sub-rule $4b$ and thus $\beta$ will be consumed and the variables $a$ and $b'$ obtained.

\para{Semantics}

To formalise the semantics embodied in the preceding examples, we define a transition relation $\leftrightarrow$ that maps a bag of terms to another bag of terms according to the rules defined for the current program. By a bag of terms, we aim to evoke the notion of a concoction of computational molecules in solution; a bag is a multiset, which can contain multiple copies of elements, but unlike a physical solution it lacks a concept of space. If desired, spatial locations can be readily simulated by appropriately chosen top-level contexts. Recalling that the global dynamics of the system are irreversible, we should expect that $\leftrightarrow$ is a non-deterministic relation. That is, application of the relation to a set of possible bags of terms may increase the size of this set, and the logarithm of the set size is identified with the entropy of the system. Concretely, the set of bags is the macrostate of the system, and each bag within the set is a microstate.

The calculus being reversible, the relation should share properties with equivalence relations. Namely, if $S,T,U$ are bags of terms, then we have
\begin{equation*}\begin{aligned}
  && S &\leftrightarrow S && \ruleName{refl} \\
  S \leftrightarrow T \implies&& T &\leftrightarrow S && \ruleName{symm} \\
  S \leftrightarrow T \land T \leftrightarrow U \implies&& S &\leftrightarrow U && \ruleName{trans}
\end{aligned}\end{equation*} 
but we also make clear that these transitions can occur within an environment of other non-participating terms, $\Gamma$,
\begin{equation*}\begin{aligned}
  S \leftrightarrow T \implies&& \forall\Gamma.~ \Gamma \cup S &\leftrightarrow \Gamma \cup T && \ruleName{ext} \\
  s \leftrightarrow t \implies&& \forall\Gamma.~ \Gamma \cup \{s\} &\leftrightarrow \Gamma \cup \{t\} && \ruleName{ext$'$}
\end{aligned}\end{equation*}
where the second rule, in which $s$ and $t$ are single terms rather than bags of terms, is introduced as a convenient abuse of notation for later definitions.

In order to physically effect rule transitions, we introduce `mediator terms' delimited by angle brackets, which interact with computational terms and can represent each of the intermediate states. If the program is given by $\mathcal P$, then the mechanism by which these mediator terms come into and out of existence is as follows,
\begin{equation*}\begin{aligned}
  \mathcal P\vdash (I=O:R) \implies&& \{\}&\leftrightarrow\{\langle I\varnothing\varnothing R\varnothing O\rangle\} && \ruleName{inst--comp} \\
  \mathcal P\vdash \rlap{$(\haltSym~{\tau})$}\phantom{(I=O:R)} \implies&& \{\}&\leftrightarrow\{\langle \tau\rangle\} && \ruleName{inst--halt}
\end{aligned}\end{equation*}
from which we see that the environment can contain an arbitrary number of copies of each\footnote{This complicates the aforementioned entropic interpretation of sets of term-bags, as the entropy will tend to diverge as the number of mediator terms tends to infinity. A more realistic implementation would leave the mean number of extant mediator terms finite and bounded, and perhaps even fix the number. This would have a further consequence on the kinetics and thermodynamics of the system, with the well-characterised Mich{\ae}lis-Menten kinetics a good starting point to the analysis thereof.}. The mediator terms for computational rules are sextuples $\langle II'BR\Gamma O\rangle$ consisting of, respectively, the bag of unmatched input patterns $I$, the bag of matched input patterns $I'$, the bag of resultant bindings $B$, the bag of sub-rules $R$, the local environment for internal sub-rules $\Gamma$, and the bag of output patterns $O$. The mediator terms for halting rules are trivial singletons $\langle\tau\rangle$ containing the relevant pattern, $\tau$.

To assist rules in binding composite terms, we permit the current focus of a term to vary over time,
\begin{equation*}\begin{aligned}
  c:(\vec l~{\ooalign{\hss$t$\hss\cr\phantom{$\bullet$}}}~\vec r) &\leftrightarrow (c:(\vec l~{\bullet}~\vec r)):t && \ruleName{ohc$_1$} \\
  c:[\vec l~{\ooalign{\hss$t$\hss\cr\phantom{$\bullet$}}}~\vec r] &\leftrightarrow (c:[\vec l~{\bullet}~\vec r]):t && \ruleName{ohc$_2$}
\end{aligned}\end{equation*}
where $\vec l$ and $\vec r$ are (possibly empty) strings of terms and $t$ is the term focus. The bracketed terms in the second rule will be explained shortly. Note that these one-hole-context rules may not be needed for all architectures, being implicitly true in a molecular architecture for example.

The action of halting mediators is simply to mark terms which are in a known halting state,
\begin{equation*}\begin{aligned}
  \exists b'.~t \overset{\tau}{\sim}b' \implies&& \{x:t,\langle\tau\rangle\} &\leftrightarrow \{x:[t],\langle\tau\rangle\} && \ruleName{term}
\end{aligned}\end{equation*}
with $[t]$ serving as an indicator of a halting state and where $t\overset\tau\sim b'$ means that $t$ unifies against $\tau$ with bindings $b'$ (see \Cref{lst:aleph-sem-unif}).

Computational mediators are somewhat more complicated. We first render their temporal symmetry manifest by two reversibility rules,
\begin{equation*}\begin{aligned}
  \langle I\varnothing\varnothing R\varnothing O\rangle &\leftrightarrow \langle O\varnothing\varnothing R\varnothing I\rangle && \ruleName{rev$_1$} \\
  \langle\varnothing IBR\Gamma O\rangle &\leftrightarrow \langle\varnothing OBR\Gamma I\rangle && \ruleName{rev$_2$}
\end{aligned}\end{equation*}
where these rules will be seen to be necessary even for computation in a single direction. Applying a computational rule consists first of binding against a matching term for each input pattern, followed by substituting the bindings into sub-terms per the sub-rules. The sub-terms may either be top-level or local. These transitions may all occur in parallel, i.e.\ a sub-term may be instantiated if all of its requisite variables are bound, even if not all the input patterns have matched a term.
{\def\ruleHeader#1#2{\llap{$\displaystyle #1 \hspace{5em}$}\rlap{\hspace{9.5em}\ruleName{#2}}}
\begin{equation*}\begin{aligned}
  &\ruleHeader{x:t\overset\pi\sim b \implies}{inp} \\
  \{x:t, \langle(I\cup\{\pi\})I'BR\Gamma O\rangle &\leftrightarrow
  \{\langle I(I'\cup\{\pi\})(B\cup b)R\Gamma O\rangle\} \\
  &\ruleHeader{\pi\in R \land x:t\overset\pi\sim b \implies}{sub$_1$} \\
  \langle II'(B\cup b)R\Gamma O\rangle &\leftrightarrow
  \langle II'BR(\Gamma\cup\{x:[t]\})O\rangle \\
  &\ruleHeader{x \notin \bnfPrim{var}' \implies}{sub$_2$} \\
  \{\langle II'BR(\Gamma\cup\{x:[t]\})O\rangle\} &\leftrightarrow
  \{x:[t],\langle II' BS\Gamma O\rangle\}
\end{aligned}\end{equation*}}
where $x$ is a context. The \ruleName{sub$_2$} transition enables top-level terms to escape the locally scoped environment. Completion of computation occurs by application of \ruleName{rev$_2$} followed by the \ruleName{inp} and \ruleName{sub$_{1,2}$} in reverse; clearly if there is no route from the set of input variables to the set of output variables then the computation will stall. It may be desirable to augment the transition rules with implicit variable duplication,
\begin{equation*}\begin{aligned}
  \langle II'(B\cup\{b\})R\Gamma O\rangle &\leftrightarrow\langle II'(B\cup\{b,b\})R\Gamma O\rangle && \ruleName{dup}
\end{aligned}\end{equation*}
otherwise rules which wish to increase their parallelisability should explicitly duplicate variables as needed. We shall also need to enable the sub-environment to evolve,
\begin{equation*}\begin{aligned}
  \Gamma\leftrightarrow\Gamma' \implies && \langle II'BR\Gamma O\rangle &\leftrightarrow\langle II'BR\Gamma'O\rangle && \ruleName{sub$_3$}
\end{aligned}\end{equation*}

These semantics, encapsulated by the $\leftrightarrow$ transition rule, are summarised in \Cref{lst:aleph-sem}. An example of their application is provided in \Cref{lst:semex-add}. It remains to describe the operation of binding/unification. This operation is defined as one would expect: a pattern matches if it is a variable, if it is an atom and the term is the same atom, or if it is a composite pattern and the term is a \emph{halting} composite term of the same length and if the pattern and term match pair-wise. This is summarised in \Cref{lst:aleph-sem-unif}.

\doublepage{%
\begin{listing}[!p]\centering
\fbox{\begin{minipage}{0.9\textwidth}\vspace{\baselineskip}\begin{equation*}\begin{aligned}
    && S &\leftrightarrow S && \ruleName{refl} \\
    S \leftrightarrow T \implies&& T &\leftrightarrow S && \ruleName{symm} \\
    S \leftrightarrow T \land T \leftrightarrow U \implies&& S &\leftrightarrow U && \ruleName{trans} \\
    S \leftrightarrow T \implies&& \forall\Gamma.~ \Gamma \cup S &\leftrightarrow \Gamma \cup T && \ruleName{ext} \\
    s \leftrightarrow t \implies&& \forall\Gamma.~ \Gamma \cup \{s\} &\leftrightarrow \Gamma \cup \{t\} && \ruleName{ext$'$} \\
    \mathcal P\vdash (I=O:R) \implies&& \{\}&\leftrightarrow\{\langle I\varnothing\varnothing R\varnothing O\rangle\} && \ruleName{inst--comp} \\
    \mathcal P\vdash \rlap{$(\haltSym~{\tau})$}\phantom{(I=O:R)} \implies&& \{\}&\leftrightarrow\{\langle \tau\rangle\} && \ruleName{inst--halt} \\
    && c:(\vec l~{\ooalign{\hss$t$\hss\cr\phantom{$\bullet$}}}~\vec r) &\leftrightarrow (c:(\vec l~{\bullet}~\vec r)):t && \ruleName{ohc$_1$} \\
    && c:[\vec l~{\ooalign{\hss$t$\hss\cr\phantom{$\bullet$}}}~\vec r] &\leftrightarrow (c:[\vec l~{\bullet}~\vec r]):t && \ruleName{ohc$_2$} \\
    \exists b'.~t \overset{\tau}{\sim}b' \implies&& \{x:t,\langle\tau\rangle\} &\leftrightarrow \{x:[t],\langle\tau\rangle\} && \ruleName{term} \\
    && \langle I\varnothing\varnothing R\varnothing O\rangle &\leftrightarrow \langle O\varnothing\varnothing R\varnothing I\rangle && \ruleName{rev$_1$} \\
    && \langle\varnothing IBR\Gamma O\rangle &\leftrightarrow \langle\varnothing OBR\Gamma I\rangle && \ruleName{rev$_2$} \\
    x:t\overset\pi\sim b \implies &&&&& \ruleName{inp} \\
    &&& \omit{\llap{$\displaystyle\{x:t, \langle(I\cup\{\pi\})I'BR\Gamma O\rangle\}$}\rlap{$\displaystyle{}\leftrightarrow\{\langle I(I'\cup\{\pi\})(B\cup b)R\Gamma O\rangle\}$}} \\
    \pi\in R \land x:t\overset\pi\sim b \implies &&&&& \ruleName{sub$_1$} \\
    &&& \omit{\llap{$\displaystyle\langle II'(B\cup b)R\Gamma O\rangle$}\rlap{$\displaystyle{}\leftrightarrow\langle II'BR(\Gamma\cup\{x:[t]\})O\rangle$}} \\
    x \notin \bnfPrim{var}' \implies &&&&& \ruleName{sub$_2$} \\
    &&& \omit{\llap{$\displaystyle\{\langle II'BR(\Gamma\cup\{x:[t]\})O\rangle$}\rlap{$\displaystyle{}\leftrightarrow\{x:[t],\langle II' BR\Gamma O\rangle\}$}} \\
    \Gamma\leftrightarrow\Gamma' \implies && \langle II'BR\Gamma O\rangle &\leftrightarrow\langle II'BR\Gamma'O\rangle && \ruleName{sub$_3$}
\end{aligned}\end{equation*}\vspace{\baselineskip}\end{minipage}}
  \caption{Summary of \textAleph\ semantics.}
  \label{lst:aleph-sem}
\end{listing}%
\begin{listing}[!p]\centering
\fbox{\begin{minipage}{0.9\textwidth}\vspace{.5\baselineskip}\begin{equation*}\begin{array}{clcl}
    \cfrac{\alpha\in\bnfPrim{atom}}{\alpha \overset{\alpha}{\sim} \varnothing} & \ruleName{unif--atom} &
    \cfrac{v\in\bnfPrim{var}}{t \overset{v}{\sim} \{v\mapsto t\}} & \ruleName{unif--var} \\
    \cfrac{\bigwedge_i t_i \overset{\pi_i}{\sim} b_i}{[t_1\cdots t_n] \mathrel{\overset{\pi_1\cdots\pi_n}{\scalebox{1.75}[1]{$\sim$}}} \bigcup_i b_i} & \ruleName{unif--sub} &
    \cfrac{\lambda\in\bnfPrim{var}'\quad t\overset{\pi}{\sim}b}{\gamma:t \overset{\lambda:\pi}{\sim} \{\lambda\mapsto\gamma\}\cup b} & \ruleName{unif--ctxt$_1$} \\
    \cfrac{\gamma\overset{\pi_1}{\sim}b_1\quad t\overset{\pi_2}{\sim}b_2}{\gamma:t \mathrel{\overset{\pi_1:\pi_2}{\scalebox{1.5}[1]{$\sim$}}} b_1\cup b_2} & \ruleName{unif--ctxt$_2$}
\end{array}\end{equation*}\vspace{.5\baselineskip}\end{minipage}}
  \caption{Summary of unification semantics for \textAleph.}
  \label{lst:aleph-sem-unif}
\end{listing}}{%
\begin{listing}[!p]
  \centering
  \begingroup
\begin{sublisting}{\linewidth}
  \def\ruleLabel#1{\llap{\tiny$#1.$}~}
  \begin{align*}
    \ruleLabel{1}& \haltSym~{+}~{a}~{b}~{\unit}{;} \\
    \ruleLabel{2}& \haltSym~{\unit}~{a}~{b}~{+}{;} \\
    \ruleLabel{3}& \{\alpha:{+}~{\Z}~{b}~{\unit}\} = \{\alpha:{\unit}~{\Z}~{b}~{+}\}{;} && \ruleName{add--base} \\
    \ruleLabel{4}& \{\alpha:{+}~({\S} {a})~{b}~{\unit}\} = \{\alpha:{\unit}~({\S} {a})~({\S} {b'})~{+}\}{:} && \ruleName{add--step} \\
    \raisebox{0pt}[1.9em]{}%
      {\begin{aligned}\ruleLabel{4a}&\\\ruleLabel{4b}&\end{aligned}}&%
      {\begin{aligned}
        &\qquad  \beta: {+}~{a}~{b}~{\unit}{.} \\
        &\qquad  \beta: {\unit}~{a}~{b'}~{+}{.}
      \end{aligned}} && \ruleName{add--step--sub}
  \end{align*}
  \caption{Desugared definition of reversible natural addition in \textAleph. It will be convenient to make the definitions $I_4=\{\alpha:{+}~({\S} {a})~{b}~{\unit}\}$, $O_4=\{\alpha:{\unit}~({\S} {a})~({\S} {b'})~{+}\}$ and $R_4=\{\beta: {+}~{a}~{b}~{\unit}, \beta: {\unit}~{a}~{b'}~{+}\}$.}
  \label{lst:semex-add-def}
\end{sublisting}

\begin{sublisting}{\linewidth}
  \def\adj{~}
  \def\smS{\textsc{\scriptsize S}}
  \begin{align*}
  &\phantom{{}\leftrightarrow{}} \{ {\unit}{:}[{+}{3}{4}{\unit}] \} && \\
  \mathcal P \vdash \haltSym{+}{a}{b}{\unit} \implies\adj& \leftrightarrow \{\langle{+}{a}{b}{\unit}\rangle,{\unit}{:}[{+}{3}{4}{\unit}]\} && \ruleName{inst--halt} \\
  & \leftrightarrow \{\langle{+}{a}{b}{\unit}\rangle,{\unit}{:}({+}{3}{4}{\unit})\} && \ruleName{term} \\
  \mathcal P \vdash \haltSym{+}{a}{b}{\unit} \implies\adj& \leftrightarrow \{{\unit}{:}({+}{3}{4}{\unit})\} && \ruleName{inst--halt} \\
  \mathcal P \vdash (I_4=O_4:R_4) \implies\adj& \leftrightarrow \{\langle I_4\varnothing\varnothing R_4\varnothing O_4\rangle,{\unit}{:}({+}{3}{4}{\unit})\} && \ruleName{inst--comp} \\
  \smash{{\unit}{:}({+}{3}{4}{\unit})\overset{{\alpha}{:}{+}({\smS}{a}){b}{\unit}}\sim\{{\cdots}\}} \implies\adj& \leftrightarrow \{\langle\varnothing I_4\{\alpha{\mapsto}{\unit},a{\mapsto}2,b{\mapsto}4\}R_4\varnothing O_4\rangle\} && \ruleName{inp} \\
  {\cdots} \implies\adj& \leftrightarrow \{\varnothing I_4\{\alpha{\mapsto}{\unit}\}R_4\{\beta{:}[{+}{2}{4}{\unit}]\}O_4\} && \ruleName{sub$_1$} \\
  \{\beta{:}[{+}{2}{4}{\unit}]\}\leftrightarrow\{\beta{:}[{\unit}{2}{6}{+}]\} \implies\adj& \leftrightarrow \{\varnothing I_4\{\alpha{\mapsto}{\unit}\}R_4\{\beta{:}[{\unit}{2}{6}{+}]\}O_4\} && \ruleName{sub$_3$} \\
  {\cdots} \implies\adj& \leftrightarrow \{\langle\varnothing I_4\{\alpha{\mapsto}{\unit},a{\mapsto}2,b'{\mapsto}6\}R_4\varnothing O_4\rangle\} && \ruleName{sub$_1$} \\
  & \leftrightarrow \{\langle\varnothing O_4\{\alpha{\mapsto}{\unit},a{\mapsto}2,b'{\mapsto}6\}R_4\varnothing I_4\rangle\} && \ruleName{rev$_2$} \\
  \smash{{\unit}{:}({\unit}{3}{7}{+})\overset{{\alpha}{:}{\unit}({\smS}{a})({\smS}{b'}){+}}\sim\{{\cdots}\}} \implies\adj& \leftrightarrow \{\langle O_4\varnothing\varnothing R_4\varnothing I_4\rangle,{\unit}{:}({\unit}{3}{7}{+})\} && \ruleName{inp} \\
  & \leftrightarrow \{\langle I_4\varnothing\varnothing R_4\varnothing O_4\rangle,{\unit}{:}({\unit}{3}{7}{+})\} && \ruleName{rev$_1$} \\
  \mathcal P \vdash (I_4=O_4:R_4) \implies\adj& \leftrightarrow \{{\unit}{:}({\unit}{3}{7}{+})\} && \ruleName{inst--comp}\\
  \mathcal P \vdash \haltSym{\unit}{a}{b}{+} \implies\adj& \leftrightarrow \{\langle{\unit}{a}{b}{+}\rangle,{\unit}{:}({\unit}{3}{7}{+})\} && \ruleName{inst--halt} \\
  & \leftrightarrow \{\langle{\unit}{a}{b}{+}\rangle,{\unit}{:}[{\unit}{3}{7}{+}]\} && \ruleName{term} \\
  \mathcal P \vdash \haltSym{\unit}{a}{b}{+} \implies\adj& \leftrightarrow \{{\unit}{:}[{\unit}{3}{7}{+}]\} && \ruleName{inst--halt}
  \end{align*}
  \caption{One possible derivation of ${\unit}:[{+}~{3}~{4}~{\unit}] \leftrightarrow {\unit}:[{\unit}~{3}~{7}~{+}]$ using the semantics for $\textAleph$.}
  \label{lst:semex-add-sem}
\end{sublisting}
\endgroup
  \caption{\textAleph\ semantics as applied to the example of natural addition.}
  \label{lst:semex-add}
\end{listing}}

\para{Reversibility}

The \ruleName{symm} transition renders the semantics trivially reversible, but this is fairly weak. We conclude the discussion of the semantics of \textAleph\ by proving that the semantics are \emph{microscopically} reversible.

\begin{dfn}[Microscopic Reversibility]A transition is primitively microscopically reversible if it is a uniquely invertible structural rearrangement. A transition is microscopically reversible if it is decomposable into a series of primitively microscopically reversible transitions.\end{dfn}

\begin{thm}\label{thm:mrev}The semantics of\kern0.1ex\ \textslAleph\ are microscopically reversible.\end{thm}
\begin{proof}
  The rules \ruleName{refl,symm,trans,ext,ext$'$,term,rev$_1$,rev$_2$,sub$_3$} are microscopically reversible either trivially or inductively.

  The instantiation rules \ruleName{inst--comp,inst--halt} are less obvious, but can be realised in a microscopically reversible fashion in much the same way that cells translate an mRNA template to a polypeptide product. We shall first need a microscopically reversible way to duplicate a term, for which we assume that the environment contains an excess of structural building blocks (i.e.\ free atoms, variables, units $\unit$, etc). A term $t$ to be duplicated is first marked as such i.e.\ $t \mapsto \overline{t}$. These modified terms deviate from the definitions introduced at the beginning of this section, and are instead intermediate transitional structures used for the small-step microscopically reversible semantics described here. This marking is then propagated throughout the structure to all atoms, variables, and units by the following two microscopically reversible transitions:
  \begin{align*}
    \overline{\gamma:t} &\leftrightarrow \overline{\gamma}:\overline{t} && \ruleName{dup--prop$_1$} \\
    \overline{({x}~\vec\cursivexS)} &\leftrightarrow (\hat{x}~\vec\cursivexS) && \ruleName{dup--prop$_2$} \\
    (\vec\cursivexS~\hat{x}~{y}~\vec\cursiveyS) &\leftrightarrow (\vec\cursivexS~\overline{x}~\hat{y}~\vec\cursiveyS) && \ruleName{dup--prop$_3$}
  \end{align*}
  where rules \ruleName{dup--prop$_{2,3}$} are used to distribute the marking throughout composite terms. Note the use of the auxiliary `hat' marker $\hat{\,\cdot\,}$ to sequence this propagation in a microscopically reversible manner.
  Elementary terms thus marked recruit fresh copies of themselves from the free structural building blocks,
  \begin{align*}
    \frac{}{a} &\leftrightarrow \frac{a}{a} &
    \frac{}{u} &\leftrightarrow \frac{u}{u} &
    \frac{}{v} &\leftrightarrow \frac{v}{v} &
    \frac{}{\unit} &\leftrightarrow \frac{\unit}{\unit} &
    &\ruleName{dup--fresh}
  \end{align*}
  where $a\in\bnfPrim{atom}$, $u\in\bnfPrim{var}$ and $v\in\bnfPrim{var}'$, and where the drawing of a fresh copy from the environment is implicit. These are then ligated together (in a prescribed right-to-left order to ensure microscopic reversibility), and the composite structure disentangled,
  \begin{align*}
    \qty\Big(\vec{s}~\frac{t}{t}~\frac{\vec u}{\vec u}) &\leftrightarrow
    \qty\Big(\vec{s}~\frac{t~\vec u}{t~\vec u}) & 
    &\ruleName{dup--lig$_1$} &
    \qty\Big(\frac{\vec t}{\vec t}) &\leftrightarrow\frac{(\vec t)}{(\vec t)} &
    &\ruleName{dup--topo$_1$} \\
    \qty\Big{\vec{s}~\frac{t}{t}~\frac{\vec u}{\vec u}} &\leftrightarrow
    \qty\Big{\vec{s}~\frac{t~\vec u}{t~\vec u}} & 
    &\ruleName{dup--lig$_2$} & 
    \qty\Big{\frac{\vec\pi}{\vec\pi}} &\leftrightarrow\frac{\{\vec\pi\}}{\{\vec\pi\}} &
    &\ruleName{dup--topo$_2$} \\
    &&&& \frac{\gamma}{\gamma} : \frac{t}{t} &\leftrightarrow \frac{\gamma:t}{\gamma:t} &
    &\ruleName{dup--topo$_3$}
  \end{align*}
  finally resulting in a fully duplicated and disentangled structure $\frac tt$ from $\overline{t}$ for $t$ any term, party, or party-bag. Finally, the instantiation rules can be realised microscopically reversibly thus,
  \begin{align*}
    (I=O:S) &\leftrightarrow (\overline{I}=\overline{O}:\overline{S}) &&\ruleName{inst--comp$_1$} \\
    \qty\Big{\qty\Big(\frac II=\frac OO:\frac SS)} &\leftrightarrow \{(I=O:S), \langle I\varnothing\varnothing S\varnothing O\rangle\} &&\ruleName{inst--comp$_2$} \\
    (\haltSym~{\tau}) &\leftrightarrow (\haltSym~\overline{\tau}) &&\ruleName{inst--halt$_1$} \\
    \qty\Big{\qty\Big(\haltSym~\frac \tau\tau)} &\leftrightarrow \{(\haltSym~{\tau}), \langle\tau\rangle\} &&\ruleName{inst--halt$_2$}
  \end{align*}

  The one-hole-context rules are nearly trivially microscopically reversible, except for the choice of partition. This can be achieved by a movable marker like so,
  {\def\mark#1{\underset{\hat{}}{#1}}
  \begin{align*}
    c:({x}~\vec\cursivexS) &\leftrightarrow c:\llparenthesis\mark{x}~\vec\cursivexS\rrparenthesis &&\ruleName{ohc$_{11}$} &
    c:[{x}~\vec\cursivexS] &\leftrightarrow c:\llbracket\mark{x}~\vec\cursivexS\rrbracket &&\ruleName{ohc$_{21}$} \\
    c:\llparenthesis\vec\cursivexS~\mark{x}~{y}~\vec\cursiveyS\rrparenthesis &\leftrightarrow c:\llparenthesis\vec\cursivexS~{x}~\mark{y}~\vec\cursiveyS\rrparenthesis &&\ruleName{ohc$_{12}$} &
    c:\llbracket\vec\cursivexS~\mark{x}~{y}~\vec\cursiveyS\rrbracket &\leftrightarrow c:\llbracket\vec\cursivexS~{x}~\mark{y}~\vec\cursiveyS\rrbracket &&\ruleName{ohc$_{22}$} \\
    c:\llparenthesis\vec{\ell}~\mark{t}~\vec{r}\rrparenthesis &\leftrightarrow (c:(\vec{\ell}~{\bullet}~\vec{r})):t &&\ruleName{ohc$_{13}$} &
    c:\llbracket\vec{\ell}~\mark{t}~\vec{r}\rrbracket &\leftrightarrow (c:[\vec{\ell}~{\bullet}~\vec{r}]):t &&\ruleName{ohc$_{23}$}
  \end{align*}}\\[-2.5\baselineskip]

  The rules \ruleName{inp,sub$_1$,sub$_2$} require a microscopically reversible realisation of a bag with random draws. In order to be microscopically reversible, the draw needs to be deterministic at any given time, even if draws at different times are uncorrelated. This is achieved by representing a bag as a string that can be shuffled via swap operations, i.e.
  \begin{align*}
    \{\vec\cursivexS\cdot{x}~\vec\cursiveyS\} &\leftrightarrow \{\vec\cursivexS~{x}\cdot\vec\cursiveyS\} && \ruleName{bag--shuff$_1$} \\
    \{\vec\cursivexS~{x}\cdot{y}~\vec\cursiveyS\} &\leftrightarrow \{\vec\cursivexS~{y}\cdot{x}~\vec\cursiveyS\} && \ruleName{bag--shuff$_2$}
  \end{align*}
  where $\cdot$ is the focus of the swap operation, required for microscopic reversibility.
  Random draws are then implemented by simply picking the first element of the string.
  
  The last rules to demonstrate microscopic reversibility for are the unification semantics. We leave this as an exercise for the reader, with the hint that its realisation is similar to that of the instantiation rules.
\end{proof}

\para{Computability}

The earlier examples give reasonable assurance that \textAleph\ is Turing complete and that programming in it is `easy' in the sense that programs can be readily composed. For avoidance of doubt, however, we prove that \textAleph\ is Turing complete in two ways: the first proves a stronger claim, that \textAleph\ is Reversible-Turing complete, meaning that it can efficiently and faithfully simulate a Reversible Turing Machine without generating garbage; the second proves Turing completeness in a more conventional fashion in order to demonstrate its high level of composability.

\begin{listing}[tb]
  \centering
  \begingroup\centering%
\def\nl{\\[1.5em]}%
\begin{minipage}[t]{.6\linewidth}\begin{align*}
  & {[\atom{Blank}~{x}~{\cdot}~\cursivexS]}~\infix{\atom{Pop}}~\atom{Blank}~{[{x}~{\cdot}~\cursivexS]}{;} \\
  & {[{(\atom{Sym}~{x})}~{\cdot}~\cursivexS]}~\infix{\atom{Pop}}~\atom{Blank}~\cursivexS{;} \\
  & {[]}~\infix{\atom{Pop}}~\atom{Blank}~{[]}{;}
  \nl
  & {(\atom{Tape}~{\ell}~{x}~{r})}~\infix{\atom{Left}}~{(\atom{Tape}~{\ell'}~{x'}~{r'})}{:} \\
  &\qquad {\ell}~\infix{\atom{Pop}}~{x'}~{\ell'}{.} \\
  &\qquad {r'}~\infix{\atom{Pop}}~{x}~{r}{.}
\end{align*}\end{minipage}%
\begin{minipage}[t]{.3\linewidth}\begin{align*}
  & \haltSym~\atom{Tape}~{\ell}~{x}~{r}{;} \\
  & \haltSym~\atom{Sym}~{x}{;} \\
  & \haltSym~\atom{Blank}{;}
  \nl
  & {t}~\infix{\atom{Right}}~{t'}{:} \\
  &\qquad {t'}~\infix{\atom{Left}}~{t}{.}
\end{align*}\end{minipage}%
\endgroup
  \caption[An \textAleph\ program implementing bi-infinite tapes for Reversible Turing Machines.]{This \textAleph\ program defines all the ingredients necessary to simulate any Reversible Turing Machine. We represent a bi-infinite tape as its one-hole-context; that is, a tape is given by the square in the current position, $x$, as well as `all' the squares to the left of it, $\ell$, and `all' the squares to the right of it, $r$. Obviously we can't actually represent \emph{all} the squares to the left and the right. Instead we make use of the fact that, at any one time, only a finite bounded region of the tape is non-blank. As such, $\ell$ contains all the squares to the left up to the last symbol, and similarly for $r$. If we keep going left, past the final symbol, then $\ell$ will be the empty list, and \atom{Blank} squares will be created as needed. The rules \atom{Left} and \atom{Right} are convenience functions for changing the focus of a tape.}
  \label{lst:aleph-tape}
\end{listing}
\begin{listing}
  \centering
  \begingroup%
\def\nl{\\[.8em]}%
\def\Mu#1{\infix{\atom{Mu}~{#1}}}%
\def\Garb#1{(\atom{Garbage}~{#1})}%
\def\Dup#1#2{\infix{\atom{Dup}~{#1}}~{#2}{.}}%
\def\xyz{x\hspace{-0.3ex}y\hspace{-0.3ex}z}%
\begin{minipage}[t]{.68\linewidth}\begin{align*}
  & \cursivexS~\Mu{(\atom{Const}~{n})}~{n'}~\Garb\cursivexS{:} \\
  &\qquad \Dup{n}{n'}
  \nl
  & {[x]}~\Mu{\atom{Succ}}~{(\S x)}~{(\atom{Garbage})}{;}
  \nl
  & \cursivexS~\Mu{(\atom{Proj}~{i})}~{y}~\Garb{\cursiveyS~\cursivezS}{:} \\
  &\qquad \Dup{i}{i'} \\
  &\qquad \llapHalt\atomLocal{Go}~{i'}~{[]}~\cursivexS = \atomLocal{Go}~{\Z}~{[{y}~{\cdot}~\cursiveyS]}~\cursivezS{.} \\
  &\qquad \atomLocal{Go}~{(\S n)}~\cursiveyS~{[{x}~{\cdot}~\cursivexS]} = \atomLocal{Go}~{n}~{[{y}~{\cdot}~\cursiveyS]}~\cursivexS{;}
  \nl
  & \cursivexS~\Mu{(\atom{Sub}~{g}~\cursivefS)}~{z}~\Garb{\cursivexS~{h}~\cursivehS}{:} \\
  &\qquad \cursiveyS~\Mu{g}~{z}~{h}{.} \\
  &\qquad \infix{\atomLocal{Go}~\cursivefS~\cursivexS}~\cursiveyS~\cursivehS{.} \\
  &\qquad \infix{\atomLocal{Go}~{[]}~\cursivexS}~{[]}~{[]}{;} \\
  &\qquad \infix{\atomLocal{Go}~{[{f}~{\cdot}~\cursivefS~\cursivexS]}}~{[{y}~{\cdot}~\cursiveyS]}~{[{h}~{\cdot}~\cursivehS]}{:} \\
  &\qquad\qquad \Dup\cursivexS{\cursivexS'} \\
  &\qquad\qquad {\cursivexS'}~\Mu{f}~{y}~{h}{.} \\
  &\qquad\qquad \infix{\atomLocal[2]{Go}~\cursivefS~\cursivexS}~\cursiveyS~\cursivehS{.}
  \nl
  & {[{n}~{\cdot}~\cursivexS]}~\Mu{(\atom{Rec}~{f}~{g})}~{y}~\Garb{{n'}~\cursivexS~\cursivehS}{:} \\
  &\qquad \Dup\cursivexS{\cursivexS'} \\
  &\qquad {\cursivexS'}~\Mu{f}~{x}~{h}{.} \\
  &\qquad \llapHalt\atomLocal{Go}~{n}~{\Z}~{g}~{x}~\cursivexS~{[h]} =
                     \atomLocal{Go}~{\Z}~{n'}~{g}~{y}~\cursivexS~\cursivehS{.} \\
  &\qquad \atomLocal{Go}~{(\S n)}~{m}~{g}~{x}~\cursivexS~{hs} =
          \atomLocal{Go}~{n}~{(\S m)}~{g}~{y}~\cursivexS~{[{h}~{\cdot}~\cursivehS]}{:} \\
  &\qquad\qquad \Dup{m}{m'}~\Dup\cursivexS{\cursivexS'} \\
  &\qquad\qquad [{m'}~{x}~{\cdot}~{\cursivexS'}]~\Mu{g}~{y}~{h}{.}
  \nl
  & \cursivexS~\Mu{(\atom{Minim}~{f})}~{i}~\Garb{\cursivexS~\cursivehS}{:} \\
  &\qquad \llapHalt\atomLocal{Go}~{f}~{\Z}~\cursivexS~{(\S\Z)}~{[]} =
                     \atomLocal{Go}~{f}~{(\S i)}~\cursivexS~{\Z}~\cursivehS{.} \\
  &\qquad \atomLocal{Go}~{f}~{i}~\cursivexS~{(\S n)}~\cursivehS =
          \atomLocal{Go}~{f}~{(\S i)}~\cursivexS~{y}~{[{n}~{h}~{\cdot}~\cursivehS]}{:} \\
  &\qquad\qquad \Dup{i}{i'}~\Dup\cursivexS{\cursivexS'} \\
  &\qquad\qquad {[{i'}~{\cdot}~{\cursivexS'}]}~\Mu{f}~{y}~{h}{.}
\end{align*}\end{minipage}%
\def\expln#1{\comment{\cmtSLopen~$#1$}}
\begin{minipage}[t]{.3\linewidth}\vspace{-1.5em}%
\begin{equation*}\begin{array}{|rl}
  \qquad&\\[.5em]
  & \dataDef{\atom{Const}~{n}}{;} \\
  & \expln{C_n(\vec\cursivexS) = n} \nl
  & \dataDef{\atom{Succ}}{;} \\
  & \expln{S(x) = x + 1} \nl
  & \dataDef{\atom{Proj}~{i}}{;} \\
  & \expln{P_i(\vec\cursivexS) = x_i} \\
  & \comment{\cmtSLopen~N.B:\ $\vec\cursivexS$\, is 1-indexed.}\nl
  & \dataDef{\atom{Sub}~{g}~\cursivefS}{;} \\
  & \expln{(g\circ\vec\cursivefS)(\vec\cursivexS) = } \\
  & \expln{\quad g(f_1(\vec\cursivexS)\cdots f_n(\vec\cursivexS))} \nl
  & \dataDef{\atom{Rec}~{f}~{g}}{;} \\
  & \expln{\rho(f,g)(0,\vec\cursivexS) = f(\vec\cursivexS)} \\
  & \expln{\rho(f,g)(n+1,\vec\cursivexS) = } \\
  & \expln{\quad g(n,\rho(f,g)(n,\vec\cursivexS),\vec\cursivexS)} \nl
  & \dataDef{\atom{Minim}~{f}}{;} \\
  & \expln{\mu(f)(\vec\cursivexS) = } \\
  & \expln{\quad\min\{n:f(n,\vec\cursivexS)=0\}}
  \\[-.5em]&
\end{array}\end{equation*}
\end{minipage}%
\endgroup
  \caption[An \textAleph\ program implementing the $\mu$-recursive functions.]{An \textAleph\ program$^\ast$ implementing the $\mu$-recursive functions. The $\mu$-recursive functions are generated by the three functions, $C_n$, $S$ and $P_i$, and the three operators, $\circ$, $\rho$ and $\mu$, whose definitions are given in the comments above. The \textAleph\ values corresponding to these functions can be directly composed using the operators (e.g.\ $(\atom{Rec}~(\atom{Const}~0)~(\atom{Proj}~1))$), and then evaluated with the \atom{Mu} atom.}
  ~\\\parbox{\textwidth}{\hangindent=0.3cm $^\ast$\scriptsize Some additional sugar has been used in this program, mainly in the form of nested definitions and the more concise expression of looping constructs. See \Cref{sec:alethe} for an in-depth definition of these sugared forms.}
  \label{lst:aleph-murec}
\end{listing}

\begin{thm}The \textslAleph-calculus is Reversible-Turing Complete, i.e.\ it can simulate any Reversible Turing Machine (RTM) as defined by \textcite{bennett-tm}.\end{thm}
\begin{proof}
  An RTM is a collection of one or more bi-infinite tapes divided into squares, each of which can be blank or can contain a symbol, as well as a machine head with an internal state. Symbols are drawn from a finite alphabet, and internal states from a distinct finite alphabet. At any time, the RTM head looks at a single square on each tape and executes one of a finite set of reversible rules. The rule chosen is uniquely determined by the current state of the system. Each rule depends on the current internal state of the machine head, which it may alter, and for each tape it can either ignore its value, possibly moving the tape one square to the left or right, or it can depend on the square taking a certain value $u$, which it can replace with a certain other value $v$.

  For example, a 6-tape RTM with symbol alphabet $\{A,B,C,D,E,F\}$ and state alphabet $\{\atom{Start},S_1,S_2,\atom{Stop}\}$ might have a rule $S_1~C~\varnothing~/~D~/~/~\rightarrow~B~A~-~F~0~+~S_2$. This rule applies if and only if the current internal state is $S_1$, the values of tapes 1 and 4 are the symbols $C$ and $D$ respectively, and tape 2 is blank. If so, it will change the internal state to $S_2$, replace the values of tapes 1, 2 and 4 with the symbols $B$, $A$ and $F$ respectively, and move tapes 3 and 6 one square to the left and right respectively, leaving tape 5 alone. The reverse of the rule is given by $S_2~B~A~/~F~/~/~\rightarrow~C~\varnothing~+~D~0~-~S_1$. The necessary conditions for the ruleset to be deterministic and reversible are that all the domains of the forward rules are mutually exclusive, as are all the domains of the reversed rules. Of course, the domains of the forward and reverse rules will intersect in any useful program.

  It is easy to translate any description of an RTM into \textAleph. Bi-infinite tapes can be easily represented and manipulated, per the program in \Cref{lst:aleph-tape}, and any rule (such as the above) can be mechanically translated as so,
  \begin{align*}
    &\phantom{{}={}} {S_1}~%
        (\atom{Tape}~{\ell_1}~(\atom{Sym}~C)~{r_1})~%
        (\atom{Tape}~{\ell_2}~\atom{Blank}~{r_2})~{t_3}~%
        (\atom{Tape}~{\ell_4}~(\atom{Sym}~D)~{r_4})~{t_5}~{t_6} \\
    &= {S_2}~%
        (\atom{Tape}~{\ell_1}~(\atom{Sym}~B)~{r_1})~%
        (\atom{Tape}~{\ell_2}~(\atom{Sym}~A)~{r_2})~{t_3'}~%
        (\atom{Tape}~{\ell_4}~(\atom{Sym}~F)~{r_4})~{t_5}~{t_6'}{:} \\
    &\phantom{{}={}}\qquad {t_3}~\infix{\atom{Left}}~{t_3'}{.} \\
    &\phantom{{}={}}\qquad {t_6}~\infix{\atom{Right}}~{t_6'}{.}
  \end{align*}
  whilst the special \atom{Start} and \atom{Stop} states are marked thus,
  \begin{align*}
    & \haltSym~\atom{Start}~{t_1}~{t_2}~{t_3}~{t_4}~{t_5}~{t_6}{;} \\
    & \haltSym~\atom{Stop}~{t_1}~{t_2}~{t_3}~{t_4}~{t_5}~{t_6}{;}
  \end{align*}
  That this translation faithfully reproduces the operation of the RTM is obvious from the high level semantics of \textAleph.
\end{proof}
\begin{crl}The \textslAleph-calculus is Turing Complete.\end{crl}
\begin{proof}
  \textcite{bennett-tm} proved that a Reversible Turing Machine can simulate any Turing Machine (and vice-versa), and so the corollary follows immediately from \textAleph's Reversible-Turing completeness. To drive the point home, however, we also implement the $\mu$-Recursive Functions (\Cref{lst:aleph-murec})---three functions, and three composition operators, which together are capable of representing any computable function over the naturals and hence are Turing complete. For example, addition, multiplication and factorials can be defined in terms of each other (in our \textAleph\ realisation) as
  \begin{align*}
    \textit{add} &\mapsto (\atom{Rec}~(\atom{Proj}~1)~(\atom{Sub}~\atom{Succ}~[(\atom{Proj}~2)])) \\
    \textit{mul} &\mapsto (\atom{Rec}~(\atom{Const}~0)~(\atom{Sub}~\textit{add}~[(\atom{Proj}~2)~(\atom{Proj}~3)])) \\
    \textit{fac} &\mapsto (\atom{Rec}~(\atom{Const}~1)~(\atom{Sub}~\textit{mul}~[(\atom{Proj}~2)~(\atom{Sub}~\atom{Succ}~[(\atom{Proj}~1)])]))
  \end{align*}
  and then one can evaluate, e.g., $(\atom{Mu}~\textit{fac})~{[7]}~{\unit}$, obtaining $\unit~5040~\textit{garbage}~(\atom{Mu}~\textit{fac})$.
\end{proof}

\section{Implementation Concerns}
\label{sec:impl}

Although the \textAleph-calculus is microscopically reversible by \Cref{thm:mrev}, \textAleph\ has many degrees of freedom wherein macroscopic reversibility can be violated. There are two sources of entropy generation; the first is ambiguity in rule application, and the second is intrinsic to any useful implementation of concurrency. We already saw in \Cref{sec:ex2} how concurrency leads to entropy generation, but in this section we will expand on the other mechanism of entropy generation and how it can be avoided at the compiler-level. This is important because the primary motivation for reversible programming is to avoid unexpected entropy leaks, and so generally the appearance of ambiguity is a bug rather than intentional. In addition to ambiguity, we will also discuss other aspects of a realistic implementation to avoid or minimise the presence of random walks: optimising the evaluation order of sub-rules (`serialisation') and inferring the direction of computation. A reference implementation of these algorithms (as well as additional sugar) is made available as a language and interpreter, \alethe\ (see \Cref{sec:alethe}).

Nevertheless, in some cases the appearance of ambiguity/non-determinism may be intentional and so we should provide the programmer a mechanism to \emph{explicitly} weaken the ambiguity checker when desired. An example of this might be in implementing a fair coin toss for generating randomness,
\begin{align*}
  & \infix{\atom{Coin}}~\atom{Tails}{;} \\
  & \infix{\atom{Coin}}~\atom{Heads}{;}
\end{align*}
If the computational architecture is Brownian, such as a molecular computer, then this can be used to exploit the thermal noise of the environment to get as close as classical physics allows to a true random number generator (RNG). The way this RNG is used is by constructing a term $\atom{Coin}~{\unit}$. This term then matches both rules, and so depending on which is chosen the term may evolve to ${\unit}~\atom{Tails}~\atom{Coin}$ or to ${\unit}~\atom{Heads}~\atom{Coin}$. Moreover, this process will depend on the probability distribution realised by the implementation; if we assume this is uniform\footnote{A more sophisticated implementation might reify and expose control over the distribution to the programmer, allowing arbitrary probabilities to be assigned. Going further, an interesting research direction would be to what extent \textAleph\ can be augmented quantum mechanically.}, then the coin toss will be fair with each term observed with probability $\frac12$. If the probability distribution is non-uniform, the situation is less clear and depends on the exact rates of the forward and reverse transitions for each rule. The analysis is beyond the scope of this chapter, but if the dynamics takes the form
\begin{align*}
  \ce{ ${\unit}~\atom{Tails}~\atom{Coin}$ <=>[$\lambda_t'$][$\lambda_t$] $\atom{Coin}~{\unit}$ <=>[$\lambda_h$][$\lambda_h'$] ${\unit}~\atom{Heads}~\atom{Coin}$ },
\end{align*}
then the steady-state ratio of heads to tails will be $\frac{\lambda_h/\lambda_h'}{\lambda_t/\lambda_t'}$. By microscopic reversibility, we expect $\lambda_h=\lambda_h'$ and similarly for $\lambda_t$ (unless the system is biased away from equilibrium), and therefore even a non-uniform distribution will have a uniformly distributed steady-state. Note the interesting property that, in reverse, an ambiguous/non-deterministic reversible rule is indistinguishable from irreversibility. Here, for example, the inverse of a fair coin toss is a process that consumes\footnote{There is a subtle point to make here: if the bit being consumed has no computational content and is uniformly distributed, then this consumption does not generate any entropy. That is, drawing a random bit from the environment and then replacing it is a completely isentropic process. For non-uniform distributions, the process can also be isentropic provided that the producer/consumer has a matching distribution. The problem arises when the distribution of the bit does not fit that of the producer/consumer, with deterministic computation being an extreme case wherein the bit can only take on a prescribed value (this is true even if the bit is pseudorandom).}\ a bit.

\para{Ambiguity Checker}

To exclude these sorts of process, we thus introduce an ambiguity checker. Not only does the introduction of an ambiguity checker reassure the programmer that they are writing information-preserving code, but the algorithm also serves as an invaluable debugging tool. In a large codebase, it can be hard to keep track of all the different patterns employed (the standard library has nearly two thousand), let alone ensure that there are no ambiguities. The ambiguity checker in \alethe\ performs this check for the programmer and, most importantly, is able to tell the programmer where the ambiguities lie.

In an unambiguous program, there are only a few valid scenarios for each term. It can match two computational rules (where we think of the forward and reverse directions of each rule as distinct rules for the current purpose), in which case it is an intermediate state of a computation. It can match one computational rule and one or more halting rules, in which case it is a terminal state of the computation. It can match one computational rule, in which case the program has entered an erroneous state and stalled. It can match one or more halting rules, in which case the term has no computational capacity, and most likely represents a data value. Lastly, it can match no rules, in which case it is an invalid term that cannot be constructed under normal conditions.

The remaining possibilities can be summarised two-fold: a term can match two computational rules and one or more halting rules, in which case the computational path from the terminus is ambiguous, or a term can match more than two computational rules (and zero or more halting rules), which is the more obvious ambiguity scenario.

The above cases can be simplified: for a given term, consider all the rules it matches but coalesce all halting rules, if any, into a single rule. Then, a term is unambiguous if it matches at most two such rules, and ambiguous if it matches more than two.

Clearly it is impractical to test this for all possible terms, there being a countable infinity of them; instead, one can determine whether it is possible to construct an ambiguous term. To do so, build a graph $G$ whose nodes are all the patterns occurring on either side of a computational rule---which we will colour white for reasons that will become apparent---and all the patterns occurring in a halting rule---which we will colour black. An edge is drawn between two nodes if it is possible to construct a term satisfying both patterns. This condition can be determined inductively: if either pattern is variable, then draw an edge. If neither pattern is variable, then draw an edge only if they are both the same atom, or they are both composite terms of the same length and their respective sub-patterns matches pairwise. In order to proceed, we first prove a lemma about this graph.

\begin{lem}There exists a term that simultaneously satisfies each of a set of patterns $\Pi$ if and only if the subgraph of $G$ formed from the nodes $\Pi$ is complete.\label{lem:amb-subgraph}\end{lem}
\begin{proof}
  The $(\implies)$ case is obvious. We prove the converse by structural induction. As base cases, we note that the statement is trivially true if $|\Pi|=0,1,2$. Now consider the possible forms of some $\pi\in\Pi$:
  \begin{itemize}
    \item[(\bnfPrim{atom})] If $\pi$ is an atom, $a$, then the only term satisfying it is the same atom, $a$. As $\pi$ has edges to every other pattern in $\Pi$, $a$ must satisfy each of these other patterns too.
    \item[(\bnfPrim{var})] Suppose there is a term $t$ that satisfies all the patterns $\Pi\setminus\{\pi\}$. As $t$ satisfies the variable pattern $\pi$ vacuously, $t$ must satisfy all of $\Pi$. Without loss of generality then, we can ignore all variable patterns.
    \item[($\tau_1\cdots\tau_n$)] If $\pi$ is a composite term, then there can only be an edge to the other patterns if these are composite terms of the same length or are the variable pattern. By the previous case, we can ignore these variable patterns. Now, index the patterns by $i=1\ldots m$ and write $\pi_i=(\tau^{(i)}_1\cdots\tau^{(i)}_n)$. Then construct graphs $H_1\cdots H_n$ such that the nodes of $H_j$ are the patterns $\{\tau^{(i)}_j:i=1\ldots m\}$. If we draw edges according to the same rules as for $G$, then by the inductive definition of the edge condition for $G$, each of the graphs $H_i$ will be complete. By inductive assumption, a term $t_i$ exists satisfying all patterns in $H_i$ for each $i$, and therefore the term $(t_1\cdots t_n)$ must satisfy all the patterns $\Pi$.
  \end{itemize}
  The lemma follows.
\end{proof}

Using \Cref{lem:amb-subgraph}, we see that the different cases of unambiguous and ambiguous terms reduce to considering the structures of the complete subgraphs of $G$. Moreover, it suffices to only consider subgraphs formed from three nodes, i.e.\ \emph{triangles}: the first ambiguous case, of two white nodes and one or more black nodes, implies the existence of a triangle with two white nodes and one black node; the second case, of three or more white nodes and zero or more black nodes, implies the existence of a triangle with three white nodes. Conversely, a triangle with two or more white nodes implies one of these ambiguous cases. The only triangles present for the unambiguous cases contain two or more black nodes. Therefore we can enumerate all the triangles of $G$ containing two or more white nodes; if there are any, then there is an ambiguity and the patterns present in the triangles should be reported to the programmer. If there are none, then the program is unambiguous and deterministic.

In fact, there is one more potential source of ambiguity. Suppose that the two patterns of a sub-rule are not orthogonal (a common occurrence), then it is in principle possible that a halting term might match either side, and therefore lead to a 2-way ambiguity. The evaluation rules of our reference interpreter for \alethe\ do not suffer from this ambiguity, as it determines ahead of time the order and directionality of all sub-rules, but it is possible for a faithful implementation of \textAleph\ to suffer so. There are two ways to avoid this: the first is to require sub-rule patterns be orthogonal, though this also forbids many legal unambiguous programs---requiring additional levels of tedious and unnecessary indirection to circumvent the restriction---and so is undesirable; the second is to apply type inference (work towards which is presented in \Cref{app:typing}) to gain the extra knowledge needed to distinguish between legal and ambiguous programs. Namely, type inference allows the compiler to infer what sequences of patterns and terms occur in a program, and hence to determine which halting patterns are computationally connected. With this knowledge, and knowledge of the types of variables in a rule, it can determine what forms the halting terms of a sub-rule will take and therefore whether there is a unique mapping between terms and patterns or not.

\para{Serialisation Heuristics}

The Brownian semantics of \textAleph, particularly in the absence of coupling to a bias source, are not very performant. Specifically, sub-rule transitions execute a random walk in a phase space whose size---in the worst case---scales exponentially with the number of variables present in the current rule. Whilst explicit bias coupling circumvents this by essentially serving to annotate the preferred order and directionality of the sub-rules, it undermines the declarative nature of \textAleph. Instead, it is possible to algorithmically infer an appropriate execution path. Moreover, it can sometimes be the case (such as in the example of fractional addition) that the execution path is non-trivial and so granting the compiler the power to perform this routing automatically makes the job of the programmer easier: the programmer need only specify how all the variables relate to one another, and may even specify this excessively, and the compiler can figure out which relations are sensible to use.

\begingroup
 \def\fr#1#2{#1\!/\!#2}%
 \def\Fr#1#2{(\atom{Frac}~{#1}~{#2})}%
 \def\ruleLabel#1{\llap{\tiny$#1.$}}%
\begin{figure}
  \centering
  \begin{subfigure}{\textwidth}
    \centering
    \input{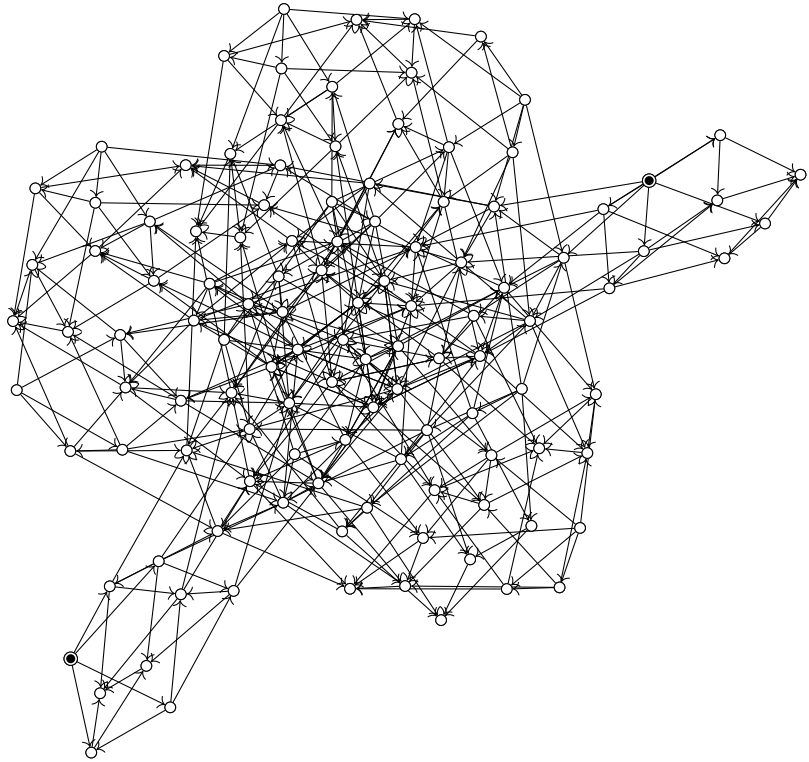}
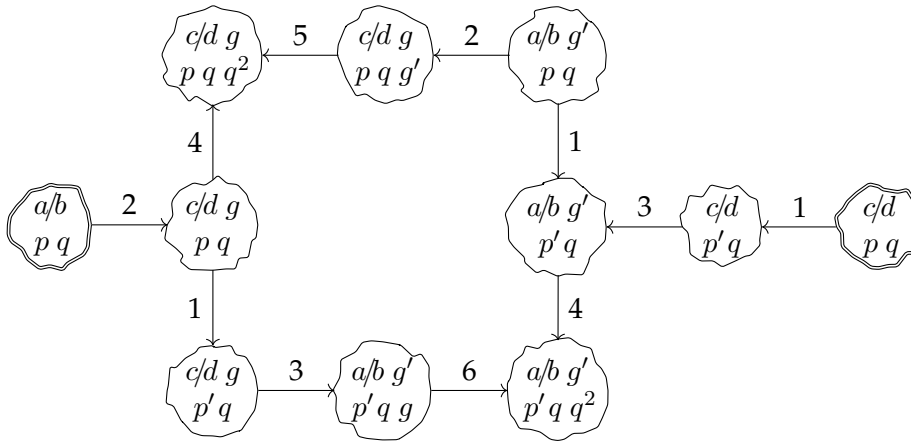
    \vspace{0.5em}
    \caption{The full transition graph. Due to the size of the graph, labels have been suppressed. Halting states are marked black.}
    \label{fig:tg-frac-full}
  \end{subfigure}
  \begin{subfigure}{\textwidth}
  \centering
    \begingroup
 \def\fr#1#2{#1\!/\!#2}%
 \def\Fr#1#2{(\atom{Frac}~{#1}~{#2})}%
 \def\ruleLabel#1{\llap{\tiny$#1.$}}%
  \pgfmathsetseed{2736}
  \begin{tikzpicture}[
    every state/.style={inner sep=-0.5ex},
    decoration={random steps,amplitude=.375ex,segment length=.75ex,
                pre=lineto,pre length=1ex,post=lineto,post length=1ex},
    rounded corners=.15ex
  ]
    \long\def\hag(#1) [#2] #3{\node[state,accepting] (#1) [#2,decorate] {$\begin{array}{c}#3\end{array}$}}
    \long\def\bag(#1) [#2] #3{\node[state] (#1) [#2,decorate] {$\begin{array}{c}#3\end{array}$}}
    \hag (A) [] {\fr ab\\\,p~q};
    \bag (B) [right=of A] {\fr cd~g\\p~q};
    \bag (C1) [above=of B] {\fr cd~g\\p~q~q^2\!\!};
    \bag (D1) [right=of C1] {\fr cd~g\\p~q~g'\!\!};
    \bag (E1) [right=of D1] {\fr ab~g'\!\\p~q};
    \bag (C2) [below=of B] {\fr cd~g\\p'\!~q};
    \bag (D2) [right=of C2] {\fr ab~g'\!\\p'\!~q~g};
    \bag (E2) [right=of D2] {\fr ab~g'\!\\p'\!~q~q^2\!\!};
    \bag (F) [below=of E1] {\fr ab~g'\!\\p'\!~q};
    \bag (G) [right=of F] {\fr cd\\p'\!~q};
    \hag (H) [right=of G] {\fr cd\\\,p~q};
    \path (A) edge[->,above] node {2} (B)
          (B) edge[->,left] node {4} (C1)
              edge[->,left] node {1} (C2)
          (C1) edge[<-,above] node {5} (D1)
          (D1) edge[<-,above] node {2} (E1)
          (E1) edge[->,right] node {1} (F)
          (C2) edge[->,above] node {3} (D2)
          (D2) edge[->,above] node {6} (E2)
          (E2) edge[<-,right] node {4} (F)
          (F) edge[<-,above,shorten >=0.1pt] node {3} (G)
          (G) edge[<-,above,shorten >=0.8pt] node {1} (H);
  \end{tikzpicture}
\endgroup
    \vspace{0.5em}
    \caption{A more practical excerpt of the full transition graph.}
    \label{fig:tg-frac-excerpt}
  \end{subfigure}
  \def\en{$-$}%
  \caption[Two views of the transition graph for the fraction-addition routine.]{Two views of the transition graph for the fraction-addition routine. Nodes correspond to bags of variables in a `known' state. An edge with label $\ell$ is drawn from node $\alpha$ to node $\beta$ if sub-rule $\ell$ can map the knowledge state $\alpha$ to the knowledge state $\beta$ (with possible variable duplication/elision). It can be seen that there are two paths from $\{\fr ab,p,q\}$ to $\{\fr cd,p,q\}$ worth considering: \fwd2\en\fwd4\en\bwd5\en\bwd2\en\fwd1\en\bwd3\en\bwd1, and \fwd2\en\fwd1\en\fwd3\en\fwd6\en\bwd4\en\bwd3\en\bwd1. These are of roughly equal computational complexity, and so either is a viable choice.}
  \label{fig:tg-frac}
\end{figure}
To illustrate the problem, consider the following program excerpt which adds two simplified fractions:
  \begin{align*}
    & {\fr ab}~\infix{{+}~{\Fr pq}}~{\fr cd}{:} \\
    \ruleLabel{1}&\qquad {p}~{-}~{p'}{.} \\
    \ruleLabel{2}&\qquad {\fr ab}~\infix{\atomLocal{}~{\Fr pq}}~{\fr cd}~{g}{.} \\
    \ruleLabel{3}&\qquad {\fr cd}~\infix{\atomLocal{}~{\Fr{p'}q}}~{\fr ab}~{g'}{.} \\
    \ruleLabel{4}&\qquad {(\S q)}~{\square}~{q^2}{.} \\
    \ruleLabel{5}&\qquad {g'}~\infix{{\times}~{g}}~{q^2}{.} \\
    \ruleLabel{6}&\qquad {g}~\infix{{\times}~{g'}}~{q^2}{.}
  \end{align*}
where a fraction is represented by $\frac{p}{q+1} \equiv \Fr pq$ with $p\in\mathbb Z$ and $q\in\mathbb N$. Maintaining the invariant that fractions are in their simplest form is non-trivial, but involves finding the greatest common divisor of the numerators and denominators of two intermediate results, $g$ and $g'$, and then showing that their product $gg'=(q+1)^2$. It turns out that this fact can then be used to eliminate the intermediate garbage. In the unassisted semantics of \textAleph, the graph of all intermediate states generated by applying these rules transitively is given by \Cref{fig:tg-frac-full}. This graph has 117 nodes. It can therefore be seen that a computation is very likely to get stuck in one of the 115 intermediates, and will take a long time to find its way from the input variables, $\{\fr ab,p,q\}$, to the output variables, $\{p,q,\fr cd\}$. Many of these states are not interesting computational, and indeed only a small subgraph of 11 (or even 8) nodes is needed to perform the desired computation as shown in \Cref{fig:tg-frac-excerpt}.

From this transition subgraph, we can see that there are two possible paths to get from the input variables to the output variables, and the goal of the serialisation algorithm is to find these paths and to pick the most optimal. Here both of these paths are of equal complexity, and so either is valid. Having made a choice, the compiler can then direct the computer to perform the prescribed series of transitions, rather than executing a random walk. In a Brownian computational system, for example, this would be achieved by adding `fake' dependencies to the sub-rules to force a linear ordering (or possibly partially parallel, where this makes sense) and automatically coupling the sub-rules to the bias source in the correct direction. As we shall see, automatically determining the cost of a path is non-trivial and sometimes even impossible, and so a perfect solution to the serialisation problem is not generally possible. Nonetheless, for many programs it is possible to find the appropriate path, or to give the compiler enough information to make a better choice, and so further discussion of serialisation is warranted.
\endgroup

\doublepage{%
    \begin{listing}[!p]
      \centering
      \begingroup%
\def\nl{\\[1.5em]}%
\def\nfx#1{\infix{\atom{#1}}}%
\def\infx#1#2{\infix{\atom{#1}~{#2}}}%
\def\ruleLabel#1{\llap{\tiny$#1.$}}%
\begin{sublisting}{\linewidth}
  \centering
  \begin{align*}
    & \dataDef{\atom{Tree}~{\alpha}} = \atom{Tree}~{\alpha}~{[\atom{Tree}~{\alpha}]} \\
    & \atom{polish}~{:\kern0.2ex:}~\atom{Tree}~{\alpha}~{\rightarrow}~{[(\atom{Int},{\alpha})]} \\
    & \atom{polish}~{(\atom{Tree}~{x}~\cursivexS)} = (\atom{length}~\cursivexS, x)~{:}~\atom{concatMap}~\atom{polish}~\cursivexS
  \end{align*}
  \caption{A snippet of \haskell\ code for converting an arbitrary tree into Polish notation. Polish notation is an isomorphic linear representation of tree-like structures that is well known for its use for arithmetic expressions. For example, the expression $(3 + 4) \times 7 - 2^5$ can be represented as ${-}~{\times}~{+}~{3}~{4}~{7}~{\raisebox{-0.5ex}{${}^\wedge$}}~{2}~{5}$.}
  \label{lst:serial-polish-hs}
\end{sublisting}
\begin{sublisting}{\linewidth}
  \centering\vspace{-.5em}
  \begin{minipage}[t]{.6\linewidth}
    \begin{align*}
      & (\atom{Tree}~{x}~\cursivetS)~\nfx{Polish}~{[({,}~{x}~{n})~{\cdot}~{p}]}{:} \\
      \ruleLabel{1}&\qquad \infx{Length}\cursivetS~{n}{.} \\
      \ruleLabel{2}&\qquad \cursivetS~\infx{ConcatMap}{\atom{Polish}}~{p}~\cursivelS{..} \\
      \ruleLabel{3}&\qquad {p}~\infx{PolishReads}{n}~\cursivetS~\cursivelS{.}
    \end{align*}\vspace{-1em}
  \end{minipage}%
  \begin{minipage}[t]{.3\linewidth}%
    \begin{equation*}\begin{array}{|rl}
      ~ & \dataDef{\atom{Tree}~{x}~\cursivetS}{;} \\
        & \dataDef{{,}~{a}~{b}}{;} \\
        &\quad \comment{\cmtSLopen~2-tuples}
    \end{array}\end{equation*}
  \end{minipage}
  \caption{An excerpt of the program in (e). This is a rough translation of the \haskell\ implementation into \alethe. Unfortunately, the implementation of \atom{ConcatMap} in \alethe\ has additional garbage in the form of a list of the lengths of the intermediate lists ($\cursivelS$). To eliminate this garbage, we write a function \atom{PolishReads} which takes a string of trees in Polish notation and converts them back into trees, also generating the same garbage value.}
  \label{lst:serial-polish-exc}
\end{sublisting}
\begin{sublisting}{\linewidth}
  \centering\vspace{.5em}
  \pgfmathsetseed{18976}
  \begin{tikzpicture}[
    every state/.style={inner sep=-0.5ex},
    decoration={random steps,amplitude=.375ex,segment length=.75ex,
                pre=lineto,pre length=1ex,post=lineto,post length=1ex},
    rounded corners=.15ex
  ]
    \long\def\hag(#1) [#2] #3{\node[state,accepting] (#1) [#2,decorate] {$\begin{array}{c}#3\end{array}$}}
    \long\def\bag(#1) [#2] #3{\node[state] (#1) [#2,decorate] {$\begin{array}{c}#3\end{array}$}}
    \hag (A) [] {x\\\cursivetS};
    \bag (B) [right=of A] {x~n\\\cursivetS};
    \bag (C) [right=of B] {x~n\\p~\cursivelS};
    \bag (D) [right=of C] {x~n\\\cursivetS~\cursivelS};
    \hag (E) [right=of D] {x~n\\p};
    \path (A) edge[->,above,shorten >=0.3pt] node {1} (B)
          (B) edge[->,above,shorten >=0.5pt] node {2} (C)
          (C) edge[<-,above,bend left,shorten <=0.8pt,shorten >=0pt] node {2} (D)
              edge[->,below,bend right,shorten <=0.2pt] node {3} (D)
          (D) edge[<-,above,shorten >=0.8pt] node {3} (E);
  \end{tikzpicture}\vspace{.5em}
  \def\en{$-$}%
  \caption{The motif employed above---generating the garbage in two different ways---can be used to eliminate it and achieve the desired (partial) bijection, as shown by its transition graph. There are two possible execution paths, \smfwd1\en\smfwd2\en\smbwd2\en\smbwd3 and \smfwd1\en\smfwd2\en\smfwd3\en\smbwd3. As explained in the text, the second route is \emph{strongly} preferred--- hence the cost annotation on \ruleName{2}.}
  \label{lst:serial-polish-tg}
\end{sublisting}
\begin{sublisting}{\linewidth}
  \centering
  \begin{minipage}[t]{.58\linewidth}
    \begin{align*}
      & {t}~\nfx{Polish}~{p}{:} \\
      &\qquad {p}~\nfx{PolishRead}~{[]}~{t}{.}
      \nl
      & {[{({,}~{x}~{n})}~{\cdot}~\cursivexS]}~\nfx{PolishRead}~{\cursivexS'}~{(\atom{Tree}~{x}~\cursivetS)}{:} \\
      &\qquad \cursivexS~{n}~\nfx{PolishReads}~{\cursivexS'}~\cursivetS{.}
    \end{align*}\vspace{-1em}
  \end{minipage}%
  \begin{minipage}[t]{.4\linewidth}
    \begin{align*}
      & \cursivexS~{\Z}~\nfx{PolishReads}~\cursivexS~{[]}{;} \\
      & \cursivexS~{(\S n)}~\nfx{PolishReads}~{\cursivexS''}~{[{t}~{\cdot}~\cursivetS]}{:} \\
      &\qquad \cursivexS~\nfx{PolishRead}~{\cursivexS'}~{t}{.} \\
      &\qquad {\cursivexS'}~{n}~\nfx{PolishReads}~{\cursivexS'}~\cursivetS{.}
    \end{align*}
  \end{minipage}
  \caption{In fact, programming in a `reversible-first' manner permits us to find the far more efficient and naturally bijective program above. This program was written by considering how to read a tree from its Polish representation: In our experience, picking the `harder' direction of a bijective computation to implement helps to suppress one's learned intuition for irreversible programming, and hence makes it less likely to fall into the traps of `shortcuts' with their accompanying garbage.}
  \label{lst:serial-polish-nice}
\end{sublisting}
\endgroup
      \caption[An example of serialisation as applied to the interconversion between trees and Polish notation.]{An example of serialisation as applied to the interconversion between trees and Polish notation... (continued on next page)}
      \label{lst:serial-polish}
    \end{listing}%
}{%
    \begin{listing}[!p]
      \ContinuedFloat%
      \centering%
      \begin{sublisting}{\linewidth}
        \centering
        \begingroup%
\def\nl{\\[1.5em]}%
\def\nfx#1{\infix{\atom{#1}}}%
\def\infx#1#2{\infix{\atom{#1}~{#2}}}%
\begin{minipage}[t]{0.45\linewidth}\begin{align*}
  & \infx{Length}{[]}~{\Z}{;} \\
  & \infx{Length}{[{x}~{\cdot}~\cursivexS]}~{(\S\ell)}{:} \\
  &\qquad \infx{Length}\cursivexS~{\ell}{.}
  \nl
  & \infx{Sum}{[]}~{\Z}{;} \\
  & \infx{Sum}{[\Z~{\cdot}~\cursivenS]}~{\sigma}{:} \\
  &\qquad \infx{Sum}\cursivenS~{\sigma}{.} \\
  & \infx{Sum}{[{(\S n)}~{\cdot}~\cursivenS]}~{(\S\sigma)}{:} \\
  &\qquad \infx{Sum}{[{n}~{\cdot}~\cursivenS]}~{\sigma}{.}
  \nl
  & {[]}~\infx{Map}{f}~{[]}{;} \\
  & {[{x}~{\cdot}~\cursivexS]}~\infx{Map}{f}~{[{y}~{\cdot}~\cursiveyS]}{;} \\
  &\qquad {x}~\infix{f}~{y}{.} \\
  &\qquad \cursivexS~\infx{Map}{f}~\cursiveyS{.}
  \nl
  & \infx{Map}{{f}~{[]}}~{[]}{;} \\
  & \infx{Map}{{f}~{[{x}~{\cdot}~\cursivexS]}}~{[{y}~{\cdot}~\cursiveyS]}{;} \\
  &\qquad \infix{f~x}~{y}{.} \\
  &\qquad \infx{Map}{f~\cursivexS}~\cursiveyS{.}
  \nl
  & {[]}~\nfx{Concat}~{[]}~{[]}{;} \\
  & {[{[]}~{\cdot}~\cursivexSS]}~\nfx{Concat}~\cursiveyS~{[{\Z}~{\cdot}~\cursivelS]}{:} \\
  &\qquad \cursivexSS~\nfx{Concat}~\cursiveyS~\cursivelS{.} \\
  & {[{[x~{\cdot}~\cursivexS]}~{\cdot}~\cursivexSS]}~\nfx{Concat}~{{x}~{\cdot}~\cursiveyS}~{[{(\S\ell)}~{\cdot}~\cursivelS]}{:} \\
  &\qquad {[\cursivexS~{\cdot}~\cursivexSS]}~\nfx{Concat}~\cursiveyS~{[{\ell}~{\cdot}~\cursivelS]}{.}
  \nl
  & \cursivexS~\infx{ConcatMap}{f}~\cursiveyS~\cursivelS{:} \\
  &\qquad \cursivexS~\infx{Map}{f}~\cursiveySS{.} \\
  &\qquad \cursiveySS~\nfx{Concat}~\cursiveyS~\cursivelS{.}
\end{align*}\vspace{-1em}\end{minipage}%
\quad\def\nl{\\[1.805em]}
\begin{minipage}[t]{.45\linewidth}\begin{align*}
  & \dataDef{\atom{Tree}~{x}~\cursivetS}{;} \\
  & \dataDef{{,}~{a}~{b}}{;}
  \nl
  &{n}~\infx{TreeSize'}{(\atom{Tree}~{x}~\cursivetS)}~{(\S n')}{:} \\
  &\qquad {n}~\infix{\atomLocal{Go}~\cursivetS}~{n'}{.} \\
  &\qquad {n}~\infix{\atomLocal{Go}~{[]}}~{n}{;} \\
  &\qquad {n}~\infix{\atomLocal{Go}~{[{t}~{\cdot}~\cursivetS]}}~{n''}{:} \\
  &\qquad\qquad {n}~\infx{TreeSize'}{t}~{n'}{.} \\
  &\qquad\qquad {n'}~\infix{\atomLocal[2]{Go}~\cursivetS}~{n''}{.}
  \nl
  & \infx{TreeSize}{t}~{n}{:} \\
  &\qquad \Z~\infx{TreeSize'}{t}~{n}{.}
  \nl
  & (\atom{Tree}~{x}~\cursivetS)~\nfx{Polish}~{[({,}~{x}~{n})~{\cdot}~{p}]}{:} \\
  &\qquad \infx{Length}\cursivetS~{n}{.} \\
  &\qquad \cursivetS~\infx{ConcatMap}{\atom{Polish}}~{p}~\cursivelS{.} \\
  &\qquad {p}~\infx{PolishReads}{n}~\cursivetS~\cursivelS{.}
  \nl
  & {[{({,}~{x}~{n})}~{\cdot}~\cursivexS]}~\nfx{PolishRead}~{\cursivexS'}~{(\atom{Tree}~{x}~\cursivetS)}~{(\S\ell)}{:} \\
  &\qquad \cursivexS~\infx{PolishReads}{n}~{\cursivexS'}~\cursivetS~\cursivelS{.} \\
  &\qquad \infx{Length}\cursivelS~{n}{.} \\
  &\qquad \infx{Sum}\cursivelS~{\ell}{.} \\
  &\qquad \infx{Map}{\atom{TreeSize}~\cursivetS}~\cursivelS{.}
  \nl
  & \cursivexS~\infx{PolishReads}{\Z}~\cursivexS~{[]}~{[]}{;} \\
  & \cursivexS~\infx{PolishReads}{(\S n)}~{\cursivexS''}~{[{t}~{\cdot}~\cursivetS]}~{[{\ell}~{\cdot}~\cursivelS]}{:} \\
  &\qquad \cursivexS~\nfx{PolishRead}~{\cursivexS'}~{t}~{\ell}{.} \\
  &\qquad {\cursivexS'}~\infx{PolishReads}{n}~{\cursivexS'}~\cursivetS~\cursivelS{.}
\end{align*}\end{minipage}%
\endgroup
        \caption{The full \alethe\ program implementing interconversion between trees and Polish notation, following the \haskell\ snippet in (a). In fact, a little thought and optimisation will reveal that the $\ell$ outputs of the read functions can be eliminated, which then simplifies the definition of \atom{Polish} to that given in (d) and obviates the need for the auxiliary functions \atom{Length}, \atom{Sum}, \atom{Map}, etc. Nevertheless, the point of this is to show that it is not unreasonable to arrive at the program listed here, and it may not be entirely obvious that a simpler implementation with less garbage is possible.}
        \label{lst:supp-polish}
      \end{sublisting}
      \caption[]{(continued) ...and the full \alethe\ program being considered here.}
    \end{listing}%
}

The simplest and least efficient serialisation strategy is as follows: suppose you start with the input variables $\{x,y,z\}$. Enumerate all the sub-rules that can be evaluated given this current `knowledge state', e.g.\ ${x}~\infix{\atom{Foo}}~{y}~\cursivewS$ and ${z}~\atomLocal{}~{x}$ would qualify but $\cursivewS~\infix{\atom{Map}~\atom{Bar}}~\cursivevS$ would not because neither the variable $\cursivewS$ nor the variable $\cursivevS$ are currently available. Evaluate all of these sub-rules in parallel, making sure to duplicate variables as necessary to avoid unlearning any. Here, we would end up with the knowledge state $\{\cursivewS,x,y,z\}$. Now ignore these sub-rules going forward and repeat this process until all sub-rules have been used once (and only once), and all variables are known. If some sub-rule has not been used or a variable remains unknown, then there is a logic error and the programmer should be appropriately chastised. Now, perform this entire process again, but starting from the output variables. If all is well, then we will have obtained two routes: one from the input variables to all variables, and one from the output variables to all variables. As these routes are reversible, we are immediately rewarded with a route from the input to the output variables, using each sub-rule twice. Whilst clearly inefficient, this strategy demonstrates that each sub-rule need be used at most twice (once in each direction).

A better strategy is to construct a `transition graph' as follows: take as nodes the powerset of all variables, i.e.\ $\varnothing$, $\{x\}$, $\{y\}$, $\{x,y\}$, etc. Mark the input and output nodes specially for later. Draw an edge between two nodes if a sub-rule (with possible variable duplication/elision) can be used to map between the two knowledge states, and label the edge with this sub-rule and in which direction it is to be applied. With the transition graph thus constructed, all possible routes between the input and output nodes can be enumerated\footnote{In fact there are infinitely many routes. It is therefore appropriate to restrict the routes under consideration to a sane (and finite) subset; namely, we avoid revisiting any node, and we also exclude routes that make use of a sub-rule more than twice.}. To proceed, a cost heuristic is needed in order to rank routes by preference. For example, if the goal is simply to minimise the number of sub-rules used then Dijkstra's algorithm will suffice. Note that, in the reference interpreter, transition graphs are constructed more restrictively. Namely, variable duplication/elision is not implemented as this always\footnote{The serialisation algorithm in \alethe\ is still subject to an exponential worst case complexity. The size of the transition graph constructed depends on the inter-dependency of the variables: the more dependency, the closer to linear in the number of variables the graph size is. If there is minimal dependency, such as in the automatically generated implemention of \atom{Dup} for $\dataDef{{,}~{a}~{b}~{c}~{d}~{e}}$, then the worst case will be realised.}\ leads to an exponentially large transition graph and, in practice, it is rarely needed. Where it is needed, the programmer must instead explicitly use \atom{Dup} (and create a new variable). This provides dramatically better performance at the expense of a slight inconvenience.

{\def\en{$-$}%
We saw an example of this algorithm earlier in \Cref{fig:tg-frac}. Another example, deserving of further comment, is given in \Cref{lst:serial-polish}. There are two routes of apparent equal cost, \fwd1\en\fwd2\en\bwd2\en\bwd3 and \fwd1\en\fwd2\en\fwd3\en\bwd3, but this is misleading. If the tree has $n$ nodes across $\ell\sim\log n$ levels, a single (non-recursive) step of \atom{Polish} has time complexity $\alpha$, and a single (non-recursive) step of \atom{PolishReads} has time complexity $\beta$, then the first route can be shown to have time complexity $\bigOO{2^\ell n(\alpha+\beta)}$ whilst the second has complexity $\bigOO{n(\alpha+2\ell\beta)}$. That is, if the serialisation algorithm picks a route which repeats a recursive step then consequently there will be an exponential overhead. Moreover, this becomes less obvious in the cases of corecursion, or higher order functions where a function may be passed into itself (recursively or corecursively)---compare, for example, the Y combinator. As such it is not generally possible to algorithmically discriminate this scenario (although the extra knowledge afforded by the type inference algorithm developed in \Cref{app:typing} may help). In fact, the reference interpreter for \alethe\ does not even try to do so. Instead, it exposes a method by which the programmer can adjust the heuristic cost function used in serialisation: each edge in the transition graph has a weight, given by the number of full-stops following the sub-rule in the source; therefore, in this case rule \ruleName{2} should be annotated with a cost of 2 (via two full-stops) in order to penalise its repeated use. To inspect the route chosen, the \code{:p} directive may be issued to the interpreter.}

\para{Directional Evaluation}

Consider an intermediate term along a reversible computation path (\Cref{fig:state-space}). In isolation, one cannot know in which direction computation should proceed. Bias coupling helps with this, but \alethe\ does not make use of this. Moreover, it would be nice to avoid explicit bias-coupling where practical. As such, we need some concept of computational `momentum' to maintain a consistent direction. This requires that the programs are deterministic such that phase space is branchless, which fortunately we have guaranteed through the ambiguity checker. Recall that, treating the forward and reverse directions of a rule as distinct, every term matches at most two rules; if we maintain a consistent direction of computation, then one of these rules will be the converse of the rule that was employed to reach the current state. More concretely, let the terms of the computation be labelled as some contiguous subset of $\{\ket{n}:n\in\mathbb Z\}$. By determinism, there is a unique rule $r(n\mapsto n+1)$ enacting each transition $\ket{n}\mapsto\ket{n+1}$, and by reversibility the unique rule enacting the transition $\ket{n+1}\mapsto\ket{n}$ is $\bar r$, the converse of $r$. Therefore, given a term $\ket{n+1}$, the rule $\bar r(n\mapsto n+1)$ can \emph{only} take us to $\ket{n}$; it can never also be the rule taking us to $\ket{n+2}$. It is, however, possible that the rule taking us to $\ket{n+2}$ is the same as $r(n\mapsto n+1)$. To summarise, all we need do is keep track of which rule we just applied, $r$; then, when considering the new term, we find any rules it matches and exclude $\bar r$ from this set. If the set is empty, we have entered an error state and should report it to the user. Otherwise, our ambiguity check has ensured that it either contains a single computational rule, which we duly apply, or contains one or more halting rules, whence we would halt the computation and report the result. This evaluation logic applies just as well to sub-rules. The one edge case is when evaluating a new term; as only halting terms can be constructed, it will have at most one computational rule to choose from and thus there is no ambiguity.

\section{A Spoonful of Sugar: \alethe}
\label{sec:alethe}

For clarity and convenience, we have introduced a number of syntactic shorthands in \Cref{sec:ex1,sec:ex2}. We now summarise and extend this sugar to construct a programming language, \alethe, the definition of which is given in \Cref{lst:alethe-dfn}. Implementing the additional measures discussed in \Cref{sec:impl}, a reference interpreter for \alethe\ is also made available\footnote{\url{https://github.com/hannah-earley/alethe-repl}} together with a `standard library' and select examples\footnote{\url{https://github.com/hannah-earley/alethe-examples}}. We proceed with the definition of \alethe\ by desugaring each form in turn.

\begin{listing}\centering
\fbox{\begin{minipage}{0.9\textwidth}\vspace{\baselineskip}\begin{equation*}\begin{aligned}
    \ruleName{pattern term} && \tau &::= \alpha ~|~ \bnfPrim{var} ~|~ (\,\tau^*\,) ~|~ \sigma \\
    \ruleName{atom} && \alpha &::= \bnfPrim{atom} ~|~ \sim\alpha ~|~ \#{}^{\backprime\backprime}\,\bnfPrim{char}^\ast\,'' ~|~ \#\alpha \\
    \ruleName{value} && \sigma &::= \mathbb{N} ~|~ {}^{\backprime\backprime}\, \bnfPrim{char}^\ast \,'' ~|~ [\,\tau^\ast\,] ~|~ [\, \tau^+ \,.\, \tau \,] ~|~ \blank \\
    \ruleName{party} && \pi &::= \tau:\tau^\ast ~|~ \bnfPrim{var}':\tau^\ast \\
    \ruleName{parties} && \Pi &::= \pi ~|~ \Pi\,;\,\pi \\
    \ruleName{relation} && \rho &::= \tau^\ast=\tau^\ast ~|~ \tau^\ast~{}^\backprime\tau^\ast{}^\prime~\tau^\ast \\
    \ruleName{definition head} && \delta_h &::= \rho ~|~ \{\,\Pi\,\}=\{\,\Pi\,\} \\
    \ruleName{definition rule} && \delta_r &::= \delta_h \,; ~|~ \delta_h\!: \Delta^+ \\
    \ruleName{definition halt} && \delta_t &::=~ \haltSym\,\tau^\ast; ~|~ \haltSym\,\rho\,; \\
    \ruleName{definition} && \delta &::= \delta_r ~|~ \delta_t \\
    \ruleName{cost annotation} && \xi &::= .^\ast \\
    \ruleName{declaration} && \Delta &::= \delta ~|~ \pi\,.\xi ~|~ \rho\,.\xi ~|~ \haltSym\,\rho\,.\xi \\
    \ruleName{statement} && \Sigma &::= \delta ~|~ \dataDef{\tau^\ast}; ~|~ \kw{import}~ {}^{\backprime\backprime}\bnfPrim{module path}''; \\
    \ruleName{program} && \mathcal P &::= \Sigma^\ast
\end{aligned}\end{equation*}\vspace{\baselineskip}\end{minipage}}
  \caption{Definition of \alethe\ syntax.}
  \label{lst:alethe-dfn}
\end{listing}

{\def\ir#1#2{\needspace{2\baselineskip}\item[\ruleName{#1}] $#2$\\}\begin{itemize}
  \ir{patt.\ term}{\tau ::= \alpha ~|~ \bnfPrim{var} ~|~ (\,\tau^*\,) ~|~ \sigma}
    The definition of a pattern term only differs from its definition in \textAleph\ by the inclusion of sugared values, $\sigma$.
  \ir{atom}{\alpha ::= \bnfPrim{atom} ~|~ \sim\alpha ~|~ \#{}^{\backprime\backprime}\,\bnfPrim{char}^\ast\,'' ~|~ \#\alpha}
    Atoms are fundamentally the same as in \textAleph, but there are four ways of inputting an atom. If the name of an atom begins with an uppercase letter or a symbol, then it can be typed directly. In fact, any string of non-reserved and non-whitespace characters that does not begin with a lowercase letter (as defined by Unicode) qualifies as an atom; if it does begin with a lowercase letter, it qualifies as a variable. If one wishes to use an atom name that does not follow this rule, one can use the form $\#{}^{\backprime\backprime}\,\bnfPrim{char}^\ast\,''$, where $\bnfPrim{char}^\ast$ is any string (using \haskell-style character escapes if needed). Additionally you can prefix a symbol with $\#$ to suppress its interpretation as a relation (e.g.\ $\#+$). Finally, if an atom is prefixed with some number of tildes then it is locally scoped: that is, you can nest rule definitions and make these unavailable outside their scope, and you can refer to locally scoped atoms in outer scopes by using more tildes (there is no scope inheritance). You can even use an empty atom name if preceded by a tilde, which can be useful for introducing anonymous definitions. Example code using this sugar is given in \Cref{lst:ex-sort2}.
  \ir{value}{\sigma ::= \mathbb{N} ~|~ {}^{\backprime\backprime}\, \bnfPrim{char}^\ast \,'' ~|~ [\,\tau^\ast\,] ~|~ [\, \tau^+ \,.\, \tau \,] ~|~ \blank}
    The sugar here for natural numbers ($\mathbb N$) and lists is familiar from earlier, but there are two additional sugared value types: strings\footnote{Once again, using \haskell-style character escapes if needed}, which are lists of character atoms where the character atom for, e.g., the letter `x' is $'\text{x}$, and units, which can be inserted with $\blank$ instead of $\unit$.
  \ir{party}{\pi ::= \tau:\tau^\ast ~|~ \bnfPrim{var}':\tau^\ast}
    Parties are the same as in raw \textAleph. Note, however, that the reference interpreter of \alethe\ has no separate representation for opaque variables: if the context is a variable, then it is assumed to be opaque rather than the variable pattern. When concurrency and contextual evaluation is supported in a future version, this deficiency will need to be addressed.
  \ir{parties}{\Pi ::= \pi ~|~ \Pi\,;\,\pi}
    Bags of parties are delimited by semicolons.
  \ir{relation}{\rho ::= \tau^\ast=\tau^\ast ~|~ \tau^\ast~{}^\backprime\tau^\ast{}^\prime~\tau^\ast}
    Relations are shorthand for where there is a single, variable context. These are familiar from the previous examples.
  \ir{def.\ head}{\delta_h ::= \rho ~|~ \{\,\Pi\,\}=\{\,\Pi\,\}}
    A rule is either a relation, or a mapping between party-bags which are enclosed in braces.
  \ir{def.\ rule}{\delta_r ::= \delta_h \,; ~|~ \delta_h\!: \Delta^+}
    If a rule has no sub-rules, then it is followed by a semicolon, otherwise it is followed by a colon and then its sub-rules/any locally scoped definitions. These declarations must either be in the same line, or should be indented more than the rule head similar to \haskell's off-side rule or \python's block syntax.
  \ir{def.\ halt}{\delta_t ::=~ \haltSym\,\tau^\ast; ~|~ \haltSym\,\rho\,;}
    As well as the canonical halting pattern form, there is also sugar for a relation wherein both sides of the relation each beget a halting pattern, as does the infix term (if any).
  \ir{def.}{\delta ::= \delta_r ~|~ \delta_t}
    This is the same as for \textAleph.
  \ir{cost ann.}{\xi ::= .^\ast}
    Sub-rules can be followed by more than one full-stop, with the number of full-stops used quantifying the cost of the sub-rule for use in the serialisation algorithm heuristics as explained towards the end of \Cref{sec:impl}.
  \ir{decl.}{\Delta ::= \delta ~|~ \pi\,.\xi ~|~ \rho\,.\xi ~|~ \haltSym\,\rho\,.\xi}
    In addition to the party sub-rules present in \textAleph, there are three sugared forms. If a declaration is a definition then, after introducing and resolving fresh anonymous names for any locally scoped atoms, the behaviour is the same as if the definition was written in the global scope. If a declaration takes the form of a relation, then it is the same as if we introduced a fresh context variable and bound each side of the relation to it. If the relation is preceded by $\haltSym$, then it both declares a sub-rule and defines a pair of halting patterns, e.g.
    \begin{align*}
      \haltSym~\atomLocal{Go}~{p}~\cursivexS~{[]}~{[]} =
          \atomLocal{Go}~{p}~\cursivexS~{[{n}~{\cdot}~\cursivenS]}~{\cursiveyS'}{.}
    \end{align*}
    becomes
    \begin{align*}
      & \haltSym~\atomLocal{Go}~{p}~{\cursivexS}~{[]}~{[]}{;} \\
      & \haltSym~\atomLocal{Go}~{p}~{\cursivexS}~{[{n}~{\cdot}~{\cursivenS}]}~{\cursiveyS'}{;} \\
      & \phantom\haltSym~\atomLocal{Go}~{p}~\cursivexS~{[]}~{[]} =
              \atomLocal{Go}~{p}~\cursivexS~{[{n}~{\cdot}~\cursivenS]}~{\cursiveyS'}{.}
    \end{align*}
  \ir{statement}{\Sigma ::= \delta ~|~ \dataDef{\tau^\ast}; ~|~ \kw{import}~ {}^{\backprime\backprime}\bnfPrim{module path}'';}
    A statement is a definition, a data definition, or an import statement. A data definition $\dataDef{\S~{n}}{;}$ is simply sugar for
    \begin{align*}
      & \haltSym~{\S}~{n}{;} \\
      & \infix{\atom{Dup}~({\S}~{n})}~{({\S}~{n'})}{:} \\
      &\qquad \infix{\atom{Dup}~{n}}~{n'}{.}
    \end{align*}
    i.e.\ it marks the pattern as halting and automatically writes a definition of \atom{Dup}. Future versions may automatically write other definitions, such as comparison functions, or include a derivation syntax akin to \haskell's. The $\kw{import}$ statement imports all the definitions of the referenced file, as if they were one file, and supports mutual dependency. Future versions may support partial, qualified and renaming import variants.
  \ir{program}{\mathcal P ::= \Sigma^\ast}
    A program is a series of statements. Additionally, comments can be added in \haskell-style:
    \begin{align*}
      & \comment{\cmtSLopen~this is a comment.} \\
      & \haltSym~{\S}~\comment{\cmtMLopen~so is this...~\cmtMLclose}~{n}{;}~\comment{\cmtSLopen~...and this.}\\
      & \comment{\cmtMLopen~and this is}\\
      & \comment{\cmtMLcont~a multiline comment~\cmtMLclose}
    \end{align*}
\end{itemize}}

{\def\midtilde{\raisebox{-0.23em}{\textasciitilde}}
\noindent We also note that the interpreter for \alethe\ treats the atoms \Z, \S, \Cons, \Nil, \atom{Garbage}, and character atoms specially when printing to the terminal. Specifically, they are automatically recognised and printed according to the sugared representations defined above. That is, except for the \atom{Garbage} atom which, when in the first position of a term, hides its contents---rendering as \code{\{\midtilde GARBAGE\midtilde\}}. This is useful when working with reversible simulations of irreversible functions that generate copious amounts of garbage data. If one assigns the garbage to a variable, then the special \code{:g} directive can be used to inspect its contents. Other directives include \code{:q} to quit, \code{:l file1 ...} to load the specified files as the current program, \code{:r} to reload the current program, \code{:v} to list the currently assigned variables after the course of the interpreter session, and \code{:p} to show all the loaded rules of the current program (and the derived serialisation strategy for each, see later). Computations can be performed in one of two ways; the first, \code{| (+ 3) 4 \raisebox{-0.3em}{-}}, takes as input a halting term, attempts to evaluate it to completion, and returns the output if successful, i.e.\ \code{() 7 (+ 3)}. The second, \code{> 4 `+ 3` y}, takes a relation as input and attempts to get from one side to the other. For example, here we evaluate \code{(+ 3) 4 ()}, obtaining \code{() 7 (+ 3)} which we then unify against \code{() y (+ 3)}, finally resulting in the variable assignment $y\mapsto7$. To run in the opposite direction, use \code{<} instead. Note that, whilst the interpreter does understand concurrent rules, it is not yet able to evaluate them.}

\para{Unification Efficiency}

Beyond the implementation concerns expressed in \Cref{sec:impl}, another issue of practical importance is the efficiency of identifying patterns matching a given term: to date, the standard library alone contains just shy of two thousand patterns, and so searching the known patterns linearly is impractical. Our interpreter constructs a trie-like structure for this purpose such that the time complexity for unification scales as $\bigOO{m+\log n}$ where $m$ is the number of patterns which match and $n$ the total number of patterns. Though asymptotically this is the best complexity possible, there are certainly improvements that could be made to the unification algorithm. In particular, the data structures used are primarily binary trees; hash tables with fixed lookup time for reasonable values of $n$ would yield significant improvements in distinguishing atoms, and random access arrays would yield significant improvements for distinguishing between composite terms of differing length. Notwithstanding these possible improvements, the reference interpreter has proven sufficiently fast for demonstrating non-trivial computations in \alethe. For example, computing $8! = \num{40320}$ takes only a few seconds despite the unary representation of natural numbers.

\begingroup
\def\isot{\mathrel{\ensuremath{:\kern0.11ex:}}}
\section{Towards a Type System}
\label{app:typing}

Type systems are very useful, in that they allow a compiler to statically analyse a program before it is run and identify whether it is ever possible for a value to be passed to a function that does not understand such an input. This is extremely powerful, and leads to the \haskell\ maxim `if it compiles, it runs': that is, although there may be logical errors, the program should not crash.

We may be tempted to try to start with a simple polymorphic type system, such as Hindley-Milner, but unfortunately Hindley-Milner is unsuitable for \textAleph\ being that it does not have objects that can be uniquely assigned a computational role. Consider attempting to type the symbol $\square$ from \Cref{lst:ex-square-def}, for example, which implements squaring/square rooting for natural numbers. One is tempted to say something like $\square\isot\mathbb{N}~\unit\leftrightarrow\unit~\mathbb{N}\isot\square$, but this does not work. Inspecting the definition more closely, we see that $\square$ appears in a number of contexts: ${\square}~{n}~{\unit}$, ${\square}~{n}~{m}~{\square}$, and ${\unit}~{m}~{\square}$. It is the second context where our attempt to type $\square$ breaks down: not only does its apparent arity change, but it appears twice in the same expression! Similar problems are encountered when trying to type other declarative languages, such as \prolog. To proceed, we shall require that any candidate type system treats terms \emph{holistically}, not ascribing any special importance to any particular part of a term. Such a holistic type system will most likely also be declarative in nature.

A fruitful way of thinking of types of terms is by considering their endpoints---their halting states. A well-formed term will be able to eventually transition to one or more halting terms, and the properties of these halting terms are what we are most interested in. For example, we would like to know that if we construct a term from $+$, two natural numbers, and $\unit$, then we will be able to extract two natural numbers---rather than, say, a boolean and a string---once computation finishes. But, the `end' state is free to wander back to the `beginning', and so it is more appropriate to consider the type of the term as an unordered pair, consisting of the two possible halting states. Checking this type will then consist of enumerating which rules may be encountered along a computational path, verifying that all these rules are consistent with the initial type and that the other state is reachable in principle. If the ambiguity checker is turned off, then we see that it is possible to have a term with more than two halting states and so its type should be a set of types corresponding to each possible halting state. We then notice that conventional data---such as a natural number---is of the same ilk, save that it only has one halting state. We call these type-sets \emph{isotypes}, and their inhabitants \emph{isovalues}, as each isovalue has one or more isomorphically equivalent representations.

How might these isotypes look? We begin with conventional data; let us call the isotype of a natural number \atom{Nat}, with inhabitants $\Z$ and $\S n$ where $n$ inhabits \atom{Nat}. Let us now attempt to type addition. As a first approximation, we write $\isot{+}~\atom{Nat}~\atom{Nat}~{\unit} = {\unit}~\atom{Nat}~\atom{Nat}~{+}{;}$ where $\isot$ at the beginning marks this as an isotype. This is in the right spirit, but doesn't readily admit a consistent type theory. In particular, how would we represent the isotype of a continuation in the term $\atom{Walk}~{\beta}~{(\atom{Charlie}'~x)}$, as introduced towards the end of \Cref{sec:ex2}? Being more careful, we can say that an addition isovalue is the unordered pair consisting of ${+}~{a}~{b}~{\unit}$ and ${\unit}~{c}~{d}~{+}$ where $a$, $b$, $c$ and $d$ are all isovalues of isotype \atom{Nat}. We need this explication, because the isotypes of the arguments could be far more complicated. Where it is unambiguous, however, it will be reasonable to write this isotype as $\isot{+}~\atom{Nat}~\atom{Nat}~{\unit} = {\unit}~\atom{Nat}~\atom{Nat}~{+}{;}$. Notice that, as this isotype is holistic, it is completely orthogonal to the isotype $\isot{+}~\atom{Rat}~\atom{Rat}~{\unit} = {\unit}~\atom{Rat}~\atom{Rat}~{+}{;}$ where \atom{Rat} is the isotype of rationals. This reveals that we get overloading for free in this type system. It is nevertheless desirable to introduce a notion similar to \haskell's typeclasses, however, or else (partial) isomorphisms employing addition, but polymorphic in the isotype they are adding, will have much too generic an isotype (in particular, the type inference algorithm would not be able to conclude that an isotype $\isot{+}~{a}~{b}~{\unit} = {\unit}~{c}~{d}~{+}{;}$ does not exist, and so we would have four isotype variables rather than the 1 expected). Typeclasses permit the programmer to specify that an isotype of the form $\isot{+}~{a}~{b}~{\unit}={\unit}~{c}~{d}~{+}{;}$ is only legal if it is of the more restrictive form $\isot{+}~{a}~{a}~{\unit}={\unit}~{a}~{a}~{+}{;}$ for some isotype $a$. Furthermore, the introduction of typeclasses permits the abbreviation of complex (ad-hoc) polymorphic isotypes.

What, then, might a polymorphic isotype look like? A good example is given by the map function, which would have as isotype the unordered pair of $(\atom{Map}~{f})~{\cursivexS}~{\unit}$ and ${\unit}~{\cursiveyS}~(\atom{Map}~{f})$ where $\cursivexS$ is an inhabitant of the isotype $[a]$, i.e.\ the isotype of lists of elements of isotype $a$, and where $\cursiveyS$ is an inhabitant of $[b]$. $f$ is required to be some term which, when used to construct the halting states ${f}~{x}~{\unit}$ and ${\unit}~{y}~{f}$, yields the isotype $\isot{f}~{a}~{\unit} = {\unit}~{b}~{f}$. More concisely, we write this as
\begin{align*}
  \isot{}& {[a]}~\infix{\atom{Map}~{f}}~{[b]}{:} \\
  \isot{}&\qquad {a}~\infix{f}~{b}{.} \\
  & []~\infix{\atom{Map}~{f}}~[]{;} \\
  & [{x}~{\cdot}~\cursivexS]~\infix{\atom{Map}~{f}}~[{y}~{\cdot}~\cursiveyS]{:} \\
  &\qquad {x}~\infix{f}~{y}{.} \\
  &\qquad \cursivexS~\infix{\atom{Map}~{f}}~\cursiveyS{.}
\end{align*}
and we note the similarity of the isotype to the implementation of \atom{Map} at the isovalue level. The composite term $\atom{Map}~{f}$ deserves further attention; this is itself halting, and so can be considered to be an implicitly declared anonymous isotype with one representation.

As briefly mentioned earlier, a type checker (and a type inference algorithm, if it exists) will perform a reachability analysis---starting from one halting pattern, it will enumerate all the accessible rules, finding the most general type consistent with these. The information gleaned this way is very valuable for all sorts of static analyses. We can determine what other halting states are accessible, as well as whether there are any missing cases along the way that could lead to a halting state. We can also resolve the issue of sub-rule ambiguity raised earlier, helping strengthen the phase space restrictions and avoid unintentional generation of entropy. It may also be possible to use this information to refine the cost heuristic used for serialisation, improving runtimes. Furthermore, knowledge of the execution path gives opportunities for identifying code motifs and performing structural transformations, a necessary ingredient for performing the kind of code optimisations a compiler is wont to do.

There is, however, a potential hiccup in the reachability analysis that occurs when a continuation-passing-style is used. In this case, there are rules which are shared between different computations. For example, in $\atom{Walk}~{\beta}~{c}\leftrightarrow\atom{Walk'}~{\alpha}~{c}$, $c$ is a continuation. There will be a number of computations which each pass control to such a \atom{Walk} term, and later accept control back from such a \atom{Walk'} term. Consequently, no definitive isotype can be assigned to these intermediate terms. Moreover, consider these two potential computations:
\begin{align*}
  & \atom{Foo} = \atom{Walk}~{\varphi}~\atom{Foo'}{;} \\
  & \atom{Bar} = \atom{Walk}~{\beta}~\atom{Bar'}{;}
\end{align*}
It is conceivable that the reachability analyser would, starting from \atom{Foo}, (correctly) determine that the pattern $\atom{Walk}~{\beta}~{c}$ is reachable, and then (incorrectly) presume the term \atom{Bar} is reachable from \atom{Foo}. If so, the type checker would then (rightly) complain of a type error, perhaps that the isotypes of $\atom{Foo'}$ and $\atom{Bar'}$ do not unify. The remedy to the first issue, of indefinite isotyping of intermediate terms, is to only assign definitive isotypes to sets of halting terms; intermediate terms are assigned definite isotypes only in the course of inferring isotypes for particular halting sets, and these intermediate isotypes are appropriately scoped such that the parallel process of isotype inference may proceed in spite of conflicting intermediate term isotypes. The second issue is remedied by specialising types as early as possible: in order to pass the ambiguity checker, any computation making use of a shared rule as above must pass control to the shared rule in a way orthogonal to any other computations passing control. In the above example, this holds because \atom{Foo'} and \atom{Bar'} are orthogonal patterns such that no term can be constructed satisfying them both simultaneously. Therefore, by immediately specialising the type of the continuation $c$ in \atom{Walk}, there is no risk of reachability `leaking' into inaccessible rules. One may be concerned that perhaps the reachability graph is much more complicated, engendering other reachability leaks, but each such confluence point between distinct computational paths \emph{must} have this same explicit orthogonality property in order to satisfy the ambiguity checker, and so there is in fact no risk. However, this remedy may be somewhat overzealous; suppose that the example above is part of a larger program,
\begin{align*}
  & \atom{Qux}~\atom{False} = \atom{Foo}{;} \\
  & \atom{Qux}~\atom{True} = \atom{Bar}{;}
\end{align*}
then the eager specialisation of the continuation isotype will yet result in a type-checking error, similar to \haskell's (optional) monomorphism restriction. To circumvent such monomorphism, it would be advantageous to allow the type checker to introduce additional scopes for intermediate types where such branches occur, providing that these alternate paths eventually converge to a single type. If this is not possible, then the solution may simply be to expect the programmer to annotate a `partial isotype' of any shared polymorphic rules in order to give the type checker sufficient information.

We are optimistic that, in programs without shared rules, it is possible to perform automatic type inference without any type annotations as in the Hindley-Milner type system. As for programs with shared rules, we are less certain; nonetheless, the power and generality of type inference algorithms in type systems as sophisticated as \haskell's most recent iterations suggests that it is attainable. We also note that the above is just a sketch of a possible type system, and will require further refinement and formalisation. Finally, we have neglected as yet to consider concurrency. Concurrency severely undermines the power of the type system, as there is no longer any continuity between terms: at any point, a term could be consumed, produced, or merged with another. It may be possible to extend the notion of isotype to include all the terms that might interact with one another, but it is unclear how useful this would be. Alternatively, one could simply ascribe types to single halting patterns, instead of seeking isotypes and an enumeration of all possible isomorphic representations of an isovalue. Again, this severely weakens the utility of the type system. A desirable compromise would be to identify the non-concurrent parts of a program, and apply type checking and inference to just these parts; a more rudimentary level of type checking could then be applied to the concurrent parts to at least verify the consistency of any non-concurrent sub-rules employed.

\endgroup

\section[The \texorpdfstring{$\Sigma$}{Sigma}-Calculus]{The \texorpdfstring{\textbsfSigma}{Sigma}-Calculus}
\label{app:sigma}

An earlier prototype, $\Sigma$ was more functional in nature. Its terms were nested applications of `general permutations', which could be named for convenience (and for achieving recursion without a fixed point combinator). Its definition was
\begin{equation*}\begin{aligned}
  \ruleName{term} && \tau &::= \langle\pi \tau^\ast : \tau^\ast \pi\rangle ~|~ \bnfPrim{var} ~|~ \bnfPrim{ref} ~|~ \tau* \\
\end{aligned}\end{equation*}
where \bnfPrim{ref} is a reference to a named pattern. A generalised permutation is given by $\langle\pi \tau^\ast : \tau^\ast \pi\rangle$, where $\pi$ is a special variable that matches the permutation itself. That is, the term $(\langle\pi~x~y : y~x~\pi\rangle~1~2)$ would transition to $(2~1~\langle\pi~x~y : y~x~\pi\rangle)$. Notice that this is the convention taken in \textAleph, that the atom or term in the first position tends to take the `active' role, and tends to place itself at the end of the term after rule execution. In $\Sigma$, however, this convention is enshrined in the language itself because all the information concerning the transition rules is held within the permutation terms themselves, and so it is required that there be two special locations within a term corresponding to the current and previous permutation. The first position is used for the current permutation in analogy to the $\lambda$-calculus, and the last position is used for the previous permutation for symmetry. Permutations are `general' in the sense that they can alter tree structure, and can copy/elide variables.

The $\Sigma$-calculus can also be proven to be microscopically reversible and Reverse-Turing complete, but the absence of sub-rules renders composition more unwieldy. Additionally, discriminating different cases is tricky. As with the $\lambda$-calculus, where data can be represented with functions via Church encoding, we can represent data in $\Sigma$ with permutations and these permutations will perform case statements. Taking lists as an example, we have
\begin{align*}
  \atom{cons} &\equiv \langle \pi \{nc\}zf : c\{nf\}z\pi \rangle \\
  \atom{nil} &\equiv \langle \pi \{nc\}zf : n\{fc\}z\pi \rangle
\end{align*}
where $\{xyz\}$ is sugar for $(\bot xyz\top)$ used to represent an inert expression ($\top\equiv\bot\equiv\varepsilon\equiv\unit$, not being a permutation, is an inert object that does nothing and is used to represent halting states in $\Sigma$). Using these definitions, list reversal can be implemented thus,
\begin{align*}
  \atom{rev} &\equiv \langle\pi l \top : \atom{nil} \{\lambda\nu\}\{\{\varepsilon\varepsilon\}l\} \pi\rangle \\
  \lambda &\equiv \langle\pi \{\atom{rev}\,\nu\}\{\{r'r''\}\{\tilde ll'l''\}\} \tilde r : \tilde l \{\atom{rev}'\,\nu\}\{\{\tilde rr'r''\}\{l'l''\}\} \pi\rangle \\
  \nu &\equiv \langle\pi \{\atom{rev}'\,\lambda\}\{r\{l'l''\}\}\atom{cons} : \atom{cons}\{\atom{rev}\,\lambda\}\{\{l'r\}l''\} \pi\rangle \\
  \atom{rev}' &\equiv \langle\pi \{\lambda\nu\}\{r\{\varepsilon\varepsilon\}\}\atom{nil} : \bot r\pi\rangle
\end{align*}
which is certainly not the most edifying program in the world! It would be used as so:
\begin{align*}
  (\atom{rev}[1\,4\,6\,2]\top) &\overset\ast\longleftrightarrow (\bot[2\,6\,4\,1]\atom{rev}')
\end{align*}

\para{\texorpdfstring{\textbsfmu}{Mu}-Recursive Functions}

To give a deeper flavour of $\Sigma$, we implement the $\mu$-recursive functions (recall their definition and \alethe\ implementation from \Cref{lst:aleph-murec}). In order to simulate a notion of function composition, we first adopt a `$\Sigma$-function' motif: a function is represented by the triple $F=\{f~z~f'\}$ where $f$ and $f'$ are the initial and terminal permutations, and $z$ is any data we wish to bind to $F$. The permutations $f$ and $f'$ must be such that $(f~z~f'~x~\top)\overset\ast\longleftrightarrow(\bot~y~g~f~z~f')$, where $x$ is the input, $y$ the output and $g$ any garbage data.

The base $\mu$-recursive functions can then be implemented thus:
\begin{align*}
  Z&\equiv \{z\varepsilon z\} & z&\equiv\langle z\varepsilon z~n~\top : \bot~0~n~z\varepsilon z \rangle \\
  S&\equiv \{s\varepsilon s\} & s&\equiv\langle s\varepsilon s~n~\top : \bot~\{\atom{succ} n\}~\varepsilon~s\varepsilon s\rangle \\
  \Pi_i^k&\equiv \{\pi_i^k\varepsilon\pi_i^k\} & \pi_i^k&\equiv\langle\pi_i^k\varepsilon\pi_i^k~\{n_1,\dots,n_k\}~\top : \\
  &&&\qquad\bot~n_i~\{n_1,\dots,n_{i-1},n_{i+1},\dots,n_k\}~\pi_i^k\varepsilon\pi_i^k\rangle
\end{align*}
Note that we have written $\langle z\varepsilon z\cdots$ instead of $\langle\pi\varepsilon z\cdots$ or $\langle\pi\varepsilon\pi\cdots$ for clarity.

The composition operator for $F$ of arity $k$ over functions $G_i$ is given by $\{c_k\{F\vec G\}c_k'\}$, where
\begin{align*}
  c_k &\equiv \langle c_k\{F\vec G\}c_k'~\vec n~\top : c_k''~F~(G_1\wedge\vec n)~\cdots~(G_k\wedge\vec n)~c_k\rangle \\
  c_k'' &\equiv \langle c_k''~F~(m_1~h_1\vee G_1)~\cdots~(m_k~h_k\vee G_k)~c_k : c_k'~\vec h~\vec G~(F\wedge\vec m)~c_k'' \rangle \\
  c_k' &\equiv \langle c_k'~\vec h~\vec G~(y~h_f\vee F)~c_k'' : \bot~y~\{h_f\,\vec h\}~c_k\{F\vec G\}c_k' \rangle
\end{align*}
where we have introduced sugared forms $(F\wedge x)\equiv(fzf'~x~\top)$ and $(y~h\vee F)\equiv(\bot~y~h~fzf')$.

The primitive recursion operator is a little more complicated. To understand the implementation below, it is important to know that the number 0 is represented as $\{\atom{zero}~\varepsilon\}$ and the number $m+1$ is given by $\{\atom{succ}~m\}$. When the form is unknown, we can write a unified representation as $\{\tilde mm\}$ where $\tilde m$ is the constructor and $m$ is either $\varepsilon$ or the predecessor. Furthermore, the constructors are defined as follows:
\begin{align*}
  \atom{zero} &\equiv \langle\pi \{pq\} z r : p\{rq\} z \pi\rangle \\
  \atom{succ} &\equiv \langle\pi \{pq\} z r : q\{pr\} z \pi\rangle
\end{align*}
We write the primitive recursion of $F$ and $G$ (arities $k$ and $k+2$) as $R_{FG}=\{\rho_k\{FG\}\rho_k'\}$, defined by:
\begin{align*}
  \rho_k &= \langle\rho_k\{FG\}\rho_k'~\{\vec n\{\tilde mm\}\}~\top : \tilde m~\{\zeta_k\,\sigma_k\}~\{FG\vec nm\}~\rho_k \rangle \\
  \rho_k' &= \langle \rho_k'~\{\zeta_k'\,\sigma_k''\}~\{yhFG\}~\tilde m : \bot~y~\{\tilde mh\}~\rho_k\{FG\}\rho_k' \rangle 
  \\[0.4\baselineskip]
  \zeta_k &= \langle\zeta_k~\{\rho_k\,\sigma_k\}~\{FG\vec n\varepsilon\}~\atom{zero} : \zeta_k'~(F\wedge\vec n)~G~\zeta_k\rangle \\
  \zeta_k' &= \langle\zeta_k'~(y~h\vee F)~G~\zeta_k : \atom{zero}~\{\rho_k'\,\sigma_k''\}~\{yhFG\}~\zeta_k'\rangle
  \\[0.4\baselineskip]
  \sigma_k &= \langle\sigma_k~\{\zeta_k\,\rho_k\}~\{FG\vec nm\}~\atom{succ} : \sigma_k'~(R_{FG}\wedge\vec nm)~\vec nm~\sigma_k \rangle \\
  \sigma_k' &= \langle\sigma_k'~(y~h\vee R_{FG})~\vec nm~\sigma_k : \sigma_k''~(G\wedge\vec nmy)~F~h~\sigma_k'\rangle \\
  \sigma_k'' &= \langle\sigma_k''~(z~h'\vee G)~F~h~\sigma_k' : \atom{succ}~\{\zeta_k'\,\rho_k'\}~\{z\{hh'\}FG\}~\sigma_k''\rangle
\end{align*}
In the above, the $\rho$ permutation switches between the $\zeta$ and $\sigma$ permutations depending on if $\{\tilde mm\}=0$ or $>0$; the $\zeta$ route computes $y=F(\vec n)$, whilst the $\sigma$ route recurses down to compute $y=R_{FG}(\vec n,m)$ and then $z=G(\vec n,m,y)$. Both these routes then put their results into a common representation and merge the control flow via $\tilde m$.

Finally, we implement the minimalisation operator as $M_F^m=\{\mu_k\{Fm\}\mu_k'\}$; $M_F^m$ is a generalisation of the more canonical minimalisation operator (which is equivalent to $M_F^0$),
\[ M_F^m(\vec n) = \min\{ \ell\ge m : f(\vec n; \ell)=0 \} \]
Interestingly, this is much simpler to implement than primitive recursion---only requiring four permutations---yet is the necessary ingredient for universality:
\begin{align*}
  \mu_k &= \langle\mu_k\{Fm\}\mu_k'~\vec n~\top : \mu_k''~m~(F\wedge \vec nm)~(M_F^{\{\atom{succ}\,m\}}\wedge\vec n)~\mu_k\rangle \\
  \mu_k'' &= \langle\mu_k''~m~(\{\tilde yy\}h\vee F)~q~\mu_k : \tilde y~\{\mu_k'\mu_k'''\}~m~(\{\tilde yy\}h\vee F)~q~\mu_k''\rangle \\
  \mu_k''' &= \langle\mu_k'''~\{\mu_k'\mu_k'''\}~m~p~(m'h\vee M_F^{\{\atom{succ}\,m\}})~\atom{succ} : \\
  &\qquad \atom{succ}~\{\mu_k''\mu_k'\}~m'~p~(m'h\vee M_F^{\{\atom{succ}\,m\}})~\mu_k'''\rangle \\
  \mu_k' &= \langle\mu_k'~\{\mu_k''\mu_k'''\}~m'~(F\wedge\vec nm)~(M_F^{\{\atom{succ}\,m\}}\wedge\vec n)~\tilde\mu : \bot~m'~\{\tilde\mu\vec n\}~\mu_k\{Fm\}\mu_k'\rangle
\end{align*}
In the above, we compute $F(\vec n;m)$ and check if it's 0, if so we return $m$ as $m'$, otherwise we compute $M_F^{m+1}(\vec n)$ in a recursive manner and return its result as $m'$. Clearly, if there is no $m$ such that $F(\vec n;m)=0$ then $M_F^m(\vec n)=\bot~\forall m$. For concision, we compute both functions simultaneously, making use of $\Sigma$'s parallelism to approximate laziness. We then conveniently reverse both computations in a Bennett-like manner to produce only a single bit of garbage (the history data we return is $\tilde\mu$, representing whether or not $m=m'$, and the input vector $\vec n$).

\para{The Interpreter} Whilst the implementation of the $\mu$-recursive functions demonstrates that realising meaningful programs in $\Sigma$ is certainly possible, it is inelegant and clunky. Rather than programming in $\Sigma$ feeling like the $\lambda$-calculus, as was intended, it feels much more like programming in assembly. Nevertheless, we release the source code of the interpreter\footnote{\url{https://github.com/hannah-earley/sigma-repl}}, example code\footnote{\url{https://github.com/hannah-earley/sigma-examples}}, and syntax highlighters\footnote{\url{https://github.com/hannah-earley/sigma-syntax}}\ for completeness.

\section{Conclusion}

The examples furnishing this chapter demonstrate the utility of the \textAleph-calculus, as well as its appropriateness as a candidate model targeting reversible molecular and Brownian computational architectures. Though a molecular implementation is as yet unrealised, an interpreter for the programming language \alethe\ serves as a useful testbed for \textAleph. Going forward, we hope to extend the interpreter to support the full concurrent language. Additionally, development of a type system appropriate for a language that is both reversible and declarative would be helpful for improving code analysis and reducing programmer errors, and work towards this is underway. Lastly, it is hoped that \textAleph\ can be realised experimentally in a molecular context.

\para{Towards a Molecular Implementation}

\Cref{fig:ex-plus} and \Cref{lst:ex-square-def-mol} are suggestive of a high level molecular implementation of the \textAleph-calculus.
Nevertheless, much more work is needed to develop a viable molecular implementation.
Translation from \textAleph\ terms and programs to the abstract molecular schemes is fairly routine: one should first eliminate sub-terms by introducing intermediate terms, an example of which can be seen in the orange plus `atoms' in \Cref{fig:ex-plus}.
One should also be careful to conserve mass, as discussed in \Cref{sec:ex2}.
The resulting reaction schemes will in general still be difficult to embed in a realistic chemical system.
The reason for this is that arbitrary structural rearrangements may occur in a single \textAleph\ reaction, but realistic chemical reactions tend to operate in small, local steps.
Therefore, additional compilation steps besides mass conservation and sub-term promotion will be needed.
A further concern is that the arbitrary nesting of terms in \textAleph\ means they may not be easily embeddable in flat three-dimensional space.
Methods to address this issue may consist of extensible `tethers', or terms could be split in two and linked `virtually' via a unique address (akin to pointers in conventional computational architectures).

Before this, however, it is necessary to identify the compilation target.
This will require construction of a lower level reversible molecular computational model that is Turing Complete, dynamic, and term-based.
As discussed in \Cref{chap:intro}, the current predominant methods for molecular computation are unsuitable.
To our knowledge, there are not any existing molecular architectures that meet the combined requirements of \textAleph: reversibility, ability to couple to a common free energy currency, arbitrary structural assemblies, and dynamic reassembly of said structures.
Whilst the absence of an existing architecture satisfying these properties might suggest that none exist, we believe that it is possible to construct one.
Indeed, biological enzymes routinely perform all of these tasks.
The challenge is in designing a simple programmable system to perform particular transitions.
Ideally such a system will be constructed from nucleic acids, as they are much easier to work with and design than polypetides are.

\endgroup

\begingroup
\chapter{Conclusion}
\label{chap:conc}

In this thesis we explore the intersection between reversible computation and molecular computation. The benefits of reversible computation are primarily from the perspective of energy efficiency, and it has been known for some time (see, e.g., \textcite{frank-thesis}) that the gains in energy efficiency are asymptotic. That is, reversible computers can exhibit a rate of computation that scales much better with system size than an equivalent irreversible computer. In \Cref{chap:revi} we refined these results, proving that it is not possible to do better than adiabatic scaling in any system, classical or quantum. This adiabatic scaling corresponds to a computation rate $\sim\sqrt{PM}$ where $P$ is the rate of supply of free energy and $M$ the computational mass-energy. In contrast, the scaling for irreversible computers is $\sim P$. Asymptotically, $P$ can scale no better than the convex bounding surface area of the system, $A$, and $M$ no better than the volume of the system, $V$, leading to respective scaling laws $\sqrt{AV}$ and $A$. Strictly speaking these apply to the mesoscopic regime; for the microscopic and cosmic scales we further refined these results.

This performance metric of net computation rate ignores the importance of interaction between computer subunits. In \Cref{chap:revii,chap:reviii} we therefore extended this analysis to consider, respectively, computer-computer communication/synchronisation interactions and interactions between computers and shared resources.
These turned out to be expensive in the case of Brownian architectures, and we conjectured about the applicability of these results to any architecture in which individual computational units evolve independently.
If the same aggregate rate of interactions as the rate of computation were to be maintained, then we showed that the net computation rate would have to fall to that of irreversible computation, $A$.
From the point of view of molecular or Brownian computers, which we showed in \Cref{sec:crn} were easily able to saturate the adiabatic scaling bounds, this was because the timescale to perform an interaction was on the order of $\sim1/b^2$ where $b$ quantifies the \emph{computational bias}, a measure of the free energy density of the system.
As $b\ll1$ is vanishing for large systems, this is untenable: ideally, all computational transitions in a Brownian computer should operate at a characteristic timescale $\lesssim1/b$.
Whilst concerning from an engineering perspective, we presented various mitigation steps in their respective conclusions.
Primarily these centred on reducing the instantaneous proportion of computational subunits permitted to participate in such interactions.
Moreover, for current computational scales up to the order of metres it is likely that reasonable values of $b\gtrsim0.1$ could be achieved and thus no mitigation may be necessary at present.

Finally we proposed a model of computation that might be appropriate for programming such a reversible Brownian computer. The motivation for this was the limited attention which this intersectional computation paradigm has received, and thus this model fills a vital gap. It is also novel in being the first (to our knowledge) declarative model of reversible programming. Whilst the \textAleph-calculus is fully functional as a programming language, more work needs to be done to develop it into a viable specification for a reversible molecular computer. The examples in \Cref{sec:ex1,app:seq-klona} serve as an initial blueprint for this work, and we hope to further develop this over the coming years.

\endgroup

\backmatter
\printbibliography[title=References]

\end{document}